\documentclass[a4paper,12pt]{amsart}
\usepackage[hmarginratio={1:1},vmarginratio={1:1},heightrounded,textwidth=455pt]{geometry}
\usepackage{amsfonts}
\usepackage{amsmath}
\usepackage[utf8]{inputenc}
\usepackage{amssymb, colonequals}
\usepackage{tikz}
\usetikzlibrary{arrows}
\usetikzlibrary{patterns}
\usetikzlibrary{arrows}
\usetikzlibrary{patterns}
\usetikzlibrary{calc}
\usepackage[textsize=small]{todonotes}
\usepackage{url}

\usepackage{breakurl}
\usepackage[breaklinks]{hyperref}
\usepackage{enumitem}
\usepackage{multirow}
\usepackage{chngcntr}
\usepackage{algorithmicx}
\usepackage{algpseudocode}
\usepackage{algorithm}
\usepackage{verbatim}
\usepackage{caption}
\usepackage{subcaption}
\usepackage{hyperref}
\hypersetup{
colorlinks=true,
linkcolor=black,
citecolor=black,
filecolor=black,
urlcolor=black,
pdfproducer={LaTeX},
pdfcreator={pdfLaTeX}
}
\counterwithin{figure}{section}
\setlist[itemize]{leftmargin=0.2in}
\setlist[enumerate]{leftmargin=0.3in}

\newcommand{\leftside}{\mathsf{left}}
\newcommand{\rightside}{\mathsf{right}}

\newcommand{\strongModules}{\mathcal{M}}

\newcommand{\camodules}{\mathcal{S}}
\newcommand{\slots}{\mathcal{S}^*}
\newcommand{\metachords}{\mathcal{MC}}

\newcommand{\owners}{\mathcal{P}}

\newcommand{\priv}{\mathsf{Priv}}
\newcommand{\Priv}{\mathsf{Priv}}

\newcommand{\cliques}{\mathsf{Cl}}
\newcommand{\signatures}{\mathsf{sgn}}

\renewcommand{\circle}{\mathsf{O}}

\newcommand{\pqmtree}{\mathcal{T}}
\newcommand{\HCP}{\mathbf{HCP}}
\newcommand{\blocks}{\mathsf{blocks}}
\newcommand{\sides}{\mathsf{sides}}

\newcommand{\TTT}{\mathbb{T}}
\newcommand{\NNN}{\mathbb{N}}
\newcommand{\MMM}{\mathbb{M}}
\newcommand{\SSS}{\mathbb{S}}
\newcommand{\KKK}{\mathbb{K}}
\newcommand{\LLL}{\mathbb{L}}

\newcommand{\DDD}{\mathbb{D}}
\let\leq\leqslant
\let\geq\geqslant
\let\setminus\smallsetminus

\let\rho\varrho

\newcommand{\brac}[1]{{\left(#1\right)}}

\newcommand{\Oh}[1]{O\brac{#1}}

\newcommand{\classfont}{\mathsf}
\newcommand{\NP}{\ensuremath{\classfont{NP}}}
\newcommand{\hcp}{{\sc{Helly Cliques}}}
\newcommand{\total}{{\sc{Total Ordering}}}
\newcommand{\yes}{{\bf{yes}}}
\newcommand{\no}{{\bf{no}}}

\makeatletter
\newcommand\ie{i.e\@ifnextchar.{}{.\@}}
\newcommand\etc{etc\@ifnextchar.{}{.\@}}
\newcommand\etal{et~al\@ifnextchar.{}{.\@}}
\def\namedlabel#1#2{\begingroup
    #2%
    \def\@currentlabel{#2}%
    \phantomsection\label{#1}\endgroup
}
\makeatother

\newcounter{dummy} 
\numberwithin{dummy}{section}
\newtheorem{theorem}[dummy]{Theorem}
\newtheorem{claim}[dummy]{Claim}
\newtheorem{lemma}[dummy]{Lemma}

\newtheorem{observation}[dummy]{Observation}
\newtheorem{definition}[dummy]{Definition}

\newcounter{hackcount}

\title[Circular-arc graphs and the Helly property]{Circular-arc graphs and the Helly property}

\author[Jan Derbisz and Tomasz Krawczyk]{Jan Derbisz$^1$$^\dagger$ \and Tomasz Krawczyk$^2$}
\address{$^1$Theoretical Computer Science Department, 
Faculty of Mathematics and Computer Science, Jagiellonian University in Krak\'ow, Poland}
\address{Doctoral School of Exact and Natural Sciences, Jagiellonian University, Poland}
\email{jan.derbisz@doctoral.uj.edu.pl}
\address{$^2$Faculty of Mathematics and Information Science, Warsaw University of Technology, Poland}
\email{tomasz.krawczyk@pw.edu.pl}

\thanks{$^\dagger$Research of this author was partially funded by Polish National Science Center (NCN) grant 2021/41/N/ST6/03671 and by the Priority Research Area SciMat under the program Excellence Initiative - Research University at the Jagiellonian University in Krak\'{o}w}
\begin{document}
\thispagestyle{empty}

\begin{abstract}
\begin{scriptsize}
One of the most common and best studied models for representing graphs is the intersection model. 
For a family $\mathcal{R}$ of geometric objects, the \emph{$\mathcal{R}$-intersection model} of a graph $G$ in the family~$\mathcal{R}$ is a
mapping $\phi : V(G) \to \mathcal{R}$ which assigns objects in~$\mathcal{R}$ 
to the vertices of $G$ so that $\phi(u) \cap \phi(v) \neq \emptyset \iff uv \in E(G)$ for any two vertices $u$ and $v$ of $G$. 
Given a graph~$G$, an $\mathcal{R}$-intersection model~$\phi$ of~$G$, and a clique~$C$ of~$G$, 
we say \emph{$C$ is Helly} in $\phi$ if the set $\bigcap_{c \in C} \phi(c)$ is non-empty and we say the \emph{model $\phi$ satisfies the Helly property} if every clique of $G$ is Helly in $\phi$.

In this paper we investigate some problems related to the Helly properties of \emph{circular-arc graphs}, which are defined as intersection graphs of arcs of a fixed circle.
As such, circular-arc graphs are among the simplest classes of intersection graphs whose models might not satisfy the Helly property. 
In particular, some cliques of a circular-arc graph might be Helly in some but not all arc intersection models of the graph.
For these reasons, designing algorithms in the class of circular-arc graphs 
seems to be more challenging than in classes of graphs whose models naturally satisfy the Helly property, such as in the class of interval or chordal graphs.

Our first result is an alternative proof of a theorem by Lin and Szwarcfiter (\emph{LNCS}, vol.~4112, 73--82) which asserts that for every circular-arc graph $G$ either every normalized model of $G$ satisfies the Helly property or no normalized model of $G$ satisfies this property. 
Normalized models of a circular-arc graph~$G$ reflect the neighbourhood relation
between the vertices of $G$; in particular, every model of~$G$ can be made normalized by extending some of its arcs.

Further, we study the Helly properties of a single clique of a circular-arc graph~$G$.  
We divide the cliques of $G$ into three types:
a clique $C$ of $G$ is \emph{always-Helly}/\emph{always-non-Helly}/\emph{ambiguous} if $C$ is Helly in every/no/(some but not all) normalized model of~$G$.
We provide a combinatorial description for the cliques of each type, and based on it, we devise a polynomial time algorithm which determines the type of a given clique.

Finally, we study the \hcp\ problem, in which we are given an $n$-vertex circular-arc graph~$G$ and some of its cliques $C_1, \ldots, C_k$ and we ask if there is an arc intersection model of $G$ in which all the cliques $C_1, \ldots, C_k$ satisfy the Helly property.
\hcp\ was introduced by Agaoglu {\c{C}}agirici, {\c{C}}agirici, Derbisz, Hartmann, Hlinen{\'{y}}, Kratochv{\'{\i}}l, Krawczyk, and Zeman in relation with their work on $H$-graphs recognition problems (\emph{LIPIcs}, vol.~272, 8:1--8:14) and was shown by Agaoglu {\c{C}}agirici and Zeman to be \NP-complete (\emph{CoRR}, abs/2206.13372).
We show that:
\begin{itemize}
    \item \hcp\ admits a $2^{O(k\log{k})}n^{O(1)}$-time algorithm (that is, \hcp\ is FPT when parametrized by the number of cliques given in the input),
    \item assuming Exponential Time Hypothesis (ETH), \hcp\ cannot be solved in time $2^{o(k)}n^{O(1)}$,
    \item \hcp\ admits a polynomial kernel of size $O(k^6)$.
\end{itemize}

All our results use a data structure, called a \emph{PQM-tree}, which maintains all normalized models of a circular-arc graph~$G$.
\end{scriptsize}
\end{abstract}

\maketitle

\section{Introduction}
\label{sec:introduction}

One of the most common and best studied models for representing graphs is the intersection model. 
For a family $\mathcal{R}$ of geometric objects, the intersection model (or representation) of a graph $G$ in the family~$\mathcal{R}$ is a
mapping $\phi : V(G) \to \mathcal{R}$ which assigns objects in~$\mathcal{R}$ 
to the vertices of $G$ so that $\phi(u) \cap \phi(v) \neq \emptyset \iff uv \in E(G)$ for any two vertices $u$ and $v$ of $G$. 
The graphs that have such a representation in $\mathcal{R}$ form the class of $\mathcal{R}$-intersection graphs.
By considering families $\mathcal{R}$ of objects of some particular kind, for example, having a specific geometric shape, we obtain various classes of graphs. For instance:
\begin{itemize}
 \item \emph{interval graphs} are the intersection graphs of intervals on a line,
 \item \emph{circular-arc graphs} are the intersection graphs of arcs of a fixed circle,
 \item \emph{chordal graphs} are the intersection graphs of subtrees of a tree
(or, equivalently, graphs with no induced cycles
of length greater than 3).
\end{itemize}
Clearly, the class of interval graphs is contained in the class of chordal graphs and the class of circular-arc graphs.
There are known linear time recognition algorithms for interval graphs~\cite{BoothLueker76} and chordal graphs~\cite{RoseTL76}.
The first polynomial time algorithm recognizing circular-arc graphs was given by Tucker~\cite{Tuck80}.
Currently, two linear-time algorithms recognizing circular-arc graph are known~\cite{McConnel03,KapNus11}.

The subject of this paper are circular-arc graphs.
Although circular-arc graphs seem to be closely related to interval graphs, 
they have significantly different algorithmic and combinatorial properties.
In particular, quite number of problems that are solved or showed to admit polynomial time solutions in the class of interval graphs, 
in the class of circular-arc graphs are still open or are computationally hard.
A good example illustrating our remark is the minimum coloring problem, which admits a simple linear algorithm in the class of interval graphs, but
in the class of circular-arc graphs is \NP-complete~\cite{GareyJMP80}.
Another example is related to the problem of characterizing circular-arc graphs in terms of forbidden structures,
in particular, by the list of all minimal induced graphs forbidden for this class of graphs.
For interval graphs such a list was completed by Lekkerkerker and Boland~\cite{LekBol62} in the 1960's;
for circular-arc graphs, despite a flurry of research~\cite{Bon09,Fed99,Fra15,Klee69,Tro76,cao2024characterization}, it is still unknown.

One of the properties that seems to distinguish interval and chordal graphs from circular-arc graphs is the Helly property.
In the context of intersection graphs, we use the following notation.
Given a graph~$G$, an intersection model~$\phi$ of~$G$, and a clique~$C$ of~$G$, 
we say the \emph{clique $C$ satisfies the Helly property} (or \emph{is Helly}) in $\phi$ if the set $\bigcap_{c \in C} \phi(c)$ is non-empty and we say the \emph{model $\phi$ satisfies the Helly property} if every clique of $G$ is Helly in $\phi$.

Due to the properties of geometric objects defining the class of interval and chordal graphs, every intersection model of a graph from these classes satisfies the Helly property.
Therefore, one can easily deduce that every interval/chordal graph contains linearly many, in the size of its vertex set, maximal cliques.
The situation is different for circular-arc graphs;
for example, a circular-arc graph $G$ which arises from a complete graph by removing the edges of a perfect matching,
contains exponentially many maximal cliques, and in any model of $G$ only linearly many of them may satisfy the Helly property.
It is worth noting here that possessing linearly many maximal cliques is a property often used in algorithm design. 
For example, in the class of interval graphs, a vast number of algorithms use a data structure called \emph{PQ-trees}. 
The PQ-tree of an interval graph~$G$ represents all possible orderings of all maximal cliques in the models of~$G$.
Thus, the PQ-tree of $G$ also represents all possible interval models of $G$.
For chordal graphs the structure describing all their intersection models is not known, 
however, the property of possessing linearly many cliques is used e.g. in polynomial-time algorithms for the recognition and the isomorphism testing of $T$-graphs, for every fixed tree~$T$~\cite{ChaplickTVZ21, CagiriciH22}.
Similarly, in the class of circular-arc graphs the structure of all their intersection models is not known.
Nevertheless, quite recently the structure of so-called \emph{normalized models} of circular-arc graphs was described~\cite{Krawczyk20}.
Normalized models of a circular-arc graph reflect the neighbourhood relation in this graph and can be seen as its canonical representations -- see Section~\ref{sec:preliminaries} for their precise definition.
In particular, every intersection model $\psi$ of a circular-arc graph~$G$ can be easily transformed into a normalized model $\phi$ of~$G$
by possibly extending some of the arcs of $\psi$.
In particular, if a clique $C$ of $G$ is Helly in $\psi$, it is also Helly in $\phi$.
Krawczyk~\cite{Krawczyk20}, based on the approach proposed by Hsu \cite{Hsu95}, devised a linear-time algorithm computing a data structure, called a \emph{PQM-tree}, which represents all normalized models of a circular-arc graph.
Since the PQM-trees of circular-arc graphs are heavily exploited in our paper, we describe them precisely in Section~\ref{sec:PQM-trees}.

An important subclass of circular-arc graphs, which inherits many properties of interval graphs, is formed by \emph{Helly circular-arc graphs}.
A circular-arc graph $G$ is \emph{Helly} if it admits a circular-arc model which satisfies the Helly property.
In particular, like interval graphs, every Helly circular-arc graph has linearly many maximal cliques; 
on the other hand, the minimum coloring problem in this class of graphs remains \NP-complete~\cite{Gavril1996}.
Heading towards a linear-time algorithm for the recognition of 
the Helly circular-arc graphs, Lin and Szwarcfiter proved the following theorem.
\begin{theorem}[\cite{LinSchw06}]
\label{thm:intro:Lin-Szwarcfiter}
Let $G$ be a circular-arc graph.
Either every normalized model of $G$ satisfies the Helly property or no normalized model of $G$ satisfies the Helly property.
\end{theorem}
They prove the theorem by listing so-called obstacles, which have to be avoided by Helly circular-arc graphs.

\subsection{Our results}
The goal of this paper is to study the Helly properties in the class of
circluar-arc graphs.

First, in Section~\ref{sec:Lin-Szwarcfiter}, we provide an alternative proof
of Theorem \ref{thm:intro:Lin-Szwarcfiter}, based on the properties of PQM-trees for circular-arc graphs.

In Section~\ref{sec:clique-type-problem} we study the Helly properties of a single clique of a circular-arc graph~$G$.  
We divide the cliques of $G$ into three types:
a clique~$C$ of~$G$ is \emph{always-Helly}/\emph{always-non-Helly}/\emph{ambiguous} if $C$ is Helly in every/no/(some but not all) normalized model of $G$.
We provide a combinatorial description for the cliques of each type, 
and based on it, we prove the following:
\begin{theorem}
\label{thm:intro:clique-type}
There is a polynomial-time algorithm which, given a circular-arc graph~$G$ and its clique $C$, determines the type of $C$.
\end{theorem}

Next, we turn our attention to the \hcp\ problem.
In the \hcp\ problem we are given a circular-arc graph~$G$ and some of its cliques $C_1, \ldots, C_k$ and we ask if there is a model of~$G$ in which all the cliques $C_1, \ldots, C_k$ are Helly
(if some model of $G$ satisfies these properties, then the model obtained by its normalization satisfies these properties as well).
\hcp\ was introduced by Agaoglu {\c{C}}agirici, {\c{C}}agirici, Derbisz, Hartmann, Hlinen{\'{y}}, Kratochv{\'{\i}}l, Krawczyk, and Zeman  \cite{CagiriciCDHHKK023} as an intermediate problem in their study of $H$-graphs recognition problems (see next subsection for details) and was shown by Agaoglu {\c{C}}agirici and Zeman to be \NP-complete~\cite{AgaogluZeman22}.
Regarding \hcp, we show the following theorems:
\begin{theorem} \ 
\label{thm:intro:helly-cliques-par}
\begin{enumerate}
\item \label{item:intro:helly-cliques-fpt} \hcp\ can be solved in time $2^{\Oh{k\log{k}}}n^{\Oh{1}}$. 
\item \label{item:intro:helly-cliques-eth} Assuming ETH, \hcp\ cannot be solved in time $2^{o(k)}n^{\Oh{1}}$.
\end{enumerate}
\end{theorem}
\begin{theorem}
\label{thm:intro:helly-cliques-poly-kernel}
\hcp\ admits a kernel of size $O(k^6)$.
\end{theorem}
The proof of Theorem~\ref{thm:intro:helly-cliques-par}.\eqref{item:intro:helly-cliques-fpt} uses the PQM-tree of $G$ to find a normalized model in which $C_1,\ldots,C_k$ are Helly.
Besides, it also uses some purely geometric observations, such as so-called \emph{Trapezoid Lemma}, which asserts that for a family $T_1,\ldots,T_n$ of pairwise intersecting trapezoids spanned between two parallel lines $A$ and $B$ one can draw inside each trapezoid $T_i$ a segment $s_i$ spanned between $A$ and $B$ so as $s_1,\ldots,s_n$ are pairwise intersecting.
We also give an alternative proof of the fact that \hcp\ is \NP-complete, which yields Theorem~\ref{thm:intro:helly-cliques-par}.\eqref{item:intro:helly-cliques-eth} as a by-product.

Our kernelization procedure from Theorem \ref{thm:intro:helly-cliques-poly-kernel} is done in two steps.
In the first step we mark a set~$R$ of \emph{important vertices} in $G$, where the size of $R$ is $O(k^6)$. 
The important vertices in $R$ encode all relevant information needed to conclude that we are dealing with \no-instance; in particular, the instances $G,C_1,\ldots,C_k$ and $G,C'_1,\ldots,C'_k$ of \hcp\ are equivalent, where $C'_i = C_i \cap R$ for $i \in [k]$.
In the second step we construct so-called \emph{reduct $G'$} of~$G$ with respect to the set~$R$. 
The vertex set of the reduct~$G'$ contains the set~$R$ as its subset and has size linear in $|R|$. 
The main property of the reduct~$G'$ of $G$ with respect to $R$ is that the 
configurations of the arcs representing~$R$ occurring in the models of $G$ and 
in the models of $G'$ coincide.
In particular, we may return $G',C'_1,\ldots,C'_k$ as the kernel of $G,C_1,\ldots,C_k$.

\subsection{Applications}
The results presented in the previous subsection have been used as tools in the resolution of some problems related to the recognition of so-called $H$-graphs~\cite{CagiriciCDHHKK023}.
For a fixed connected graph~$H$, the class of \emph{$H$-graphs} contains the intersection graphs of connected subgraphs of some subdivisions of~$H$.
Many known geometric intersection graph classes are $H$-graphs for an appropriately chosen graph $H$: $K_2$-graphs coincide with interval graphs, $K_3$-graphs coincide with circular-arc graphs, the class $\bigcup \text{$T$-graph}$, where the union ranges over all trees $T$, coincides with the class of chordal graphs.
Since $H$-graphs generalize many known geometric intersection graph classes,
they form a~good background that allows to study basic computational problems in some systematic way.
Recently, quite a lot of research has been devoted to determine the tractability border for various computational problems, such as recognition or isomorphism testing, in classes of $H$-graphs with respect to the combinatorial properties of graphs $H$.
The research of Agaoglu {\c{C}}agirici at al.~\cite{CagiriciCDHHKK023} shows that \hcp\ is strongly related with the recognition problems of~$H$-graphs for the cases of this problem which remain open.
Let us briefly describe what we know about $H$-graphs recognition.
Chaplick, T{\"{o}}pfer, Voborn{\'{\i}}k, and Zeman~\cite{ChaplickTVZ21} showed that, for every fixed tree~$T$, the recognition of $T$-graphs can be solved in polynomial time. 
They also showed that $H$-graph recognition is \NP-hard when 
$H$~contains two cycles sharing an edge~\cite{ChaplickTVZ21}.
We already mentioned that $K_3$-graphs (circular-arc graphs) can be recognized in linear-time \cite{McConnel03, KapNus11}.
Agaoglu {\c{C}}agirici at al.~\cite{CagiriciCDHHKK023} strengthened the result of 
Chaplick at al.~\cite{ChaplickTVZ21} and showed that the recognition of $H$-graphs is \NP-hard when $H$ contains two distinct cycles.
They also showed that $H$-graph recognition is polynomial-time solvable when $H$ is a graph consisting of a cycle and an edge attached to it, called a \emph{lollipop}.
In particular, their algorithm recognizing lollipop-graphs uses as a subroutine a polynomial-time algorithm testing whether there is a model of a circular-arc graph in which some particular clique is Helly.
Such an algorithm is asserted by Theorem~\ref{thm:intro:clique-type}.
The above result leave open $H$-graph recognition problems for the cases where 
$H$ is a unicyclic graph different from a cycle and a lollipop.
Following the arguments from~\cite{CagiriciCDHHKK023}, one can observe that the recognition of $H$-graphs, where $H$ is a cycle with $k$ disjoint edges adjacent to it, is as difficult as \hcp\ with the number of cliques equal to~$k$.
Finally, \cite{CagiriciCDHHKK023} shows that the recognition of $\bigcup \text{$H$-graph}$, where $H$ ranges over all unicyclic graphs, is polynomial time equivalent to \hcp.

\subsection{The structure of the paper}
Our paper is organized as follows:
\begin{itemize}
 \item in Section~\ref{sec:preliminaries} we define notation and basic concepts related to circular-arc graphs,
 \item in Section \ref{sec:PQM-trees} we describe PQM-trees representing normalized models of circular-arc graphs,
 \item in Section \ref{sec:Lin-Szwarcfiter} we provide an alternative proof of Theorem~\ref{thm:intro:Lin-Szwarcfiter} by Lin and Szwarcfiter,
 \item in Section~\ref{sec:clique-type-problem} we provide a polynomial-time algorithm testing the type of a clique in a circular-arc graph,
 \item in Section~\ref{sec:fpt-Helly-Cliques-problem} we prove Theorem \ref{thm:intro:helly-cliques-par}.\eqref{item:intro:helly-cliques-fpt}.
 \item in Section~\ref{sec:helly-cliques-kernel} we prove Theorem \ref{thm:intro:helly-cliques-poly-kernel}.
 \item in Section~\ref{sec:npc} an alternative proof of \NP-completeness of \hcp\ is contained, which shows Theorem \ref{thm:intro:helly-cliques-par}.\ref{item:intro:helly-cliques-eth} as a corollary.
 \item Section~\ref{sec:appendix} (Appendix) contains the construction of the reduct of a circluar arc graph and the proof of the Trapezoid Lemma. 
 \end{itemize}

\section{Preliminaries}
\label{sec:preliminaries}
\subsection{Graphs}
All graphs considered in this paper are \emph{simple}, that is, they have no multiedges and no loops.
Let $G = (V,E)$ be a graph with the vertex set~$V$ and the edge set~$E$.
The \emph{neighbourhood} of a vertex $u \in V$ in $G$ is the set $N(u) = \{v \in V(G) \mid uv \in E(G)\}$,
the \emph{neighbourhood} of a set $U \subseteq V$ in $G$ is the set $N(U)=\bigcup_{u\in U} N(u)\setminus U$.
We also define $N[u]=N(u)\cup \{u\}$ for a vertex $u\in V$ and we write $N[U]=\bigcup_{u\in U} N[u]$ for a set $U\subseteq V$. 
A vertex $u$ of $G$ is \emph{universal} in $G$ if $N[u] = V(G)$.
Two distinct vertices $u,v$ of $G$ are \emph{twins} in $G$ if $N[u] = N[v]$.
Note that twins are adjacent in $G$.

\subsection{Circular-arc graphs}
Let $G$ be a circular-arc graph with no universal vertices and no twins and let $\psi$ be a circular-arc model of $G$ on the circle $\circle$.
Let $(v,u)$ be a pair of distinct vertices in $G$.
We say that:
\begin{itemize}
 \item $\psi(v)$ and $\psi(u)$ are \emph{disjoint} if $\psi(v) \cap \psi(u) = \emptyset$,
 \item $\psi(v)$ \emph{contains} $\psi(u)$ if $\psi(v) \supsetneq \psi(u)$,
 \item $\psi(v)$ \emph{is contained} in $\psi(u)$ if $\psi(v) \subsetneq \psi(u)$,
 \item $\psi(v)$ and $\psi(u)$ \emph{cover the circle} if $\psi(v) \cup \psi(u) = \circle$,
 \item $\psi(v)$ and $\psi(u)$ \emph{overlap}, otherwise.
\end{itemize}
See Figure~\ref{fig:mutual_arc_position} for an illustration.

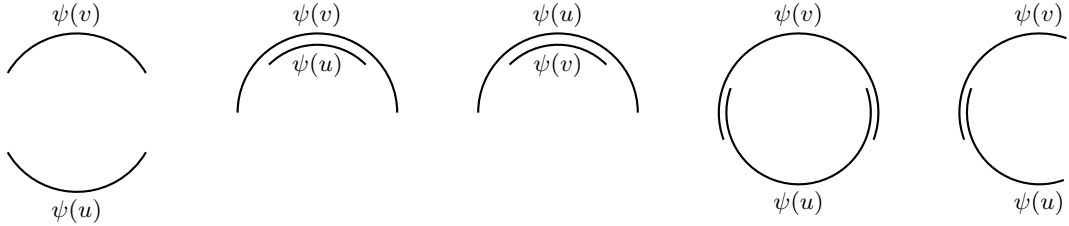
\begin{figure}[htp!]
\begin{tikzpicture}[scale=0.5]
\coordinate (center) at (0,0) {};
\coordinate (v) at ($(center)+(90:2cm)$) {};
\coordinate (u) at ($(center)+(270:2cm)$) {};

\coordinate (lv) at ($(center)+(90:2.6cm)$) {};
\coordinate (lu) at ($(center)+(270:2.6cm)$) {};

\tikzstyle{every node}=[inner sep=1pt]
\begin{scriptsize}
\node at (lv) {$\psi(v)$};
\node at (lu) {$\psi(u)$};
\end{scriptsize}

\draw[thick] ([shift=(30:2.1cm)]0,0) arc (30:150:2.1cm);
\draw[thick] ([shift=(210:2.1cm)]0,0) arc (210:330:2.1cm);

\draw[thick, white] (-3.0,-3)--(-3.0,-2.8);
\draw[thick, white] (3.0,3)--(3.0,2.8);

\end{tikzpicture}
\begin{tikzpicture}[scale=0.5]
\coordinate (center) at (0,0) {};
\coordinate (v) at ($(center)+(90:2cm)$) {};
\coordinate (u) at ($(center)+(270:2cm)$) {};

\coordinate (lv) at ($(center)+(90:2.6cm)$) {};
\coordinate (lu) at ($(center)+(90:1.25cm)$) {};

\tikzstyle{every node}=[inner sep=1pt]
\begin{scriptsize}
\node at (lv) {$\psi(v)$};
\node at (lu) {$\psi(u)$};
\end{scriptsize}

\draw[thick] ([shift=(0:2.1cm)]0,0) arc (0:180:2.1cm);
\draw[thick] ([shift=(45:1.8cm)]0,0) arc (45:135:1.8cm);

\draw[thick, white] (-3.0,-3)--(-3.0,-2.8);
\draw[thick, white] (3.0,3)--(3.0,2.8);

\end{tikzpicture}
\begin{tikzpicture}[scale=0.5]
\coordinate (center) at (0,0) {};
\coordinate (v) at ($(center)+(90:2cm)$) {};
\coordinate (u) at ($(center)+(270:2cm)$) {};

\coordinate (lu) at ($(center)+(90:2.6cm)$) {};
\coordinate (lv) at ($(center)+(90:1.25cm)$) {};

\tikzstyle{every node}=[inner sep=1pt]
\begin{scriptsize}
\node at (lv) {$\psi(v)$};
\node at (lu) {$\psi(u)$};
\end{scriptsize}

\draw[thick] ([shift=(0:2.1cm)]0,0) arc (0:180:2.1cm);
\draw[thick] ([shift=(45:1.8cm)]0,0) arc (45:135:1.8cm);

\draw[thick, white] (-3.0,-3)--(-3.0,-2.8);
\draw[thick, white] (3.0,3)--(3.0,2.8);

\end{tikzpicture}
\begin{tikzpicture}[scale=0.5]
\coordinate (center) at (0,0) {};
\coordinate (v) at ($(center)+(90:2cm)$) {};
\coordinate (u) at ($(center)+(270:2cm)$) {};

\coordinate (lv) at ($(center)+(90:2.6cm)$) {};
\coordinate (lu) at ($(center)+(270:2.4cm)$) {};

\tikzstyle{every node}=[inner sep=1pt]
\begin{scriptsize}
\node at (lv) {$\psi(v)$};
\node at (lu) {$\psi(u)$};
\end{scriptsize}

\draw[thick] ([shift=(-20:2.1cm)]0,0) arc (-20:200:2.1cm);
\draw[thick] ([shift=(160:1.9cm)]0,0) arc (160:380:1.9cm);

\draw[thick, white] (-3.0,-3)--(-3.0,-2.8);
\draw[thick, white] (3.0,3)--(3.0,2.8);

\end{tikzpicture}
\begin{tikzpicture}[scale=0.5]
\coordinate (center) at (0,0) {};
\coordinate (v) at ($(center)+(90:2cm)$) {};
\coordinate (u) at ($(center)+(270:2cm)$) {};

\coordinate (lv) at ($(center)+(90:2.6cm)$) {};
\coordinate (lu) at ($(center)+(270:2.4cm)$) {};

\tikzstyle{every node}=[inner sep=1pt]
\begin{scriptsize}
\node at (lv) {$\psi(v)$};
\node at (lu) {$\psi(u)$};
\end{scriptsize}
\draw[thick] ([shift=(70:2.1cm)]0,0) arc (70:200:2.1cm);
\draw[thick] ([shift=(160:1.9cm)]0,0) arc (160:290:1.9cm);

\draw[thick, white] (-3.0,-3)--(-3.0,-2.8);
\draw[thick, white] (3.0,3)--(3.0,2.8);
\end{tikzpicture}

\caption{\label{fig:mutual_arc_position} From left to right:
$\psi(v)$ and $\psi(u)$ are disjoint, $\psi(v)$ contains $\psi(u)$, $\psi(v)$ is contained in $\psi(u)$,
$\psi(v)$ and $\psi(u)$ cover the circle, and $\psi(v)$ and $\psi(u)$ overlap.
}

\end{figure}

In so-called \emph{normalized models}, introduced by Hsu in \cite{Hsu95}, the relative relation between the arcs reflects the neighbourhood relation between the vertices of $G$, as follows.
\begin{definition}
\label{def:normalized_model}
Let $G$ be a circular-arc graph with no universal vertices and no twins.
A circular-arc model $\psi$ of $G$ is \emph{normalized} if for every pair $(v,u)$ of distinct vertices of $G$ the following conditions are satisfied:
\begin{enumerate}
 \item \label{item:u_v_disjoint} if $uv \notin E(G)$, then $\psi(v)$ and $\psi(u)$ are disjoint,
 \item \label{item:v_contains_u} if $N[u] \subsetneq N[v]$, then $\psi(v)$ contains $\psi(u)$,
 \item \label{item:u_contains_v} if $N[v] \subsetneq N[u]$, then $\psi(v)$ is contained in $\psi(u)$,
 \item \label{item:u_v_cover_the_circle} if $uv \in E(G)$, 
 $N[v] \cup N[u] = V(G)$, $N[w] \subsetneq N[v]$ for every $w \in N[v]\setminus N[u]$, and $N[w] \subsetneq N[u]$ for every $w \in N[u]\setminus N[v]$, then $\psi(v)$ and $\psi(u)$ cover the circle, 
 \item \label{item:u_v_overlap} If none of the above condition holds, then $\psi(v)$ and $\psi(u)$ overlap.
\end{enumerate}
Furthermore, for a pair $(v,u)$ of vertices from $G$, we say that $v$ \emph{contains} $u$, $v$ \emph{is contained} in $u$, $v$ and $u$ \emph{cover the circle}, and 
$v$ and $u$ \emph{overlap} if the pair $(v,u)$ satisfies the assumption of statement \eqref{item:v_contains_u},
\eqref{item:u_contains_v}, \eqref{item:u_v_cover_the_circle}, and \eqref{item:u_v_overlap}, respectively.
\end{definition}
Hsu~\cite{Hsu95} showed that every circular-arc model $\psi$ of $G$ can be turned into a normalized model by 
possibly extending some arcs of $\psi$. 
We associate with every vertex $v \in V$ two sets, $\leftside(v)$ and $\rightside(v)$, where:
$$
\begin{array}{rcl}
\leftside(v) &=& \{ u \in V(G): \quad
\text{$v$ contains $u$} \quad \text{or} \quad \text{$v$ and $u$ cover the circle}\}, 
\\
\rightside(v) &=& \{ u \in V(G):
\quad \text{$v$ and $u$ are disjoint}  \quad \text{or} \quad \text{$v$ is contained in $u$}\}.
\end{array}
$$
Following \cite{Hsu95,Krawczyk20} we define the \emph{overlap graph} $G_{ov} = (V,{\sim})$ of $G$ which joins with an edge $u \sim v$ every two vertices $u,v$ of $G$ that overlap in $G$.
Given a normalized model $\psi$ of $G$ we obtain an intersection model $\phi$ of the overlap graph $G_{ov}$ by 
transforming every arc $\psi(v)$ for $v \in V$ to the oriented chord $\phi(v)$ such that $\psi(v)$ and $\phi(v)$ have the same endpoints and 
the oriented chord $\phi(v)$ has the arc $\psi(v)$ on its left side.
Note that the intersection model~$\phi$ of $G_{ov}$ satisfies the conditions: 
\begin{itemize}
\item $u \in \leftside(v)$ if and only if $\phi(u)$ lies on the left side of $\phi(v)$,
\item $u \in \rightside(v)$ if and only if $\phi(u)$ lies on the right side of $\phi(v)$.
\end{itemize}
See Figure~\ref{fig:uv_pairs} for an illustration.
The intersection models of $G_{ov}$ in the set of oriented chords satisfying the above two properties are called the \emph{conformal models} of $G$.
Clearly, there is one-to-one correspondence between the normalized models of $G$ and the conformal models of $G$. 
See~\cite{Hsu95,Krawczyk20} for more details.
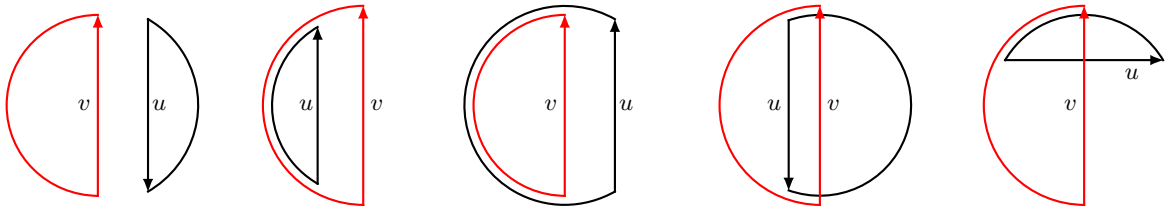
\begin{figure}[htp!]
\begin{tikzpicture}[scale=0.6,>=latex,shorten >=-0.4pt,shorten <=-0.4pt]
\coordinate (label) at (0,-3) {};

\coordinate (lv) at (-0.3,0) {};
\coordinate (lu) at (1.35,0) {};

\tikzstyle{every node}=[inner sep=1pt]
\begin{scriptsize}
\node at (lv) {$v$};
\node at (lu) {$u$};
\end{scriptsize}

\draw[thick] ([shift=(-60:2.2cm)]0,0) arc (-60:60:2.2cm);
\draw[thick,<-] ([shift=(-60:2.2cm)]0,0) -- ([shift=(60:2.2cm)]0,0);

\draw[red, thick] ([shift=(90:2.0cm)]0,0) arc (90:270:2.0cm);
\draw[red, thick,<-] ([shift=(90:2.0cm)]0,0) -- ([shift=(270:2.0cm)]0,0);

\draw[white] (-2.5,-2.5)--(-2.5,-2.3);
\draw[white] (2.5,2.5)--(2.5,2.3);
\end{tikzpicture}
\hspace{0.2cm}
\begin{tikzpicture}[scale=0.6,>=latex,shorten >=-0.4pt,shorten <=-0.4pt]
\coordinate (label) at (0,-3) {};

\coordinate (lv) at (0.3,0) {};
\coordinate (lu) at (-1.25,0) {};

\tikzstyle{every node}=[inner sep=1pt]
\begin{scriptsize}
\node at (lv) {$v$};
\node at (lu) {$u$};
\end{scriptsize}

\draw[thick] ([shift=(120:2.0cm)]0,0) arc (120:240:2.0cm);
\draw[thick,->] ([shift=(240:2.0cm)]0,0) -- ([shift=(120:2.0cm)]0,0);

\draw[red, thick] ([shift=(90:2.2cm)]0,0) arc (90:270:2.2cm);
\draw[red, thick,<-] ([shift=(90:2.2cm)]0,0) -- ([shift=(270:2.2cm)]0,0);

\draw[white] (-2.5,-2.5)--(-2.5,-2.3);
\draw[white] (1.1,2.5)--(1.1,2.3);
\end{tikzpicture}
\hspace{0.2cm}
\begin{tikzpicture}[scale=0.6,>=latex,shorten >=-0.4pt,shorten <=-0.4pt]
\coordinate (label) at (0,-3) {};

\coordinate (lv) at (-0.3,0) {};
\coordinate (lu) at (1.35,0) {};

\tikzstyle{every node}=[inner sep=1pt]
\begin{scriptsize}
\node at (lv) {$v$};
\node at (lu) {$u$};
\end{scriptsize}

\draw[thick] ([shift=(60:2.2cm)]0,0) arc (60:300:2.2cm);
\draw[thick,->] ([shift=(300:2.2cm)]0,0) -- ([shift=(60:2.2cm)]0,0);

\draw[red, thick] ([shift=(90:2.0cm)]0,0) arc (90:270:2.0cm);
\draw[red, thick,<-] ([shift=(90:2.0cm)]0,0) -- ([shift=(270:2.0cm)]0,0);

\draw[white] (-2.5,-2.5)--(-2.5,-2.3);
\draw[white] (2.3,2.5)--(2.3,2.3);
\end{tikzpicture}
\hspace{0.2cm}
\begin{tikzpicture}[scale=0.6,>=latex,shorten >=-0.4pt,shorten <=-0.4pt]
\coordinate (label) at (0,-3) {};

\coordinate (lv) at (0.3,0) {};
\coordinate (lu) at (-1.0,0) {};

\tikzstyle{every node}=[inner sep=1pt]
\begin{scriptsize}
\node at (lv) {$v$};
\node at (lu) {$u$};
\end{scriptsize}

\draw[thick] ([shift=(250:2.0cm)]0,0) arc (250:470:2.0cm);
\draw[thick,<-] ([shift=(250:2.0cm)]0,0) -- ([shift=(470:2.0cm)]0,0);

\draw[red, thick] ([shift=(90:2.2cm)]0,0) arc (90:270:2.2cm);
\draw[red, thick,<-] ([shift=(90:2.2cm)]0,0) -- ([shift=(270:2.2cm)]0,0);

\draw[white] (-2.5,-2.5)--(-2.5,-2.3);
\draw[white] (2.5,2.5)--(2.5,2.3);
\end{tikzpicture}
\hspace{0.2cm}
\begin{tikzpicture}[scale=0.6,>=latex,shorten >=-0.4pt,shorten <=-0.4pt]
\coordinate (label) at (0,-3) {};

\coordinate (lv) at (-0.3,0) {};
\coordinate (lu) at (1.05,0.7) {};

\tikzstyle{every node}=[inner sep=1pt]
\begin{scriptsize}
\node at (lv) {$v$};
\node at (lu) {$u$};
\end{scriptsize}

\draw[thick] ([shift=(30:2.0cm)]0,0) arc (30:150:2.0cm);
\draw[thick,<-] ([shift=(30:2.0cm)]0,0) -- ([shift=(150:2.0cm)]0,0);

\draw[red, thick] ([shift=(90:2.2cm)]0,0) arc (90:270:2.2cm);
\draw[red, thick,<-] ([shift=(90:2.2cm)]0,0) -- ([shift=(270:2.2cm)]0,0);

\draw[white] (-2.5,-2.5)--(-2.5,-2.3);
\draw[white] (2.5,2.5)--(2.5,2.3);
\end{tikzpicture}
\caption{\label{fig:uv_pairs} Mutual relation between the arcs $\psi(v)$ and $\psi(u)$ and between the corresponding oriented chords 
$\phi(v)$ and $\phi(u)$ for the cases: $v$ and $u$ are disjoint, $v$ contains $u$, $v$ is contained in $u$, $v$ and $u$ cover the circle, and $v$ and $u$ overlap, respectively.}
\end{figure}

Usually, the notion of normalized models is extended on all circular-arc graphs $G$ 
by requiring $\psi(v)=\circle$ for any universal vertex $v$ in $G$ and $\psi(u)=\psi(v)$ for any
pair $(u,v)$ of twin vertices in $G$.

\subsection{Representing the oriented chords and points of the circle}
Let $\circle$ be a fixed circle. 
Let $S$ be a collection of oriented chords of the circle $\circle$
and $P$ be a collection of points on the circle $\circle$ such that all the points from $P$ 
and the endpoints of the chords from $S$ are distinct.
We represent the set $B= S \cup P$ by a circular word $cw(B)$ on the letter set $S^* \cup P$, where 
$S^* = \{s^0,s^1:s \in S\}$, obtained as follows.
We start at some point on the circle with the empty word $cw(B)$. 
Then we follow the circle in the clockwise order and we append to $cw(B)$:
\begin{itemize}
\item the letter $p$ whenever we pass a point $p$ from $P$,
\item the letter $s^0$ whenever we pass the tail of a chord $s$ from $S$,
\item the letter $s^1$ whenever we pass the head of a chord $s$ from $S$.
\end{itemize}
Usually we use the same symbol to denote the set $B$ and its word representation $cw(B)$. 

For a collection $B$ of oriented chords and points on the circle
we define the \emph{reflection $B^R$} of $B$ obtained by mirroring every object from $B$ over some fixed line $L$ and then by reorienting every chord in $B^R$ 
(hence the arc to the left of an oriented chord $s^0s^1$ in $B^R$ 
is obtained by mirroring the arc to the left of the oriented chord $s^0s^1$ in $B$) -- see Figure~\ref{fig:reflection} for an illustration.
Note that the word $cw(B^R)$ representing $B^R$ is obtained from $cw(B)$ by reversing the order of the letters in $cw(B)$ and by exchanging the superscripts $0$ to $1$ and $1$ to $0$;
any word $w^R$ (not only circluar) obtained this way from a word $w$ is called the \emph{reflection} of $w$.

\begin{figure}[htp!]
\centering
\begin{tikzpicture}[yscale=0.80,xscale=0.80,>=latex,shorten >=-0.4pt,shorten <=-0.4pt]

\coordinate (center) at (0,0) {};

\coordinate (lv10) at ($(center)+(270:2.45cm)$) {};
\coordinate (lv11) at ($(center)+(90:2.45cm)$) {};

\coordinate (lv20) at ($(center)+(0:2.45cm)$) {};
\coordinate (lv21) at ($(center)+(180:2.45cm)$) {};

\coordinate (lv30) at ($(center)+(165:2.45cm)$) {};
\coordinate (lv31) at ($(center)+(30:2.45cm)$) {};

\coordinate (lv40) at ($(center)+(60:2.45cm)$) {};
\coordinate (lv41) at ($(center)+(-60:2.45cm)$) {};

\coordinate (p) at ($(center)+(125:2cm)$) {};
\coordinate (q) at ($(center)+(225:2cm)$) {};

\coordinate (lp) at ($(center)+(125:2.45cm)$) {};
\coordinate (lq) at ($(center)+(225:2.45cm)$) {};

\tikzstyle{every node}=[inner sep=1pt]
\node at (lv10) {$s^0_1$};
\node at (lv11) {$s^1_1$};
\node at (lv20) {$s^0_2$};
\node at (lv21) {$s^1_2$};
\node at (lv30) {$s^0_3$};
\node at (lv31) {$s^1_3$};
\node at (lv40) {$s^0_4$};
\node at (lv41) {$s^1_4$};
\node at (lp) {$p$};
\node at (lq) {$q$};

\draw (0,0) circle (2cm);

\draw[very thick] ([shift=(180:1.92cm)]0,0) arc (180:360:1.92cm);
\draw[very thick,->] ([shift=(360:1.92cm)]0,0) -- ([shift=(180:1.92cm)]0,0);

\draw[very thick] ([shift=(30:1.92cm)]0,0) arc (30:165:1.92cm);
\draw[very thick,->] ([shift=(165:1.92cm)]0,0) -- ([shift=(30:1.92cm)]0,0);

\draw[red,very thick] ([shift=(90:2.08cm)]0,0) arc (90:270:2.08cm);
\draw[red,very thick,->] ([shift=(270:2.08cm)]0,0) -- ([shift=(90:2.08cm)]0,0) ;

\draw[red,very thick] ([shift=(-60:2.08cm)]0,0) arc (-60:60:2.08cm);
\draw[red,very thick,->] ([shift=(60:2.08cm)]0,0) -- ([shift=(-60:2.08cm)]0,0);

\tikzstyle{every node}=[circle,minimum size=7pt,inner sep=0pt,draw,fill]
\node at (p) {};
\node at (q) {};

\draw[white] (-2.5,-2.8)--(-2.5,-2);
\draw[white] (2.5,2.8)--(2.5,2);
\end{tikzpicture} 
\begin{tikzpicture}[yscale=0.80,xscale=0.80,>=latex,shorten >=-0.4pt,shorten <=-0.4pt]
\draw[black,dashed] (0,-2.2)--(0,2.8);
\coordinate (L) at (0,-2.55) {};
\tikzstyle{every node}=[inner sep=1pt]
\node at (L) {$L$};
\draw[white] (-0.5,-2.8)--(-0.5,-2);
\draw[white] (0.5,2.8)--(0.5,2);
\end{tikzpicture} 
\begin{tikzpicture}[yscale=0.80,xscale=-0.80,>=latex,shorten >=-0.4pt,shorten <=-0.4pt]
\coordinate (center) at (0,0) {};

\coordinate (lv10) at ($(center)+(270:2.45cm)$) {};
\coordinate (lv11) at ($(center)+(90:2.45cm)$) {};

\coordinate (lv20) at ($(center)+(0:2.45cm)$) {};
\coordinate (lv21) at ($(center)+(180:2.45cm)$) {};

\coordinate (lv30) at ($(center)+(165:2.45cm)$) {};
\coordinate (lv31) at ($(center)+(30:2.45cm)$) {};

\coordinate (lv40) at ($(center)+(60:2.45cm)$) {};
\coordinate (lv41) at ($(center)+(-60:2.45cm)$) {};

\coordinate (p) at ($(center)+(125:2cm)$) {};
\coordinate (q) at ($(center)+(225:2cm)$) {};

\coordinate (lp) at ($(center)+(125:2.45cm)$) {};
\coordinate (lq) at ($(center)+(225:2.45cm)$) {};

\tikzstyle{every node}=[inner sep=1pt]
\node at (lv10) {$s^1_1$};
\node at (lv11) {$s^0_1$};
\node at (lv20) {$s^1_2$};
\node at (lv21) {$s^0_2$};
\node at (lv30) {$s^1_3$};
\node at (lv31) {$s^0_3$};
\node at (lv40) {$s^1_4$};
\node at (lv41) {$s^0_4$};
\node at (lp) {$p$};
\node at (lq) {$q$};

\draw (0,0) circle (2cm);

\draw[very thick] ([shift=(180:1.92cm)]0,0) arc (180:360:1.92cm);
\draw[very thick,<-] ([shift=(360:1.92cm)]0,0) -- ([shift=(180:1.92cm)]0,0);

\draw[very thick] ([shift=(30:1.92cm)]0,0) arc (30:165:1.92cm);
\draw[very thick,<-] ([shift=(165:1.92cm)]0,0) -- ([shift=(30:1.92cm)]0,0);

\draw[red,very thick] ([shift=(90:2.08cm)]0,0) arc (90:270:2.08cm);
\draw[red,very thick,<-] ([shift=(270:2.08cm)]0,0) -- ([shift=(90:2.08cm)]0,0) ;

\draw[red,very thick] ([shift=(-60:2.08cm)]0,0) arc (-60:60:2.08cm);
\draw[red,very thick,<-] ([shift=(60:2.08cm)]0,0) -- ([shift=(-60:2.08cm)]0,0);

\tikzstyle{every node}=[circle,minimum size=7pt,inner sep=0pt,draw,fill]
\node at (p) {};
\node at (q) {};

\draw[white] (-2.5,-2.8)--(-2.5,-2);
\draw[white] (2.5,2.8)--(2.5,2);
\end{tikzpicture} 

\caption{\label{fig:reflection} A collection $B$ of oriented chords and points and its reflection $B^R$.
The set $B$ is represented by the circular word $cw(B) = s^0_4s^1_3s^0_2s^1_4s^0_1qs^1_2s^0_3ps^1_1$ and the set $B^R$ is represented by $cw^R(B) = s^0_1ps^1_3s^0_2qs^1_1s^0_4s^1_2s^0_3s^1_4$.}
\end{figure}

We represent any conformal model $\phi$ of $G$ by its word representation (note that $\phi$ is a circular word over $V^*(G)$); we treat two conformal models of $G$ as \emph{equivalent} if their word representations are equal.
Given a conformal model $\phi$ of $G$, possibly extended by some points from the set $P$ (usually witnessing the Helly property of some cliques of $G$), given a subset $U$ of $V(G)$ and a subset $Q$ of $P$, 
by $\phi|(U^* \cup Q)$ we denote the circular order $\phi$ restricted to the letters from the set $U^* \cup Q$, which represents the restriction of the model $\phi$ to the chords representing the vertices from $U$ and to the points from $Q$.
Finally, we use operator $\equiv$ to stress that the equality holds between two circular words.

Before we describe the structure representing all non-equivalent conformal models of $G$, we need some preparation.
In particular, we need to introduce the basic concepts related to the modular decomposition of the graph~$G_{ov}$ and its connection to the transitive orientations of some induced subgraphs of~$G_{ov}$.

\subsection{Modular decomposition of $G_{ov}$}
\label{subsec:modular_decomposition}

Let $G=(V,E)$ be a circular-arc graph and let $G_{ov} = (V,{\sim})$ be the overlap graph of~$G$. 
By $(V, {\parallel})$ we denote the complement of $(V,{\sim})$.
For a set $U \subseteq V$, by $(U,{\sim})$, and $(U,{\parallel})$ we denote the subgraphs of $(V,{\sim})$ and $(V,{\parallel})$ 
induced by the set $U$, respectively.
For two sets $U_1,U_2 \subset V$, we write $U_1 \sim U_2$ ($U_1 \parallel U_2$)
if $u_1 \sim u_2$ ($u_1 \parallel u_2$, respectively) for every $u_1 \in U_1$ and $u_2 \in U_2$.
For a set $U \subseteq V$, by $U^*$ we denote the set $\{u^0,u^1: u \in U\}$.

The results presented below are due to Gallai \cite{Gal67}, 
except of Theorem~\ref{thm:permutation_models_transitive_orientations} by Dushnik and Miller~\cite{DM41}. 

A non-empty set $M\subseteq V$ is a \emph{module} in $G_{ov}$ 
if $x\sim M$ or $x \parallel M$ for every $x\in V\smallsetminus M$.
The singleton sets and the whole $V$ are the \emph{trivial modules} of $G_{ov}$. 
A module $M$ of $G_{ov}$ is \emph{strong} if $M\subseteq N$, $N\subseteq M$, or $M\cap N=\emptyset$ for every other module $N$ in $G_{ov}$.
In particular, two strong modules of $G_{ov}$ are either nested or disjoint.
The \emph{modular decomposition} of~$G_{ov}$, denoted by~$\strongModules(G_{ov})$, 
consists of all strong modules of~$G_{ov}$.
Since every two modules from $\strongModules(G_{ov})$ 
are either disjoint or nested, the modules from $\strongModules(G_{ov})$ can be organized in a tree in which $V$ is the root, the maximal proper subsets from of $M \in \strongModules(G_{ov})$ from $\strongModules(G_{ov})$ are the children of $M$ (the children of $M$ form a partition of~$M$), 
and the singleton modules $\{x\}$ for $x\in V$ are the leaves.

A module $M \in \strongModules(G_{ov})$ is \emph{serial} if $M_1\sim M_2$ for every two children $M_1$ and $M_2$ of~$M$, 
\emph{parallel} if $M_1 \parallel M_2$ for every two children $M_1$ and $M_2$ of~$M$, and \emph{prime} otherwise.
Equivalently, $M \in \strongModules$ is serial if $(M,{\parallel})$ is disconnected, 
parallel if $(M,{\sim})$ is disconnected, and prime if both $(M,{\sim})$ and $(M,{\parallel})$ are connected.

Suppose $U$ is a module in $G_{ov}$.
A graph $(U,{\sim})$ is a \emph{permutation subgraph} of $G_{ov}$ if there
exists a pair $(\tau^0,\tau^1)$, where $\tau^{0}$ and $\tau^{1}$ are permutations of $U$, such
that for every $x,y \in U$:
$$x \sim y \iff  
\begin{array}{l}
\text{$x$ occurs before $y$ in both $\tau^{0}$ and $\tau^{1}$, or} \\
\text{$y$ occurs before $x$ in both $\tau^{0}$ and $\tau^{1}$.}
\end{array}
$$
If this is the case, $(\tau^{0},\tau^{1})$ is called a \emph{permutation model} of $(U,{\sim})$.
See Figure \ref{fig:non_oriented_permutation_graph} for an example of a permutation graph and its permutation model.

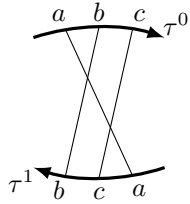
\begin{figure}[htp!]
\centering
\begin{tikzpicture}[xscale=0.85,yscale=0.5,>=latex,shorten >=-0.4pt,shorten <=-0.4pt]
\coordinate (center) at (0,0) {};
\coordinate (ua) at ($(center)+(105:2cm)$) {};
\coordinate (ub) at ($(center)+(90:2cm)$) {};
\coordinate (uc) at ($(center)+(75:2cm)$) {};

\coordinate (lua) at ($(center)+(105:2.4cm)$) {};
\coordinate (lub) at ($(center)+(90:2.4cm)$) {};
\coordinate (luc) at ($(center)+(75:2.4cm)$) {};

\coordinate (bc) at ($(center)+(270:2cm)$) {};
\coordinate (bb) at ($(center)+(255:2cm)$) {};
\coordinate (ba) at ($(center)+(285:2cm)$) {};

\coordinate (lbc) at ($(center)+(270:2.4cm)$) {};
\coordinate (lbb) at ($(center)+(255:2.4cm)$) {};
\coordinate (lba) at ($(center)+(285:2.4cm)$) {};

\coordinate (tau0) at ($(center)+(60:2.4cm)$) {};
\coordinate (tau1) at ($(center)+(240:2.4cm)$) {};

\tikzstyle{every node}=[inner sep=1pt]
\begin{footnotesize}

\node at (lua) {$a$};
\node at (lba) {$a$};

\node at (lub) {$b$};
\node at (lbb) {$b$};

\node at (luc) {$c$};
\node at (lbc) {$c$};

\node at (tau0) {$\tau^0$};
\node at (tau1) {$\tau^1$};

\end{footnotesize}

\draw[very thick,<-] ([shift=(60:2cm)]0,0) arc (60:120:2cm);
\draw[very thick,<-] ([shift=(240:2cm)]0,0) arc (240:300:2cm);

\draw (ua)--(ba);
\draw (ub)--(bb);
\draw (uc)--(bc);

\end{tikzpicture}

\caption{\label{fig:non_oriented_permutation_graph} 
A permutation model $(\tau^0,\tau^1) = (abc,acb)$ of the permutation graph $(\{a,b,c\}, \{a \sim b, a \sim c)\}$.
}
\end{figure}

Let $U$ be a module in $G_{ov}$ such that $(U,{\sim})$ is a permutation subgraph in $G_{ov}$.
The structure of the permutation models of $(U,{\sim})$ can be described by means of transitive 
orientations of the graphs $(U,{\sim})$ and $(U,{\parallel})$.
An \emph{orientation} $(U,{\prec})$ of $(U,{\sim})$ arises by orienting every edge $u \sim v$ in $(U,{\sim})$ either
from $v$ to $u$, denoted $u \prec v$, or from $u$ to $v$, denoted $v \prec u$.
An orientation $(U,{\prec})$ of $(U,{\sim})$ is \emph{transitive} if ${\prec}$ is a transitive relation on $U$.
We define transitive orientations of $(U,{\parallel})$ similarly.
Since $(U,{\sim})$ is a permutation subgraph in $G_{ov}$, 
$(U,{\sim})$ has a permutation model $\tau = (\tau^{0}, \tau^{1})$.
Note that $\tau = (\tau^0,\tau^1)$ yields transitive orientations ${\prec^{\tau}}$ and ${<^{\tau}}$ of the graphs $(U,{\sim})$ and $(U,{\parallel})$, respectively,
given by:
\begin{equation}
\label{eq:models_of_permutation_graphs_1}
\begin{array}{lll}
x \prec^{\tau} y & \iff & x \text{ occurs before } y \text{ in } \tau^0 \text{ and } x \sim y, \\
x <^{\tau} y & \iff &  x \text{ occurs before } y \text{ in } \tau^0 \text{ and } x \parallel y. \\
\end{array}
\end{equation}
On the other hand, given transitive orientations ${\prec^{\tau}}$ and ${<^{\tau}}$ of $(U,{\sim})$ and $(U,{\parallel})$, respectively, 
one can construct a permutation model $(\tau^{0}, \tau^{1})$ of $(U,{\sim})$ such that
\begin{equation}
\label{eq:models_of_permutation_graphs_2}
\begin{array}{lll}
x \text{ occurs before } y \text{ in } \tau^{0} \iff x \prec^{\tau} y \text{ or } x <^{\tau} y,\\
x \text{ occurs before } y \text{ in } \tau^{1} \iff x \prec^{\tau} y \text{ or } y <^{\tau} x.\\
\end{array}
\end{equation}
\begin{theorem}[\cite{DM41}]
\label{thm:permutation_models_transitive_orientations}
Let $(U,{\sim})$ be a permutation subgraph of $G_{ov}$.
There is a one-to-one correspondence between permutation models $\tau = (\tau^{0}, \tau^{1})$ of $(U,{\sim})$
and the pairs $({<^{\tau}}, {\prec^{\tau}})$ of transitive orientations of $(U,{\parallel})$ and $(U,{\sim})$, respectively, given by equations \eqref{eq:models_of_permutation_graphs_1} and~\eqref{eq:models_of_permutation_graphs_2}.
\end{theorem}
The transitive orientations of $(U,{\sim})$ and $(U,\parallel)$ can in turn be described by means of the modular decomposition trees
of $\strongModules(U,{\sim})$ and $\strongModules(U,\parallel)$.
Note that $\strongModules(U,{\sim}) = \strongModules(U,{\parallel})$ and since $U$ is a module in $(U,{\sim})$, we have
$\strongModules(U,{\sim}) = \{M \in \strongModules(G_{ov}): M \subseteq U\} \cup \{U\}$.
The relations between the transitive orientations of $(U,{\sim})$ and the modular decomposition tree of $(U,{\sim})$
were described by Gallai~\cite{Gal67}.
\begin{theorem}[\cite{Gal67}]
\label{thm:transitive_orientation_of_edges_between_children}
If $M_1,M_2 \in \strongModules(U,{\sim})$ are such that $M_1 \sim M_2$, then every
transitive orientation $(U,{\prec})$ satisfies either $M_1 \prec M_2$ or $M_2 \prec M_1$
(that is, either $x \prec y$ for every $x \in M_1$ and $y \in M_2$ or $y \prec x$ for every $x \in M_1$ and $y \in M_2$).
\end{theorem}
Let $M$ be a strong module in $\strongModules(U,{\sim})$.
The edge relation $\sim$ in $(M,{\sim})$ restricted to the edges joining the vertices from two different children of $M$ is denoted by $\sim_M$.
If $x\sim y$, then $x\sim_My$ for exactly one module $M\in\strongModules(U,{\sim})$.
Hence, the set $\{{\sim_M} : M \in \strongModules(U,{\sim})\}$ forms a partition of the edge set ${\sim}$ of the graph $(U,{\sim})$.
\begin{theorem}[\cite{Gal67}]
\label{thm:transitive_orientations_versus_transitive_orientations_of_strong_modules}
There is a one-to-one correspondence between the set of transitive orientations $(U,{\prec})$ of $(U,{\sim})$
and the families $$\big{\{}(M,{\prec_M}): M \in \strongModules(U,{\prec}) \text{ and } \prec_M \text{ is a transitive orientation of $(M, \sim_M)$}\big{\}}$$
given by $x \prec y \iff x \prec_M y$, where $M$ is the module in $\strongModules(U,{\sim})$ such that $x \sim_M y$.
\end{theorem}
The above theorem asserts that for every $M \in \strongModules(U,{\sim})$ every transitive orientation of $(U,{\sim})$ restricted to the edges of the graph $(M, {\sim_M})$ 
induces a transitive orientation of $(M, {\sim_M})$, 
and that every transitive orientation of $(U,{\sim})$ can be obtained by independent transitive orientation of $(M,{\sim_M})$ for $M \in \strongModules(U,{\sim})$.
Gallai~\cite{Gal67} characterized all possible transitive orientation of $(M,{\sim_M})$, where $M$ is a strong module of $\strongModules(U,{\sim})$.
\begin{theorem}[\cite{Gal67}]
Let $M$ be a prime module in $\strongModules(U,{\sim})$. 
Then, $(M,{\sim_M})$ has two transitive orientations, one being the reverse of the other.
\end{theorem}
For a parallel module $M$ the graph $(M,{\sim_M})$ has exactly one (empty) transitive orientation.
For a parallel module $M$ the transitive orientations of $(M,{\sim_M})$ correspond to the total orderings of its children, that
is, every transitive orientation of $(M,{\sim_M})$ has the form $M_{i_1} \prec \ldots \prec M_{i_k}$, where
$i_1 \ldots i_k$ is a permutation of $[k]$ and $M_1, \ldots,M_k$ are the children of $M$ in $\strongModules(U,{\sim})$.

\section{Data structure representing the conformal models of $G$}
\label{sec:PQM-trees}
Let $G = (V,E)$ be a circular-arc graph with no twins and no universal vertices and
let $G_{ov}=(V,{\sim})$ be the overlap graph of~$G$.
In this section we describe a data structure, called a \emph{PQM-tree}~$\pqmtree$ of $G$, 
which represents all conformal models of $G$.
First, we describe the basic properties and the role played by 
each component of $\pqmtree$, then we show how to read out the components of~$\pqmtree$
from some fixed conformal model of~$G$.
For the combinatorial definition of $\pqmtree$ and the formal proofs of its properties we refer the reader to~\cite{Krawczyk20}.

\subsection{Components of $\pqmtree$ and their properties} 
The PQM-tree $\pqmtree$ representing the conformal models of $G$ consists of:
\begin{itemize}
\item the set $\camodules = \{S_1,\ldots,S_t\}$ of \emph{CA-modules} of $G$.
The set $\camodules$ is defined such that:
\begin{itemize}
\item $\camodules$ forms a partition of~$V$,
\item for every $i \in [t]$ the graph $(S_i,{\sim})$ is a permutation subgraph of~$G_{ov}$ contained in some connected component of $G_{ov}$.
\end{itemize}
Additionally, for every $i \in [t]$ a vertex $s_i \in S_i$, called the \emph{representant} of $S_i$, is fixed.
\item the set $\metachords = \{\mathbb{S}_1,\ldots,\mathbb{S}_t\}$ of \emph{metachords} of $G$. 
Each metachord $\mathbb{S}_i$ is a triple $(S^0_i,S^1_i,{<_{S_i}})$, where:
\begin{itemize}
\item the set $\{S_i^0 , S_i^1\}$ forms a partition of $S^*_i$ such that $s^0_i \in S^0_i$, $s^1_i \in S^1_i$, and $|\{u_0,u_1\} \cap S^0_i| = |\{u^0 , u^1\} \cap S^1_i| = 1$ for every $u \in S$.
\item ${<_{S_i}}$ is a fixed transitive orientation of $(S_i,{\parallel})$.
\end{itemize}
The elements of the set $\slots = \{S^0_1,S^1_1,\ldots,S^0_t,S^1_t\}$ are called \emph{the slots of $G$}.  
\item the set $\Pi$ of \emph{admissible orders of the slots} of $G$.
Every member of $\Pi$ is a circular word over $\slots$ (contains every slot of $G$ exactly once).
\end{itemize}
The CA-modules and the slots of~$G$ are defined such that:
\begin{description}
 \item[\namedlabel{prop:slots}{(P1)}] For every conformal model $\phi$ of $G$ and every $i \in [t]$:
 \begin{itemize}
  \item for every $j \in \{0,1\}$ the elements of the set $S^j_i$ form a contiguous subword in $\phi$, denoted by $\phi|S^j_i$,
  \item the pair $(\phi|S^0_i,\phi|S^1_i)$ is an oriented permutation model of $(S_i,{\sim})$.
 \end{itemize}
\end{description}
The above property allows us to treat any conformal model $\phi$ of $G$ as a collection of $t$ oriented permutation models $(\phi|S^0_i,\phi|S^1_i)$
of $(S_i,{\sim})$ spanned between the slots $S^0_i,S^1_i$ -- see Figure~\ref{fig:example}. 

The set of oriented permutation models that might be spanned between the slots    
$S^0_i$ and $S^1_i$ is represented by the metachord $\SSS_i = (S^0_1,S^1_i,{<_{S_i}})$, as follows.
\begin{definition}
\label{def:admissible-models-for-metachord}
An oriented permutation model $\tau = (\tau^0,\tau^1)$ of $(S_i,{\sim})$ is \emph{admissible} by the metachord $\mathbb{S}_i$ if:
\begin{itemize}
 \item $\tau^j$ is a permutation of $S^j_i$ for $j \in \{0,1\}$,
 \item we have ${<^{\tau}} = {<_{S_i}}$ (${\prec^{\tau}}$ is not restricted), where ${<^{\tau}}$ and ${\prec^{\tau}}$ are transitive orientations 
 of $(S_i,{\parallel})$ and $(S_i,{\sim})$, respectively, corresponding to $\tau$.
\end{itemize}
\end{definition}
Note that the admissible models for $\SSS_i$, regardless of the conformal model $\phi$, keep the left/right relation between every two non-intersecting chords from $(\phi|S^0_i,\phi|S^1_i)$ in the same relation (${<_{S_i}}$ is defined so as it asserts the consistency with the sets $\leftside(\cdot)$ and $\rightside(\cdot)$).
The metachords $\SSS_1,\ldots,\SSS_t$ are defined such that:
\begin{description}
 \item[\namedlabel{prop:metachords}{(P2)}] For every conformal model $\phi$ of $G_{ov}$ and every $i \in [t]$ the oriented permutation model $(\phi|S^0_i, \phi|S^1_i)$ of $(S_i,{\sim})$ is admissible by the metachord $\mathbb{S}_i$.
\end{description}

Property~\ref{prop:slots} allows us to denote by $\pi(\phi)$ the \emph{circular order of the slots in $\phi$}, that is, the word obtained from $\phi$ by substituting every contiguous subword $\phi|S^j_i$ of $\phi$ by the letter $S^j_i$, for $i \in [t]$ and $j \in \{0,1\}$ -- see Figure~\ref{fig:example} for an illustration.
Clearly, $\pi(\phi)$ is a circular word over $\slots$.
The last component $\Pi$ maintains the set of circular orders of the slots in the conformal models of~$G$.
In particular, $\Pi$ is defined such that:
\begin{description}
 \item [\namedlabel{prop:circular-orders-of-the-slots}{(P3)}] For every conformal model $\phi$ of $G$ the circular word $\pi(\phi)$ is a member of $\Pi$.
\end{description}
Eventually, all the components of $\pqmtree$ are defined such that the following holds:
\begin{description}
\item [\namedlabel{prop:completeness}{(P4)}] We can generate any conformal model of $G$ by:  
\begin{itemize}
 \item picking a circular order of the slots $\pi$ from the set $\Pi$, 
 \item replacing the slots $S^0_i$ and $S^1_i$ in $\pi$ by words $\tau^0_i$ and $\tau^1_i$, 
 where $(\tau^0_i,\tau^1_i)$ is an oriented permutation model of $(S_i,{\sim})$ admissible by the metachord $\mathbb{S}_i$.
\end{itemize}
\end{description}
The set $\Pi$ might have exponentially many members, however, it has a linear-size representation in $\pqmtree$ by means of a \emph{PQ-tree} $\pqmtree^{PQ}$ of $G$ --
see Subsection~\ref{subsec:PQ-tree}. 

Figure~\ref{fig:example} shows a conformal model $\phi$ of some circular-arc graph $G$ (to the left).
The circular-arc graph $G$ is defined on the vertex set $V = \{a,b,c,d,e,f,g,h,i\}$;
the edges of~$G$ can be read out from $\phi$.
The data structure $\mathcal{T}$ representing the conformal models of $G$ consists of:
\begin{itemize}
 \item four CA-modules of $G$, $S_1 = \{a,b,c,d\}$, $S_2 = \{f\}$, $S_3 = \{g\}$, $S_4 = \{h,i\}$, represented by vertices
 $a$, $f$, $g$, and $h$, respectively.
 \item the slots $S^0_1 = \{a^0,b^0,c^1,d^0,e^1\}$, $S^1_1 = \{a^1,b^1,c^0,d^1,e^0\}$,
 $S^0_2 = \{f^0\}$, $S^1_2 = \{f^1\}$, $S^0_3 = \{g^0\}$, $S^1_3 = \{g^1\}$, $S^0_4 = \{h^0,i^1\}$, $S^1_4 = \{h^1,i^0\}$,
 \item the transitive orientations ${<_{S_1}},{<_{S_2}},{<_{S_3}},{<_{S_4}}$, where ${<_{S_1}}$ consists of the pairs $\{b <_{S_1} a, c <_{S_1} a, d <_{S_1} a, e <_{S_1} a, c <_{S_1} b, e <_{S_1} d\}$,
 ${<_{S_4}}$ consists of the pair $\{i <_{S_4} h\}$, ${<_{S_2}}$ and ${<_{S_3}}$ are empty,
 \item the set $\Pi = \{\pi,\pi^R\}$ of circular order of the slots, where 
 $\pi = S^1_1S^1_2S^1_4S^0_1S^0_3S^0_4S^1_3S^0_2$ and $\pi^R$ is the reflection of $\pi$.
\end{itemize}
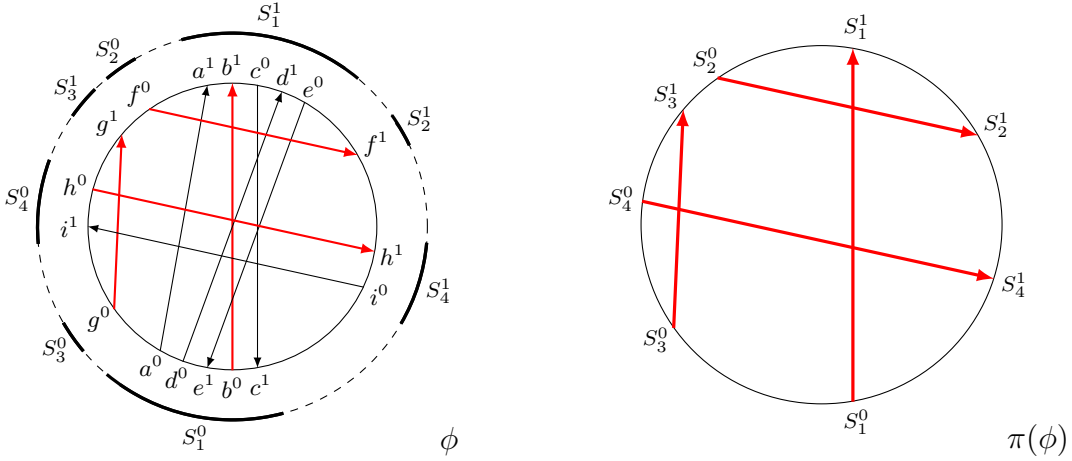
\begin{figure}[ht]
\begin{tikzpicture}[scale=0.95,>=latex,shorten >=-0.2pt,shorten <=-0.2pt]
\coordinate (center) at (0.0,0.0) {};

\coordinate (A0) at ($(center)+(240:2cm)$) {};
\coordinate (A1) at ($(center)+(100:2cm)$) {};

\coordinate (B0) at ($(center)+(270:2cm)$) {};
\coordinate (B1) at ($(center)+(90:2cm)$) {};

\coordinate (C0) at ($(center)+(80:2cm)$) {};
\coordinate (C1) at ($(center)+(280:2cm)$) {};

\coordinate (D0) at ($(center)+(250:2cm)$) {};
\coordinate (D1) at ($(center)+(70:2cm)$) {};

\coordinate (E0) at ($(center)+(60:2cm)$) {};
\coordinate (E1) at ($(center)+(260:2cm)$) {};

\coordinate (F0) at ($(center)+(125:2cm)$) {};
\coordinate (F1) at ($(center)+(30:2cm)$) {};

\coordinate (G0) at ($(center)+(215:2cm)$) {};
\coordinate (G1) at ($(center)+(140:2cm)$) {};

\coordinate (H0) at ($(center)+(165:2cm)$) {};
\coordinate (H1) at ($(center)+(350:2cm)$) {};

\coordinate (I0) at ($(center)+(335:2cm)$) {};
\coordinate (I1) at ($(center)+(180:2cm)$) {};

\coordinate (lA0) at ($(center)+(240:2.25cm)$) {};
\coordinate (lA1) at ($(center)+(100:2.25cm)$) {};

\coordinate (lB0) at ($(center)+(270:2.25cm)$) {};
\coordinate (lB1) at ($(center)+(90:2.25cm)$) {};

\coordinate (lC0) at ($(center)+(80:2.25cm)$) {};
\coordinate (lC1) at ($(center)+(280:2.25cm)$) {};

\coordinate (lD0) at ($(center)+(250:2.25cm)$) {};
\coordinate (lD1) at ($(center)+(70:2.25cm)$) {};

\coordinate (lE0) at ($(center)+(60:2.25cm)$) {};
\coordinate (lE1) at ($(center)+(260:2.25cm)$) {};

\coordinate (lF0) at ($(center)+(125:2.25cm)$) {};
\coordinate (lF1) at ($(center)+(30:2.3cm)$) {};

\coordinate (lG0) at ($(center)+(215:2.25cm)$) {};
\coordinate (lG1) at ($(center)+(140:2.25cm)$) {};

\coordinate (lH0) at ($(center)+(165:2.25cm)$) {};
\coordinate (lH1) at ($(center)+(350:2.25cm)$) {};

\coordinate (lI0) at ($(center)+(335:2.25cm)$) {};
\coordinate (lI1) at ($(center)+(180:2.25cm)$) {};

\coordinate (LA0) at ($(center)+(240:3cm)$) {};
\coordinate (LA1) at ($(center)+(100:3cm)$) {};

\coordinate (LB0) at ($(center)+(270:3cm)$) {};
\coordinate (LB1) at ($(center)+(90:3cm)$) {};

\coordinate (LC0) at ($(center)+(80:3cm)$) {};
\coordinate (LC1) at ($(center)+(280:3cm)$) {};

\coordinate (LD0) at ($(center)+(250:3cm)$) {};
\coordinate (LD1) at ($(center)+(70:3cm)$) {};

\coordinate (LE0) at ($(center)+(60:3cm)$) {};
\coordinate (LE1) at ($(center)+(260:3cm)$) {};

\coordinate (LF0) at ($(center)+(125:3cm)$) {};
\coordinate (LF1) at ($(center)+(30:3cm)$) {};

\coordinate (LG0) at ($(center)+(215:3cm)$) {};
\coordinate (LG1) at ($(center)+(140:3cm)$) {};

\coordinate (LH0) at ($(center)+(165:3cm)$) {};
\coordinate (LH1) at ($(center)+(350:3cm)$) {};

\coordinate (LI0) at ($(center)+(335:3cm)$) {};
\coordinate (LI1) at ($(center)+(180:3cm)$) {};

\coordinate (LHI0) at ($(center)+(172.5:3cm)$) {};
\coordinate (LHI1) at ($(center)+(342.5:3cm)$) {};

\coordinate (L) at (3,-3) {};

\draw (center) circle (2cm);
\draw[dashed] (center) circle (2.7cm);

\draw[->] (A0)--(A1);
\draw[->,red, thick] (B0)--(B1);
\draw[->] (C0)--(C1);
\draw[->] (D0)--(D1);
\draw[->] (E0)--(E1);

\draw[->,thick, red] (F0)--(F1);

\draw[->,thick, red] (G0)--(G1);

\draw[->, thick, red] (H0)--(H1);

\draw[->] (I0)--(I1);

\draw[very thick, -] ([shift=(105:2.7cm)]0,0) arc (105:50:2.7cm);
\draw[very thick, -] ([shift=(230:2.7cm)]0,0) arc (230:285:2.7cm);

\draw[very thick, -] ([shift=(120:2.7cm)]0,0) arc (120:130:2.7cm);
\draw[very thick, -] ([shift=(25:2.7cm)]0,0) arc (25:35:2.7cm);

\draw[very thick, -] ([shift=(135:2.7cm)]0,0) arc (135:145:2.7cm);
\draw[very thick, -] ([shift=(210:2.7cm)]0,0) arc (210:220:2.7cm);

\draw[very thick, -] ([shift=(160:2.7cm)]0,0) arc (160:185:2.7cm);
\draw[very thick, -] ([shift=(330:2.7cm)]0,0) arc (330:355:2.7cm);

\draw[very thick, -] ([shift=(105:2.7cm)]0,0) arc (105:55:2.7cm);
\draw[very thick, -] ([shift=(235:2.7cm)]0,0) arc (235:285:2.7cm);

\draw[very thick, -] ([shift=(120:2.7cm)]0,0) arc (120:130:2.7cm);
\draw[very thick, -] ([shift=(25:2.7cm)]0,0) arc (25:35:2cm);

\draw[very thick, -] ([shift=(135:2.7cm)]0,0) arc (135:145:2.7cm);
\draw[very thick, -] ([shift=(210:2.7cm)]0,0) arc (210:220:2.7cm);

\draw[very thick, -] ([shift=(160:2.7cm)]0,0) arc (160:185:2.7cm);
\draw[very thick, -] ([shift=(330:2.7cm)]0,0) arc (330:355:2.7cm);

\tikzstyle{every node}=[inner sep=1pt]
\begin{scriptsize}
\node at (lA0) {$a^0$};
\node at (lA1) {$a^1$};
\node at (lB0) {$b^0$};
\node at (lB1) {$b^1$};
\node at (lC0) {$c^0$};
\node at (lC1) {$c^1$};
\node at (lD0) {$d^0$};
\node at (lD1) {$d^1$};
\node at (lE0) {$e^0$};
\node at (lE1) {$e^1$};
\node at (lF0) {$f^0$};
\node at (lF1) {$f^1$};
\node at (lG0) {$g^0$};
\node at (lG1) {$g^1$};
\node at (lH0) {$h^0$};
\node at (lH1) {$h^1$};
\node at (lI0) {$i^0$};
\node at (lI1) {$i^1$};
\end{scriptsize}

\tikzstyle{every node}=[inner sep=1pt]
\begin{tiny}
\node at (LC0) {$S^1_1$};
\node at (LE1) {$S^0_1$};
\node at (LF0) {$S^0_2$};
\node at (LF1) {$S^1_2$};
\node at (LG0) {$S^0_3$};
\node at (LG1) {$S^1_3$};
\node at (LHI0) {$S^0_4$};
\node at (LHI1) {$S^1_4$};
\end{tiny}
\node at (L) {$\phi$};

\draw[white] (-4,-3)--(-4,-2.5);
\draw[white] (4,3)--(4,2.5);
\end{tikzpicture}
\begin{tikzpicture}[scale=0.95,>=latex,shorten >=-0.2pt,shorten <=-0.2pt]
\coordinate (center) at (0.0,0.0) {};

\coordinate (C0) at ($(center)+(80:2.5cm)$) {};
\coordinate (C1) at ($(center)+(280:2.5cm)$) {};

\coordinate (F0) at ($(center)+(125:2.5cm)$) {};
\coordinate (F1) at ($(center)+(30:2.5cm)$) {};

\coordinate (G0) at ($(center)+(215:2.5cm)$) {};
\coordinate (G1) at ($(center)+(140:2.5cm)$) {};

\coordinate (HI0) at ($(center)+(172.5:2.5cm)$) {};
\coordinate (HI1) at ($(center)+(342.5:2.5cm)$) {};

\coordinate (LC0) at ($(center)+(80:2.8cm)$) {};
\coordinate (LC1) at ($(center)+(280:2.8cm)$) {};

\coordinate (LF0) at ($(center)+(125:2.8cm)$) {};
\coordinate (LF1) at ($(center)+(30:2.8cm)$) {};

\coordinate (LG0) at ($(center)+(215:2.8cm)$) {};
\coordinate (LG1) at ($(center)+(140:2.8cm)$) {};

\coordinate (LHI0) at ($(center)+(172.5:2.8cm)$) {};
\coordinate (LHI1) at ($(center)+(342.5:2.8cm)$) {};

\coordinate (L) at (3,-3) {};

\draw (center) circle (2.5cm);

\draw[->,red,very thick] (C1)--(C0);

\draw[->,red,very thick] (F0)--(F1);

\draw[->,red,very thick] (G0)--(G1);

\draw[->,red,very thick] (HI0)--(HI1);

\tikzstyle{every node}=[inner sep=1pt]
\begin{tiny}
\node at (LC0) {$S^1_1$};
\node at (LC1) {$S^0_1$};
\node at (LF0) {$S^0_2$};
\node at (LF1) {$S^1_2$};
\node at (LG0) {$S^0_3$};
\node at (LG1) {$S^1_3$};
\node at (LHI0) {$S^0_4$};
\node at (LHI1) {$S^1_4$};
\end{tiny}
\node at (L) {$\pi(\phi)$};

\draw[white] (-4,-3)--(-4,-2.5);
\draw[white] (4,3)--(4,2.5);
\end{tikzpicture}

\caption{\label{fig:example} A conformal model $\phi$ of $G$ and the circular order of the slots $\pi(\phi)$ in $\phi$.}
\end{figure}

\subsection{CA-modules, slots, metachords and their admissible models}
In this subsection we first describe the structure of the admissible models for a single metachord of~$G$.
Then, we show how we can read out CA-modules, slots, and metachords of $G$ from 
some conformal model of $G$.

\subsubsection{The structure of the admissible models}
Let $S$ be a CA-module of $G$ and let 
$\SSS = (S^0,S^1,{<_{S}})$ be the metachord associated with~$S$.
Due to Theorem~\ref{thm:permutation_models_transitive_orientations}, the models $\tau$ admissible by $\mathbb{S}$ 
are in the correspondence with the transitive orientations ${\prec_{\tau}}$ of the permutation graph $(S,{\sim})$.
Transitive orientations of $(S,{\sim})$ are in turn represented by the modular decomposition tree $\strongModules(S,{\sim})$ of the graph $(S,{\sim})$:
Theorem~\ref{thm:transitive_orientations_versus_transitive_orientations_of_strong_modules} 
asserts that each transitive orientation ${\prec_{\tau}}$ of $(S,{\sim})$ is uniquely determined by transitive orientations ${\prec^{M}_{\tau}}$ of the graphs 
$(M,{\sim}_M)$, where $M$ runs over all inner nodes in $\strongModules(S,{\sim})$.
The modular decomposition tree $\strongModules(S,{\sim})$ will be the part of $\pqmtree$,
denoted by $\pqmtree_S$, and the inner nodes of $\pqmtree_S$ for $S \in \camodules$ 
will be called \emph{M-nodes} of $\pqmtree$.

Figure~\ref{fig:admissible-model-structure} shows 
a modular decomposition tree $\pqmtree_S$ of $(S,{\sim})$ for some CA-module $S$ and its admissible model $\tau = (\tau^0,\tau^1)$.
M-node $S$ is prime and hence $(S,{\sim_S})$ has two transitive orientations,
M-node $A_2$ is serial with three children and hence $(A_2,{\sim_{A_2}})$ 
has $3!$ transitive orientations corresponding to the permutations of the children of $A_2$.
The remaining M-nodes $A_3, B_2, B_3$ are parallel and each graph $(A_3,{\sim_{A_3}})$, $(B_2,{\sim_{B_2}})$, $(B_3,{\sim_{B_3}})$ has one empty transitive orientation. 
Hence, the metachord $\mathbb{S}$ has $12$ different admissible models, 
corresponding to $12=2 \cdot 3!$ transitive orientations of $(S,{\sim})$.

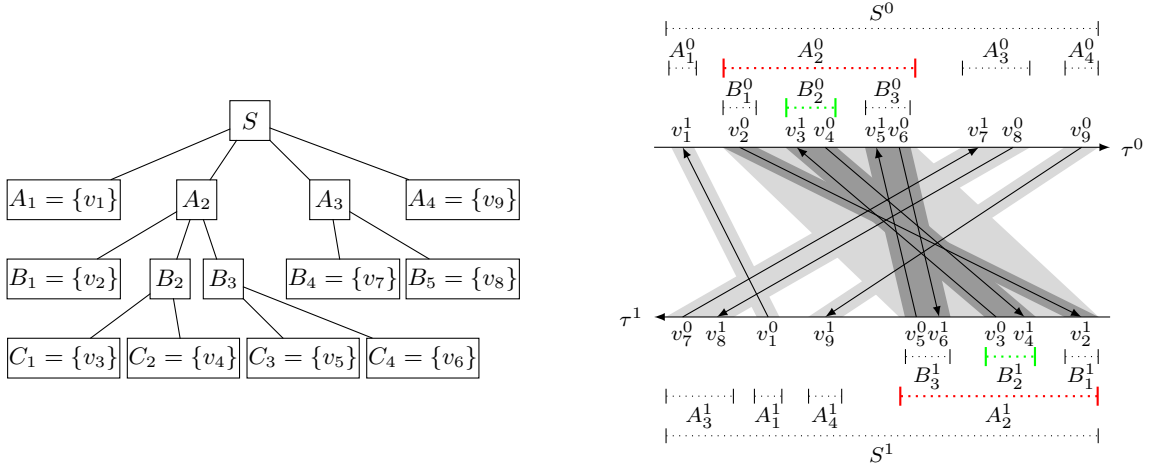
\begin{figure}[ht]
\begin{tikzpicture}[scale=0.7,>=latex,shorten >=-0.4pt,shorten <=-0.4pt]
\tikzstyle{every node}=[rectangle,minimum size=15pt,inner sep=0.5,draw];
  \begin{scriptsize}
  \node (M) at (4.0,7.5) {$S$};
  \node (v1) at (0.5,6) {$A_1 = \{v_1\}$};
  \node (A) at (3,6) {$A_2$};
  \node (B) at (5.5,6) {$A_3$};
  \node (v9) at (8,6) {$A_4 = \{v_9\}$};
  \node (v2) at (0.5,4.5) {$B_1 = \{v_2\}$};
  \node (C) at (2.5,4.5) {$B_2$};
  \node (D) at (3.5,4.5) {$B_3$};
  \node (v7) at (5.75,4.5) {$B_4=\{v_7\}$};
  \node (v8) at (8,4.5) {$B_5 = \{v_8\}$};
  \node (v3) at (0.5,3) {$C_1 = \{v_3\}$};
  \node (v4) at (2.75,3) {$C_2 = \{v_4\}$};
  \node (v5) at (5,3) {$C_3 = \{v_5\}$};
  \node (v6) at (7.25,3) {$C_4 = \{v_6\}$};
 
\end{scriptsize}
\tikzstyle{every node}=[square,minimum size=15pt,inner sep=0.5,draw];
\path (M) edge (v1); 
\path (M) edge (A); 
\path (M) edge (B); 
\path (M) edge (v9); 
\path (A) edge (v2); 
\path (A) edge (C); 
\path (A) edge (D); 
\path (B) edge (v7); 
\path (B) edge (v8); 
\path (C) edge (v3); 
\path (C) edge (v4); 
\path (D) edge (v5); 
\path (D) edge (v6); 

\draw[-,white] (-0.5,1) -- (0,1);
\end{tikzpicture}
\hspace{1cm}
\begin{tikzpicture}[xscale=0.75,yscale=0.75,>=latex]
\coordinate (a_v1) at (0,3) {};
\coordinate (a_v2) at (1,3) {};
\coordinate (a_v3) at (2,3) {};
\coordinate (a_v4) at (2.5,3) {};
\coordinate (a_v5) at (3.4,3) {};
\coordinate (a_v6) at (3.8,3) {};
\coordinate (a_v7) at (5.2,3) {};
\coordinate (a_v8) at (5.8,3) {};
\coordinate (a_v9) at (7,3) {};

\coordinate (b_v2) at (7,0) {};
\coordinate (b_v4) at (6,0) {};
\coordinate (b_v3) at (5.5,0) {};
\coordinate (b_v6) at (4.5,0) {};
\coordinate (b_v5) at (4.1,0) {};
\coordinate (b_v9) at (2.5,0) {};
\coordinate (b_v1) at (1.5,0) {};
\coordinate (b_v8) at (0.6,0) {};
\coordinate (b_v7) at (0,0) {};

\coordinate (la_v1) at (0,3.3) {};
\coordinate (la_v2) at (1,3.3) {};
\coordinate (la_v3) at (2,3.3) {};
\coordinate (la_v4) at (2.5,3.3) {};
\coordinate (la_v5) at (3.4,3.3) {};
\coordinate (la_v6) at (3.8,3.3) {};
\coordinate (la_v7) at (5.2,3.3) {};
\coordinate (la_v8) at (5.8,3.3) {};
\coordinate (la_v9) at (7,3.3) {};

\coordinate (lb_v2) at (7,-0.3) {};
\coordinate (lb_v4) at (6,-0.3) {};
\coordinate (lb_v3) at (5.5,-0.3) {};
\coordinate (lb_v6) at (4.5,-0.3) {};
\coordinate (lb_v5) at (4.1,-0.3) {};
\coordinate (lb_v9) at (2.5,-0.3) {};
\coordinate (lb_v1) at (1.5,-0.3) {};
\coordinate (lb_v8) at (0.6,-0.3) {};
\coordinate (lb_v7) at (0,-0.3) {};

\coordinate (lS0) at (3.5,5.1) {};

\coordinate (tau0) at (7.9,3);
\coordinate (tau1) at (-0.9,0);

\tikzstyle{every node}=[inner sep=2pt,fill=white]

\coordinate (lS0) at (3.5,5.4) {};
\draw[dotted,|-|] (-0.3,5.1) -- (7.3,5.1);

\coordinate (lS1) at (3.5,-2.4) {};
\draw[dotted,|-|] (-0.3,-2.1) -- (7.3,-2.1);

\draw[fill=gray!30, draw=none] (-0.2,3) -- (0.2,3) -- (1.7,0) -- (1.3,0) -- cycle;
\coordinate (lA01) at (0,4.7) {};
\draw[dotted,|-|] (-0.25,4.4) -- (0.25,4.4);
\coordinate (lA11) at (1.5,-1.75) {};
\draw[dotted,|-|] (1.25,-1.4) -- (1.75,-1.4);

\draw[fill=gray!30, draw=none] (0.7,3) -- (4.1,3) -- (7.3,0) -- (3.8,0) -- cycle;
\coordinate (lA02) at (2.25,4.7) {};
\draw[dotted,red,thick,|-|] (0.7,4.4) -- (4.1,4.4);
\coordinate (lA12) at (5.55,-1.75) {};
\draw[dotted,red,thick,|-|] (3.8,-1.4) -- (7.3,-1.4);

\draw[fill=gray!30, draw=none] (4.9,3) -- (6.1,3) -- (0.9,0) -- (-0.3,0) -- cycle;
\coordinate (lA03) at (5.5,4.7) {};
\draw[dotted,|-|] (4.9,4.4) -- (6.1,4.4);
\coordinate (lA13) at (0.3,-1.75) {};
\draw[dotted,|-|] (-0.3,-1.4) -- (0.9,-1.4);

\draw[fill=gray!30, draw=none] (6.7,3) -- (7.3,3) -- (2.8,0) -- (2.2,0) -- cycle;
\coordinate (lA04) at (7,4.7) {};
\draw[dotted,|-|] (6.7,4.4) -- (7.3,4.4);
\coordinate (lA14) at (2.5,-1.75) {};
\draw[dotted,|-|] (2.2,-1.4) -- (2.8,-1.4);

\draw[fill=gray!80, draw=none] (0.7,3) -- (1.3,3) -- (7.3,0) -- (6.7,0) -- cycle;
\coordinate (lB01) at (1,4) {};
\draw[dotted,|-|] (0.7,3.7) -- (1.3,3.7);
\coordinate (lB11) at (7,-1.05) {};
\draw[dotted,|-|] (6.7,-0.7) -- (7.3,-0.7);

\draw[fill=gray!80, draw=none] (1.8,3) -- (2.7,3) -- (6.2,0) -- (5.3,0) -- cycle;
\coordinate (lB02) at (2.25,4) {};
\draw[dotted,green,thick,|-|] (1.8,3.7) -- (2.7,3.7);
\coordinate (lB12) at (5.75,-1.05) {};
\draw[dotted,green,thick,|-|] (5.3,-0.7) -- (6.2,-0.7);

\draw[fill=gray!80, draw=none] (3.2,3) -- (4,3) -- (4.7,0) -- (3.9,0) -- cycle;
\coordinate (lB03) at (3.6,4) {};
\draw[dotted,|-|] (3.2,3.7) -- (4,3.7);
\coordinate (lB13) at (4.3,-1.05) {};
\draw[dotted,|-|] (3.9,-0.7) -- (4.7,-0.7);

\draw[<-] (a_v1)--(b_v1);
\draw[->] (a_v2)--(b_v2);
\draw[<-] (a_v3)--(b_v3);
\draw[->] (a_v4)--(b_v4);
\draw[<-] (a_v5)--(b_v5);
\draw[->] (a_v6)--(b_v6);
\draw[<-] (a_v7)--(b_v7);
\draw[->] (a_v8)--(b_v8);
\draw[->] (a_v9)--(b_v9);

\draw[->] (-0.5,3) -- (7.5,3);
\draw[<-] (-0.5,0) -- (7.5,0);

\tikzstyle{every node}=[inner sep=1pt]
\begin{tiny}
\node at (lb_v1) {$v^0_1$};
\node at (lb_v2) {$v^1_2$};
\node at (lb_v3) {$v^0_3$};
\node at (lb_v4) {$v^1_4$};
\node at (lb_v5) {$v^0_5$};
\node at (lb_v6) {$v^1_6$};
\node at (lb_v7) {$v^0_7$};
\node at (lb_v8) {$v^1_8$};
\node at (lb_v9) {$v^1_9$};

\node at (la_v1) {$v^1_1$};
\node at (la_v2) {$v^0_2$};
\node at (la_v3) {$v^1_3$};
\node at (la_v4) {$v^0_4$};
\node at (la_v5) {$v^1_5$};
\node at (la_v6) {$v^0_6$};
\node at (la_v7) {$v^1_7$};
\node at (la_v8) {$v^0_8$};
\node at (la_v9) {$v^0_9$};

\node at (lB01) {$B^0_1$};
\node at (lB02) {$B^0_2$};
\node at (lB03) {$B^0_3$};

\node at (lB11) {$B^1_1$};
\node at (lB12) {$B^1_2$};
\node at (lB13) {$B^1_3$};

\node at (lA01) {$A^0_1$};
\node at (lA02) {$A^0_2$};
\node at (lA03) {$A^0_3$};
\node at (lA04) {$A^0_4$};

\node at (lA11) {$A^1_1$};
\node at (lA12) {$A^1_2$};
\node at (lA13) {$A^1_3$};
\node at (lA14) {$A^1_4$};

\node at (lS0) {$S^0$};
\node at (lS1) {$S^1$};

\node at (tau0) {$\tau^0$};
\node at (tau1) {$\tau^1$};
\end{tiny}
\end{tikzpicture}
\caption{\label{fig:admissible-model-structure} Modular decomposition tree
$\pqmtree_S$ of CA-module $S$ (to the left) and its admissible model $\tau = (\tau^0,\tau^1)$.
We have $\tau^0|A^0_2 = v^0_2v^1_3v^0_4v^1_5v^0_6$ and $\tau^1|A^1_2 = v^1_2v^1_4v^0_3v^1_6v^0_5$ (depicted by red intervals) and 
$\tau^0|B^0_2 = v^1_3v^0_4$ and $\tau^1|B^1_2 = v^1_4v^0_3$ (depicted by green intervals),
$\tau_{|S} = (A^0_1A^0_2A^0_3A^0_4, A^1_2A^1_4A^1_1A^1_3)$ and $\tau_{|A_2} = (B^0_1B^0_2B^0_3,B^1_1B^1_2B^1_3)$.}
\end{figure}

In this paper we represent the models admissible for $\SSS$ the following way.
Firstly, we extend the notion of the metachord on the set of all nodes in 
$\pqmtree_{S}$: for a node $M \in \pqmtree_{S}$ the \emph{metachord $\MMM$} associated with $M$ is the triple $(M^0,M^1,{<_M})$, where 
$M^0 = M^* \cap S^0$, $M^1 = M^* \cap S^1$, and ${<_M}$ equals to ${<_{S}}$ 
restricted to $M$.
Secondly, for every M-node $M$ we define the set $\Pi(M)$ of \emph{admissible orderings} of $M$, as follows: for every transitive orientation ${\prec_M}$ of $(M,{\sim}_M)$ we define an element $\pi_M=(\pi^0, \pi^1)$ in the set $\Pi(M)$, where for $j \in \{0,1\}$ the word $\pi^j$ is a permutation of the set $\{K^j: K \text{ is a child of }M\}$ such that for every two distinct children $K,L$ of $M$ we have:
\begin{equation}
\begin{array}{lll}
K^0 \text{ occurs before } L^0 \text{ in } \pi^{0} \iff K \prec_M L \text{ or } K <_{M} L,\\
K^1 \text{ occurs before } L^1 \text{ in } \pi^{1} \iff K \prec_M L \text{ or } L <_{M} K.\\
\end{array}
\end{equation}
Note that, by Theorem~\ref{thm:transitive_orientation_of_edges_between_children}, 
for every model $\tau = (\tau^0,\tau^1)$ admissible for $\SSS$ and every node $M$ in $\pqmtree_{S}$ the sets $M^0$ and $M^1$ form contiguous subwords, say $\tau^0|M^0$ and $\tau^1|M^1$, in the words $\tau^0$ and $\tau^1$, respectively.
Now, for every M-node $M$ in $\pqmtree_{S}$ let $\tau_{|M}$ denote the pair
$({\tau^0}_{|M}, {\tau^1}_{|M})$, where the word ${\tau^j}_{|M}$ is obtained from the word $\tau^{j}|M^j$ by replacing its contiguous subword $\tau^j|K^j$ by the letter $K^j$, for every child $K$ of $M$ in $\pqmtree_{S}$ -- see Figure~\ref{fig:admissible-model-structure} for an illustration.
Note that $\tau_{|M}$ is a member of $\Pi(M)$.
Summing up, Theorems \ref{thm:permutation_models_transitive_orientations} and~\ref{thm:transitive_orientations_versus_transitive_orientations_of_strong_modules} yield the following:
\begin{theorem}
There is a bijection between the set $\Phi^{\SSS}$ of admissible models for the metachord $\SSS$ 
and the set $\Phi^{\SSS}_{\diamond}$, where
$$\Phi^{\SSS}_{\diamond} = \Bigg{\{}\Big{\{}\big{(}M,\pi_M\big{)}: 
\begin{array}{c}
\text{$M$ is an inner node in $\pqmtree_{S}$ and } \\ 
\text{$\pi_M$ is an admissible ordering from $\Pi(M)$}
\end{array}
\Big{\}}\Bigg{\}},$$
established by 
$$ \Phi^{\SSS} \ni \tau  \quad \longrightarrow \quad \Big{\{}\big{(}M,\tau_{|M}\big{)}: M \text{ is an M-node in } \pqmtree_{S}\Big{\}} \in \Phi^{\SSS}_{\diamond}.$$
\end{theorem}

Finally, suppose $\phi$ is a conformal model of $G$ and $\tau = (\tau^0,\tau^1)$ is an admissible model for~$\SSS$ spanned between $S^0$ and $S^1$ (that is, $(\phi|S^0, \phi|S^1) = (\tau^0, \tau^1)$).
Let $\phi^R$ be the reflection of $\phi$ and let $\mu = (\mu^0,\mu^1)$ be an admissible model for $\SSS$ spanned between $S^0$ and $S^1$ (that is, $(\phi^R|S^0, \phi^R|S^1) = (\mu^0, \mu^1)$).
Note that $\mu$ and $\tau$ are related: $\mu^0$ is the reflection of $\tau^1$ and 
$\mu^1$ is the reflection of $\mu^0$, and hence we say $\mu$ is the \emph{reflection of $\tau$} -- see Figure \ref{fig:admissible-model-reflection}.
Observe also that, if $\tau$ and $\mu$ correspond to transitive orientations 
${\prec^{\tau}}$ and ${\prec^{\mu}}$ of $(S,{\sim})$, then $\prec^{\mu}$ is the reverse of $\prec^{\tau}$.
In particular, for every prime M-node $M$ in $\pqmtree_S$, the set $\Pi(M)$ has two admissible orderings, one being the reflection  of the other - see Figure~\ref{fig:admissible-model-reflection}.

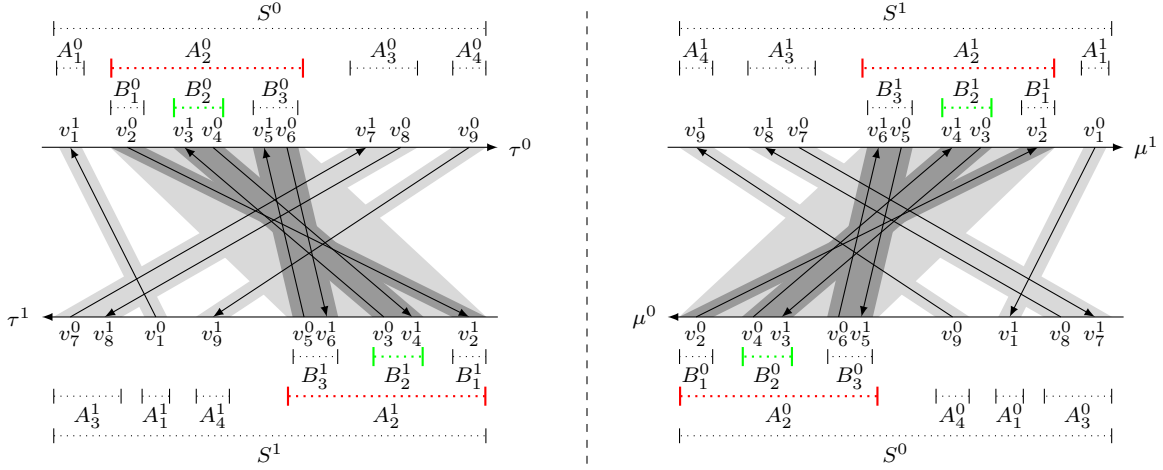
\begin{figure}[ht]
\begin{tikzpicture}[xscale=0.75,yscale=0.75,>=latex]
\coordinate (a_v1) at (0,3) {};
\coordinate (a_v2) at (1,3) {};
\coordinate (a_v3) at (2,3) {};
\coordinate (a_v4) at (2.5,3) {};
\coordinate (a_v5) at (3.4,3) {};
\coordinate (a_v6) at (3.8,3) {};
\coordinate (a_v7) at (5.2,3) {};
\coordinate (a_v8) at (5.8,3) {};
\coordinate (a_v9) at (7,3) {};

\coordinate (b_v2) at (7,0) {};
\coordinate (b_v4) at (6,0) {};
\coordinate (b_v3) at (5.5,0) {};
\coordinate (b_v6) at (4.5,0) {};
\coordinate (b_v5) at (4.1,0) {};
\coordinate (b_v9) at (2.5,0) {};
\coordinate (b_v1) at (1.5,0) {};
\coordinate (b_v8) at (0.6,0) {};
\coordinate (b_v7) at (0,0) {};

\coordinate (la_v1) at (0,3.3) {};
\coordinate (la_v2) at (1,3.3) {};
\coordinate (la_v3) at (2,3.3) {};
\coordinate (la_v4) at (2.5,3.3) {};
\coordinate (la_v5) at (3.4,3.3) {};
\coordinate (la_v6) at (3.8,3.3) {};
\coordinate (la_v7) at (5.2,3.3) {};
\coordinate (la_v8) at (5.8,3.3) {};
\coordinate (la_v9) at (7,3.3) {};

\coordinate (lb_v2) at (7,-0.3) {};
\coordinate (lb_v4) at (6,-0.3) {};
\coordinate (lb_v3) at (5.5,-0.3) {};
\coordinate (lb_v6) at (4.5,-0.3) {};
\coordinate (lb_v5) at (4.1,-0.3) {};
\coordinate (lb_v9) at (2.5,-0.3) {};
\coordinate (lb_v1) at (1.5,-0.3) {};
\coordinate (lb_v8) at (0.6,-0.3) {};
\coordinate (lb_v7) at (0,-0.3) {};

\coordinate (lS0) at (3.5,5.1) {};

\coordinate (tau0) at (7.9,3);
\coordinate (tau1) at (-0.9,0);

\tikzstyle{every node}=[inner sep=2pt,fill=white]

\coordinate (lS0) at (3.5,5.4) {};
\draw[dotted,|-|] (-0.3,5.1) -- (7.3,5.1);

\coordinate (lS1) at (3.5,-2.4) {};
\draw[dotted,|-|] (-0.3,-2.1) -- (7.3,-2.1);

\draw[fill=gray!30, draw=none] (-0.2,3) -- (0.2,3) -- (1.7,0) -- (1.3,0) -- cycle;
\coordinate (lA01) at (0,4.7) {};
\draw[dotted,|-|] (-0.25,4.4) -- (0.25,4.4);
\coordinate (lA11) at (1.5,-1.75) {};
\draw[dotted,|-|] (1.25,-1.4) -- (1.75,-1.4);

\draw[fill=gray!30, draw=none] (0.7,3) -- (4.1,3) -- (7.3,0) -- (3.8,0) -- cycle;
\coordinate (lA02) at (2.25,4.7) {};
\draw[dotted,red,thick,|-|] (0.7,4.4) -- (4.1,4.4);
\coordinate (lA12) at (5.55,-1.75) {};
\draw[dotted,red,thick,|-|] (3.8,-1.4) -- (7.3,-1.4);

\draw[fill=gray!30, draw=none] (4.9,3) -- (6.1,3) -- (0.9,0) -- (-0.3,0) -- cycle;
\coordinate (lA03) at (5.5,4.7) {};
\draw[dotted,|-|] (4.9,4.4) -- (6.1,4.4);
\coordinate (lA13) at (0.3,-1.75) {};
\draw[dotted,|-|] (-0.3,-1.4) -- (0.9,-1.4);

\draw[fill=gray!30, draw=none] (6.7,3) -- (7.3,3) -- (2.8,0) -- (2.2,0) -- cycle;
\coordinate (lA04) at (7,4.7) {};
\draw[dotted,|-|] (6.7,4.4) -- (7.3,4.4);
\coordinate (lA14) at (2.5,-1.75) {};
\draw[dotted,|-|] (2.2,-1.4) -- (2.8,-1.4);

\draw[fill=gray!80, draw=none] (0.7,3) -- (1.3,3) -- (7.3,0) -- (6.7,0) -- cycle;
\coordinate (lB01) at (1,4) {};
\draw[dotted,|-|] (0.7,3.7) -- (1.3,3.7);
\coordinate (lB11) at (7,-1.05) {};
\draw[dotted,|-|] (6.7,-0.7) -- (7.3,-0.7);

\draw[fill=gray!80, draw=none] (1.8,3) -- (2.7,3) -- (6.2,0) -- (5.3,0) -- cycle;
\coordinate (lB02) at (2.25,4) {};
\draw[dotted,green,thick,|-|] (1.8,3.7) -- (2.7,3.7);
\coordinate (lB12) at (5.75,-1.05) {};
\draw[dotted,green,thick,|-|] (5.3,-0.7) -- (6.2,-0.7);

\draw[fill=gray!80, draw=none] (3.2,3) -- (4,3) -- (4.7,0) -- (3.9,0) -- cycle;
\coordinate (lB03) at (3.6,4) {};
\draw[dotted,|-|] (3.2,3.7) -- (4,3.7);
\coordinate (lB13) at (4.3,-1.05) {};
\draw[dotted,|-|] (3.9,-0.7) -- (4.7,-0.7);

\draw[<-] (a_v1)--(b_v1);
\draw[->] (a_v2)--(b_v2);
\draw[<-] (a_v3)--(b_v3);
\draw[->] (a_v4)--(b_v4);
\draw[<-] (a_v5)--(b_v5);
\draw[->] (a_v6)--(b_v6);
\draw[<-] (a_v7)--(b_v7);
\draw[->] (a_v8)--(b_v8);
\draw[->] (a_v9)--(b_v9);

\draw[->] (-0.5,3) -- (7.5,3);
\draw[<-] (-0.5,0) -- (7.5,0);

\tikzstyle{every node}=[inner sep=1pt]
\begin{tiny}
\node at (lb_v1) {$v^0_1$};
\node at (lb_v2) {$v^1_2$};
\node at (lb_v3) {$v^0_3$};
\node at (lb_v4) {$v^1_4$};
\node at (lb_v5) {$v^0_5$};
\node at (lb_v6) {$v^1_6$};
\node at (lb_v7) {$v^0_7$};
\node at (lb_v8) {$v^1_8$};
\node at (lb_v9) {$v^1_9$};

\node at (la_v1) {$v^1_1$};
\node at (la_v2) {$v^0_2$};
\node at (la_v3) {$v^1_3$};
\node at (la_v4) {$v^0_4$};
\node at (la_v5) {$v^1_5$};
\node at (la_v6) {$v^0_6$};
\node at (la_v7) {$v^1_7$};
\node at (la_v8) {$v^0_8$};
\node at (la_v9) {$v^0_9$};

\node at (lB01) {$B^0_1$};
\node at (lB02) {$B^0_2$};
\node at (lB03) {$B^0_3$};

\node at (lB11) {$B^1_1$};
\node at (lB12) {$B^1_2$};
\node at (lB13) {$B^1_3$};

\node at (lA01) {$A^0_1$};
\node at (lA02) {$A^0_2$};
\node at (lA03) {$A^0_3$};
\node at (lA04) {$A^0_4$};

\node at (lA11) {$A^1_1$};
\node at (lA12) {$A^1_2$};
\node at (lA13) {$A^1_3$};
\node at (lA14) {$A^1_4$};

\node at (lS0) {$S^0$};
\node at (lS1) {$S^1$};

\node at (tau0) {$\tau^0$};
\node at (tau1) {$\tau^1$};
\end{tiny}
\end{tikzpicture}
\hspace{0.4cm}
\begin{tikzpicture}[xscale=0.75,yscale=0.75,>=latex]
\draw[dashed] (0,-4)--(0,4);
\end{tikzpicture}
\hspace{0.28cm}
\begin{tikzpicture}[xscale=-0.75,yscale=0.75,>=latex]
\coordinate (a_v1) at (0,3) {};
\coordinate (a_v2) at (1,3) {};
\coordinate (a_v3) at (2,3) {};
\coordinate (a_v4) at (2.5,3) {};
\coordinate (a_v5) at (3.4,3) {};
\coordinate (a_v6) at (3.8,3) {};
\coordinate (a_v7) at (5.2,3) {};
\coordinate (a_v8) at (5.8,3) {};
\coordinate (a_v9) at (7,3) {};

\coordinate (b_v2) at (7,0) {};
\coordinate (b_v4) at (6,0) {};
\coordinate (b_v3) at (5.5,0) {};
\coordinate (b_v6) at (4.5,0) {};
\coordinate (b_v5) at (4.1,0) {};
\coordinate (b_v9) at (2.5,0) {};
\coordinate (b_v1) at (1.5,0) {};
\coordinate (b_v8) at (0.6,0) {};
\coordinate (b_v7) at (0,0) {};

\coordinate (la_v1) at (0,3.3) {};
\coordinate (la_v2) at (1,3.3) {};
\coordinate (la_v3) at (2,3.3) {};
\coordinate (la_v4) at (2.5,3.3) {};
\coordinate (la_v5) at (3.4,3.3) {};
\coordinate (la_v6) at (3.8,3.3) {};
\coordinate (la_v7) at (5.2,3.3) {};
\coordinate (la_v8) at (5.8,3.3) {};
\coordinate (la_v9) at (7,3.3) {};

\coordinate (lb_v2) at (7,-0.3) {};
\coordinate (lb_v4) at (6,-0.3) {};
\coordinate (lb_v3) at (5.5,-0.3) {};
\coordinate (lb_v6) at (4.5,-0.3) {};
\coordinate (lb_v5) at (4.1,-0.3) {};
\coordinate (lb_v9) at (2.5,-0.3) {};
\coordinate (lb_v1) at (1.5,-0.3) {};
\coordinate (lb_v8) at (0.6,-0.3) {};
\coordinate (lb_v7) at (0,-0.3) {};

\coordinate (lS0) at (3.5,5.1) {};

\coordinate (tau0) at (7.9,0);
\coordinate (tau1) at (-0.9,3);

\tikzstyle{every node}=[inner sep=2pt,fill=white]

\coordinate (lS0) at (3.5,5.4) {};
\draw[dotted,|-|] (-0.3,5.1) -- (7.3,5.1);

\coordinate (lS1) at (3.5,-2.4) {};
\draw[dotted,|-|] (-0.3,-2.1) -- (7.3,-2.1);

\draw[fill=gray!30, draw=none] (-0.2,3) -- (0.2,3) -- (1.7,0) -- (1.3,0) -- cycle;
\coordinate (lA01) at (0,4.7) {};
\draw[dotted,|-|] (-0.25,4.4) -- (0.25,4.4);
\coordinate (lA11) at (1.5,-1.75) {};
\draw[dotted,|-|] (1.25,-1.4) -- (1.75,-1.4);

\draw[fill=gray!30, draw=none] (0.7,3) -- (4.1,3) -- (7.3,0) -- (3.8,0) -- cycle;
\coordinate (lA02) at (2.25,4.7) {};
\draw[dotted,red,thick,|-|] (0.7,4.4) -- (4.1,4.4);
\coordinate (lA12) at (5.55,-1.75) {};
\draw[dotted,red,thick,|-|] (3.8,-1.4) -- (7.3,-1.4);

\draw[fill=gray!30, draw=none] (4.9,3) -- (6.1,3) -- (0.9,0) -- (-0.3,0) -- cycle;
\coordinate (lA03) at (5.5,4.7) {};
\draw[dotted,|-|] (4.9,4.4) -- (6.1,4.4);
\coordinate (lA13) at (0.3,-1.75) {};
\draw[dotted,|-|] (-0.3,-1.4) -- (0.9,-1.4);

\draw[fill=gray!30, draw=none] (6.7,3) -- (7.3,3) -- (2.8,0) -- (2.2,0) -- cycle;
\coordinate (lA04) at (7,4.7) {};
\draw[dotted,|-|] (6.7,4.4) -- (7.3,4.4);
\coordinate (lA14) at (2.5,-1.75) {};
\draw[dotted,|-|] (2.2,-1.4) -- (2.8,-1.4);

\draw[fill=gray!80, draw=none] (0.7,3) -- (1.3,3) -- (7.3,0) -- (6.7,0) -- cycle;
\coordinate (lB01) at (1,4) {};
\draw[dotted,|-|] (0.7,3.7) -- (1.3,3.7);
\coordinate (lB11) at (7,-1.05) {};
\draw[dotted,|-|] (6.7,-0.7) -- (7.3,-0.7);

\draw[fill=gray!80, draw=none] (1.8,3) -- (2.7,3) -- (6.2,0) -- (5.3,0) -- cycle;
\coordinate (lB02) at (2.25,4) {};
\draw[dotted,green,thick,|-|] (1.8,3.7) -- (2.7,3.7);
\coordinate (lB12) at (5.75,-1.05) {};
\draw[dotted,green,thick,|-|] (5.3,-0.7) -- (6.2,-0.7);

\draw[fill=gray!80, draw=none] (3.2,3) -- (4,3) -- (4.7,0) -- (3.9,0) -- cycle;
\coordinate (lB03) at (3.6,4) {};
\draw[dotted,|-|] (3.2,3.7) -- (4,3.7);
\coordinate (lB13) at (4.3,-1.05) {};
\draw[dotted,|-|] (3.9,-0.7) -- (4.7,-0.7);

\draw[->] (a_v1)--(b_v1);
\draw[<-] (a_v2)--(b_v2);
\draw[->] (a_v3)--(b_v3);
\draw[<-] (a_v4)--(b_v4);
\draw[->] (a_v5)--(b_v5);
\draw[<-] (a_v6)--(b_v6);
\draw[->] (a_v7)--(b_v7);
\draw[<-] (a_v8)--(b_v8);
\draw[<-] (a_v9)--(b_v9);

\draw[<-] (-0.5,3) -- (7.5,3);
\draw[->] (-0.5,0) -- (7.5,0);

\tikzstyle{every node}=[inner sep=1pt]
\begin{tiny}
\node at (lb_v1) {$v^1_1$};
\node at (lb_v2) {$v^0_2$};
\node at (lb_v3) {$v^1_3$};
\node at (lb_v4) {$v^0_4$};
\node at (lb_v5) {$v^1_5$};
\node at (lb_v6) {$v^0_6$};
\node at (lb_v7) {$v^1_7$};
\node at (lb_v8) {$v^0_8$};
\node at (lb_v9) {$v^0_9$};

\node at (la_v1) {$v^0_1$};
\node at (la_v2) {$v^1_2$};
\node at (la_v3) {$v^0_3$};
\node at (la_v4) {$v^1_4$};
\node at (la_v5) {$v^0_5$};
\node at (la_v6) {$v^1_6$};
\node at (la_v7) {$v^0_7$};
\node at (la_v8) {$v^1_8$};
\node at (la_v9) {$v^1_9$};

\node at (lB01) {$B^1_1$};
\node at (lB02) {$B^1_2$};
\node at (lB03) {$B^1_3$};

\node at (lB11) {$B^0_1$};
\node at (lB12) {$B^0_2$};
\node at (lB13) {$B^0_3$};

\node at (lA01) {$A^1_1$};
\node at (lA02) {$A^1_2$};
\node at (lA03) {$A^1_3$};
\node at (lA04) {$A^1_4$};

\node at (lA11) {$A^0_1$};
\node at (lA12) {$A^0_2$};
\node at (lA13) {$A^0_3$};
\node at (lA14) {$A^0_4$};

\node at (lS0) {$S^1$};
\node at (lS1) {$S^0$};

\node at (tau0) {$\mu^0$};
\node at (tau1) {$\mu^1$};
\end{tiny}
\end{tikzpicture}
\caption{\label{fig:admissible-model-reflection} 
Admissible model $\tau = (\tau^0, \tau^1)$ (to the left) and its reflection $\mu = (\mu^0, \mu^1)$ (to the right).
Prime M-node $S$ has two admissible orderings: $(A^0_1A^0_2A^0_3A^0_4, A^1_2A^1_4A^1_1A^1_3)$ and its reflection $(A^0_3A^0_1A^0_4A^0_2,A^1_4A^1_3A^1_2A^1_1)$.
}
\end{figure}

\subsubsection{Determining CA-modules, slots, and metachords}
In this section we show how to read out CA-modules, slots, and metachords of $G$ from some conformal model $\phi$ of~$G$.
For a combinatorial definition we refer the reader to \cite{Krawczyk20}.

First, for every child $M$ of the root node $V$ of $\strongModules(G_{ov})$, we define
the set $\camodules(M)$ of \emph{CA-modules of $G$ contained in $M$}, as follows:
$S$~is a member of $\camodules(M)$ if $S$ is a maximal submodule of $M$ (not necessarily strong) such that the chords of $\phi(S)$ form a \emph{valid oriented permutation model of $(S,{\sim})$} in $\phi$, which means that the letters of $S^*$ form either two contiguous subwords $\tau^0(S)$ and $\tau^1(S)$ in $\phi$ or one contiguous subword of the form $\tau^0(S)\tau^1(S)$, where $(\tau^0(S), \tau^1(S))$ is an oriented permutation model of $(S,{\sim})$.
It is shown in \cite{Krawczyk20} that for every child $M$ of $V$ in $\strongModules(G_{ov})$ the set $\camodules(M)$ forms a partition of $M$.
Finally, we set $$\camodules = \bigcup \{\camodules(M): M \text{ is a child of $V$ in } \strongModules(G_{ov})\}.$$

Assume that $\camodules = \{S_1,\ldots,S_t\}$.
For every $S_i \in \camodules$ we pick a vertex $s_i \in S_i$ representing the set~$S_i$.
We assume $s^0_i \in \tau^0(S_i)$ and $s^1_i \in \tau^1(S_i)$ 
(we assert $s^0_i \in \tau^0(S_i)$ by possibly swapping the superscripts in $\tau^0(S_i)$ and $\tau^1(S_i)$).
We define the metachord $\mathbb{S}_i = (S^0_i,S^1_i,{<_{S_i}})$ for CA-module $S_i$ such that:
\begin{itemize}
  \item $S^0_i$ contains the letters from the word $\tau^0(S_i)$,
  \item $S^1_i$ contains the letters from the word $\tau^1(S_i)$,
  \item assuming ${<^\tau}$ and ${\prec^\tau}$ are the transitive orientations of $(S_i,{\parallel})$ and $(S_i,{\sim})$ corresponding to the oriented permutation model 
  $(\tau^0(S_i), \tau^{1}(S_i))$ of $(S_i,{\sim})$, we set ${<_{S_i}} = {<^\tau}$.
\end{itemize}
It is shown in \cite{Krawczyk20} that the slots $S^0_1,S^1_1,\ldots,S^0_t,S^1_t$ satisfy property~\ref{prop:slots} and 
the metachords $\mathbb{S}_1,\ldots,\mathbb{S}_t$ satisfy property~\ref{prop:metachords}.

\subsection{PQ-trees}
\label{subsec:PQ-tree}
As we mentioned, the set $\Pi$ may contain exponentially many members, but it has a linear-size
representation by means of the PQ-tree $\pqmtree^{PQ}$, which is the part of $\pqmtree$.
To define $\pqmtree^{PQ}$ we need some preparation.
Let $Q_0$ and $Q_1$ be two components of $G_{ov}$.
We say the components $Q_{0}$ and $Q_{1}$ are \emph{separated} 
if there is $v \in V \setminus (Q_0 \cup Q_1)$ such that $Q_i \subseteq \leftside(v)$ and $Q_{1-i} \subseteq \rightside(v)$ for some $i \in \{0,1\}$; otherwise, $Q_0$ and $Q_1$ are \emph{neighbouring}.
See Figure~\ref{fig:PQ_tree} for an illustration.

The PQ-tree $\pqmtree^{PQ}$ is an unrooted tree.
The leaf nodes of $\pqmtree^{PQ}$ are in the correspondence with the slots of $G$.
The non-leaf nodes of~$\pqmtree^{PQ}$ are labelled either by the letter P (\emph{P-nodes}) or the letter Q (\emph{Q-nodes}):
Q-nodes are in the correspondence with the connected components of $G_{ov}$ and 
P-nodes are in the correspondence with the maximal sets consisting of at least two pairwise neighbouring Q-nodes.
We refer to non-leaf nodes of $\pqmtree^{PQ}$ as \emph{PQ-nodes} of $\pqmtree^{PQ}$.
We add an edge between a Q-node $Q$ and a $P$-node $P$
if $Q \in P$ and we add an edge between a slot $S^i_j$ and a Q-node $Q$ if $S_i \subseteq Q$.
See Figure~\ref{fig:PQ_tree} for an illustration.

Note that $\pqmtree^{PQ}$ consist a single inner node $V$ in the case when $V$ is serial/prime in $\strongModules(G_{ov})$.


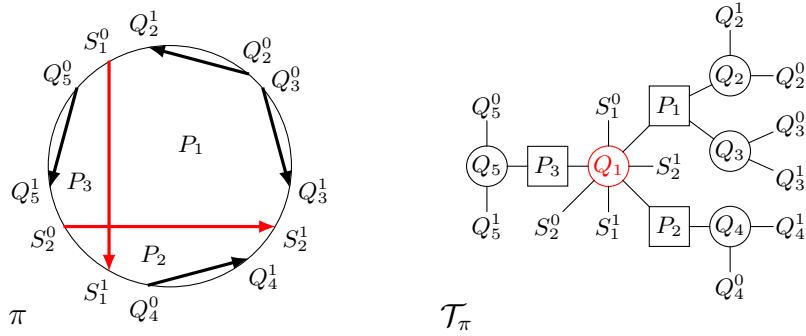
\begin{figure}[htp!]
\begin{tikzpicture}[scale=0.8,>=latex,shorten >=-0.4pt,shorten <=-0.4pt]
\coordinate (center) at (0,0) {};
\coordinate (label) at (-2.5,-2.5) {};

\coordinate (s10) at ($(center)+(120:2.0cm)$) {};
\coordinate (s11) at ($(center)+(240:2.0cm)$) {};

\coordinate (s20) at ($(center)+(210:2.0cm)$) {};
\coordinate (s21) at ($(center)+(330:2.0cm)$) {};

\coordinate (m20) at ($(center)+(50:2.0cm)$) {};
\coordinate (m21) at ($(center)+(100:2.0cm)$) {};

\coordinate (m31) at ($(center)+(-10:2.0cm)$) {};
\coordinate (m30) at ($(center)+(40:2.0cm)$) {};

\coordinate (m40) at ($(center)+(260:2.0cm)$) {};
\coordinate (m41) at ($(center)+(310:2.0cm)$) {};

\coordinate (m50) at ($(center)+(140:2.0cm)$) {};
\coordinate (m51) at ($(center)+(190:2.0cm)$) {};

\coordinate (lp1) at ($(center)+(45:0.5cm)$) {};
\coordinate (lp3) at ($(center)+(190:1.5cm)$) {};
\coordinate (lp2) at ($(center)+(260:1.5cm)$) {};

\coordinate (ls10) at ($(center)+(120:2.4cm)$) {};
\coordinate (ls11) at ($(center)+(240:2.4cm)$) {};

\coordinate (ls20) at ($(center)+(210:2.4cm)$) {};
\coordinate (ls21) at ($(center)+(330:2.4cm)$) {};

\coordinate (lm20) at ($(center)+(53:2.4cm)$) {};
\coordinate (lm21) at ($(center)+(100:2.4cm)$) {};

\coordinate (lm31) at ($(center)+(-10:2.4cm)$) {};
\coordinate (lm30) at ($(center)+(37:2.4cm)$) {};

\coordinate (lm40) at ($(center)+(260:2.4cm)$) {};
\coordinate (lm41) at ($(center)+(310:2.4cm)$) {};

\coordinate (lm50) at ($(center)+(140:2.4cm)$) {};
\coordinate (lm51) at ($(center)+(190:2.4cm)$) {};

\tikzstyle{every node}=[inner sep=1pt]
\begin{scriptsize}
\node at (lp1) {$P_1$};
\node at (lp2) {$P_2$};
\node at (lp3) {$P_3$};

\node at (ls10) {$S^0_1$};
\node at (ls11) {$S^1_1$};

\node at (ls20) {$S^0_2$};
\node at (ls21) {$S^1_2$};

\node at (lm20) {$Q^0_2$};
\node at (lm21) {$Q^1_2$};

\node at (lm30) {$Q^0_3$};
\node at (lm31) {$Q^1_3$};

\node at (lm40) {$Q^0_4$};
\node at (lm41) {$Q^1_4$};

\node at (lm50) {$Q^0_5$};
\node at (lm51) {$Q^1_5$};

\end{scriptsize}
\node at (label) {$\pi$};

\draw (0,0) circle (2cm);

\draw[very thick,red,->] (s10) -- (s11);
\draw[very thick,red,->] (s20) -- (s21);
\draw[very thick,->] (m20) -- (m21);
\draw[very thick,->] (m30) -- (m31);
\draw[very thick,->] (m40) -- (m41);
\draw[very thick,->] (m50) -- (m51);

\draw[white] (-2.8,-2.8)--(-2.8,-2);
\draw[white] (2.8,2.8)--(2.8,2);

\end{tikzpicture}
\hspace{1cm}
\begin{tikzpicture}[scale=0.8,>=latex,shorten >=-0.4pt,shorten <=-0.4pt]
\coordinate (label) at (-2.5,-2.5) {};

  \tikzstyle{every node}=[circle,minimum size=10pt,inner sep=0.5,draw];
  \begin{scriptsize}
  \node[red] (m1) at (0,0) {$Q_1$};
  \node (m2) at (2,1.5) {$Q_2$};
  \node (m3) at (2,0.25) {$Q_3$};
  \node (m4) at (2,-1) {$Q_4$};
  \node (m5) at (-2,0.0) {$Q_5$};
  \end{scriptsize}
  \tikzstyle{every node}=[rectangle,minimum size=15pt,inner sep=0.5,draw];
  \begin{scriptsize}
  \node (p1) at (1,1.0) {$P_1$};
  \node (p2) at (1,-1.0) {$P_2$};
  \node (p3) at (-1.0,0) {$P_3$};
  \end{scriptsize}
\tikzstyle{every node}=[inner sep=1pt]
  \begin{scriptsize}
  \node (m12) at (2,2.5) {$Q^1_2$};
  \node (m02) at (3,1.5) {$Q^0_2$};
  \node (m03) at (3,0.7) {$Q^0_3$};
  \node (m13) at (3,-0.2) {$Q^1_3$};
  \node (m14) at (3,-1) {$Q^1_4$};
  \node (m04) at (2,-2) {$Q^0_4$};

  \node (s01) at (0,1) {$S^0_1$};
  \node (s11) at (0,-1) {$S^1_1$};
  \node (s02) at (-1,-1) {$S^0_2$};
  \node (s12) at (1,0) {$S^1_2$};
  \node (m05) at (-2,1) {$Q^0_5$};
  \node (m15) at (-2,-1) {$Q^1_5$};

  \end{scriptsize}
  \node at (label) {$\pqmtree_{\pi}$};

\path (m1) edge (p1); 
\path (m1) edge (p2); 
\path (m1) edge (p3); 
\path (m2) edge (m02); 
\path (m2) edge (m12); 
\path (m3) edge (m03); 
\path (m3) edge (m13); 
\path (m4) edge (m04); 
\path (m4) edge (m14); 
\path (m5) edge (m05); 
\path (m5) edge (m15); 

\path (m1) edge (s01); 
\path (m1) edge (s11); 
\path (m1) edge (s02); 
\path (m1) edge (s12); 

\path (p1) edge (m2); 
\path (p1) edge (m3); 
\path (p2) edge (m4); 
\path (p3) edge (m5); 

\draw[white] (-2.8,-2.8)--(-2.8,-2);
\draw[white] (2.8,2.8)--(2.8,2);
\end{tikzpicture}
\caption{\label{fig:PQ_tree} To the left: circular order of the slots $\pi$ in some conformal model
of a circular-arc graph $G$ for the case when $V$ is parallel in $\strongModules(G_{ov})$,
$Q_1,Q_2,Q_3,Q_4,Q_5$ are the connected components of $G_{ov}$, 
$\camodules(Q_1)=\{S_1,S_2\}$ and $\camodules(Q_i)=\{Q_i\}$ for $i \in [2,5]$.
Components $Q_4$ and $Q_5$ are separated by any vertex from $Q_1$, components $Q_2$ and $Q_5$ are separated by any vertex from $S_1$.
P-nodes $P_1$, $P_2$, and $P_3$ correspond to 
maximal sets consisting of $\geq 2$ pairwise neighbouring modules: $\{Q_1,Q_2,Q_3\}$,
$\{Q_1,Q_4\}$, and $\{Q_1,Q_5\}$, respectively.
We have $\slots(Q_1 \to P_1) = \{Q^0_2,Q^1_2,Q^0_3,Q^1_3\}$,
$\slots(P_1 \to Q_1) = \{S^0_1,S^1_1, S^0_2, S^1_2, Q^0_4,Q^1_4,Q^0_5,Q^1_5\}$.
We have $\pi_{|Q_1} \equiv S^0_1P_1S^1_2P_2S^1_1S^0_2P_3$ and 
$\pi_{|P_1} \equiv Q_1Q_2Q_3$.
To the right: PQ-tree $\pqmtree^{PQ}_{\pi}$ representing $\pi$.
}
\end{figure}

Before we describe the way in which the PQ-tree $\pqmtree^{PQ}$ represents $\Pi$, we need to list some of its properties.
Let $Q$ be a Q-node and $P$ be a P-node such that $Q$ and $P$ are adjacent in $\pqmtree^{PQ}$.
When we delete the edge $QP$ from $\pqmtree^{PQ}$, we obtain a forest consisting of two trees, $\pqmtree^{PQ}_{Q \to P}$ and $\pqmtree^{PQ}_{P \to Q}$, 
where the node $P$ is in $\pqmtree^{PQ}_{Q \to P}$ and the node $Q$ is in $\pqmtree^{PQ}_{P \to Q}$.
Let $\slots(Q \to P)$ and $\slots(P \to Q)$ denote the sets of slots contained in the trees 
$\pqmtree^{PQ}_{Q \to P}$ and $\pqmtree^{PQ}_{P \to Q}$, respectively.
In \cite{Krawczyk20} the following property of every $\pi \in \Pi$ with respect to $Q$ and $P$ is shown:
\begin{itemize}
 \item the slots from the set $\slots(Q \to P)$ form a contiguous subword in $\pi$,
 \item the slots from the set $\slots(P \to Q)$ form a contiguous subword in $\pi$.
\end{itemize}
See Figure~\ref{fig:PQ_tree}. 
The above properties allows us to denote:
\begin{itemize}
\item for every Q-node $Q$ by $\pi_{|Q}$ the circular word obtained from $\pi$ by substituting the contiguous subword $\slots(Q\to P)$ by the letter $P$, for every P-node $P$ adjacent to $M$,
\item for every P-node $P$ by $\pi_{|P}$ the circular word obtained from $\pi$ by substituting the contiguous subword $\slots(P\to Q)$ by the letter $Q$, for every Q-node $Q$ adjacent to $P$.
\end{itemize} 
See Figure~\ref{fig:PQ_tree}. 
Thus, we have defined the word $\pi_{|N}$ for every inner node $N$ in $\pqmtree^{PQ}$; note that
$\pi_{|N}$ is a circular ordering of the nodes adjacent to $N$ in $\pqmtree^{PQ}$.

The set $\Pi$ is represented by means of the sets of \emph{admissible orderings} $\Pi(N)$ of inner nodes $N$ of~$\pqmtree^{PQ}$,
where each member of $\Pi(N)$ is a circular order of the nodes adjacent to $N$ in~$\pqmtree^{PQ}$.
The sets $\Pi(N)$ are defined so as the following property holds:
\begin{description}
 \item [\namedlabel{prop:pm-tree}{(P4)}] There is a one-to-one correspondence between the set $\Pi$ and the set $\Pi_{\diamond}$, where
 $$\Pi_{\diamond} = \Bigg{\{}\Big{\{}\big{(}N,\pi_N\big{)}: 
 \begin{array}{c}
 \text{$N$ is an inner node in $\pqmtree^{PQ}$ and} \\
 \text{$\pi_N$ is an admissible ordering from $\Pi(N)$}
 \end{array}
\Big{\}}\Bigg{\}},$$
 established by 
 $$\Pi \ni \pi \quad \longrightarrow \quad \Big{\{}\big{(}N, \pi_{|N}\big{)}: N \text{ is an  inner node in } \pqmtree^{PQ}\Big{\}} \in \Pi_{\diamond}.$$
\end{description}

Let $\pi \in \Pi$. 
The set $\Big{\{}\big{(}N, \pi_{|N}\big{)}: N \text{ is a PQ-node in } \pqmtree^{PQ}\Big{\}}$ from $\Pi_{\diamond}$ 
corresponding to~$\pi$ induces a plane drawing $\pqmtree^{PQ}_{\pi}$ of the tree $\pqmtree^{PQ}$
in which for every PQ-node~$N$ the clockwise order of the neighbours of $N$ is given by $\pi_{|N}$.
Note that the circular word obtained by listing all the slots when walking the boundary of $\pqmtree^{PQ}_{\pi}$ in the clockwise order coincides with~$\pi$,
and hence we say that $\pqmtree^{PQ}_{\pi}$ \emph{represents}~$\pi$.

The sets $\Pi(\cdot)$ are defined differently depending on whether the module $V$ is
prime, serial, or parallel in $\strongModules(G_{ov})$. 
In any case, however, the sets $\Pi(\cdot)$ can be read out easily from $\pi(\phi)$. 

\subsubsection{$V$ is serial in $\strongModules(G_{ov})$}
It is shown in \cite{Krawczyk20} that we have $\camodules(M)=\{M\}$ for any child $M$ of $V$ in $\strongModules(G_{ov})$, 
and hence the set $\camodules$ of CA-modules of $G$ consists of the children of $V$ in $\strongModules(G_{ov})$.
Assume that $\camodules = \{S_1,\ldots,S_t\}$.
Since $V$ is serial, $(S_i,{\parallel})$ is connected for every $i \in [t]$ and $S_i \sim S_j$ for every two distinct $i,j \in [t]$.
The set $\Pi(V)$ is defined such that
$$
\Pi(V) =  \left \{ \pi :
\begin{array}{c}
\text{$\pi$ is a circular order of $S^0_1,S^1_1,\ldots,S^0_t,S^1_t$ such that for every } \\
\text{distinct $i,j \in [t]$ the slots corresponding to $S_i$ and $S_j$ overlap} 
\end{array}
\right \}.
$$
The PQ-tree $\pqmtree^{PQ}$ consists of the single \emph{serial} Q-node~$V$ and the 
leaf nodes from $\slots$ adjacent to $V$.
In particular, we have $\Pi = \Pi(V)$.
See Figure~\ref{fig:Pi-serial} for an illustration.

\begin{center}
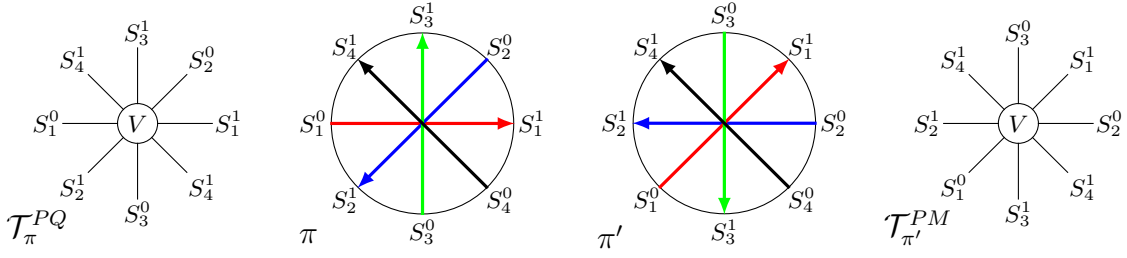
\begin{figure}[ht]
\begin{tikzpicture}[xscale=0.6,yscale=0.6,>=latex,shorten >=-0.4pt,shorten <=-0.4pt]
\coordinate (label) at (-2.2,-2.3) {};
\tikzstyle{every node}=[circle,minimum size=15pt,inner sep=0.5,draw];
\begin{scriptsize}
\node (V) at (0.0,0) {$V$};
\end{scriptsize}
\tikzstyle{every node}=[inner sep=1pt]
\begin{scriptsize}
\node (s1) at ($(center)+(0:2.0cm)$) {$S^1_1$};
\node (s2) at ($(center)+(45:2.0cm)$) {$S^0_2$};
\node (s3) at ($(center)+(90:2.0cm)$) {$S^1_3$};
\node (s4) at ($(center)+(135:2.0cm)$) {$S^1_4$};
\node (s5) at ($(center)+(180:2.0cm)$) {$S^0_1$};
\node (s6) at ($(center)+(225:2.0cm)$) {$S^1_2$};
\node (s7) at ($(center)+(270:2.0cm)$) {$S^0_3$};
\node (s8) at ($(center)+(315:2.0cm)$) {$S^1_4$};
\end{scriptsize}
\node at (label) {$\pqmtree^{PQ}_{\pi}$};
\path (V) edge (s1); 
\path (V) edge (s2); 
\path (V) edge (s3); 
\path (V) edge (s4); 
\path (V) edge (s5); 
\path (V) edge (s6); 
\path (V) edge (s7); 
\path (V) edge (s8); 
\draw[white] (-2.8,-2.8)--(-2.8,-2);
\draw[white] (2.8,2.8)--(2.8,2);
\end{tikzpicture}
\hspace{0.1cm}
\begin{tikzpicture}[yscale=0.6,xscale=0.6,>=latex,shorten >=-0.4pt,shorten <=-0.4pt]
\coordinate (center) at (0,0) {};
\coordinate (label) at (-2.5,-2.5) {};

\coordinate (s1) at ($(center)+(0:2.0cm)$) {};
\coordinate (s2) at ($(center)+(45:2.0cm)$) {};
\coordinate (s3) at ($(center)+(90:2.0cm)$) {};
\coordinate (s4) at ($(center)+(135:2.0cm)$) {};
\coordinate (s5) at ($(center)+(180:2.0cm)$) {};
\coordinate (s6) at ($(center)+(225:2.0cm)$) {};
\coordinate (s7) at ($(center)+(270:2.0cm)$) {};
\coordinate (s8) at ($(center)+(315:2.0cm)$) {};

\coordinate (ls1) at ($(center)+(0:2.4cm)$) {};
\coordinate (ls2) at ($(center)+(45:2.4cm)$) {};
\coordinate (ls3) at ($(center)+(90:2.4cm)$) {};
\coordinate (ls4) at ($(center)+(135:2.4cm)$) {};
\coordinate (ls5) at ($(center)+(180:2.4cm)$) {};
\coordinate (ls6) at ($(center)+(225:2.4cm)$) {};
\coordinate (ls7) at ($(center)+(270:2.4cm)$) {};
\coordinate (ls8) at ($(center)+(315:2.4cm)$) {};

\tikzstyle{every node}=[inner sep=1pt]
\begin{scriptsize}
\node at (ls1) {$S^1_1$};
\node at (ls2) {$S^0_2$};
\node at (ls3) {$S^1_3$};
\node at (ls4) {$S^1_4$};
\node at (ls5) {$S^0_1$};
\node at (ls6) {$S^1_2$};
\node at (ls7) {$S^0_3$};
\node at (ls8) {$S^0_4$};
\end{scriptsize}
\node at (label) {$\pi$};

\draw (0,0) circle (2cm);

\draw[very thick,red,<-] (s1) -- (s5);
\draw[very thick,blue,->] (s2) -- (s6);
\draw[very thick,green,<-] (s3) -- (s7);
\draw[very thick,<-] (s4) -- (s8);

\draw[white] (-2.8,-2.8)--(-2.8,-2);
\draw[white] (2.8,2.8)--(2.8,2);
\end{tikzpicture}
\hspace{0.3cm}
\begin{tikzpicture}[yscale=0.6,xscale=0.6,>=latex,shorten >=-0.4pt,shorten <=-0.4pt]
\coordinate (center) at (0,0) {};
\coordinate (label) at (-2.5,-2.5) {};

\coordinate (s1) at ($(center)+(0:2.0cm)$) {};
\coordinate (s2) at ($(center)+(45:2.0cm)$) {};
\coordinate (s3) at ($(center)+(90:2.0cm)$) {};
\coordinate (s4) at ($(center)+(135:2.0cm)$) {};
\coordinate (s5) at ($(center)+(180:2.0cm)$) {};
\coordinate (s6) at ($(center)+(225:2.0cm)$) {};
\coordinate (s7) at ($(center)+(270:2.0cm)$) {};
\coordinate (s8) at ($(center)+(315:2.0cm)$) {};

\coordinate (ls1) at ($(center)+(0:2.4cm)$) {};
\coordinate (ls2) at ($(center)+(45:2.4cm)$) {};
\coordinate (ls3) at ($(center)+(90:2.4cm)$) {};
\coordinate (ls4) at ($(center)+(135:2.4cm)$) {};
\coordinate (ls5) at ($(center)+(180:2.4cm)$) {};
\coordinate (ls6) at ($(center)+(225:2.4cm)$) {};
\coordinate (ls7) at ($(center)+(270:2.4cm)$) {};
\coordinate (ls8) at ($(center)+(315:2.4cm)$) {};

\tikzstyle{every node}=[inner sep=1pt]
\begin{scriptsize}
\node at (ls1) {$S^0_2$};
\node at (ls2) {$S^1_1$};
\node at (ls3) {$S^0_3$};
\node at (ls4) {$S^1_4$};
\node at (ls5) {$S^1_2$};
\node at (ls6) {$S^0_1$};
\node at (ls7) {$S^1_3$};
\node at (ls8) {$S^0_4$};
\end{scriptsize}
\node at (label) {$\pi'$};

\draw (0,0) circle (2cm);

\draw[very thick,blue,->] (s1) -- (s5);
\draw[very thick,red,<-] (s2) -- (s6);
\draw[very thick,green,->] (s3) -- (s7);
\draw[very thick,<-] (s4) -- (s8);
\draw[white] (-2.8,-2.8)--(-2.8,-2);
\draw[white] (2.8,2.8)--(2.8,2);
\end{tikzpicture}
\hspace{0.1cm}
\begin{tikzpicture}[xscale=0.6,yscale=0.6,>=latex,shorten >=-0.4pt,shorten <=-0.4pt]
\coordinate (label) at (-2.2,-2.3) {};
\tikzstyle{every node}=[circle,minimum size=15pt,inner sep=0.5,draw];
\begin{scriptsize}
\node (V) at (0.0,0) {$V$};
\end{scriptsize}
\tikzstyle{every node}=[inner sep=1pt]
\begin{scriptsize}
\node (s1) at ($(center)+(0:2.0cm)$) {$S^0_2$};
\node (s2) at ($(center)+(45:2.0cm)$) {$S^1_1$};
\node (s3) at ($(center)+(90:2.0cm)$) {$S^0_3$};
\node (s4) at ($(center)+(135:2.0cm)$) {$S^1_4$};
\node (s5) at ($(center)+(180:2.0cm)$) {$S^1_2$};
\node (s6) at ($(center)+(225:2.0cm)$) {$S^0_1$};
\node (s7) at ($(center)+(270:2.0cm)$) {$S^1_3$};
\node (s8) at ($(center)+(315:2.0cm)$) {$S^1_4$};
\end{scriptsize}
\node at (label) {$\pqmtree^{PM}_{\pi'}$};
\path (V) edge (s1); 
\path (V) edge (s2); 
\path (V) edge (s3); 
\path (V) edge (s4); 
\path (V) edge (s5); 
\path (V) edge (s6); 
\path (V) edge (s7); 
\path (V) edge (s8); 
\draw[white] (-2.8,-2.8)--(-2.8,-2);
\draw[white] (2.8,2.8)--(2.8,2);
\end{tikzpicture}
\caption{\label{fig:Pi-serial} Two members $\pi$ and $\pi'$ of the set $\Pi$ and the trees $\pqmtree^{PQ}_{\pi}$ and $\pqmtree^{PQ}_{\pi'}$ representing $\pi$ and $\pi'$ for the case when $V$ is serial in $\strongModules(G_{ov})$: $\pi'$ is obtained from $\pi$ by swapping the red and blue metachords and by flipping the orientation of the green metachord.}
\end{figure}

\end{center}

\subsubsection{$V$ is prime in $\strongModules(G_{ov})$}
Let $M$ be a child of $V$ in $\strongModules(G_{ov})$.
It is shown in \cite{Krawczyk20} that if $M$ is prime, then $\camodules(M)=\{M\}$ and
if $M$ is serial/parallel, then every module from $\camodules(M)$ is the union of some children of $M$ in $\strongModules(G_{ov})$. 
Assume that $\camodules = \{S_1,\ldots,S_t\}$ and $\pi = \pi(\phi)$.
We set
$$\Pi(V) = \{\pi,\pi^R\},$$
where $\pi^R$ is the reflection of $\pi$. 
The PQ-tree $\pqmtree^{PQ}$ consists of the single \emph{prime} Q-node~$V$ and the 
leaf nodes from $\slots$ adjacent to $V$.
We have $\Pi = \Pi(V)$.
See Figure~\ref{fig:Pi-prime} for an illustration.
\begin{center}
\begin{figure}[ht]
\begin{tikzpicture}[xscale=0.6,yscale=0.6,>=latex,shorten >=-0.4pt,shorten <=-0.4pt]
\coordinate (label) at (-2.2,-2.3) {};
\tikzstyle{every node}=[circle,minimum size=15pt,inner sep=0.5,draw];
\begin{scriptsize}
\node (V) at (0.0,0) {$V$};
\end{scriptsize}
\tikzstyle{every node}=[inner sep=1pt]
\begin{scriptsize}
\node (s1) at ($(center)+(0:2.0cm)$) {$S^1_1$};
\node (s2) at ($(center)+(45:2.0cm)$) {$S^1_2$};
\node (s3) at ($(center)+(90:2.0cm)$) {$S^1_3$};
\node (s4) at ($(center)+(135:2.0cm)$) {$S^0_2$};
\node (s5) at ($(center)+(180:2.0cm)$) {$S^1_4$};
\node (s6) at ($(center)+(225:2.0cm)$) {$S^0_1$};
\node (s7) at ($(center)+(270:2.0cm)$) {$S^0_4$};
\node (s8) at ($(center)+(315:2.0cm)$) {$S^1_3$};
\end{scriptsize}
\node at (label) {$\pqmtree^{PQ}_{\pi}$};
\path (V) edge (s1); 
\path (V) edge (s2); 
\path (V) edge (s3); 
\path (V) edge (s4); 
\path (V) edge (s5); 
\path (V) edge (s6); 
\path (V) edge (s7); 
\path (V) edge (s8); 
\draw[white] (-2.8,-2.8)--(-2.8,-2);
\draw[white] (2.8,2.8)--(2.8,2);
\end{tikzpicture}
\hspace{0.1cm}
\begin{tikzpicture}[yscale=0.6,xscale=0.6,>=latex,shorten >=-0.4pt,shorten <=-0.4pt]
\coordinate (center) at (0,0) {};
\coordinate (label) at (-2.2,-2.3) {};

\coordinate (s1) at ($(center)+(30:2.0cm)$) {};
\coordinate (s2) at ($(center)+(90:2.0cm)$) {};
\coordinate (s3) at ($(center)+(150:2.0cm)$) {};
\coordinate (s4) at ($(center)+(180:2.0cm)$) {};
\coordinate (s5) at ($(center)+(220:2.0cm)$) {};
\coordinate (s6) at ($(center)+(260:2.0cm)$) {};
\coordinate (s7) at ($(center)+(290:2.0cm)$) {};
\coordinate (s8) at ($(center)+(360:2.0cm)$) {};

\coordinate (ls1) at ($(center)+(30:2.4cm)$) {};
\coordinate (ls2) at ($(center)+(90:2.4cm)$) {};
\coordinate (ls3) at ($(center)+(150:2.4cm)$) {};
\coordinate (ls4) at ($(center)+(180:2.4cm)$) {};
\coordinate (ls5) at ($(center)+(220:2.4cm)$) {};
\coordinate (ls6) at ($(center)+(260:2.4cm)$) {};
\coordinate (ls7) at ($(center)+(290:2.4cm)$) {};
\coordinate (ls8) at ($(center)+(360:2.4cm)$) {};

\tikzstyle{every node}=[inner sep=1pt]
\begin{scriptsize}
\node at (ls1) {$S^1_2$};
\node at (ls2) {$S^1_3$};
\node at (ls3) {$S^0_2$};
\node at (ls4) {$S^1_4$};
\node at (ls5) {$S^0_1$};
\node at (ls6) {$S^0_4$};
\node at (ls7) {$S^0_3$};
\node at (ls8) {$S^1_1$};
\end{scriptsize}
\node at (label) {$\pi$};

\draw (0,0) circle (2cm);

\draw[very thick,->] (s3) -- (s1);
\draw[very thick,red,->] (s7) -- (s2);
\draw[very thick,->] (s6) -- (s4);
\draw[very thick,->] (s5) -- (s8);

\draw[thick,red] ([shift=(90:2.0cm)]0,0) arc (90:290:2.0cm);

\draw[white] (-2.8,-2.8)--(-2.8,-2);
\draw[white] (2.8,2.8)--(2.8,2);
\end{tikzpicture}
\begin{tikzpicture}[yscale=0.6,xscale=-0.6,>=latex,shorten >=-0.4pt,shorten <=-0.4pt]
\draw[white] (-0.5,-3)--(-0.5,-3);
\draw[white] (0.5,3)--(0.5,3);
\draw[black, dashed] (0,-2.5)--(0,2.5);
\end{tikzpicture}
\begin{tikzpicture}[yscale=0.6,xscale=-0.6,>=latex,shorten >=-0.4pt,shorten <=-0.4pt]
\coordinate (center) at (0,0) {};
\coordinate (label) at (-2.2,-2.3) {};

\coordinate (s1) at ($(center)+(30:2.0cm)$) {};
\coordinate (s2) at ($(center)+(90:2.0cm)$) {};
\coordinate (s3) at ($(center)+(150:2.0cm)$) {};
\coordinate (s4) at ($(center)+(180:2.0cm)$) {};
\coordinate (s5) at ($(center)+(220:2.0cm)$) {};
\coordinate (s6) at ($(center)+(260:2.0cm)$) {};
\coordinate (s7) at ($(center)+(290:2.0cm)$) {};
\coordinate (s8) at ($(center)+(360:2.0cm)$) {};

\coordinate (ls1) at ($(center)+(30:2.4cm)$) {};
\coordinate (ls2) at ($(center)+(90:2.4cm)$) {};
\coordinate (ls3) at ($(center)+(150:2.4cm)$) {};
\coordinate (ls4) at ($(center)+(180:2.4cm)$) {};
\coordinate (ls5) at ($(center)+(220:2.4cm)$) {};
\coordinate (ls6) at ($(center)+(260:2.4cm)$) {};
\coordinate (ls7) at ($(center)+(290:2.4cm)$) {};
\coordinate (ls8) at ($(center)+(360:2.4cm)$) {};

\tikzstyle{every node}=[inner sep=1pt]
\begin{scriptsize}
\node at (ls1) {$S^0_2$};
\node at (ls2) {$S^0_3$};
\node at (ls3) {$S^1_2$};
\node at (ls4) {$S^0_4$};
\node at (ls5) {$S^1_1$};
\node at (ls6) {$S^1_4$};
\node at (ls7) {$S^1_3$};
\node at (ls8) {$S^0_1$};
\end{scriptsize}
\node at (label) {$\pi^R$};

\draw (0,0) circle (2cm);

\draw[very thick,<-] (s3) -- (s1);
\draw[very thick,red,<-] (s7) -- (s2);
\draw[very thick,<-] (s6) -- (s4);
\draw[very thick,<-] (s5) -- (s8);

\draw[thick,red] ([shift=(90:2.0cm)]0,0) arc (90:290:2.0cm);

\draw[white] (-2.8,-2.8)--(-2.8,-2);
\draw[white] (2.8,2.8)--(2.8,2);
\end{tikzpicture}
\begin{tikzpicture}[xscale=-0.6,yscale=0.6,>=latex,shorten >=-0.4pt,shorten <=-0.4pt]
\coordinate (label) at (-2.2,-2.3) {};
\tikzstyle{every node}=[circle,minimum size=15pt,inner sep=0.5,draw];
\begin{scriptsize}
\node (V) at (0.0,0) {$V$};
\end{scriptsize}
\tikzstyle{every node}=[inner sep=1pt]
\begin{scriptsize}
\node (s1) at ($(center)+(0:2.0cm)$) {$S^0_1$};
\node (s2) at ($(center)+(45:2.0cm)$) {$S^0_2$};
\node (s3) at ($(center)+(90:2.0cm)$) {$S^0_3$};
\node (s4) at ($(center)+(135:2.0cm)$) {$S^1_2$};
\node (s5) at ($(center)+(180:2.0cm)$) {$S^0_4$};
\node (s6) at ($(center)+(225:2.0cm)$) {$S^1_1$};
\node (s7) at ($(center)+(270:2.0cm)$) {$S^1_4$};
\node (s8) at ($(center)+(315:2.0cm)$) {$S^0_3$};
\end{scriptsize}
\node at (label) {$\pqmtree^{PQ}_{\pi^R}$};
\path (V) edge (s1); 
\path (V) edge (s2); 
\path (V) edge (s3); 
\path (V) edge (s4); 
\path (V) edge (s5); 
\path (V) edge (s6); 
\path (V) edge (s7); 
\path (V) edge (s8); 
\draw[white] (-2.8,-2.8)--(-2.8,-2);
\draw[white] (2.8,2.8)--(2.8,2);
\end{tikzpicture}
\caption{\label{fig:Pi-prime} Two members $\pi$ and $\pi^R$ of the set $\Pi$ and the trees $\pqmtree^{PQ}_{\pi}$ and $\pqmtree^{PQ}_{\pi^R}$ representing $\pi$ and $\pi^R$ for the case when $V$ is prime in $\strongModules(G_{ov})$: $\pi^R \equiv S^1_2S^0_4S^1_1S^1_4S^1_3S^0_1S^0_2S^0_3$ is the reflection of $\pi \equiv S^1_3S^1_2S^1_1S^0_3S^0_4S^0_1S^1_4S^0_2$.}
\end{figure}

\end{center}

\subsubsection{$V$ is parallel in $\strongModules(G_{ov})$}
In this case the children of $V$ in $\strongModules(G_{ov})$ correspond to the 
connected components of $G_{ov}$, which in turn correspond to Q-nodes in $\pqmtree^{PQ}$.
Also, we recall that $\camodules(Q)$ forms a partition of $Q$, for every Q-node $Q$ of $\pqmtree^{PQ}$.

Let $\pi = \pi(\phi)$ be the circular order of the slots in $\phi$.
For every PQ-node $N$ in $\pqmtree^{PQ}$ we define the set $\Pi(N)$, where:
\begin{itemize}
 \item for a P-node $P$ the set $\Pi(P)$ contains all circular orders of the M-nodes adjacent to $P$,
 \item for a Q-node $Q$ the set $\Pi(Q)$ contains two members: $\pi_{|Q}$ and its reflection $(\pi_{|Q})^R$.
\end{itemize}
We note that when $\camodules(Q)=\{Q\}$ then we have $\pi_{|Q} \equiv (\pi_{|Q})^R$.
Also, note that any circular order of the slots in $\Pi$ can be obtained by performing a sequence of the operations on the tree $\pqmtree^{PQ}_{\pi}$, where each of them either
\begin{itemize}
 \item \emph{reflects a Q-node}, or
 \item \emph{permutes arbitrarily} the neighbours of \emph{a P-node}.
\end{itemize}
The reflection of a Q-node $Q$ transforms $\pi$ into $\pi'$, 
where the only difference between the trees $\pqmtree^{PQ}_{\pi}$ and $\pqmtree^{PQ}_{\pi'}$ is on the node $M$:
$\pqmtree^{PQ}_{\pi}$ orders the neighbours of $Q$ consistently with $\pi_{|Q}$ and
$\pqmtree^{PQ}_{\pi'}$ orders the neighbours of $Q$ consistently with the reflection of $\pi_{|Q}$ (that is, we have $\pi'_{|Q} \equiv (\pi_{|Q})^R$).
See Figure~\ref{fig:Q-node-reflection} for an illustration.
\input ./figures/data_structure/parallel/Q_node_reflection.tex

A permutation of a P-node $P$ transforms $\pi$ into $\pi'$, 
where the only difference between the trees $\pqmtree^{PQ}_{\pi}$ and $\pqmtree^{PQ}_{\pi'}$ is on the node $P$:
$\pqmtree^{PQ}_{\pi}$ orders the neighbours of $Q$ consistently with $\pi_{|P}$ and 
$\pqmtree^{PQ}_{\pi'}$ orders the neighbours of $Q$ consistently with $\pi'_{|P}$, where
$\pi'_{|P}$ is any circular ordering of the neighbours of $P$.
See Figure~\ref{fig:P-node-permutation} for an illustration.
\input ./figures/data_structure/parallel/P_node_permutation.tex

\subsection{The structure of the conformal models of $G$}
In order to describe the structure of conformal models of $G$ in the way we will use it in this paper, we extend the tree $\pqmtree^{PQ}$ by attaching, for every Q-node $Q$ and every $S \in \camodules(Q)$, the root $S$ of the modular decomposition tree $\pqmtree_{S}$ to the node $Q$.
This way we obtain the PQM-tree $\mathcal{T}$.
Finally, to simplify our notation, for every Q-node $Q$ we denote by $\pqmtree_{Q}$ the subtree of $\pqmtree$ rooted in $Q$ restricted to the node $Q$ and the nodes in the trees $\pqmtree_S$ for $S \in \camodules(Q)$.

Let $\phi$ be any conformal model of $G$.
Now, for every PQM-node $N$ (that is, for $N$ which is either P-node, Q-node, or M-node in $\pqmtree$)
by $\phi_{|N}$ we denote an \emph{admissible ordering of $N$ induced by the model $\phi$}, defined such that:
\begin{itemize}
 \item $\phi_{|N} = \pi(\phi)_{|N}$ if $N$ is a PQ-node in $\pqmtree$,
 \item $\phi_{|N} = (\phi|S^0,\phi|S^1)_{|N}$ if $N$ is an M-node in $\pqmtree$.
\end{itemize}
Finally, we can summarize this section with the following theorem.
\begin{theorem}
There is a bijection between the set $\Phi$ of conformal models of $G$ and the set $\Phi_\diamond$, where 
$$\Phi_{\diamond} = \Bigg{\{} \Big{\{}\big{(}N,\pi_N\big{)}: 
\begin{array}{c}
\text{$N$ is an PQM-node in $\pqmtree$ and} \\
\text{$\pi_N$ is an admissible ordering from $\Pi(N)$} \\
\end{array}
 \Big{\}} \Bigg{\}},
$$
established by 
$$\Phi \ni \phi \quad \longrightarrow \quad  \Big{\{}\big{(}N,{\phi}_{|N}\big{)}: N \text{ is a PQM-node in } \pqmtree \Big{\}} \in \Phi_\diamond.$$ 
\end{theorem}

\subsection{Conformal models and the Helly property: notation and basic properties}

Let $G=(V,E)$ be a circular-arc graph and let $\pqmtree$ be the PQM-tree of $G$.

Let $M$ be an M-node in $\pqmtree$.
We say a vertex $a \in M$ is \emph{$(M^j,M^{1-j})$-oriented} in the metachord $\MMM$
if $a^0 \in M^j$ and $a^1 \in M^{1-j}$.
Clearly, if $a$ is $(M^j,M^{1-j})$-oriented in $\MMM$, then $a$ is $(S^j,S^{1-j})$-oriented in $\SSS$, where $S$ is a CA-module containing $M$.

Let $C_1,\ldots,C_k$ be some fixed cliques of $G$.
A circular word $\phi$ over the letters $V^* \cup \{C_1,\ldots,C_k\}$
is a \emph{$\{C_1,\ldots,C_k\}$-conformal model of $G$} if:
\begin{description}
 \item [\namedlabel{prop:conformal-model-1}{(H1)}] $\phi|V^*$ is a conformal model of $G$.
 \item [\namedlabel{prop:conformal-model-2}{(H2)}] For every $i \in [k]$ and every $a \in C_i$ 
 the point $\phi(C_i)$ lies on the left side of the chord $\phi(a)$.
\end{description}
Clearly, every conformal model of $G$ in which the cliques $C_1,\ldots,C_k$ satisfy the Helly property might be extended to a $\{C_1,\ldots,C_k\}$-conformal model.
We abuse slightly our notation, and for a $\{C_1,\ldots,C_k\}$-conformal model of $G$, an M-node $M$, and $j \in \{0,1\}$, by $\phi|M^j$ we denote the shortest contiguous subword of $\phi$ containing all the letters from $M^j$ and no letter from the set $M^{1-j}$.

Let $M$ be an M-node in $\pqmtree$ and let $C_i$ be a clique from $\{C_1,\ldots,C_k\}$.
We say~$C_i$ is \emph{private} for $M$ (or $M$ is an \emph{owner} of $C_i$) 
if there are two distinct vertices in $C_i \cap M$ with different orientations in $\MMM$.
We say~$C_i$ is \emph{private} if $C_i$ is private for some M-node in $\pqmtree$; otherwise we say~$C_i$ is \emph{public}.

\begin{lemma}
\label{lemma:private-nodes}
Let $M$ be an M-node in $\pqmtree$ and let $C_i \in \{C_1,\ldots,C_k\}$ be such that $C_i$ is private for $M$.
Then:
\begin{enumerate}
\item \label{item:private-nodes-model} For every $\{C_1,\ldots,C_k\}$-conformal model $\phi$ of $G$ the letter $C_i$ occurs either in $\phi|M^0$ or in $\phi|M^1$.
\item \label{item:private-nodes-path} Suppose $M \subseteq S$ for some CA-module $S$. 
Then the clique $C_i$ is private for every M-node $M'$ lying on the path from $M$ to $S$ in $\pqmtree_{S}$.
In particular, the letter $C_i$ occurs either in the slot $\phi|S^0$ 
or in the slot $\phi|S^1$ in any $\{C_1,\ldots,C_k\}$-conformal model $\phi$ of~$G$.
\end{enumerate}
\end{lemma}
\begin{proof}

First we prove~\eqref{item:private-nodes-model}.
Since $C_i$ is private for $M$, there are $a,b \in C_i\cap M$ such that $a$ is $(M^0,M^1)$-oriented and $b$ is $(M^1,M^0)$-oriented in $\MMM$.
It means that $a^0,b^1 \in M^0$ and $a^1,b^0 \in M^1$.
Since $\phi(C_i)$ is to the left of $\phi(a)$ and $\phi(b)$, 
$C_i$ must be contained either between $a^0$ and $b^1$ or between $b^0$ and $a^1$ in $\phi$.
In the first case $C_i$ is contained in $\phi|M^0$ and in the second case $C_i$ is contained in $\phi|M^1$.
Statement \eqref{item:private-nodes-path} follows from the fact that $M \subseteq M' \subseteq S$.
\end{proof}

Clearly, every $\{C_1,\ldots,C_k\}$-conformal model of $G$ can be turned into $\{C_1,\ldots,C_k\}$-conformal model $\phi$ of $G$ such that:
\begin{description}
 \item [\namedlabel{prop:conformal-model-3}{(H3)}] For every $i \in [k]$ and every M-node $M$ of $\pqmtree$ such that $C_i$ is not private for $M$ the letter $C_i$ does not occur in the words $\phi|M^0$ and $\phi|M^1$.
\end{description}
Indeed, if $C_i$ is not private for $M$, then all the vertices from $C_i \cap M$ have the 
same orientation in $\MMM$, and we can simply push out the clique letter $C_i$ from 
$\phi|M^j$ whenever $C_i$ is contained in $\phi|M^j$, for any $j \in \{0,1\}$.
In particular, if $C_i$ is public and~\ref{prop:conformal-model-3} holds, then $C_i$ occurs outside the slot $\phi|S^j$ for every $j \in \{0,1\}$ and every $S \in \camodules$.
In the rest of the paper we assume $\{C_1,\ldots,C_k\}$-conformal models of $G$
satisfy additionally the property~\ref{prop:conformal-model-3}.

\section{An alternative proof of the Lin-Szwarcfiter theorem}
\label{sec:Lin-Szwarcfiter}
In \cite{LinSchw06} Lin and Szwarcfiter have proved the following:
\begin{theorem}
\label{thm:Lin-Szwarcfiter}
Let $G$ be a circular-arc graph.
Either every normalized model of $G$ satisfies the Helly property or every normalized model of $G$ does not satisfy the Helly property.
\end{theorem}

Let $G$ be a circular-arc graph, $G_{ov}=(V,{\sim})$ be the overlap graph of $G$, and 
$\pqmtree$ be the PQM-tree representing the conformal models of $G$.

\begin{definition}
Let $C$ be a clique in $G$ of size $k \geq 3$ and let $\phi$ be a conformal model of~$G$.
We say $C$ \emph{forms a non-Helly structure} in $\phi$ if 
$$\phi|C \equiv v_0^0v_2^1v_1^0v_3^1 \ldots v_{k-2}^0v_0^1v_{k-1}^0v_1^1$$
for some circular order $v_0,\ldots,v_{k-1}$ of the vertices of $C$.
\end{definition}
See Figure~\ref{fig:forb-structures} for an illustration.
Note that whenever $C$ forms a non-Helly clique in $\phi$ and $\phi|C \equiv v_0^0v_2^1v_1^0v_3^1 \ldots v_{k-2}^0v_0^1v_{k-1}^0v_1^1$ for a circular order $v_0,\ldots,v_{k-1}$ of the vertices of $C$, then for every two distinct $i,j \in [0,k-1]$:
\begin{itemize}
\item $v_{i} \sim v_{j}$ if $(i-j)_{mod\ k} \in \{1,k-1\}$,
\item $v_i,v_j$ cover the circle if $(i-j)_{mod\ k} \notin \{1,k-1\}$.
\end{itemize}
In particular, $C$ induces a cycle in $G_{ov}$ with the consecutive vertices $v_0,\ldots,v_{k-1}$.

\begin{figure}[htp!]
\begin{tikzpicture}[scale=0.5,>=latex,shorten >=-0.4pt,shorten <=-0.4pt]
\coordinate (center) at (0,0) {};
\coordinate (label) at (0,-2.5) {};

\coordinate (s1) at ($(center)+(0:2.0cm)$) {};
\coordinate (s2) at ($(center)+(60:2.0cm)$) {};
\coordinate (s3) at ($(center)+(120:2.0cm)$) {};
\coordinate (s4) at ($(center)+(180:2.0cm)$) {};
\coordinate (s5) at ($(center)+(240:2.0cm)$) {};
\coordinate (s6) at ($(center)+(300:2.0cm)$) {};

\begin{scriptsize}
\node at ($(center)+(0:2.5cm)$) {$v_0^0$};
\node at ($(center)+(60:2.5cm)$) {$v_1^1$};
\node at ($(center)+(120:2.5cm)$) {$v_2^0$};
\node at ($(center)+(180:2.5cm)$) {$v_0^1$};
\node at ($(center)+(240:2.5cm)$) {$v_1^0$};
\node at ($(center)+(300:2.5cm)$) {$v_2^1$};
\end{scriptsize}

\draw (0,0) circle (2cm);

\draw[thick,red,->] (s1) -- (s4);
\draw[thick,red,->] (s5) -- (s2);
\draw[thick,red,->] (s3) -- (s6);

\draw[white] (-3,-3)--(-3,-2);
\draw[white] (3,3)--(3,2);
\end{tikzpicture}
\hspace{0.1cm}
\begin{tikzpicture}[scale=0.5,>=latex,shorten >=-0.4pt,shorten <=-0.4pt]
\coordinate (center) at (0,0) {};
\coordinate (label) at (0,-2.5) {};

\coordinate (s1) at ($(center)+(30:2.0cm)$) {};
\coordinate (s2) at ($(center)+(60:2.0cm)$) {};
\coordinate (s3) at ($(center)+(120:2.0cm)$) {};
\coordinate (s4) at ($(center)+(150:2.0cm)$) {};
\coordinate (s5) at ($(center)+(210:2.0cm)$) {};
\coordinate (s6) at ($(center)+(240:2.0cm)$) {};
\coordinate (s7) at ($(center)+(300:2.0cm)$) {};
\coordinate (s8) at ($(center)+(330:2.0cm)$) {};

\begin{scriptsize}
\node at ($(center)+(30:2.5cm)$) {$v_0^0$};
\node at ($(center)+(60:2.5cm)$) {$v_1^1$};
\node at ($(center)+(120:2.5cm)$) {$v_3^0$};
\node at ($(center)+(150:2.5cm)$) {$v_0^1$};
\node at ($(center)+(210:2.5cm)$) {$v_2^0$};
\node at ($(center)+(240:2.5cm)$) {$v_3^1$};
\node at ($(center)+(300:2.5cm)$) {$v_1^0$};
\node at ($(center)+(330:2.5cm)$) {$v_2^1$};
\end{scriptsize}

\draw (0,0) circle (2cm);

\draw[thick,red,->] (s1) -- (s4);
\draw[thick,red,->] (s7) -- (s2);
\draw[thick,red,->] (s3) -- (s6);
\draw[thick,red,->] (s5) -- (s8);

\draw[white] (-3,-3)--(-3,-2);
\draw[white] (3,3)--(3,2);
\end{tikzpicture}
\hspace{0.1cm}
\begin{tikzpicture}[scale=0.5,>=latex,shorten >=-0.4pt,shorten <=-0.4pt]
\coordinate (center) at (0,0) {};
\coordinate (label) at (0,-2.5) {};

\coordinate (s1) at ($(center)+(30:2.0cm)$) {};
\coordinate (s2) at ($(center)+(60:2.0cm)$) {};
\coordinate (s3) at ($(center)+(102:2.0cm)$) {};
\coordinate (s4) at ($(center)+(132:2.0cm)$) {};
\coordinate (s5) at ($(center)+(174:2.0cm)$) {};
\coordinate (s6) at ($(center)+(204:2.0cm)$) {};
\coordinate (s7) at ($(center)+(246:2.0cm)$) {};
\coordinate (s8) at ($(center)+(276:2.0cm)$) {};
\coordinate (s9) at ($(center)+(318:2.0cm)$) {};
\coordinate (s10) at ($(center)+(348:2.0cm)$) {};

\begin{scriptsize}
\node at ($(center)+(30:2.5cm)$) {$v_0^0$};
\node at ($(center)+(60:2.5cm)$) {$v_1^1$};
\node at ($(center)+(102:2.5cm)$) {$v_4^0$};
\node at ($(center)+(132:2.5cm)$) {$v_0^1$};
\node at ($(center)+(174:2.5cm)$) {$v_3^0$};
\node at ($(center)+(204:2.5cm)$) {$v_4^1$};
\node at ($(center)+(246:2.5cm)$) {$v_2^0$};
\node at ($(center)+(276:2.5cm)$) {$v_3^1$};
\node at ($(center)+(318:2.5cm)$) {$v_1^0$};
\node at ($(center)+(348:2.5cm)$) {$v_2^1$};
\end{scriptsize}

\draw (0,0) circle (2cm);

\draw[thick,red,->] (s1) -- (s4);
\draw[thick,red,->] (s7) -- (s10);
\draw[thick,red,->] (s3) -- (s6);
\draw[thick,red,->] (s5) -- (s8);
\draw[thick,red,->] (s9) -- (s2);

\draw[white] (-3,-3)--(-3,-2);
\draw[white] (3,3)--(3,2);
\end{tikzpicture}
\caption{\label{fig:forb-structures} 
Non-Helly structures for $k=3,4,5$}
\end{figure}
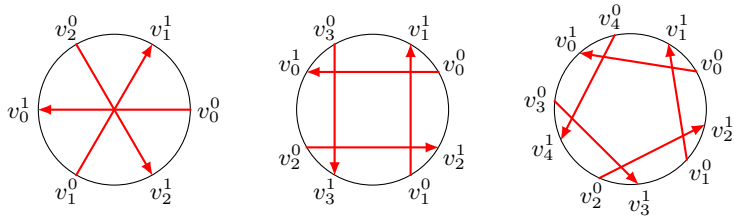

The proof of Theorem~\ref{thm:Lin-Szwarcfiter} by Lin and Szwarcfiter is based on the following:
\begin{lemma}[\cite{LinSchw06}]
\label{lem:non-Helly-clique-forces-forbidden-structure}
Let $\phi$ be a conformal model of $G$ and let $C$ be a clique in~$G$.
Then $C$ is minimal with respect to inclusion non-Helly clique in $\phi$ if and only if $C$ forms a non-Helly structure in $\phi$.
\end{lemma}
In particular, if $\phi$ does not satisfy the Helly property, there is a clique $C$ in $G$ which is non-Helly in $\phi$, and then every minimal non-Helly subclique of $C$ forms a non-Helly structure in $G$.

\begin{definition}
Let $C$ be a clique in $G$ of size $k \geq 3$.
We say $C$ \emph{is a rigid non-Helly clique} in $G$ if $C$ forms a non-Helly
structure in every conformal model of $G$.
\end{definition}

In the next two claims we describe the cliques in $G$ that form rigid non-Helly cliques in~$G$.
In the course of the proof of Theorem~\ref{thm:Lin-Szwarcfiter} we will show that $G$ contains no other rigid non-Helly cliques.
\begin{claim}
\label{claim:large_rigid_forbidden_structure}
Let $C$ be a clique of $G$ that forms a non-Helly structure in a conformal model~$\phi$ of $G$. 
If $|C| \geq 4$ then $C$ is a rigid non-Helly clique in~$G$.
\end{claim}
\begin{proof}
Since $C$ forms a non-Helly structure in $\phi$, $C$ induces a cycle in $G_{ov}$.
Since there are two possibilities to represent a cycle of size $\geq 4$ as an intersection graph of chords (one is the reflection of the other),
we easily deduce that either the circular word $\phi|C$ or
its reflection must occur in any conformal model of $G$.
\end{proof}

\begin{claim}
\label{claim:small_rigid_forbidden_structure}
Let $C$ be a clique of $G$ which forms a non-Helly structure in 
a conformal model~$\phi$ of $G$.
Let $Q$ be a component of $G_{ov}$ such that $C \subseteq Q$.
If $|C|=3$, $C \subseteq N$ for a prime node $N \in \pqmtree_Q$, and the elements of $C$ are contained in three distinct children of $N$ in $\pqmtree_Q$, 
then $C$ is a rigid non-Helly clique in $G$.
\end{claim}
\begin{proof}
Let $\phi$ be a conformal model of $G$.
Suppose $C=\{v_0,v_1,v_2\}$ and $\phi|C \equiv v^0_0v^1_1v^0_2v^1_0v^0_1v^1_2$.
Let $\pi = \phi_{|N} \in \Pi(N)$ be an ordering of $N$ in $\phi$ and let $\gamma \in \Pi(N)$ be the reflection of $\pi$.
Since $\Pi(N) = \{\pi,\gamma\}$ and $C$ has a non-empty intersection with three children of $N$, for every conformal model $\phi'$ of $G$ we have: 
\begin{itemize}
\item $\phi'|C \equiv v^0_0v^1_1v^0_2v^1_0v^0_1v^1_2$ if ${\phi'}_{|N}= \pi$,
\item $\phi'|C \equiv v^0_2v^1_1v^0_0v^1_2v^0_1v^1_0$ if ${\phi'}_{|N}= \gamma$.
\end{itemize}
\end{proof}
If $C = \{v_0,v_1,v_2\}$, $N$, and $\pi,\gamma \in \Pi(N)$ are as in the previous claim,  then $\pi$ and $\gamma$ are said to \emph{force non-Helly structures} $v^0_0v^1_1v^0_2v^1_0v^0_1v^1_2$ and $v^0_2v^1_1v^0_0v^1_2v^0_1v^1_0$, respectively.
See Figure~\ref{fig:small-rigid-non-Helly} for an illustration.
\begin{figure}[ht]
\begin{tikzpicture}[xscale=0.6,yscale=0.9,>=latex]
\coordinate (u0) at (0,2) {};
\coordinate (u1) at (1,2) {};
\coordinate (u2) at (2,2) {};
\coordinate (u3) at (3,2) {};
\coordinate (u4) at (4,2) {};

\coordinate (b0) at (0,0) {};
\coordinate (b1) at (1,0) {};
\coordinate (b2) at (2,0) {};
\coordinate (b3) at (3,0) {};
\coordinate (b4) at (4,0) {};

\coordinate (lu0) at (0,2.35) {};
\coordinate (lu1) at (1,2.35) {};
\coordinate (lu2) at (2,2.35) {};
\coordinate (lu3) at (3,2.35) {};
\coordinate (lu4) at (4,2.35) {};

\coordinate (lb0) at (0,-0.35) {};
\coordinate (lb1) at (1,-0.35) {};
\coordinate (lb2) at (2,-0.35) {};
\coordinate (lb3) at (3,-0.35) {};
\coordinate (lb4) at (4,-0.35) {};

\coordinate (tau0) at (4.9,2);
\coordinate (tau1) at (-0.9,0);

\tikzstyle{every node}=[inner sep=2pt,fill=white]

\draw[fill=gray!30, draw=none] ($(b0) + (-0.2,0)$) -- ($(b0) + (0.2,0)$) -- ($(u1) + (0.2,0)$) -- ($(u1) + (-0.2,0)$) -- cycle;
\draw[fill=gray!30, draw=none] ($(b1) + (-0.2,0)$) -- ($(b1) + (0.2,0)$) -- ($(u4) + (0.2,0)$) -- ($(u4) + (-0.2,0)$) -- cycle;
\draw[fill=gray!30, draw=none] ($(b2) + (-0.2,0)$) -- ($(b2) + (0.2,0)$) -- ($(u2) + (0.2,0)$) -- ($(u2) + (-0.2,0)$) -- cycle;
\draw[fill=gray!30, draw=none] ($(b3) + (-0.2,0)$) -- ($(b3) + (0.2,0)$) -- ($(u0) + (0.2,0)$) -- ($(u0) + (-0.2,0)$) -- cycle;
\draw[fill=gray!30, draw=none] ($(b4) + (-0.2,0)$) -- ($(b4) + (0.2,0)$) -- ($(u3) + (0.2,0)$) -- ($(u3) + (-0.2,0)$) -- cycle;

\draw[->] (-0.5,2) -- (4.5,2);
\draw[<-] (-0.5,0) -- (4.5,0);

\draw[->,thick] (u0) -- (b3);
\draw[<-,thick] (u2) -- (b2);
\draw[->,thick] (u4) -- (b1);

\tikzstyle{every node}=[inner sep=1pt]
\begin{footnotesize}
\node at (tau0) {$\pi^0$};
\node at (tau1) {$\pi^1$};
\end{footnotesize}

\begin{tiny}
\node at (lu0) {$v^0_0$};
\node at (lu2) {$v^1_1$};
\node at (lu4) {$v^0_2$};

\node at (lb1) {$v^1_2$};
\node at (lb2) {$v^0_1$};
\node at (lb3) {$v^1_0$};
\end{tiny}
\end{tikzpicture}
\hspace{0.05cm}
\begin{tikzpicture}[xscale=0.8,yscale=0.6,>=latex]
\draw[-,dashed] (0,-2.5) -- (0,2.5);
\end{tikzpicture}
\hspace{0.05cm}
\begin{tikzpicture}[xscale=-0.6,yscale=0.9,>=latex]
\coordinate (u0) at (0,2) {};
\coordinate (u1) at (1,2) {};
\coordinate (u2) at (2,2) {};
\coordinate (u3) at (3,2) {};
\coordinate (u4) at (4,2) {};

\coordinate (b0) at (0,0) {};
\coordinate (b1) at (1,0) {};
\coordinate (b2) at (2,0) {};
\coordinate (b3) at (3,0) {};
\coordinate (b4) at (4,0) {};

\coordinate (lu0) at (0,2.35) {};
\coordinate (lu1) at (1,2.35) {};
\coordinate (lu2) at (2,2.35) {};
\coordinate (lu3) at (3,2.35) {};
\coordinate (lu4) at (4,2.35) {};

\coordinate (lb0) at (0,-0.35) {};
\coordinate (lb1) at (1,-0.35) {};
\coordinate (lb2) at (2,-0.35) {};
\coordinate (lb3) at (3,-0.35) {};
\coordinate (lb4) at (4,-0.35) {};

\coordinate (tau0) at (4.9,0);
\coordinate (tau1) at (-0.9,2);

\tikzstyle{every node}=[inner sep=2pt,fill=white]

\draw[fill=gray!30, draw=none] ($(b0) + (-0.2,0)$) -- ($(b0) + (0.2,0)$) -- ($(u1) + (0.2,0)$) -- ($(u1) + (-0.2,0)$) -- cycle;
\draw[fill=gray!30, draw=none] ($(b1) + (-0.2,0)$) -- ($(b1) + (0.2,0)$) -- ($(u4) + (0.2,0)$) -- ($(u4) + (-0.2,0)$) -- cycle;
\draw[fill=gray!30, draw=none] ($(b2) + (-0.2,0)$) -- ($(b2) + (0.2,0)$) -- ($(u2) + (0.2,0)$) -- ($(u2) + (-0.2,0)$) -- cycle;
\draw[fill=gray!30, draw=none] ($(b3) + (-0.2,0)$) -- ($(b3) + (0.2,0)$) -- ($(u0) + (0.2,0)$) -- ($(u0) + (-0.2,0)$) -- cycle;
\draw[fill=gray!30, draw=none] ($(b4) + (-0.2,0)$) -- ($(b4) + (0.2,0)$) -- ($(u3) + (0.2,0)$) -- ($(u3) + (-0.2,0)$) -- cycle;

\draw[<-] (-0.5,2) -- (4.5,2);
\draw[->] (-0.5,0) -- (4.5,0);

\draw[<-,thick] (u0) -- (b3);
\draw[->,thick] (u2) -- (b2);
\draw[<-,thick] (u4) -- (b1);

\tikzstyle{every node}=[inner sep=1pt]
\begin{footnotesize}
\node at (tau0) {$\gamma^0$};
\node at (tau1) {$\gamma^1$};
\end{footnotesize}

\begin{tiny}
\node at (lu0) {$v^1_0$};
\node at (lu2) {$v^0_1$};
\node at (lu4) {$v^1_2$};

\node at (lb1) {$v^0_2$};
\node at (lb2) {$v^1_1$};
\node at (lb3) {$v^0_0$};
\end{tiny}
\end{tikzpicture}
\hspace{0.7cm}
\begin{tikzpicture}[yscale=0.5,xscale=0.5,>=latex]
\coordinate (center) at (0,0) {};

\coordinate (v0) at ($(center)+(90:2cm)$) {};
\coordinate (v1) at ($(center)+(30:2cm)$) {};
\coordinate (v2) at ($(center)+(330:2cm)$) {};
\coordinate (v3) at ($(center)+(270:2cm)$) {};
\coordinate (v4) at ($(center)+(210:2cm)$) {};
\coordinate (v5) at ($(center)+(150:2cm)$) {};

\coordinate (lv0) at ($(center)+(90:2.5cm)$) {};
\coordinate (lv1) at ($(center)+(30:2.5cm)$) {};
\coordinate (lv2) at ($(center)+(330:2.5cm)$) {};
\coordinate (lv3) at ($(center)+(270:2.5cm)$) {};
\coordinate (lv4) at ($(center)+(210:2.5cm)$) {};
\coordinate (lv5) at ($(center)+(150:2.5cm)$) {};

\coordinate (tau) at (-2.5,-2.5) {};

\draw[fill=gray!30, draw=none] ($(center)+(95:2cm)$) -- ($(center)+(90:2cm)$)-- ($(center)+(85:2cm)$) -- ($(center)+(275:2cm)$) -- ($(center)+(270:2cm)$) --($(center)+(265:2cm)$) -- cycle;

\draw[fill=gray!30, draw=none] ($(center)+(45:2cm)$) -- ($(center)+(50:2cm)$)-- ($(center)+(55:2cm)$) -- ($(center)+(125:2cm)$) -- ($(center)+(130:2cm)$) --($(center)+(135:2cm)$) -- cycle;

\draw[fill=gray!30, draw=none] ($(center)+(170:2cm)$) -- ($(center)+(175:2cm)$)-- ($(center)+(180:2cm)$) -- ($(center)+(235:2cm)$) -- ($(center)+(240:2cm)$) --($(center)+(245:2cm)$) -- cycle;

\draw[fill=gray!30, draw=none] ($(center)+(35:2cm)$) -- ($(center)+(30:2cm)$)-- ($(center)+(25:2cm)$) -- ($(center)+(215:2cm)$) -- ($(center)+(210:2cm)$) --($(center)+(205:2cm)$) -- cycle;

\draw[fill=gray!30, draw=none] ($(center)+(335:2cm)$) -- ($(center)+(330:2cm)$)-- ($(center)+(325:2cm)$) -- ($(center)+(155:2cm)$) -- ($(center)+(150:2cm)$) --($(center)+(145:2cm)$) -- cycle;

\draw (0,0) circle (2cm);

\draw[thick,->] (v0)--(v3);
\draw[thick,<-] (v1)--(v4);
\draw[thick,->] (v2)--(v5);

\tikzstyle{every node}=[inner sep=1pt]
\begin{footnotesize}
\node at (tau) {$\pi$};
\end{footnotesize}
\begin{tiny}
\node at (lv0) {$v^0_0$};
\node at (lv1) {$v^1_1$};
\node at (lv2) {$v^0_2$};
\node at (lv3) {$v^1_0$};
\node at (lv4) {$v^0_1$};
\node at (lv5) {$v^1_2$};
\end{tiny}
\draw[white,-] (-2.8,-2.8)--(-2.2,-2.8);
\draw[white,-] (2.8,2.8)--(2.2,2.8);
\end{tikzpicture}
\hspace{0.05cm}
\begin{tikzpicture}[xscale=0.8,yscale=0.6,>=latex]
\draw[-,dashed] (0,-2.5) -- (0,2.5);
\end{tikzpicture}
\hspace{0.05cm}
\begin{tikzpicture}[yscale=0.5,xscale=-0.5,>=latex]
\coordinate (center) at (0,0) {};

\coordinate (v0) at ($(center)+(90:2cm)$) {};
\coordinate (v1) at ($(center)+(30:2cm)$) {};
\coordinate (v2) at ($(center)+(330:2cm)$) {};
\coordinate (v3) at ($(center)+(270:2cm)$) {};
\coordinate (v4) at ($(center)+(210:2cm)$) {};
\coordinate (v5) at ($(center)+(150:2cm)$) {};

\coordinate (lv0) at ($(center)+(90:2.5cm)$) {};
\coordinate (lv1) at ($(center)+(30:2.5cm)$) {};
\coordinate (lv2) at ($(center)+(330:2.5cm)$) {};
\coordinate (lv3) at ($(center)+(270:2.5cm)$) {};
\coordinate (lv4) at ($(center)+(210:2.5cm)$) {};
\coordinate (lv5) at ($(center)+(150:2.5cm)$) {};

\coordinate (tau) at (2.5,-2.5) {};

\draw[fill=gray!30, draw=none] ($(center)+(95:2cm)$) -- ($(center)+(90:2cm)$)-- ($(center)+(85:2cm)$) -- ($(center)+(275:2cm)$) -- ($(center)+(270:2cm)$) --($(center)+(265:2cm)$) -- cycle;

\draw[fill=gray!30, draw=none] ($(center)+(45:2cm)$) -- ($(center)+(50:2cm)$)-- ($(center)+(55:2cm)$) -- ($(center)+(125:2cm)$) -- ($(center)+(130:2cm)$) --($(center)+(135:2cm)$) -- cycle;

\draw[fill=gray!30, draw=none] ($(center)+(170:2cm)$) -- ($(center)+(175:2cm)$)-- ($(center)+(180:2cm)$) -- ($(center)+(235:2cm)$) -- ($(center)+(240:2cm)$) --($(center)+(245:2cm)$) -- cycle;

\draw[fill=gray!30, draw=none] ($(center)+(35:2cm)$) -- ($(center)+(30:2cm)$)-- ($(center)+(25:2cm)$) -- ($(center)+(215:2cm)$) -- ($(center)+(210:2cm)$) --($(center)+(205:2cm)$) -- cycle;

\draw[fill=gray!30, draw=none] ($(center)+(335:2cm)$) -- ($(center)+(330:2cm)$)-- ($(center)+(325:2cm)$) -- ($(center)+(155:2cm)$) -- ($(center)+(150:2cm)$) --($(center)+(145:2cm)$) -- cycle;

\draw (0,0) circle (2cm);

\draw[thick,<-] (v0)--(v3);
\draw[thick,->] (v1)--(v4);
\draw[thick,<-] (v2)--(v5);

\tikzstyle{every node}=[inner sep=1pt]
\begin{footnotesize}
\node at (tau) {$\gamma$};
\end{footnotesize}
\begin{tiny}
\node at (lv0) {$v^1_0$};
\node at (lv1) {$v^0_1$};
\node at (lv2) {$v^1_2$};
\node at (lv3) {$v^0_0$};
\node at (lv4) {$v^1_1$};
\node at (lv5) {$v^0_2$};
\end{tiny}
\draw[white,-] (-2.8,-2.8)--(-2.2,-2.8);
\draw[white,-] (2.8,2.8)--(2.2,2.8);
\end{tikzpicture}

\caption{\label{fig:small-rigid-non-Helly} 
Orderings $\pi$ and $\gamma$ forcing non-Helly structures $v^0_0v^1_1v^0_2v^1_0v^0_1v^1_2$ and $v^0_2v^1_1v^0_0v^1_2v^0_1v^1_0$, respectively, for the cases $N$ is a prime M-node (to the left) and $N$ is a prime Q-node (to the right).}
\end{figure}
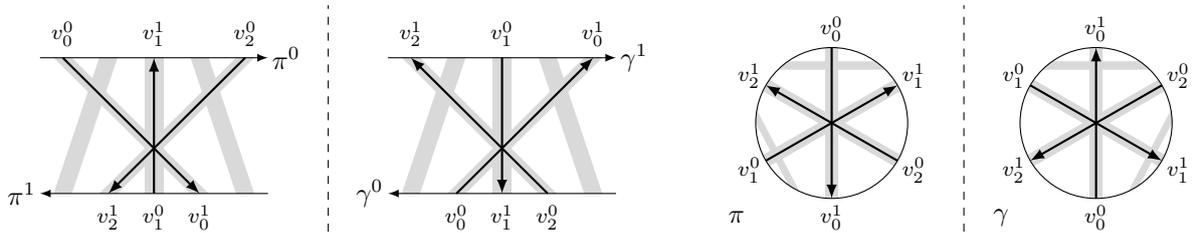

Our proof of Theorem~\ref{thm:Lin-Szwarcfiter} will use the following technical lemma.
\begin{lemma}
\label{lem:technical-lemma-crossing-chords}
Let $\phi$ be a conformal model of $G$ and let $M$ be an M-node in $\pqmtree$.
\begin{enumerate}
\item \label{item:technical-lemma-crossing-chords-same-direction} If there are $a,b \in M$ such that $a \sim b$ 
and $a^0b^0$ is a subword of $\phi|M^j$ for some $j \in \{0,1\}$,
then there are $c,d \in M$ such that $c \sim d$ and $c^1a^{0}d^1b^0$ is a subword of $\phi|M^j$ and 
$a^1c^0b^1d^0$ is a subword of $\phi|M^{1-j}$ (see Figure \ref{fig:technical_lemma} to the left).
\item \label{item:technical-lemma-crossing-chords-different-direction}
If there are $a,b \in M$ such that $a \sim b$ and $a^1b^0$ is a subword of $\phi|M^j$ for some $j \in \{0,1\}$,
then there are $c,d \in M$ such that $c \sim d$ and $a^{1}c^0d^1b^0$ is a subword of $\phi|M^j$ and 
$c^1a^0b^1d^0$ is a subword of $\phi|M^{1-j}$ (see Figure \ref{fig:technical_lemma} to the right).
\end{enumerate}
\end{lemma}
\begin{proof}
First we prove~\eqref{item:technical-lemma-crossing-chords-same-direction}.
Since $a \sim b$, $b$ is not contained in $a$, there exists $c \in V$ such that $ac \notin E(G)$ and $bc \in E(G)$.
Therefore, $a^1c^0b^1$ is a subword of $\phi|M^{1-j}$.
This implies $c\in M$.
Therefore, $c^1a^0b^0$ is a subword of $\phi|M^j$.
Similarly, since $a \sim b$ there exists $d$ such that $ad \in E(G)$ and $bd \notin E(G)$.
Therefore, $a^0d^1b^0$ is a subword of $\phi|M^j$.
This implies $d\in M$ and $a^1b^1d^0$ is a subword of $\phi|M^{1-j}$.
Combining the above conclusions we have that $c^1a^{0}d^1b^0$ is a subword of $\phi|M^j$ and $a^1c^0b^1d^0$ is a subword of $\phi|M^{1-j}$, as desired.
In particular, this implies $c \sim d$.

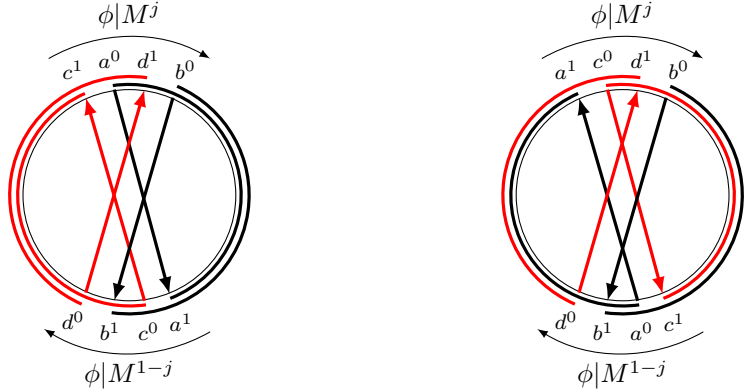
\begin{figure}[htp!]
\centering
\begin{tikzpicture}[scale=0.70,>=latex,shorten >=-0.4pt,shorten <=-0.4pt]

\coordinate (center) at (0,0) {};

\coordinate (u1) at ($(center)+(114:2cm)$) {};
\coordinate (u2) at ($(center)+(98:2cm)$) {};
\coordinate (u3) at ($(center)+(82:2cm)$) {};
\coordinate (u4) at ($(center)+(66:2cm)$) {};

\coordinate (l1) at ($(center)+(292:2cm)$) {};
\coordinate (l2) at ($(center)+(278:2cm)$) {};
\coordinate (l3) at ($(center)+(262:2cm)$) {};
\coordinate (l4) at ($(center)+(246:2cm)$) {};

\coordinate (lu1) at ($(center)+(114:2.6cm)$) {};
\coordinate (lu2) at ($(center)+(98:2.6cm)$) {};
\coordinate (lu3) at ($(center)+(82:2.6cm)$) {};
\coordinate (lu4) at ($(center)+(66:2.6cm)$) {};

\coordinate (ll1) at ($(center)+(292:2.6cm)$) {};
\coordinate (ll2) at ($(center)+(278:2.6cm)$) {};
\coordinate (ll3) at ($(center)+(262:2.6cm)$) {};
\coordinate (ll4) at ($(center)+(246:2.6cm)$) {};

\coordinate (ll) at ($(center)+(270:3.4cm)$) {};
\coordinate (lu) at ($(center)+(90:3.4cm)$) {};

\tikzstyle{every node}=[inner sep=1pt]
\begin{tiny}
\node at (lu1) {$c^1$};
\node at (lu2) {$a^0$};
\node at (lu3) {$d^1$};
\node at (lu4) {$b^0$};

\node at (ll1) {$a^1$};
\node at (ll2) {$c^0$};
\node at (ll3) {$b^1$};
\node at (ll4) {$d^0$};
\end{tiny}

\begin{footnotesize}
\node at (lu) {$\phi|M^j$};
\node at (ll) {$\phi|M^{1-j}$};
\end{footnotesize}
\draw (0,0) circle (2cm);

\draw[very thick] ([shift=(292:2.10cm)]0,0) arc (292:(360+98):2.10cm);
\draw[very thick] ([shift=(262:2.25cm)]0,0) arc (262:(360+66):2.25cm);
\draw[very thick, red] ([shift=(114:2.10cm)]0,0) arc (114:278:2.10cm);
\draw[very thick, red] ([shift=(82:2.25cm)]0,0) arc (82:246:2.25cm);
\draw[very thick,red,<-] (u1) -- (l2);
\draw[very thick,->] (u2) -- (l1);
\draw[very thick,red,<-] (u3) -- (l4);
\draw[very thick,->] (u4) -- (l3);

\draw[<-] ([shift=(60:3cm)]0,0) arc (60:120:3cm);
\draw[<-] ([shift=(238:3cm)]0,0) arc (238:300:3cm);

\draw[white] (-3,-4) -- (-3,-2.5);
\draw[white] (3,4) -- (3,2.5);

\end{tikzpicture} 
\hspace{2cm}
\begin{tikzpicture}[scale=0.70,>=latex,shorten >=-0.4pt,shorten <=-0.4pt]

\coordinate (center) at (0,0) {};

\coordinate (u1) at ($(center)+(114:2cm)$) {};
\coordinate (u2) at ($(center)+(98:2cm)$) {};
\coordinate (u3) at ($(center)+(82:2cm)$) {};
\coordinate (u4) at ($(center)+(66:2cm)$) {};

\coordinate (l1) at ($(center)+(292:2cm)$) {};
\coordinate (l2) at ($(center)+(278:2cm)$) {};
\coordinate (l3) at ($(center)+(262:2cm)$) {};
\coordinate (l4) at ($(center)+(246:2cm)$) {};

\coordinate (lu1) at ($(center)+(114:2.6cm)$) {};
\coordinate (lu2) at ($(center)+(98:2.6cm)$) {};
\coordinate (lu3) at ($(center)+(82:2.6cm)$) {};
\coordinate (lu4) at ($(center)+(66:2.6cm)$) {};

\coordinate (ll1) at ($(center)+(292:2.6cm)$) {};
\coordinate (ll2) at ($(center)+(278:2.6cm)$) {};
\coordinate (ll3) at ($(center)+(262:2.6cm)$) {};
\coordinate (ll4) at ($(center)+(246:2.6cm)$) {};

\coordinate (ll) at ($(center)+(270:3.4cm)$) {};
\coordinate (lu) at ($(center)+(90:3.4cm)$) {};

\tikzstyle{every node}=[inner sep=1pt]
\begin{tiny}
\node at (lu1) {$a^1$};
\node at (lu2) {$c^0$};
\node at (lu3) {$d^1$};
\node at (lu4) {$b^0$};

\node at (ll1) {$c^1$};
\node at (ll2) {$a^0$};
\node at (ll3) {$b^1$};
\node at (ll4) {$d^0$};
\end{tiny}

\begin{footnotesize}
\node at (lu) {$\phi|M^j$};
\node at (ll) {$\phi|M^{1-j}$};
\end{footnotesize}
\draw (0,0) circle (2cm);

\draw[very thick,red] ([shift=(292:2.10cm)]0,0) arc (292:(360+98):2.10cm);
\draw[very thick] ([shift=(262:2.25cm)]0,0) arc (262:(360+66):2.25cm);
\draw[very thick] ([shift=(114:2.10cm)]0,0) arc (114:278:2.10cm);
\draw[very thick,red] ([shift=(82:2.25cm)]0,0) arc (82:246:2.25cm);
\draw[very thick,<-] (u1) -- (l2);
\draw[very thick,red,->] (u2) -- (l1);
\draw[very thick,red,<-] (u3) -- (l4);
\draw[very thick,->] (u4) -- (l3);

\draw[<-] ([shift=(60:3cm)]0,0) arc (60:120:3cm);
\draw[<-] ([shift=(238:3cm)]0,0) arc (238:300:3cm);

\draw[white] (-3,-4) -- (-3,-2.5);
\draw[white] (3,4) -- (3,2.5);

\end{tikzpicture} 
\caption{\label{fig:technical_lemma} Illustration of statements~\eqref{item:technical-lemma-crossing-chords-same-direction} (to the left) and \eqref{item:technical-lemma-crossing-chords-different-direction} (to the right) of Lemma 2.1.}
\end{figure}

Now we prove \eqref{item:technical-lemma-crossing-chords-different-direction}.
Since $a \sim b$, $a$ and $b$ do not cover the circle, and $|M^j \cap \{x^0,x^1\}| = 1$ for every $x \in M$, 
there exists $c \in V$ such that $ac \notin E(G)$ and $b \sim c$ 
or there exists $c' \in V$ such that $a\sim c'$ and $bc' \notin E(G)$.
Suppose the first case holds (the other case is analogous).
Therefore, $a^1c^0b^0$ is a subword of $\phi|M^j$.
This implies that $c\in M$.
Therefore $c^1a^0b^1$  is a subword of $\phi|M^{1-j}$.
Now, applying Lemma~\ref{lem:technical-lemma-crossing-chords}.\eqref{item:technical-lemma-crossing-chords-same-direction} to~$b$ and~$c$, we get that there is $d$ such that $c^0d^1b^0$ is a subword of $\phi|M^j$ and $c^1b^1d^0$ is a subword of of $\phi|M^{1-j}$.
Combining the above conclusions we have that $a^1c^0d^1b^0$ is a subword of $\phi|M^j$ and $c^1a^0b^1d^0$ is a subword of $\phi|M^{1-j}$, as desired.
In particular, this implies $c\sim d$.
\end{proof}

We proceed to the proof of Theorem~\ref{thm:Lin-Szwarcfiter}.
In the course of the proof we also show that $G$ contains no other rigid non-Helly cliques than those defined by Claims~\ref{claim:large_rigid_forbidden_structure} and \ref{claim:small_rigid_forbidden_structure}.
In the result, we show Theorem~\ref{thm:rigid_forbidden_structures} as a by-product.

\begin{proof}[Proof of Theorem \ref{thm:Lin-Szwarcfiter}]
Let $\phi$ be a non-Helly conformal model of $G$.
We show that any other conformal model $\phi'$ of $G$ is also non-Helly.
Since $\phi$ is non-Helly, $G$ contains a minimal non-Helly clique $C$ in $\phi$.
By Lemma~\ref{lem:non-Helly-clique-forces-forbidden-structure}, 
$C$ forms a non-Helly structure in $\phi$.
If $C$ satisfies the assumptions of Claim~\ref{claim:large_rigid_forbidden_structure} or Claim~\ref{claim:small_rigid_forbidden_structure},
$C$ is non-Helly in $\phi'$, and $\phi'$ is non-Helly.
Suppose the other case, that is, suppose $C = \{a,b,c\}$,
$C \subseteq Q$ for a component $Q$ of $G_{ov}$, and suppose $N$ is 
the deepest node in $\pqmtree_Q$ such that $C \subseteq Q$ (that is, $C \subseteq N$ and either $C$ intersects with three children of $N$ and $N$ is serial,
or $C$ intersects two children of $N$).

Let $N$ be a serial node in $\pqmtree_Q$ and let $a \in N_a$, $b \in N_b$, and $c \in N_c$ for three distinct children of~$N$.
Clearly, $N$ is an M-node or $N=V$ and $N$ is the only inner node in $\pqmtree$. 
Observe that, since $N$ is serial, $\Pi(N)$ admits an order $\pi$ such that 
$C$ is Helly in $\phi$ whenever $\phi_{|N}=\pi$.
Now, suppose $a^0c^0b^1a^1c^1b^0$ is a circular subword of $\phi'$.
If $N$ is is an M-node, we 
use Lemma \ref{lem:technical-lemma-crossing-chords} to deduce there is $d \in N$ such that
$d^1a^0c^0b^1a^1d^0c^1b^0$ is a circular subword of $\phi'$.
Otherwise, if $N = V$ and $V$ is serial, 
there is $d \in N_a$ such that $d^1a^0c^0b^1a^1d^0c^1b^0$ is a circular subword of $\phi'$
as otherwise $a$ would be a universal vertex in $G$.
In any case we deduce $\{b,c,d\}$ forms a non-Helly structure in $\phi'$.
The other cases are proven analogously.

Assume $C$ intersects exactly two children of $N$.
Assume first that $\{a,b\} \in N_{a,b}$ and $c \in N_c$, where $N_{a,b}$ and $N_c$
are two distinct children of $N$.
Note that $N_{a,b}$ is an M-node in $\pqmtree_{Q}$.
Assume $a^0b^1$ is a subword of $\phi|N^0_{a,b}$ and 
$a^1b^0$ is a subword of $\phi|N^1_{a,b}$.
We use Lemma~\ref{lem:technical-lemma-crossing-chords}.\eqref{item:technical-lemma-crossing-chords-different-direction}
to show there are $d,e \in N_{a,b}$ such that $d^1a^0b^1e^0$ is a subword of $\phi|N^0_{a,b}$ and $a^1d^0e^1b^0$ 
is a subword of $\phi|N^1_{a,b}$.
Now, we refer to Figure~\ref{fig:abcde_relations} which shows all possible configurations of the oriented chords representing the elements in the set $\{a,b,c,d,e\}$
in conformal models of $G$.

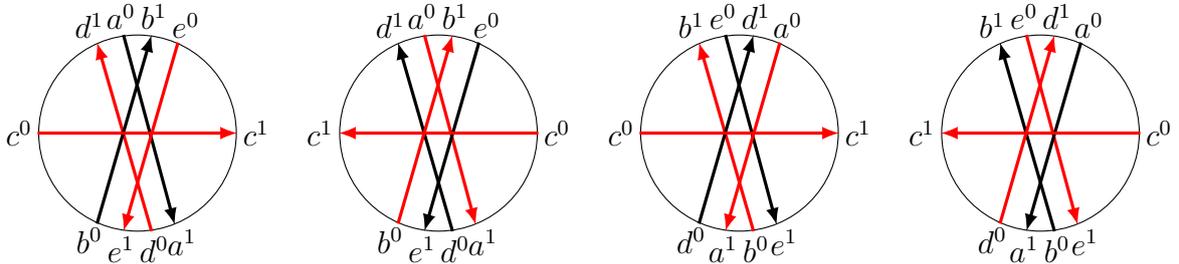
\begin{figure}[htp!]
\centering
\begin{tikzpicture}[scale=0.65,>=latex,shorten >=-0.4pt,shorten <=-0.4pt]

\coordinate (center) at (0,0) {};

\coordinate (u1) at ($(center)+(114:2cm)$) {};
\coordinate (u2) at ($(center)+(98:2cm)$) {};
\coordinate (u3) at ($(center)+(82:2cm)$) {};
\coordinate (u4) at ($(center)+(66:2cm)$) {};

\coordinate (l) at ($(center)+(180:2cm)$) {};
\coordinate (r) at ($(center)+(0:2cm)$) {};

\coordinate (l1) at ($(center)+(292:2cm)$) {};
\coordinate (l2) at ($(center)+(278:2cm)$) {};
\coordinate (l3) at ($(center)+(262:2cm)$) {};
\coordinate (l4) at ($(center)+(246:2cm)$) {};

\coordinate (lu1) at ($(center)+(114:2.4cm)$) {};
\coordinate (lu2) at ($(center)+(98:2.4cm)$) {};
\coordinate (lu3) at ($(center)+(82:2.4cm)$) {};
\coordinate (lu4) at ($(center)+(66:2.4cm)$) {};

\coordinate (ll1) at ($(center)+(292:2.4cm)$) {};
\coordinate (ll2) at ($(center)+(278:2.4cm)$) {};
\coordinate (ll3) at ($(center)+(262:2.4cm)$) {};
\coordinate (ll4) at ($(center)+(246:2.4cm)$) {};

\coordinate (ll) at ($(center)+(180:2.4cm)$) {};
\coordinate (lr) at ($(center)+(0:2.4cm)$) {};

\tikzstyle{every node}=[inner sep=1pt]
\node at (lu1) {$d^1$};
\node at (lu2) {$a^0$};
\node at (lu3) {$b^1$};
\node at (lu4) {$e^0$};

\node at (ll1) {$a^1$};
\node at (ll2) {$d^0$};
\node at (ll3) {$e^1$};
\node at (ll4) {$b^0$};
\node at (ll) {$c^0$};
\node at (lr) {$c^1$};

\draw (0,0) circle (2cm);

\draw[very thick,<-,red] (u1) -- (l2);
\draw[very thick,->] (u2) -- (l1);
\draw[very thick,<-] (u3) -- (l4);
\draw[very thick,->,red] (u4) -- (l3);

\draw[very thick,->,red] (l) -- (r);

\draw[white] (-2.5,-2.5) -- (-2.5,-2);
\draw[white] (2.5,2.5) -- (2.5,2);
\end{tikzpicture} 
\hspace{0.15cm}
\begin{tikzpicture}[scale=0.65,>=latex,shorten >=-0.4pt,shorten <=-0.4pt]

\coordinate (center) at (0,0) {};

\coordinate (u1) at ($(center)+(114:2cm)$) {};
\coordinate (u2) at ($(center)+(98:2cm)$) {};
\coordinate (u3) at ($(center)+(82:2cm)$) {};
\coordinate (u4) at ($(center)+(66:2cm)$) {};

\coordinate (l) at ($(center)+(180:2cm)$) {};
\coordinate (r) at ($(center)+(0:2cm)$) {};

\coordinate (l1) at ($(center)+(292:2cm)$) {};
\coordinate (l2) at ($(center)+(278:2cm)$) {};
\coordinate (l3) at ($(center)+(262:2cm)$) {};
\coordinate (l4) at ($(center)+(246:2cm)$) {};

\coordinate (lu1) at ($(center)+(114:2.4cm)$) {};
\coordinate (lu2) at ($(center)+(98:2.4cm)$) {};
\coordinate (lu3) at ($(center)+(82:2.4cm)$) {};
\coordinate (lu4) at ($(center)+(66:2.4cm)$) {};

\coordinate (ll1) at ($(center)+(292:2.4cm)$) {};
\coordinate (ll2) at ($(center)+(278:2.4cm)$) {};
\coordinate (ll3) at ($(center)+(262:2.4cm)$) {};
\coordinate (ll4) at ($(center)+(246:2.4cm)$) {};

\coordinate (ll) at ($(center)+(180:2.4cm)$) {};
\coordinate (lr) at ($(center)+(0:2.4cm)$) {};

\tikzstyle{every node}=[inner sep=1pt]
\node at (lu1) {$d^1$};
\node at (lu2) {$a^0$};
\node at (lu3) {$b^1$};
\node at (lu4) {$e^0$};

\node at (ll1) {$a^1$};
\node at (ll2) {$d^0$};
\node at (ll3) {$e^1$};
\node at (ll4) {$b^0$};
\node at (ll) {$c^1$};
\node at (lr) {$c^0$};

\draw (0,0) circle (2cm);

\draw[very thick,<-] (u1) -- (l2);
\draw[very thick,->,red] (u2) -- (l1);
\draw[very thick,<-,red] (u3) -- (l4);
\draw[very thick,->] (u4) -- (l3);

\draw[very thick,->,red] (r) -- (l);

\draw[white] (-2.5,-2.5) -- (-2.5,-2);
\draw[white] (2.5,2.5) -- (2.5,2);
\end{tikzpicture} 
\hspace{0.15cm}
\begin{tikzpicture}[scale=0.65,>=latex,shorten >=-0.4pt,shorten <=-0.4pt]

\coordinate (center) at (0,0) {};

\coordinate (u1) at ($(center)+(114:2cm)$) {};
\coordinate (u2) at ($(center)+(98:2cm)$) {};
\coordinate (u3) at ($(center)+(82:2cm)$) {};
\coordinate (u4) at ($(center)+(66:2cm)$) {};

\coordinate (l) at ($(center)+(180:2cm)$) {};
\coordinate (r) at ($(center)+(0:2cm)$) {};

\coordinate (l1) at ($(center)+(292:2cm)$) {};
\coordinate (l2) at ($(center)+(278:2cm)$) {};
\coordinate (l3) at ($(center)+(262:2cm)$) {};
\coordinate (l4) at ($(center)+(246:2cm)$) {};

\coordinate (lu1) at ($(center)+(114:2.4cm)$) {};
\coordinate (lu2) at ($(center)+(98:2.4cm)$) {};
\coordinate (lu3) at ($(center)+(82:2.4cm)$) {};
\coordinate (lu4) at ($(center)+(66:2.4cm)$) {};

\coordinate (ll1) at ($(center)+(292:2.4cm)$) {};
\coordinate (ll2) at ($(center)+(278:2.4cm)$) {};
\coordinate (ll3) at ($(center)+(262:2.4cm)$) {};
\coordinate (ll4) at ($(center)+(246:2.4cm)$) {};

\coordinate (ll) at ($(center)+(180:2.4cm)$) {};
\coordinate (lr) at ($(center)+(0:2.4cm)$) {};

\tikzstyle{every node}=[inner sep=1pt]
\node at (lu3) {$d^1$};
\node at (lu4) {$a^0$};
\node at (lu1) {$b^1$};
\node at (lu2) {$e^0$};

\node at (ll3) {$a^1$};
\node at (ll4) {$d^0$};
\node at (ll1) {$e^1$};
\node at (ll2) {$b^0$};
\node at (ll) {$c^0$};
\node at (lr) {$c^1$};

\draw (0,0) circle (2cm);

\draw[very thick,<-,red] (u1) -- (l2);
\draw[very thick,->] (u2) -- (l1);
\draw[very thick,<-] (u3) -- (l4);
\draw[very thick,->,red] (u4) -- (l3);

\draw[very thick,->,red] (l) -- (r);

\draw[white] (-2.5,-2.5) -- (-2.5,-2);
\draw[white] (2.5,2.5) -- (2.5,2);
\end{tikzpicture} 
\hspace{0.15cm}
\begin{tikzpicture}[scale=0.65,>=latex,shorten >=-0.4pt,shorten <=-0.4pt]

\coordinate (center) at (0,0) {};

\coordinate (u1) at ($(center)+(114:2cm)$) {};
\coordinate (u2) at ($(center)+(98:2cm)$) {};
\coordinate (u3) at ($(center)+(82:2cm)$) {};
\coordinate (u4) at ($(center)+(66:2cm)$) {};

\coordinate (l) at ($(center)+(180:2cm)$) {};
\coordinate (r) at ($(center)+(0:2cm)$) {};

\coordinate (l1) at ($(center)+(292:2cm)$) {};
\coordinate (l2) at ($(center)+(278:2cm)$) {};
\coordinate (l3) at ($(center)+(262:2cm)$) {};
\coordinate (l4) at ($(center)+(246:2cm)$) {};

\coordinate (lu1) at ($(center)+(114:2.4cm)$) {};
\coordinate (lu2) at ($(center)+(98:2.4cm)$) {};
\coordinate (lu3) at ($(center)+(82:2.4cm)$) {};
\coordinate (lu4) at ($(center)+(66:2.4cm)$) {};

\coordinate (ll1) at ($(center)+(292:2.4cm)$) {};
\coordinate (ll2) at ($(center)+(278:2.4cm)$) {};
\coordinate (ll3) at ($(center)+(262:2.4cm)$) {};
\coordinate (ll4) at ($(center)+(246:2.4cm)$) {};

\coordinate (ll) at ($(center)+(180:2.4cm)$) {};
\coordinate (lr) at ($(center)+(0:2.4cm)$) {};

\tikzstyle{every node}=[inner sep=1pt]
\node at (lu3) {$d^1$};
\node at (lu4) {$a^0$};
\node at (lu1) {$b^1$};
\node at (lu2) {$e^0$};

\node at (ll3) {$a^1$};
\node at (ll4) {$d^0$};
\node at (ll1) {$e^1$};
\node at (ll2) {$b^0$};
\node at (ll) {$c^1$};
\node at (lr) {$c^0$};

\draw (0,0) circle (2cm);

\draw[very thick,<-] (u1) -- (l2);
\draw[very thick,->,red] (u2) -- (l1);
\draw[very thick,<-,red] (u3) -- (l4);
\draw[very thick,->] (u4) -- (l3);

\draw[very thick,->,red] (r) -- (l);

\draw[white] (-2.5,-2.5) -- (-2.5,-2);
\draw[white] (2.5,2.5) -- (2.5,2);
\end{tikzpicture} 

\caption{\label{fig:abcde_relations} Forbidden structure (in red) in all possible configurations between oriented chords for $a,b,c,d,e$.}
\end{figure}

In any case, note that there is a triple in $\{a,b,c,d,e\}$ which forms a non-Helly structure in $G$.
Also, since each of these configurations occurs in some conformal model of $G$, 
$\{a,b,c\}$ satisfies the Helly condition in some conformal model of $G$.
The remaining case $\{a,c\} \in N_{a,c}$ and $b \in N_b$, where $N_{a,c}$ and $N_b$
are two distinct children of $N$ is proven analogously
(in this case we use Lemma~\ref{lem:technical-lemma-crossing-chords}.\eqref{item:technical-lemma-crossing-chords-same-direction} to obtain two additional chords $d,e$).
\end{proof}

As a by-product we obtain the theorem that characterizes rigid non-Helly cliques in $G$.
\begin{theorem}
\label{thm:rigid_forbidden_structures}
A clique $C$ of size $\geq 3$ is rigid in $G$ if and only if there is an order $v_0,\ldots,v_{k}$ of the vertices of $G$ such that for every $i \neq j$ from $[0,k]$:
\begin{itemize}
\item $v_{i} \sim v_{j}$ if $(i-j)_{mod\ k} \in \{1,k-1\}$,
\item $v_i,v_j$ cover the circle if $(i-j)_{mod\ k} \notin \{1,k-1\}$,
\end{itemize}
and exactly one of the following statements hold:
\begin{itemize}
\item $|C| \geq 4$.
\item $|C|=3$, $C$ is contained in a component $Q$ of $G_{ov}$, there is a prime node
$N \in \pqmtree_Q$ such that $C \subseteq N$, $C$ intersects with three children of $N$, and 
any ordering of $\Pi(N)$ forces either non-Helly structure
$v^0_0v^1_1v^0_2v^1_0v^0_1v^1_2$ or non-Helly structure $v^0_2v^1_1v^0_0v^1_2v^0_1v^1_0$.
\end{itemize}
\end{theorem}

\section{Types of cliques of circular-arc graphs}
\label{sec:clique-type-problem}

Let $G$ be a circular-arc graph.
\begin{definition}
\label{def:clique_types}
A clique $C$ in $G$ is said to be:
\begin{itemize}
\item \emph{always-Helly} if $C$ is Helly in every conformal model of~$G$,
\item \emph{always-non-Helly} if $C$ is not Helly in every conformal model of~$G$,
\item \emph{ambiguous}, otherwise.
\end{itemize}
\end{definition}
Of course, each clique of $G$ is exactly one of the types listed above.
The goal of this section is to present a polynomial time algorithm 
for the following problem:

\smallskip

\begin{tabular}{rl}
\textbf{Problem:} & \sc{Clique Type} \\
\textbf{Input:}& A circular-arc graph $G$ and a clique $C$ in $G$, \\
\textbf{Output:}& The type of $C$.
\end{tabular}
\medskip

The next observation allows us to restrict our attention to so-called clean cliques in circular-arc graphs.
\begin{definition}
\label{def:clean_clique}
A clique $C$ in $G$ is \emph{clean} if $C$ contains no two vertices such that one is contained in the other. 
\end{definition}
Consider the following ``cleaning'' procedure performed on a clique $C$ of $G$:
\begin{itemize}
    \item as long as there exist $u, v \in C$ such that $u$ is contained in $v$, remove $v$ from $C$.
\end{itemize}
Since the containment relation is transitive, the above procedure always (regardless of the order in which the vertices are removed from $C$) pulls out the same subset of vertices from $C$, and leaves the type of $C$ unchanged.

Suppose $G,C$ is an instance of the Clique Type problem.
Let $\pqmtree$ be the PM-tree for $G$.
The remarks of the previous paragraph allow us to assume that:
\begin{description}
 \item [\namedlabel{prop:clique-type-clean}{(C1)}] $C$ is a clean clique of $G$.
\end{description}

\begin{observation}
\label{obs:type-clique-one-module}
Suppose $C$ satisfies \ref{prop:clique-type-clean}. If $C \cap Q_1 \neq \emptyset$ and $C \cap Q_2 \neq \emptyset$ for two distinct Q-nodes $Q_1$ and $Q_2$ in $\pqmtree$, then $C$ is always-Helly.
\end{observation}
\begin{proof}
Suppose for the sake of contradiction that there is a subclique $C' \subseteq C$ which 
forms a forbidden structure in some conformal model $\phi$ of $G$.
Since $(C',{\sim})$ is connected, $C'$ is contained in some connected component of $G_{ov}$.
Without loss of generality assume $C' \subseteq Q_1$.
Therefore, there exists a vertex $v\in C\setminus C'$ such that $v\in Q_2$.
However, such a vertex cannot exist because $C$ is clean.
\end{proof}
The above observation allows us to assume that:
\begin{description}
 \item [\namedlabel{prop:clique-type-one-component}{(C2)}] $C \subseteq Q$ for some Q-node $Q$ of $\pqmtree$.
\end{description}
Clearly, if $C$ contains a rigid non-Helly subclique, then $C$ is always-non-Helly.
So, we further assume that: 
\begin{description}
 \item [\namedlabel{prop:clique-type-no-rigid-non-Helly-subclique}{(C3)}] 
 \text{$C$ contains no rigid non-Helly subclique.}
\end{description}
We now classify the type of $C$ depending on whether $C$ is public or private.
The next theorem establishes the type of $C$ when $C$ is public.
\begin{theorem}
\label{thm:public_clique_classification}
Suppose $C$ is a public clique in $G$ which satisfies conditions \ref{prop:clique-type-clean}-\ref{prop:clique-type-no-rigid-non-Helly-subclique}.
\begin{enumerate}
 \item \label{item:public_clique_classification_M_prime} If $Q$ is prime, then $C$ is always-Helly. 
 \item \label{item:public_clique_classification_M_serial} Suppose $Q$ is serial (which implies $Q=V$). 
 If $C$ intersects at most two CA-modules from $\camodules$, then $C$ is always-Helly,
 otherwise $C$ is ambiguous.
 \end{enumerate}
\end{theorem}
\begin{proof}
To show statement~\eqref{item:public_clique_classification_M_prime} suppose, for the sake of contradiction, that
there is a conformal model $\phi$ of $G$ and vertices $a,b,c \in C$ such that $\phi$ 
contains $a^0b^1c^0a^1b^0c^1$ as a circular subword.
Since $C$ satisfies~\ref{prop:clique-type-no-rigid-non-Helly-subclique}, we have that $a,b,c$ are not contained in three distinct children of $Q$.
Without loss of generality, suppose $a$ and $b$ are contained in $M$ where $M$ is a child of $Q$ in $\pqmtree_{Q}$.
Since $C$ is public, $a$ and $b$ have the same orientation.
Without loss of generality assume $\phi|M^j$ contains $b^0a^0$ as a subword for some $j\in \{0,1\}$.
However, this implies that $\phi|M^j$ contains $b^0c^1a^0$ as a subword, which cannot be the case as $C$ is public.

To show statement~\eqref{item:public_clique_classification_M_serial} suppose there are $a,b,c\in C$ contained in three different $CA$-modules from $\camodules$.
Let $\phi$ be a conformal model of $G$.
Since $Q=V$ is serial, to obtain other conformal models, we can permute the slots of~$\phi$ arbitrarily as long as the slots corresponding to different CA-modules overlap. 
Thus, we can obtain a conformal model~$\phi'$ of $G$ such that $\phi'$ contains $a^0b^1c^0a^1b^0c^1$ as a circular subword.
Since $C$ is public, we can permute the slots of $\phi$ to get a conformal model of $G$ in which $C$ is Helly.

One can easily check that $C$ is Helly in any conformal model of $G$ when $C$ is public and $C$ is contained in at most two $CA$-modules from $\camodules$.
\end{proof}

Assume now that $C$ is private. 
Denote by $\owners$ the set of all M-nodes of $\pqmtree$ which are the owners of~$C$. 
Clearly, $\owners \neq \emptyset$ as $C$ is private.
We start with a simple claim.
\begin{claim}
\label{claim:private-nodes-and-vertices-outside}
For every $M \in \owners$ and every $c \in C \setminus M$, we have 
$c \sim M$.
\end{claim}
\begin{proof}
The claim follows easily by the fact that $C$ is clean and $M \cap C$ contains two vertices with different orientation in $\MMM$.
\end{proof}

\begin{observation}
Suppose $C$ is a private clique which satisfies conditions~\ref{prop:clique-type-clean}--\ref{prop:clique-type-no-rigid-non-Helly-subclique}.
Then:
\begin{itemize}
\item If there are $M_1, M_2 \in \owners$ which are incomparable in $\pqmtree_{Q}$, then $C$ is always-non-Helly.
\item If there are $M \in \owners$ and $c_1,c_2 \in C \setminus M$ such that
$c_1 \parallel c_2$, then $C$ is always-non-Helly.
\end{itemize}
\end{observation}
\begin{proof}
The first statement follows from Lemma~\ref{lemma:private-nodes}.\eqref{item:private-nodes-model}.

For the second statement suppose for the sake of contradiction that there is a $C$-conformal model $\phi$ of G.
Due to Lemma~\ref{lemma:private-nodes}.\eqref{item:private-nodes-model} the clique letter $C$ occurs in either $\phi|M^0$ or $\phi|M^1$.
Observe that $c_1$ and $c_2$ cover the circle since $C$ is clean.
By Claim~\ref{claim:private-nodes-and-vertices-outside} we have $c_1 \sim M$ and $c_2 \sim M$.
Therefore, the intersection of the arcs $\phi(c_1)$ and $\phi(c_2)$ 
is disjoint with the slots $\phi|M^0$ and $\phi|M^1$,
and hence $\phi$ is not a $C$-conformal model of $G$.
\end{proof}

The above observation allows us to assume that:
\begin{description}
 \item [\namedlabel{prop:clique-type-private-set}{(C4)}] The set $\owners$ induces a path in $\pqmtree_Q$ starting at some CA-module $S \in \camodules(Q)$ and ending at some M-node $D$ in $\pqmtree_{S}$.
 \item [\namedlabel{prop:clique-type-vertices-of-C-from-outside-D}{(C5)}] For every $M \in \owners$ and every $c \in C \setminus M$ we have $c \sim M$ and for every two distinct $c_1,c_2 \in C \setminus M$ we have $c_1 \sim c_2$.
\end{description}
The M-nodes $S$ and $D$, determined by property \ref{prop:clique-type-private-set}, are called the \emph{lowest owner} and the \emph{deepest owner} of $C$ in $\pqmtree_{Q}$, respectively.

Now we describe the properties of the modules from the set $\owners \cup \{Q\}$.
For this purpose, for every $N \in \owners \cup \{Q\}$ we set 
$$\mathcal{I}(N) = \left\{ K: 
\begin{array}{c}
\text{$K$ is a child of $N$ in $\pqmtree_{Q}$ such that} \\
\text{$(C \cap K) \neq \emptyset$ and $K \notin \owners$}
\end{array}
\right\}.
$$
\begin{claim}
\label{claim:I_N-properties}
The following statements hold:
\begin{enumerate}
 \item \label{item:I_N-properties-clique} For every $c \in C$ there is $N \in \owners \cup \{Q\}$ such that $c \in K$ for some $K \in \mathcal{I}(N)$.
 \item \label{item:I_N-properties-consistent-orientation} For every $N \in \owners \cup \{Q\}$ and every $K \in \mathcal{I}(N)$ the edges from $(C \cap K)$ have the same orientation in~$\KKK$.
\end{enumerate}
\end{claim}
\begin{proof}
The statements of the claim follow from Condition~\ref{prop:clique-type-private-set} and Claim~\ref{claim:private-nodes-and-vertices-outside}.
\end{proof}

Let $N \in \mathcal{P} \cup \{Q\}$ be such that $\mathcal{I}(N) \neq \emptyset$.
Claim~\ref{claim:I_N-properties}.\eqref{item:I_N-properties-consistent-orientation} allows us to partition $\mathcal{I}(N)$ into 
the sets $\mathcal{L}(N)$ and $\mathcal{R}(N)$, as follows:
$$
\begin{array}{ccc}
\smallskip
 \mathcal{L}(N) &=& \big{\{}K \in \mathcal{I}(N): \text{the elements from $(C \cap K)$ are $(K^0,K^1)$-oriented in $\KKK$}\big{\}}, \\
 \mathcal{R}(N) &=& \big{\{}K \in \mathcal{I}(N): \text{the elements from $(C \cap K)$ are $(K^1,K^0)$-oriented in $\KKK$}\big{\}}.
\end{array}
$$
\begin{claim}
\label{claim:L-R-sets}
The following statements hold:
\begin{enumerate}
\item \label{item:L-R-sets-L} For every $N \in \owners \cup \{Q\}$ and distinct $L_1,L_2 \in \mathcal{L}(N)$
we have $L_1 \sim L_2$,
\item \label{item:L-R-sets-R} For every $N \in \owners \cup \{Q\}$ and distinct $R_1,R_2 \in \mathcal{R}(N)$
we have $R_1 \sim R_2$,
\item \label{item:L-R-sets-L-R} For every $N \in \big{(}\owners \cup \{Q\}\big{)} \setminus \{D\}$, $L \in \mathcal{L}(N)$, and $R \in \mathcal{R}(N)$, we have 
$L \sim K$, $K \sim R$, and $L \sim R$, where $K$ is a child of $N$ from $\owners$.
\end{enumerate}
\end{claim}
In the next lemmas we describe the properties of the modules from the set $\owners \cup \{Q\}$.
We start with the deepest owner $D$  of $C$.
\begin{lemma}
\label{lem:deepest-owner-properties}
Let $D$ be the deepest owner of $C$ in $\pqmtree_Q$.
Then:
\begin{enumerate}
\item \label{item:deepest-owner-properties-L-D-relation} We have $\mathcal{L}(D) \neq \emptyset$, $\mathcal{R}(D) \neq \emptyset$. Moreover,
\begin{enumerate}
 \item \label{item:deepest-owner-properties-L-D-relation-parallel} If $D$ is parallel, then $\mathcal{L}(D) \parallel \mathcal{R}(D)$ and $|\mathcal{L}(D)| = |\mathcal{R}(D)| = 1$.
 \item \label{item:deepest-owner-properties-L-D-relation-serial} If $D$ is serial, then $\mathcal{L}(D) \sim \mathcal{R}(D)$.
\item \label{item:deepest-owner-properties-L-D-relation-prime} If $D$ is prime, then either $\mathcal{L}(D) \parallel \mathcal{R}(D)$ or $L \sim R$ for some $L \in \mathcal{L}(D)$ and some $R \in \mathcal{R}(D)$.
\end{enumerate}
\item \label{item:deepest-owner-properties-model}
Let $\phi$ be a $C$-conformal model $\phi$ of $G$ and let $\phi_{|D} = (\pi^0,\pi^1)$ be an ordering of $N$ in~$\phi$:
\begin{enumerate}
 \item If $C$ is contained in $\phi|D^0$, then $\pi^0$ orders $\mathcal{L}^0(D)$ before $\mathcal{R}^0(D)$.
 \item If $C$ is contained in $\phi|D^1$, then $\pi^1$ orders $\mathcal{R}^1(D)$ before $\mathcal{L}^1(D)$.
\end{enumerate}
\item \label{item:deepest-owner-properties-prime}
Let $D$ be prime and let $L \in \mathcal{L}(D)$ and $R \in \mathcal{R}(D)$ be such that
$L \sim R$.
Then: 
\begin{enumerate}
\item \label{item:deepest-owner-properties-prime-bindind-in-0} There is a unique admissible ordering $\gamma = (\gamma^0, \gamma^1)$ in $\Pi(D)$ which orders
$\mathcal{L}^0(D)$ before $\mathcal{R}^0(D)$ in $\gamma^0$.
\item \label{item:deepest-owner-properties-prime-bindind-in-1} There is a unique admissible ordering $\delta = (\delta^0, \delta^1)$ in $\Pi(D)$
which orders $\mathcal{R}^1(D)$ before $\mathcal{L}^1(D)$ in $\delta^1$.
\end{enumerate}
Moreover, we have $\Pi(D)=\{\gamma,\delta\}$, which means that $\delta$ is the reflection of $\gamma$.
\end{enumerate}
\end{lemma}
\begin{proof}
We have $\mathcal{L}(D) \neq \emptyset$ and $\mathcal{R}(D) \neq \emptyset$ by Claim~\ref{claim:I_N-properties}.\eqref{item:I_N-properties-consistent-orientation} and by the fact that we have two vertices from $C \cap D$ with different orientations in $\DDD$.
Statement~\eqref{item:deepest-owner-properties-L-D-relation-parallel}
follows by the fact that $C$ is clean.
Statements~\eqref{item:deepest-owner-properties-L-D-relation-serial}-\eqref{item:deepest-owner-properties-L-D-relation-prime}
are obvious.

Statement \eqref{item:deepest-owner-properties-model} follows from the fact that the letter $C$ is on the left side of the chords of~$\phi$ representing the vertices from the set $C \cap \big{(} \bigcup \mathcal{L}(D) \big{)}$ and from the set $C \cap \big{(}\bigcup \mathcal{R}(D) \big{)}$.

Now we proceed to the proof of statement \eqref{item:deepest-owner-properties-prime}.

Since $C$ is a clique, for every $(\pi^0,\pi^1)$ in $\Pi(D)$ and every $L \in \mathcal{L}(D)$ and $R \in \mathcal{R}(D)$ such that $L \parallel R$
we have that $L^0$ occurs before $R^0$ in $\pi^0$ and $R^1$ occurs before $L^1$ in $\pi^1$.

Suppose now that there exist $L\in \mathcal{L}(D)$ and $R \in \mathcal{R}(D)$ such that $L\sim R$.
Let ${\prec_D}$ be a transitive orientation of $(D,{\sim_D})$ satisfying $L \prec_D R$ 
(note that in the other transitive orientation $\prec^{R}_D$ of $(D,{\sim_D})$ we have $R \prec^R_D L$).
To prove statement~\eqref{item:deepest-owner-properties-prime} it is enough to show $L' \prec_{D} R'$ for every $L' \in \mathcal{L}(D)$ and
$R' \in \mathcal{R}(D)$ such that $L' \sim R'$.
Then we have $R' \prec^R_{D} L'$ for every $L' \in \mathcal{L}(D)$ and
$R' \in \mathcal{R}(D)$ such that $L' \sim R'$, and the admissible orderings $\gamma$ and $\delta$ from $\Pi(D)$
corresponding to ${\prec_D}$ and ${\prec^R_D}$ satisfy statements~\eqref{item:deepest-owner-properties-prime-bindind-in-0} and \eqref{item:deepest-owner-properties-prime-bindind-in-1}, respectively.
Suppose for the sake of contradiction that there exist $L'\in \mathcal{L}(D)$ and $R' \in \mathcal{R}(D)$ such that $L'\sim R'$ and $R' \prec_D L'$.
First, suppose $L\sim R'$.
If $L \prec_D R'$, then $L \prec_D R' \prec_D L'$ which contradicts that $C$ satisfies property~\ref{prop:clique-type-no-rigid-non-Helly-subclique}.
But if $R' \prec_D L$ then $R' \prec_D L \prec_D R$ which again contradicts that $C$ satisfies property~\ref{prop:clique-type-no-rigid-non-Helly-subclique}.
Therefore, we must have $L \parallel R'$.
In particular, we have $L <_{S} R'$.
Due to symmetrical arguments, we have $L' \parallel R$.
In particular, we have $L' <_S R$.
Recall that $L\sim L'$ and $R \sim R'$ because $C$ is clean.
However, this implies that $C$ contains a rigid non-Helly clique of size $4$ 
which cannot be the case since $C$ satisfies property~\ref{prop:clique-type-no-rigid-non-Helly-subclique}.
\end{proof}
Lemma~\ref{lem:deepest-owner-properties} allows us to introduce the following terminology for admissible orders of the deepest owner $D$ of $C$:
\begin{itemize}
 \item If $D$ is prime and there are $L \in \mathcal{L}(D)$ and $R \in \mathcal{R}(D)$ such that $L \sim R$, then we say admissible ordering $\gamma = (\gamma^{0}, \gamma^1)$ which orders $\mathcal{L}^0(D)$ before $\mathcal{R}^0(D)$ in $\gamma^0$ \emph{binds $C$ in the slot~$S^0$}, and admissible ordering $\delta = (\delta^{0}, \delta^1)$ of $\Pi(D)$ which orders $\mathcal{R}^1(D)$ before $\mathcal{L}^1(D)$ in $\delta^1$ \emph{binds $C$ in the slot~$S^1$}.
 See Figure~\ref{fig:deepest-owner} to the left for an illustration.
 \item If $D$ is serial, then any admissible ordering $\gamma = (\gamma^{0}, \gamma^1)$ which orders $\mathcal{L}^0(D)$ before $\mathcal{R}^0(D)$ in $\gamma^0$ \emph{binds $C$ in the slot~$S^0$}, any admissible ordering $\delta = (\delta^{0}, \delta^1)$ of $\Pi(D)$ which orders $\mathcal{R}^1(D)$ before $\mathcal{L}^1(D)$ in $\delta^1$ \emph{binds $C$ in the slot~$S^1$}, and any other ordering $\sigma=(\sigma^0,\sigma^1)$ \emph{makes $C$ non-Helly}. 
 Note that $D$ admits an ordering which makes $C$ non-Helly iff $|\mathcal{I}(D)| \geq 3$.
 See Figure~\ref{fig:deepest-owner} (in the middle) for an illustration.
\end{itemize}
Note that when $D$ is prime and $\mathcal{L}(D) \parallel \mathcal{R}(D)$ or when $D$ is parallel, 
then the choice of the admissible ordering of $D$ has no impact on 
the slot in which $C$ occurs. 
See Figure~\ref{fig:deepest-owner} to the right for an illustration.
Also, if $D$ is prime and $L \sim D$ for some $L \in \mathcal{L}(D)$ and $R \in \mathcal{R}(D)$ or when $D$ is serial and $|\mathcal{I}(D)|=2$, then any ordering of $D$ binds $C$ is $S^j$ for some $j \in \{0,1\}$.

\begin{figure}[ht]
\begin{tikzpicture}[xscale=0.8,yscale=0.75,>=latex]
\coordinate (u1) at (0,2) {};
\coordinate (u2) at (1,2) {};
\coordinate (u3) at (2,2) {};
\coordinate (u4) at (3,2) {};

\coordinate (uC) at (1.5,2) {};
\coordinate (luC) at (1.5,2.35) {};

\coordinate (b1) at (0,0) {};
\coordinate (b2) at (1,0) {};
\coordinate (b3) at (2,0) {};
\coordinate (b4) at (3,0) {};

\coordinate (lu1) at (0,2.35) {};
\coordinate (lu2) at (1,2.35) {};
\coordinate (lu3) at (2,2.35) {};
\coordinate (lu4) at (3,2.35) {};

\coordinate (lb1) at (0,-0.35) {};
\coordinate (lb2) at (1,-0.35) {};
\coordinate (lb3) at (2,-0.35) {};
\coordinate (lb4) at (3,-0.35) {};

\coordinate (tau0) at (3.8,2);
\coordinate (tau1) at (-0.8,0);

\tikzstyle{every node}=[inner sep=2pt,fill=white]

\draw[fill=gray!30, draw=none] (-0.2,0) -- (0.2,0) -- (1.2,2) -- (0.8,2) -- cycle;
\draw[fill=gray!30, draw=none] (1.8,0) -- (2.2,0) -- (0.2,2) -- (-0.2,2) -- cycle;
\draw[fill=gray!30, draw=none] (0.8,0) -- (1.2,0) -- (3.2,2) -- (2.8,2) -- cycle;
\draw[fill=gray!30, draw=none] (2.8,0) -- (3.2,0) -- (2.2,2) -- (1.8,2) -- cycle;

\draw[thick,red,->] (u1)--(b3);
\draw[thick,red,->] (u2)--(b1);
\draw[->] (u3)--(b4);
\draw[thick,blue,<-] (u4)--(b2);

\draw[->] (-0.5,2) -- (3.5,2);
\draw[<-] (-0.5,0) -- (3.5,0);

\tikzstyle{every node}=[circle,minimum size=5pt,inner sep=0pt,draw,fill]
\node at (uC) {};

\tikzstyle{every node}=[inner sep=1pt]
\begin{footnotesize}
\node at (tau0) {$\gamma^0$};
\node at (tau1) {$\gamma^1$};

\node at (lu1) {$L^0_1$};
\node at (lu2) {$L^0_2$};
\node at (luC) {$C$};
\node at (lu4) {$R^0_1$};

\node at (lb1) {$L^1_2$};
\node at (lb2) {$R^1_1$};
\node at (lb3) {$L^1_1$};

\end{footnotesize}
\end{tikzpicture}
\hspace{0.4cm}
\begin{tikzpicture}[xscale=0.8,yscale=0.75,>=latex]
\coordinate (u1) at (0,2) {};
\coordinate (u2) at (1,2) {};
\coordinate (u3) at (2,2) {};
\coordinate (u4) at (3,2) {};

\coordinate (uC) at (1.5,2) {};
\coordinate (luC) at (1.5,2.35) {};

\coordinate (b1) at (0,0) {};
\coordinate (b2) at (1,0) {};
\coordinate (b3) at (2,0) {};
\coordinate (b4) at (3,0) {};

\coordinate (bC) at (2.5,0) {};
\coordinate (lbC) at (2.5,-0.35) {};

\coordinate (lu1) at (0,2.35) {};
\coordinate (lu2) at (1,2.35) {};
\coordinate (lu3) at (2,2.35) {};
\coordinate (lu4) at (3,2.35) {};

\coordinate (lb1) at (0,-0.35) {};
\coordinate (lb2) at (1,-0.35) {};
\coordinate (lb3) at (2,-0.35) {};
\coordinate (lb4) at (3,-0.35) {};

\coordinate (tau0) at (3.8,2);
\coordinate (tau1) at (-0.8,0);

\tikzstyle{every node}=[inner sep=2pt,fill=white]

\draw[fill=gray!30, draw=none] (-0.2,0) -- (0.2,0) -- (3.2,2) -- (2.8,2) -- cycle;
\draw[fill=gray!30, draw=none] (0.8,0) -- (1.2,0) -- (2.2,2) -- (1.8,2) -- cycle;
\draw[fill=gray!30, draw=none] (1.8,0) -- (2.2,0) -- (1.2,2) -- (0.8,2) -- cycle;
\draw[fill=gray!30, draw=none] (2.8,0) -- (3.2,0) -- (0.2,2) -- (-0.2,2) -- cycle;

\draw[thick,red,->] (u1)--(b4);
\draw[thick,red,->] (u2)--(b3);
\draw[thick,blue,<-] (u3)--(b2);
\draw[->] (u4)--(b1);

\draw[->] (-0.5,2) -- (3.5,2);
\draw[<-] (-0.5,0) -- (3.5,0);

\tikzstyle{every node}=[circle,minimum size=5pt,inner sep=0pt,draw,fill]
\node at (uC) {};

\tikzstyle{every node}=[inner sep=1pt]
\begin{footnotesize}
\node at (tau0) {$\gamma^0$};
\node at (tau1) {$\gamma^1$};

\node at (lu1) {$L^0_1$};
\node at (lu2) {$L^0_2$};
\node at (lu3) {$R^0_1$};

\node at (lb2) {$R^1_1$};
\node at (lb3) {$L^1_2$};
\node at (lb4) {$L^1_1$};

\node at (luC) {$C$};
\node at (lbC) {$C$};

\end{footnotesize}
\end{tikzpicture}
\hspace{0.4cm}
\begin{tikzpicture}[xscale=0.8,yscale=0.75,>=latex]
\coordinate (u1) at (0,2) {};
\coordinate (u2) at (1,2) {};
\coordinate (u3) at (2,2) {};
\coordinate (u4) at (3,2) {};

\coordinate (uC) at (1.5,2) {};
\coordinate (luC) at (1.5,2.35) {};

\coordinate (b1) at (0,0) {};
\coordinate (b2) at (1,0) {};
\coordinate (b3) at (2,0) {};
\coordinate (b4) at (3,0) {};

\coordinate (bC) at (2.5,0) {};
\coordinate (lbC) at (2.5,-0.35) {};

\coordinate (lu1) at (0,2.35) {};
\coordinate (lu2) at (1,2.35) {};
\coordinate (lu3) at (2,2.35) {};
\coordinate (lu4) at (3,2.35) {};

\coordinate (lb1) at (0,-0.35) {};
\coordinate (lb2) at (1,-0.35) {};
\coordinate (lb3) at (2,-0.35) {};
\coordinate (lb4) at (3,-0.35) {};

\coordinate (tau0) at (3.8,2);
\coordinate (tau1) at (-0.8,0);

\tikzstyle{every node}=[inner sep=2pt,fill=white]

\draw[fill=gray!30, draw=none] (-0.2,0) -- (0.2,0) -- (1.2,2) -- (0.8,2) -- cycle;
\draw[fill=gray!30, draw=none] (1.8,0) -- (2.2,0) -- (0.2,2) -- (-0.2,2) -- cycle;
\draw[fill=gray!30, draw=none] (0.8,0) -- (1.2,0) -- (3.2,2) -- (2.8,2) -- cycle;
\draw[fill=gray!30, draw=none] (2.8,0) -- (3.2,0) -- (2.2,2) -- (1.8,2) -- cycle;

\draw[thick,red,->] (u1)--(b3);
\draw[thick,red,->] (u2)--(b1);
\draw[thick,<-,blue] (u3)--(b4);
\draw[->] (u4)--(b2);

\draw[<-] (-0.5,2) -- (3.5,2);
\draw[->] (-0.5,0) -- (3.5,0);

\tikzstyle{every node}=[circle,minimum size=5pt,inner sep=0pt,draw,fill]
\node at (uC) {};
\node at (bC) {};

\tikzstyle{every node}=[inner sep=1pt]
\begin{footnotesize}
\node at (tau0) {$\gamma^0$};
\node at (tau1) {$\gamma^1$};

\node at (lu1) {$L^0_1$};
\node at (lu2) {$L^0_2$};
\node at (lu3) {$R^0_1$};

\node at (lb1) {$L^1_2$};
\node at (lb3) {$L^1_1$};
\node at (lb4) {$R^1_1$};

\node at (luC) {$C$};
\node at (lbC) {$C$};

\end{footnotesize}
\end{tikzpicture}

\vspace{0.4cm}

\begin{tikzpicture}[xscale=0.8,yscale=-0.75,>=latex]
\coordinate (u1) at (0,2) {};
\coordinate (u2) at (1,2) {};
\coordinate (u3) at (2,2) {};
\coordinate (u4) at (3,2) {};

\coordinate (uC) at (1.5,2) {};
\coordinate (luC) at (1.5,2.35) {};

\coordinate (b1) at (0,0) {};
\coordinate (b2) at (1,0) {};
\coordinate (b3) at (2,0) {};
\coordinate (b4) at (3,0) {};

\coordinate (lu1) at (0,2.35) {};
\coordinate (lu2) at (1,2.35) {};
\coordinate (lu3) at (2,2.35) {};
\coordinate (lu4) at (3,2.35) {};

\coordinate (lb1) at (0,-0.35) {};
\coordinate (lb2) at (1,-0.35) {};
\coordinate (lb3) at (2,-0.35) {};
\coordinate (lb4) at (3,-0.35) {};

\coordinate (tau0) at (3.8,0);
\coordinate (tau1) at (-0.8,2);

\tikzstyle{every node}=[inner sep=2pt,fill=white]

\draw[fill=gray!30, draw=none] (-0.2,0) -- (0.2,0) -- (1.2,2) -- (0.8,2) -- cycle;
\draw[fill=gray!30, draw=none] (1.8,0) -- (2.2,0) -- (0.2,2) -- (-0.2,2) -- cycle;
\draw[fill=gray!30, draw=none] (0.8,0) -- (1.2,0) -- (3.2,2) -- (2.8,2) -- cycle;
\draw[fill=gray!30, draw=none] (2.8,0) -- (3.2,0) -- (2.2,2) -- (1.8,2) -- cycle;

\draw[thick,red,<-] (u1)--(b3);
\draw[thick,red,<-] (u2)--(b1);
\draw[<-] (u3)--(b4);
\draw[thick,blue,->] (u4)--(b2);

\draw[<-] (-0.5,2) -- (3.5,2);
\draw[->] (-0.5,0) -- (3.5,0);

\tikzstyle{every node}=[circle,minimum size=5pt,inner sep=0pt,draw,fill]
\node at (uC) {};

\tikzstyle{every node}=[inner sep=1pt]
\begin{footnotesize}
\node at (tau0) {$\delta^0$};
\node at (tau1) {$\delta^1$};

\node at (lu1) {$L^1_1$};
\node at (lu2) {$L^1_2$};
\node at (luC) {$C$};
\node at (lu4) {$R^1_1$};

\node at (lb1) {$L^0_2$};
\node at (lb2) {$R^0_1$};
\node at (lb3) {$L^0_1$};

\end{footnotesize}
\end{tikzpicture}
\hspace{0.4cm}
\begin{tikzpicture}[xscale=0.8,yscale=0.75,>=latex]
\coordinate (u1) at (0,2) {};
\coordinate (u2) at (1,2) {};
\coordinate (u3) at (2,2) {};
\coordinate (u4) at (3,2) {};

\coordinate (uC) at (1.5,2) {};
\coordinate (luC) at (1.5,2.35) {};

\coordinate (b1) at (0,0) {};
\coordinate (b2) at (1,0) {};
\coordinate (b3) at (2,0) {};
\coordinate (b4) at (3,0) {};

\coordinate (bC) at (2.5,0) {};
\coordinate (lbC) at (2.5,-0.35) {};

\coordinate (lu1) at (0,2.35) {};
\coordinate (lu2) at (1,2.35) {};
\coordinate (lu3) at (2,2.35) {};
\coordinate (lu4) at (3,2.35) {};

\coordinate (lb1) at (0,-0.35) {};
\coordinate (lb2) at (1,-0.35) {};
\coordinate (lb3) at (2,-0.35) {};
\coordinate (lb4) at (3,-0.35) {};

\coordinate (tau0) at (3.8,2);
\coordinate (tau1) at (-0.8,0);

\tikzstyle{every node}=[inner sep=2pt,fill=white]

\draw[fill=gray!30, draw=none] (-0.2,0) -- (0.2,0) -- (3.2,2) -- (2.8,2) -- cycle;
\draw[fill=gray!30, draw=none] (0.8,0) -- (1.2,0) -- (2.2,2) -- (1.8,2) -- cycle;
\draw[fill=gray!30, draw=none] (1.8,0) -- (2.2,0) -- (1.2,2) -- (0.8,2) -- cycle;
\draw[fill=gray!30, draw=none] (2.8,0) -- (3.2,0) -- (0.2,2) -- (-0.2,2) -- cycle;

\draw[thick,red,->] (u1)--(b4);
\draw[thick,blue,<-] (u2)--(b3);
\draw[->] (u3)--(b2);
\draw[thick,red,->] (u4)--(b1);

\draw[->] (-0.5,2) -- (3.5,2);
\draw[<-] (-0.5,0) -- (3.5,0);

\tikzstyle{every node}=[circle,minimum size=5pt,inner sep=0pt,draw,fill]

\tikzstyle{every node}=[inner sep=1pt]
\begin{footnotesize}
\node at (tau0) {$\sigma^0$};
\node at (tau1) {$\sigma^1$};

\node at (lu1) {$L^0_1$};
\node at (lu2) {$R^0_1$};
\node at (lu4) {$L^0_2$};

\node at (lb1) {$L^1_2$};
\node at (lb3) {$R^1_1$};
\node at (lb4) {$L^1_1$};


\end{footnotesize}
\end{tikzpicture}
\hspace{0.4cm}
\begin{tikzpicture}[xscale=0.8,yscale=-0.75,>=latex]
\coordinate (u1) at (0,2) {};
\coordinate (u2) at (1,2) {};
\coordinate (u3) at (2,2) {};
\coordinate (u4) at (3,2) {};

\coordinate (uC) at (1.5,2) {};
\coordinate (luC) at (1.5,2.35) {};

\coordinate (b1) at (0,0) {};
\coordinate (b2) at (1,0) {};
\coordinate (b3) at (2,0) {};
\coordinate (b4) at (3,0) {};

\coordinate (bC) at (2.5,0) {};
\coordinate (lbC) at (2.5,-0.35) {};

\coordinate (lu1) at (0,2.35) {};
\coordinate (lu2) at (1,2.35) {};
\coordinate (lu3) at (2,2.35) {};
\coordinate (lu4) at (3,2.35) {};

\coordinate (lb1) at (0,-0.35) {};
\coordinate (lb2) at (1,-0.35) {};
\coordinate (lb3) at (2,-0.35) {};
\coordinate (lb4) at (3,-0.35) {};

\coordinate (tau0) at (3.8,0);
\coordinate (tau1) at (-0.8,2);

\tikzstyle{every node}=[inner sep=2pt,fill=white]

\draw[fill=gray!30, draw=none] (-0.2,0) -- (0.2,0) -- (1.2,2) -- (0.8,2) -- cycle;
\draw[fill=gray!30, draw=none] (1.8,0) -- (2.2,0) -- (0.2,2) -- (-0.2,2) -- cycle;
\draw[fill=gray!30, draw=none] (0.8,0) -- (1.2,0) -- (3.2,2) -- (2.8,2) -- cycle;
\draw[fill=gray!30, draw=none] (2.8,0) -- (3.2,0) -- (2.2,2) -- (1.8,2) -- cycle;

\draw[thick,red,<-] (u1)--(b3);
\draw[thick,red,<-] (u2)--(b1);
\draw[thick,->,blue] (u3)--(b4);
\draw[<-] (u4)--(b2);

\draw[<-] (-0.5,2) -- (3.5,2);
\draw[->] (-0.5,0) -- (3.5,0);

\tikzstyle{every node}=[circle,minimum size=5pt,inner sep=0pt,draw,fill]
\node at (uC) {};
\node at (bC) {};

\tikzstyle{every node}=[inner sep=1pt]
\begin{footnotesize}
\node at (tau0) {$\delta^0$};
\node at (tau1) {$\delta^1$};

\node at (lu1) {$L^1_1$};
\node at (lu2) {$L^1_2$};
\node at (lu3) {$R^1_1$};

\node at (lb1) {$L^0_2$};
\node at (lb3) {$L^0_1$};
\node at (lb4) {$R^0_1$};

\node at (luC) {$C$};
\node at (lbC) {$C$};

\end{footnotesize}
\end{tikzpicture}
\caption{\label{fig:deepest-owner} 
To the left: prime deepest owner $D$ with $L_1 \sim R_1$: ordering $\gamma=(\gamma^0,\gamma^1)$ binds $C$ in~$S^0$ and ordering $\delta=(\delta^0,\delta^1)$ binds $C$ in~$S^1$. 
In the middle: ordering $\gamma = (\gamma^0,\gamma^1)$ of serial deepest owner $D$ binds $C$ in $S^0$, ordering $\sigma = (\sigma^0,\sigma^1)$ makes $C$ non-Helly.
To the right: prime deepest owner $D$ satisfies $\mathcal{L}(D) \parallel \mathcal{R}(D)$
and hence $C$ is not bound by $D$.}
\end{figure}
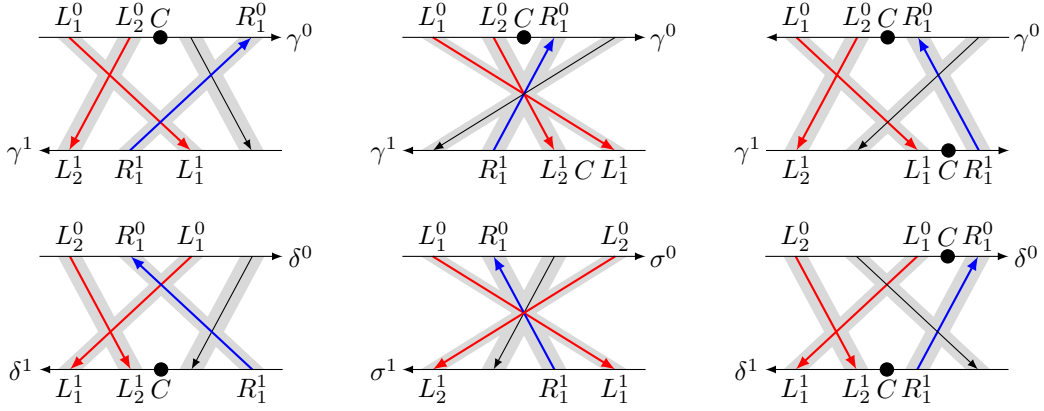

\begin{lemma}
\label{lem:M-node-owner-properties}
Let $M \in \owners \setminus \{D\}$ be such that $\mathcal{I}(M) \neq \emptyset$ and
let $K$ be the child of $M$ in the set $\owners$. 
Then $M$ is either prime or serial and:
\begin{enumerate}
\item \label{item:M-node-owner-properties-L-R-sets} We have $\mathcal{L}(M) \sim K$, $\mathcal{R}(M) \sim K$, and $\mathcal{L}(M) \sim \mathcal{R}(M)$.
\item \label{item:M-node-owner-properties-model}
Let $\phi$ be a $C$-conformal model $\phi$ of $G$ and let $\phi_{|M} = (\pi^0,\pi^1)$ be an ordering of $M$ in~$\phi$:
\begin{enumerate}
 \item If $C$ is contained in $\phi|D^0$ then $\pi^0$ orders $\mathcal{L}^0(M)$ before $K^0$ and $K^0$ before $\mathcal{R}^0(M)$,
 \item If $C$ is contained in $\phi|D^1$ then $\pi^1$ orders $\mathcal{R}^1(M)$ before $K^1$ and $K^1$ before $\mathcal{L}^1(M)$.
\end{enumerate}
\item \label{item:M-node-owner-properties-prime} Suppose $M$ is prime. 
Then: 
\begin{enumerate}
\item \label{item:M-node-owner-properties-prime-bindind-in-0} There is a unique admissible ordering $\gamma = (\gamma^0, \gamma^1)$ in $\Pi(M)$ which orders
$\mathcal{L}^0(M)$ before $K^0$ and $K^0$ before $\mathcal{R}^0(M)$ in $\gamma^0$.
\item \label{item:M-node-owner-properties-prime-bindind-in-1} There is a unique admissible ordering $\delta = (\delta^0, \delta^1)$ in $\Pi(M)$
which orders $\mathcal{R}^1(M)$ before $K^1$ and $K^1$ before $\mathcal{L}^1(M)$ in $\delta^1$.
\end{enumerate}
Moreover, we have $\Pi(M)=\{\gamma,\delta\}$, which means that $\delta$ is the reflection of $\gamma$.
\end{enumerate}
\end{lemma}
\begin{proof}
Since $C$ is clean, $\mathcal{I}(M) \neq \emptyset$, and $K \cap C$ contains two vertices with different orientation in $\KKK$, $M$ needs to be either 
prime or serial.

Statement~\eqref{item:M-node-owner-properties-L-R-sets} follows from the fact that $D \subseteq K$ and $C$ satisfies condition~\ref{prop:clique-type-vertices-of-C-from-outside-D}.

Statement~\eqref{item:M-node-owner-properties-model} follows from the fact that $D \subseteq K$ and the letter $C$ is on the left side of the chords of $\phi$ representing the vertices from the set $C \cap \big{(}\bigcup \mathcal{L}(M)\big{)}$ and from the set $C \cap \big{(}\bigcup \mathcal{R}(M)\big{)}$.

Statement~\eqref{item:M-node-owner-properties-prime} follows from Claim~\ref{claim:I_N-properties}.\eqref{item:I_N-properties-consistent-orientation} and the fact that $C$ contains no rigid non-Helly subclique.
\end{proof}
Similarly, the above lemma allows us to introduce the following terminology for admissible orderings of $\Pi(M)$ for M-nodes $M \in \owners \setminus \{D\}$ such that $\mathcal{I}(M) \neq \emptyset$:
\begin{itemize}
 \item If $M$ is prime, then admissible ordering $\gamma = (\gamma^{0}, \gamma^1)$ which orders $\mathcal{L}^0(M)$ before $K^0$ and $K^0$ before $\mathcal{R}^0(M)$ in $\gamma^0$ \emph{binds $C$ in the slot~$S^0$}, and admissible ordering $\delta = (\delta^{0}, \delta^1)$ of $\Pi(M)$
 which orders $\mathcal{R}^1(M)$ before $K^1$ and $K^1$ before $\mathcal{L}^1(M)$ in $\delta^1$ \emph{binds $C$ in the slot~$S^1$}.
 See Figure~\ref{fig:M-node-owner} (to the left) for an illustration.
 \item If $M$ is serial, then any admissible ordering $\gamma = (\gamma^{0}, \gamma^1)$ which orders $\mathcal{L}^0(M)$ before $K^0$ and $K^0$ before $\mathcal{R}^0(M)$ in $\gamma^0$ \emph{binds $C$ in the slot~$S^0$}, any admissible ordering $\delta = (\delta^{0}, \delta^1)$ of $\Pi(M)$
 which orders $\mathcal{R}^1(M)$ before $K^1$ and $K^1$ before $\mathcal{L}^1(M)$ in $\delta^1$ \emph{binds $C$ in the slot~$S^1$}, and any other ordering $\sigma=(\sigma^0,\sigma^1)$ \emph{makes $C$ non-Helly}. 
 Note that $M$ admits an ordering which makes $C$ non-Helly iff $|\mathcal{I}(M)| \geq 2$.
 See Figure~\ref{fig:M-node-owner} (to the right) for an illustration.
\end{itemize}
Note that when $M$ is prime or when $M$ is serial and $|\mathcal{I}(M)| = 1$, 
then any ordering of $M$ binds $C$ in $S^j$ for some $j \in \{0,1\}$.

\begin{figure}[ht]
\begin{tikzpicture}[xscale=0.7,yscale=0.75,>=latex]
\coordinate (u1) at (0,2) {};
\coordinate (u2) at (1,2) {};
\coordinate (u3) at (2,2) {};
\coordinate (u4) at (3,2) {};
\coordinate (u5) at (4,2) {};

\coordinate (uC) at (2,2) {};
\coordinate (luC) at (2,2.35) {};
\coordinate (bC) at (2,0) {};
\coordinate (buC) at (2,-0.35) {};

\coordinate (b1) at (0,0) {};
\coordinate (b2) at (1,0) {};
\coordinate (b3) at (2,0) {};
\coordinate (b4) at (3,0) {};
\coordinate (b5) at (4,0) {};

\coordinate (lu1) at (0,2.35) {};
\coordinate (lu2) at (1,2.35) {};
\coordinate (lu3) at (2,2.35) {};
\coordinate (lu4) at (3,2.35) {};
\coordinate (lu5) at (4,2.35) {};

\coordinate (lb1) at (0,-0.35) {};
\coordinate (lb2) at (1,-0.35) {};
\coordinate (lb3) at (2,-0.35) {};
\coordinate (lb4) at (3,-0.35) {};
\coordinate (lb5) at (4,-0.35) {};

\coordinate (tau0) at (4.8,2);
\coordinate (tau1) at (-0.8,0);

\tikzstyle{every node}=[inner sep=2pt,fill=white]

\draw[fill=gray!30, draw=none] (-0.2,0) -- (0.2,0) -- (1.2,2) -- (0.8,2) -- cycle;
\draw[fill=gray!30, draw=none] (0.8,0) -- (1.2,0) -- (4.2,2) -- (3.8,2) -- cycle;
\draw[fill=gray!30, draw=none] (1.8,0) -- (2.2,0) -- (2.2,2) -- (1.8,2) -- cycle;
\draw[fill=gray!30, draw=none] (2.8,0) -- (3.2,0) -- (0.2,2) -- (-0.2,2) -- cycle;
\draw[fill=gray!30, draw=none] (3.8,0) -- (4.2,0) -- (3.2,2) -- (2.8,2) -- cycle;

\draw[thick,red,->] (u1)--(b4);
\draw[->] (u2)--(b1);
\draw[->] (u4)--(b5);
\draw[thick,blue,<-] (u5)--(b2);

\draw[->] (-0.5,2) -- (4.5,2);
\draw[<-] (-0.5,0) -- (4.5,0);

\tikzstyle{every node}=[circle,minimum size=5pt,inner sep=0pt,draw,fill]
\node at (uC) {};

\tikzstyle{every node}=[inner sep=1pt]
\begin{footnotesize}
\node at (tau0) {$\gamma^0$};
\node at (tau1) {$\gamma^1$};

\node at (lu1) {$L^0_1$};
\node at (lu3) {$K^0$};
\node at (lu5) {$R^0_1$};

\node at (lb2) {$R^1_1$};
\node at (lb3) {$K^1$};
\node at (lb4) {$L^1_1$};


\end{footnotesize}
\end{tikzpicture}
\hspace{0.5cm}
\begin{tikzpicture}[xscale=0.7,yscale=0.75,>=latex]
\coordinate (u1) at (0,2) {};
\coordinate (u2) at (1,2) {};
\coordinate (u3) at (2,2) {};
\coordinate (u4) at (3,2) {};
\coordinate (u5) at (4,2) {};

\coordinate (uC) at (2,2) {};
\coordinate (luC) at (2,2.35) {};
\coordinate (bC) at (2,0) {};
\coordinate (buC) at (2,-0.35) {};

\coordinate (b1) at (0,0) {};
\coordinate (b2) at (1,0) {};
\coordinate (b3) at (2,0) {};
\coordinate (b4) at (3,0) {};
\coordinate (b5) at (4,0) {};

\coordinate (lu1) at (0,2.35) {};
\coordinate (lu2) at (1,2.35) {};
\coordinate (lu3) at (2,2.35) {};
\coordinate (lu4) at (3,2.35) {};
\coordinate (lu5) at (4,2.35) {};

\coordinate (lb1) at (0,-0.35) {};
\coordinate (lb2) at (1,-0.35) {};
\coordinate (lb3) at (2,-0.35) {};
\coordinate (lb4) at (3,-0.35) {};
\coordinate (lb5) at (4,-0.35) {};

\coordinate (tau0) at (4.8,2);
\coordinate (tau1) at (-0.8,0);

\tikzstyle{every node}=[inner sep=2pt,fill=white]

\draw[fill=gray!30, draw=none] (-0.2,0) -- (0.2,0) -- (4.2,2) -- (3.8,2) -- cycle;
\draw[fill=gray!30, draw=none] (0.8,0) -- (1.2,0) -- (3.2,2) -- (2.8,2) -- cycle;
\draw[fill=gray!30, draw=none] (1.8,0) -- (2.2,0) -- (2.2,2) -- (1.8,2) -- cycle;
\draw[fill=gray!30, draw=none] (2.8,0) -- (3.2,0) -- (1.2,2) -- (0.8,2) -- cycle;
\draw[fill=gray!30, draw=none] (3.8,0) -- (4.2,0) -- (0.2,2) -- (-0.2,2) -- cycle;

\draw[thick,red,->] (u1)--(b5);
\draw[thick,red,->] (u2)--(b4);
\draw[->] (u4)--(b2);
\draw[->] (u5)--(b1);

\draw[->] (-0.5,2) -- (4.5,2);
\draw[<-] (-0.5,0) -- (4.5,0);

\tikzstyle{every node}=[circle,minimum size=5pt,inner sep=0pt,draw,fill]
\node at (uC) {};

\tikzstyle{every node}=[inner sep=1pt]
\begin{footnotesize}
\node at (tau0) {$\gamma^0$};
\node at (tau1) {$\gamma^1$};

\node at (lu1) {$L^0_1$};
\node at (lu2) {$L^0_2$};
\node at (lu3) {$K^0$};

\node at (lb3) {$K^1$};
\node at (lb4) {$L^1_1$};
\node at (lb5) {$L^1_2$};


\end{footnotesize}
\end{tikzpicture}

\vspace{0.4cm}

\begin{tikzpicture}[xscale=0.7,yscale=-0.75,>=latex]
\coordinate (u1) at (0,2) {};
\coordinate (u2) at (1,2) {};
\coordinate (u3) at (2,2) {};
\coordinate (u4) at (3,2) {};
\coordinate (u5) at (4,2) {};

\coordinate (uC) at (2,2) {};
\coordinate (luC) at (2,2.35) {};
\coordinate (bC) at (2,0) {};
\coordinate (buC) at (2,-0.35) {};

\coordinate (b1) at (0,0) {};
\coordinate (b2) at (1,0) {};
\coordinate (b3) at (2,0) {};
\coordinate (b4) at (3,0) {};
\coordinate (b5) at (4,0) {};

\coordinate (lu1) at (0,2.35) {};
\coordinate (lu2) at (1,2.35) {};
\coordinate (lu3) at (2,2.35) {};
\coordinate (lu4) at (3,2.35) {};
\coordinate (lu5) at (4,2.35) {};

\coordinate (lb1) at (0,-0.35) {};
\coordinate (lb2) at (1,-0.35) {};
\coordinate (lb3) at (2,-0.35) {};
\coordinate (lb4) at (3,-0.35) {};
\coordinate (lb5) at (4,-0.35) {};

\coordinate (tau0) at (4.8,0);
\coordinate (tau1) at (-0.8,2);

\tikzstyle{every node}=[inner sep=2pt,fill=white]

\draw[fill=gray!30, draw=none] (-0.2,0) -- (0.2,0) -- (1.2,2) -- (0.8,2) -- cycle;
\draw[fill=gray!30, draw=none] (0.8,0) -- (1.2,0) -- (4.2,2) -- (3.8,2) -- cycle;
\draw[fill=gray!30, draw=none] (1.8,0) -- (2.2,0) -- (2.2,2) -- (1.8,2) -- cycle;
\draw[fill=gray!30, draw=none] (2.8,0) -- (3.2,0) -- (0.2,2) -- (-0.2,2) -- cycle;
\draw[fill=gray!30, draw=none] (3.8,0) -- (4.2,0) -- (3.2,2) -- (2.8,2) -- cycle;

\draw[thick,red,<-] (u1)--(b4);
\draw[<-] (u2)--(b1);
\draw[<-] (u4)--(b5);
\draw[thick,blue,->] (u5)--(b2);

\draw[<-] (-0.5,2) -- (4.5,2);
\draw[->] (-0.5,0) -- (4.5,0);

\tikzstyle{every node}=[circle,minimum size=5pt,inner sep=0pt,draw,fill]
\node at (uC) {};

\tikzstyle{every node}=[inner sep=1pt]
\begin{footnotesize}
\node at (tau0) {$\delta^0$};
\node at (tau1) {$\delta^1$};

\node at (lu1) {$L^1_1$};
\node at (lu3) {$K^1$};
\node at (lu5) {$R^1_1$};

\node at (lb2) {$R^0_1$};
\node at (lb3) {$K^0$};
\node at (lb4) {$L^0_1$};


\end{footnotesize}
\end{tikzpicture}
\hspace{0.5cm}
\begin{tikzpicture}[xscale=0.7,yscale=0.75,>=latex]
\coordinate (u1) at (0,2) {};
\coordinate (u2) at (1,2) {};
\coordinate (u3) at (2,2) {};
\coordinate (u4) at (3,2) {};
\coordinate (u5) at (4,2) {};

\coordinate (uC) at (2,2) {};
\coordinate (luC) at (2,2.35) {};
\coordinate (bC) at (2,0) {};
\coordinate (buC) at (2,-0.35) {};

\coordinate (b1) at (0,0) {};
\coordinate (b2) at (1,0) {};
\coordinate (b3) at (2,0) {};
\coordinate (b4) at (3,0) {};
\coordinate (b5) at (4,0) {};

\coordinate (lu1) at (0,2.35) {};
\coordinate (lu2) at (1,2.35) {};
\coordinate (lu3) at (2,2.35) {};
\coordinate (lu4) at (3,2.35) {};
\coordinate (lu5) at (4,2.35) {};

\coordinate (lb1) at (0,-0.35) {};
\coordinate (lb2) at (1,-0.35) {};
\coordinate (lb3) at (2,-0.35) {};
\coordinate (lb4) at (3,-0.35) {};
\coordinate (lb5) at (4,-0.35) {};

\coordinate (tau0) at (4.8,0);
\coordinate (tau1) at (-0.8,2);

\tikzstyle{every node}=[inner sep=2pt,fill=white]

\draw[fill=gray!30, draw=none] (-0.2,0) -- (0.2,0) -- (4.2,2) -- (3.8,2) -- cycle;
\draw[fill=gray!30, draw=none] (0.8,0) -- (1.2,0) -- (3.2,2) -- (2.8,2) -- cycle;
\draw[fill=gray!30, draw=none] (1.8,0) -- (2.2,0) -- (2.2,2) -- (1.8,2) -- cycle;
\draw[fill=gray!30, draw=none] (2.8,0) -- (3.2,0) -- (1.2,2) -- (0.8,2) -- cycle;
\draw[fill=gray!30, draw=none] (3.8,0) -- (4.2,0) -- (0.2,2) -- (-0.2,2) -- cycle;

\draw[thick,red,->] (u1)--(b5);
\draw[->] (u2)--(b4);
\draw[thick,red,->] (u4)--(b2);
\draw[->] (u5)--(b1);

\draw[->] (-0.5,2) -- (4.5,2);
\draw[<-] (-0.5,0) -- (4.5,0);

\tikzstyle{every node}=[circle,minimum size=5pt,inner sep=0pt,draw,fill]

\tikzstyle{every node}=[inner sep=1pt]
\begin{footnotesize}
\node at (tau0) {$\sigma^0$};
\node at (tau1) {$\sigma^1$};

\node at (lu1) {$L^0_1$};
\node at (lu3) {$K^0$};
\node at (lu4) {$L^0_2$};

\node at (lb2) {$L^1_2$};
\node at (lb3) {$K^1$};
\node at (lb5) {$L^1_1$};


\end{footnotesize}
\end{tikzpicture}
\caption{\label{fig:M-node-owner} 
To the left: ordering $\gamma=(\gamma^0,\gamma^1)$ of prime owner $M$ binds $C$ in~$S^0$ and ordering $\delta=(\delta^0,\delta^1)$ binds $C$ in~$S^1$. 
To the right: ordering $\gamma = (\gamma^0,\gamma^1)$ of serial owner $M$ (we have $\mathcal{R}(N) = \emptyset$) binds $C$ in $S^0$, ordering $\sigma = (\sigma^0,\sigma^1)$ makes $C$ non-Helly.}
\end{figure}
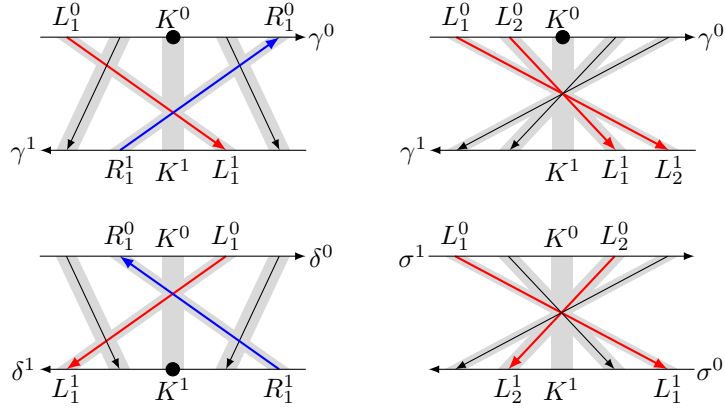

\noindent Finally, we describe the properties of the Q-node $Q$ which contains $C$.
\begin{lemma} Let $Q$ be the $Q$-node containing $C$ and let $S$ be the lowest owner of $C$ in~$\pqmtree_Q$. 
Then:
\label{lem:Q-node-owner-properties}
\begin{enumerate} 
\item \label{item:Q-node-owner-properties-L-R-sets} We have $\mathcal{L}(Q) \sim S$, $\mathcal{R}(Q) \sim S$, and $\mathcal{L}(Q) \sim \mathcal{R}(Q)$.
\item \label{item:Q-node-owner-properties-model}
Let $\phi$ be a conformal model of $G$ and let $\phi_{|Q} \equiv \pi_Q$ be an admissible ordering of $Q$ in~$\phi$. 
Then:
\begin{itemize}
\item If $C$ is in $\phi|S^0$, then $\pi_Q$ orders the slot $S^0$ between $\mathcal{L}^0(Q)$ and $\mathcal{L}^1(Q)$ and between $\mathcal{R}^1(Q)$ and $\mathcal{R}^0(Q)$.
\item If $C$ is in $\phi|S^1$, then $\pi_Q$ orders the slot $S^1$ between $\mathcal{L}^0(Q)$ and $\mathcal{L}^1(Q)$ and between $\mathcal{R}^1(Q)$ and $\mathcal{R}^0(Q)$.
\end{itemize}
\item \label{item:Q-node-owner-properties-prime} Suppose $M$ is prime.
Then: 
\begin{enumerate}
\item \label{item:Q-node-owner-properties-prime-slot-0} There is a unique admissible ordering $\gamma \in \Pi(Q)$ which orders the slot $S^0$ between $\mathcal{L}^0(Q)$ and $\mathcal{L}^1(Q)$ and between $\mathcal{R}^1(Q)$ and $\mathcal{R}^0(Q)$.
\item \label{item:Q-node-owner-properties-prime-slot-1} There is a unique admissible ordering $\delta \in \Pi(Q)$ which orders the slot $S^1$ between $\mathcal{L}^0(Q)$ and $\mathcal{R}^1(Q)$ and between $\mathcal{R}^1(Q)$ and $\mathcal{R}^0(Q)$.
\end{enumerate}
Moreover, we have $\Pi(Q) = \{\gamma, \delta\}$, which means that $\delta$ is the reflection  of $\gamma$.
\end{enumerate}
\end{lemma}
\begin{proof}
Statement~\eqref{item:Q-node-owner-properties-L-R-sets} follows from the fact that $D\subseteq S$ and $C$ satisfies condition~\ref{prop:clique-type-vertices-of-C-from-outside-D}.

Statement~\eqref{item:Q-node-owner-properties-model} follows from the fact that $D\subseteq S$ and the letter $C$ is on the left side of chords of $\phi$ representing the vertices from the set $C \cap \big{(}\bigcup \mathcal{L}(Q)\big{)}$ and from the set $C \cap \big{(}\bigcup \mathcal{R}(Q)\big{)}$.

Statement~\eqref{item:Q-node-owner-properties-prime} follows from Claim~\ref{claim:I_N-properties}.\eqref{item:I_N-properties-consistent-orientation} and the fact that $C$ contains no rigid non-Helly subclique.
\end{proof}
Following our convention, we introduce the following terminology for admissible orderings of $\Pi(Q)$ for Q-node $Q$ such that $C \subseteq Q$ and $\mathcal{I}(Q) \neq \emptyset$:
\begin{itemize}
 \item If $Q$ is prime, then admissible ordering $\gamma \in \Pi(Q)$ which orders the slot $S^0$ between $\mathcal{L}^0(Q)$ and $\mathcal{L}^1(Q)$ and between $\mathcal{R}^1(Q)$ and $\mathcal{R}^0(Q)$ \emph{binds $C$ in the slot $S^0$}, and admissible ordering $\delta$ which orders the slot $S^1$ between $\mathcal{L}^0(Q)$ and $\mathcal{L}^1(Q)$ and between $\mathcal{R}^1(Q)$ and $\mathcal{R}^0(Q)$ \emph{binds $C$ in the slot $S^1$}.
 See Figure~\ref{fig:Q-node-owner} (to the left) for an illustration.
 \item If $Q$ is serial, then any admissible ordering $\gamma \in \Pi(Q)$ which orders the slot $S^0$ between $\mathcal{L}^0(Q)$ and $\mathcal{L}^1(Q)$ and between $\mathcal{R}^1(Q)$ and $\mathcal{R}^0(Q)$ \emph{binds $C$ in the slot $S^0$}, 
 any admissible ordering $\delta$ of $\Pi(Q)$ which orders the slot $S^1$ between $\mathcal{L}^0(Q)$ and $\mathcal{L}^1(Q)$ and between $\mathcal{R}^1(Q)$ and $\mathcal{R}^0(Q)$ \emph{binds $C$ in the slot $S^1$}, and any other ordering $\sigma \in \Pi(Q)$ \emph{makes $C$ non-Helly}. 
 Note that $Q$ admits an ordering which makes $C$ non-Helly iff $|\mathcal{I}(Q)| \geq 2$.
 See Figure~\ref{fig:Q-node-owner} (to the right) for an illustration.
\end{itemize}
Note that when $Q$ is prime or when $Q$ is serial and $|\mathcal{I}(Q)| = 1$, 
then any ordering of $Q$ binds $C$ in $S^j$ for some $j \in \{0,1\}$.

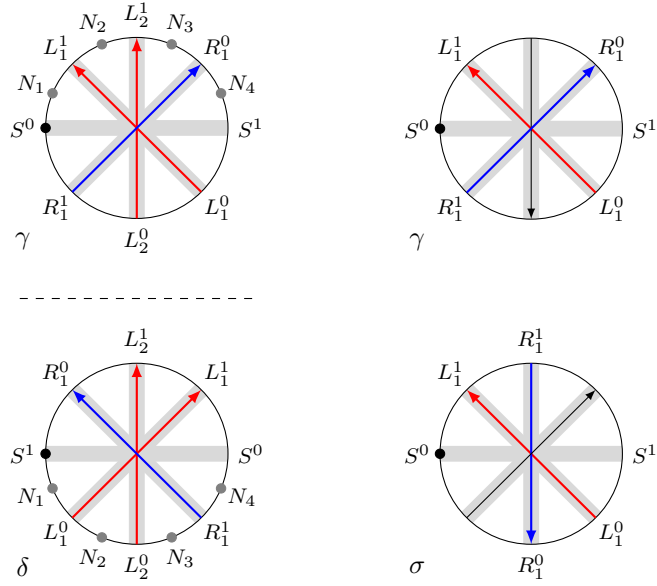
\begin{figure}[ht]
\begin{tikzpicture}[xscale=0.6,yscale=0.6,>=latex]
\coordinate (center) at (0,0) {};

\coordinate (v0) at ($(center)+(0:2cm)$) {};
\coordinate (v1) at ($(center)+(315:2cm)$) {};
\coordinate (v2) at ($(center)+(270:2cm)$) {};
\coordinate (v3) at ($(center)+(225:2cm)$) {};
\coordinate (v4) at ($(center)+(180:2cm)$) {};
\coordinate (v5) at ($(center)+(135:2cm)$) {};
\coordinate (v6) at ($(center)+(90:2cm)$) {};
\coordinate (v7) at ($(center)+(45:2cm)$) {};

\coordinate (lv0) at ($(center)+(0:2.5cm)$) {};
\coordinate (lv1) at ($(center)+(315:2.5cm)$) {};
\coordinate (lv2) at ($(center)+(270:2.5cm)$) {};
\coordinate (lv3) at ($(center)+(225:2.5cm)$) {};
\coordinate (lv4) at ($(center)+(180:2.5cm)$) {};
\coordinate (lv5) at ($(center)+(135:2.5cm)$) {};
\coordinate (lv6) at ($(center)+(90:2.5cm)$) {};
\coordinate (lv7) at ($(center)+(45:2.5cm)$) {};

\coordinate (n4) at ($(center)+(22.5:2cm)$) {};
\coordinate (n3) at ($(center)+(67.5:2cm)$) {};
\coordinate (n2) at ($(center)+(112.5:2cm)$) {};
\coordinate (n1) at ($(center)+(157.5:2cm)$) {};

\coordinate (ln4) at ($(center)+(22.5:2.5cm)$) {};
\coordinate (ln3) at ($(center)+(67.5:2.5cm)$) {};
\coordinate (ln2) at ($(center)+(112.5:2.5cm)$) {};
\coordinate (ln1) at ($(center)+(157.5:2.5cm)$) {};

\coordinate (tau) at (-2.5,-2.5) {};

\draw[fill=gray!30, draw=none] ($(center)+(-5:2cm)$) -- ($(center)+(0:2cm)$)-- ($(center)+(5:2cm)$) -- ($(center)+(175:2cm)$) -- ($(center)+(180:2cm)$) --($(center)+(185:2cm)$) -- cycle;

\draw[fill=gray!30, draw=none] ($(center)+(310:2cm)$) -- ($(center)+(315:2cm)$)-- ($(center)+(320:2cm)$) -- ($(center)+(130:2cm)$) -- ($(center)+(135:2cm)$) --($(center)+(140:2cm)$) -- cycle;

\draw[fill=gray!30, draw=none] ($(center)+(265:2cm)$) -- ($(center)+(270:2cm)$)-- ($(center)+(275:2cm)$) -- ($(center)+(85:2cm)$) -- ($(center)+(90:2cm)$) --($(center)+(95:2cm)$) -- cycle;

\draw[fill=gray!30, draw=none] ($(center)+(220:2cm)$) -- ($(center)+(225:2cm)$)-- ($(center)+(230:2cm)$) -- ($(center)+(40:2cm)$) -- ($(center)+(45:2cm)$) --($(center)+(50:2cm)$) -- cycle;

\draw (0,0) circle (2cm);

\draw[thick,red,->] (v1)--(v5);
\draw[thick,red,->] (v2)--(v6);
\draw[thick,blue,->] (v3)--(v7);

\tikzstyle{every node}=[circle,minimum size=3.5pt,inner sep=0pt,draw,fill]
\node[gray] at (n1) {};
\node[gray] at (n2) {};
\node[gray] at (n3) {};
\node[gray] at (n4) {};
\node at (v4) {};

\tikzstyle{every node}=[inner sep=1pt]
\begin{footnotesize}
\node at (tau) {$\gamma$};
\end{footnotesize}
\begin{tiny}
\node at (lv0) {$S^1$};
\node at (lv1) {$L^0_1$};
\node at (lv2) {$L^0_2$};
\node at (lv3) {$R^1_1$};
\node at (lv4) {$S^0$};
\node at (lv5) {$L^1_1$};
\node at (lv6) {$L^1_2$};
\node at (lv7) {$R^0_1$};

\node at (ln1) {$N_1$};
\node at (ln2) {$N_2$};
\node at (ln3) {$N_3$};
\node at (ln4) {$N_4$};
\end{tiny}
\draw[white,-] (-2.8,-2.8)--(-2.2,-2.8);
\draw[white,-] (2.5,2.8)--(2.2,2.8);
\end{tikzpicture}
\hspace{1.5cm}
\begin{tikzpicture}[xscale=0.6,yscale=0.6,>=latex]
\coordinate (center) at (0,0) {};

\coordinate (v0) at ($(center)+(0:2cm)$) {};
\coordinate (v1) at ($(center)+(315:2cm)$) {};
\coordinate (v2) at ($(center)+(270:2cm)$) {};
\coordinate (v3) at ($(center)+(225:2cm)$) {};
\coordinate (v4) at ($(center)+(180:2cm)$) {};
\coordinate (v5) at ($(center)+(135:2cm)$) {};
\coordinate (v6) at ($(center)+(90:2cm)$) {};
\coordinate (v7) at ($(center)+(45:2cm)$) {};

\coordinate (lv0) at ($(center)+(0:2.5cm)$) {};
\coordinate (lv1) at ($(center)+(315:2.5cm)$) {};
\coordinate (lv2) at ($(center)+(270:2.5cm)$) {};
\coordinate (lv3) at ($(center)+(225:2.5cm)$) {};
\coordinate (lv4) at ($(center)+(180:2.5cm)$) {};
\coordinate (lv5) at ($(center)+(135:2.5cm)$) {};
\coordinate (lv6) at ($(center)+(90:2.5cm)$) {};
\coordinate (lv7) at ($(center)+(45:2.5cm)$) {};

\coordinate (n4) at ($(center)+(22.5:2cm)$) {};
\coordinate (n3) at ($(center)+(67.5:2cm)$) {};
\coordinate (n2) at ($(center)+(112.5:2cm)$) {};
\coordinate (n1) at ($(center)+(157.5:2cm)$) {};

\coordinate (ln4) at ($(center)+(22.5:2.5cm)$) {};
\coordinate (ln3) at ($(center)+(67.5:2.5cm)$) {};
\coordinate (ln2) at ($(center)+(112.5:2.5cm)$) {};
\coordinate (ln1) at ($(center)+(157.5:2.5cm)$) {};

\coordinate (tau) at (-2.5,-2.5) {};

\draw[fill=gray!30, draw=none] ($(center)+(-5:2cm)$) -- ($(center)+(0:2cm)$)-- ($(center)+(5:2cm)$) -- ($(center)+(175:2cm)$) -- ($(center)+(180:2cm)$) --($(center)+(185:2cm)$) -- cycle;

\draw[fill=gray!30, draw=none] ($(center)+(310:2cm)$) -- ($(center)+(315:2cm)$)-- ($(center)+(320:2cm)$) -- ($(center)+(130:2cm)$) -- ($(center)+(135:2cm)$) --($(center)+(140:2cm)$) -- cycle;

\draw[fill=gray!30, draw=none] ($(center)+(265:2cm)$) -- ($(center)+(270:2cm)$)-- ($(center)+(275:2cm)$) -- ($(center)+(85:2cm)$) -- ($(center)+(90:2cm)$) --($(center)+(95:2cm)$) -- cycle;

\draw[fill=gray!30, draw=none] ($(center)+(220:2cm)$) -- ($(center)+(225:2cm)$)-- ($(center)+(230:2cm)$) -- ($(center)+(40:2cm)$) -- ($(center)+(45:2cm)$) --($(center)+(50:2cm)$) -- cycle;

\draw (0,0) circle (2cm);

\draw[thick,red,->] (v1)--(v5);
\draw[<-] (v2)--(v6);
\draw[thick,blue,->] (v3)--(v7);

\tikzstyle{every node}=[circle,minimum size=3.5pt,inner sep=0pt,draw,fill]
\node at (v4) {};

\tikzstyle{every node}=[inner sep=1pt]
\begin{footnotesize}
\node at (tau) {$\gamma$};
\end{footnotesize}
\begin{tiny}
\node at (lv0) {$S^1$};
\node at (lv1) {$L^0_1$};
\node at (lv3) {$R^1_1$};
\node at (lv4) {$S^0$};
\node at (lv5) {$L^1_1$};
\node at (lv7) {$R^0_1$};

\end{tiny}
\draw[white,-] (-2.8,-2.8)--(-2.2,-2.8);
\draw[white,-] (2.5,2.8)--(2.2,2.8);
\end{tikzpicture}

\vspace{0.2cm}
\begin{tikzpicture}[xscale=0.6,yscale=-0.6,>=latex]
\draw[black,dashed,-] (-2.6,0)--(2.6,0);
\draw[white,-] (-2.8,-0.2)--(-2.2,-0.2);
\draw[white,-] (2.8,0.2)--(2.2,0.2);
\end{tikzpicture}
\hspace{1.5cm}
\begin{tikzpicture}[xscale=0.6,yscale=-0.6,>=latex]
\draw[white,-] (-2.6,0)--(2.6,0);
\draw[white,-] (-2.8,-0.2)--(-2.2,-0.2);
\draw[white,-] (2.8,0.2)--(2.2,0.2);
\end{tikzpicture}

\vspace{0.2cm}

\begin{tikzpicture}[xscale=0.6,yscale=-0.6,>=latex]
\coordinate (center) at (0,0) {};

\coordinate (v0) at ($(center)+(0:2cm)$) {};
\coordinate (v1) at ($(center)+(315:2cm)$) {};
\coordinate (v2) at ($(center)+(270:2cm)$) {};
\coordinate (v3) at ($(center)+(225:2cm)$) {};
\coordinate (v4) at ($(center)+(180:2cm)$) {};
\coordinate (v5) at ($(center)+(135:2cm)$) {};
\coordinate (v6) at ($(center)+(90:2cm)$) {};
\coordinate (v7) at ($(center)+(45:2cm)$) {};

\coordinate (lv0) at ($(center)+(0:2.5cm)$) {};
\coordinate (lv1) at ($(center)+(315:2.5cm)$) {};
\coordinate (lv2) at ($(center)+(270:2.5cm)$) {};
\coordinate (lv3) at ($(center)+(225:2.5cm)$) {};
\coordinate (lv4) at ($(center)+(180:2.5cm)$) {};
\coordinate (lv5) at ($(center)+(135:2.5cm)$) {};
\coordinate (lv6) at ($(center)+(90:2.5cm)$) {};
\coordinate (lv7) at ($(center)+(45:2.5cm)$) {};

\coordinate (n4) at ($(center)+(22.5:2cm)$) {};
\coordinate (n3) at ($(center)+(67.5:2cm)$) {};
\coordinate (n2) at ($(center)+(112.5:2cm)$) {};
\coordinate (n1) at ($(center)+(157.5:2cm)$) {};

\coordinate (ln4) at ($(center)+(22.5:2.5cm)$) {};
\coordinate (ln3) at ($(center)+(67.5:2.5cm)$) {};
\coordinate (ln2) at ($(center)+(112.5:2.5cm)$) {};
\coordinate (ln1) at ($(center)+(157.5:2.5cm)$) {};

\coordinate (tau) at (-2.5,2.5) {};

\draw[fill=gray!30, draw=none] ($(center)+(-5:2cm)$) -- ($(center)+(0:2cm)$)-- ($(center)+(5:2cm)$) -- ($(center)+(175:2cm)$) -- ($(center)+(180:2cm)$) --($(center)+(185:2cm)$) -- cycle;

\draw[fill=gray!30, draw=none] ($(center)+(310:2cm)$) -- ($(center)+(315:2cm)$)-- ($(center)+(320:2cm)$) -- ($(center)+(130:2cm)$) -- ($(center)+(135:2cm)$) --($(center)+(140:2cm)$) -- cycle;

\draw[fill=gray!30, draw=none] ($(center)+(265:2cm)$) -- ($(center)+(270:2cm)$)-- ($(center)+(275:2cm)$) -- ($(center)+(85:2cm)$) -- ($(center)+(90:2cm)$) --($(center)+(95:2cm)$) -- cycle;

\draw[fill=gray!30, draw=none] ($(center)+(220:2cm)$) -- ($(center)+(225:2cm)$)-- ($(center)+(230:2cm)$) -- ($(center)+(40:2cm)$) -- ($(center)+(45:2cm)$) --($(center)+(50:2cm)$) -- cycle;

\draw (0,0) circle (2cm);

\draw[thick,red,<-] (v1)--(v5);
\draw[thick,red,<-] (v2)--(v6);
\draw[thick,blue,<-] (v3)--(v7);

\tikzstyle{every node}=[circle,minimum size=3.5pt,inner sep=0pt,draw,fill]
\node[gray] at (n1) {};
\node[gray] at (n2) {};
\node[gray] at (n3) {};
\node[gray] at (n4) {};
\node at (v4) {};

\tikzstyle{every node}=[inner sep=1pt]
\begin{footnotesize}
\node at (tau) {$\delta$};
\end{footnotesize}
\begin{tiny}
\node at (lv0) {$S^0$};
\node at (lv1) {$L^1_1$};
\node at (lv2) {$L^1_2$};
\node at (lv3) {$R^0_1$};
\node at (lv4) {$S^1$};
\node at (lv5) {$L^0_1$};
\node at (lv6) {$L^0_2$};
\node at (lv7) {$R^1_1$};
\node at (ln1) {$N_1$};
\node at (ln2) {$N_2$};
\node at (ln3) {$N_3$};
\node at (ln4) {$N_4$};
\end{tiny}
\draw[white,-] (-2.8,-2.8)--(-2.2,-2.8);
\draw[white,-] (2.5,2.8)--(2.2,2.8);
\end{tikzpicture}
\hspace{1.5cm}
\begin{tikzpicture}[xscale=0.6,yscale=0.6,>=latex]
\coordinate (center) at (0,0) {};

\coordinate (v0) at ($(center)+(0:2cm)$) {};
\coordinate (v1) at ($(center)+(315:2cm)$) {};
\coordinate (v2) at ($(center)+(270:2cm)$) {};
\coordinate (v3) at ($(center)+(225:2cm)$) {};
\coordinate (v4) at ($(center)+(180:2cm)$) {};
\coordinate (v5) at ($(center)+(135:2cm)$) {};
\coordinate (v6) at ($(center)+(90:2cm)$) {};
\coordinate (v7) at ($(center)+(45:2cm)$) {};

\coordinate (lv0) at ($(center)+(0:2.5cm)$) {};
\coordinate (lv1) at ($(center)+(315:2.5cm)$) {};
\coordinate (lv2) at ($(center)+(270:2.5cm)$) {};
\coordinate (lv3) at ($(center)+(225:2.5cm)$) {};
\coordinate (lv4) at ($(center)+(180:2.5cm)$) {};
\coordinate (lv5) at ($(center)+(135:2.5cm)$) {};
\coordinate (lv6) at ($(center)+(90:2.5cm)$) {};
\coordinate (lv7) at ($(center)+(45:2.5cm)$) {};

\coordinate (n4) at ($(center)+(22.5:2cm)$) {};
\coordinate (n3) at ($(center)+(67.5:2cm)$) {};
\coordinate (n2) at ($(center)+(112.5:2cm)$) {};
\coordinate (n1) at ($(center)+(157.5:2cm)$) {};

\coordinate (ln4) at ($(center)+(22.5:2.5cm)$) {};
\coordinate (ln3) at ($(center)+(67.5:2.5cm)$) {};
\coordinate (ln2) at ($(center)+(112.5:2.5cm)$) {};
\coordinate (ln1) at ($(center)+(157.5:2.5cm)$) {};

\coordinate (tau) at (-2.5,-2.5) {};

\draw[fill=gray!30, draw=none] ($(center)+(-5:2cm)$) -- ($(center)+(0:2cm)$)-- ($(center)+(5:2cm)$) -- ($(center)+(175:2cm)$) -- ($(center)+(180:2cm)$) --($(center)+(185:2cm)$) -- cycle;

\draw[fill=gray!30, draw=none] ($(center)+(310:2cm)$) -- ($(center)+(315:2cm)$)-- ($(center)+(320:2cm)$) -- ($(center)+(130:2cm)$) -- ($(center)+(135:2cm)$) --($(center)+(140:2cm)$) -- cycle;

\draw[fill=gray!30, draw=none] ($(center)+(265:2cm)$) -- ($(center)+(270:2cm)$)-- ($(center)+(275:2cm)$) -- ($(center)+(85:2cm)$) -- ($(center)+(90:2cm)$) --($(center)+(95:2cm)$) -- cycle;

\draw[fill=gray!30, draw=none] ($(center)+(220:2cm)$) -- ($(center)+(225:2cm)$)-- ($(center)+(230:2cm)$) -- ($(center)+(40:2cm)$) -- ($(center)+(45:2cm)$) --($(center)+(50:2cm)$) -- cycle;

\draw (0,0) circle (2cm);

\draw[thick,red,->] (v1)--(v5);
\draw[thick,blue,<-] (v2)--(v6);
\draw[->] (v3)--(v7);

\tikzstyle{every node}=[circle,minimum size=3.5pt,inner sep=0pt,draw,fill]
\node at (v4) {};

\tikzstyle{every node}=[inner sep=1pt]
\begin{footnotesize}
\node at (tau) {$\sigma$};
\end{footnotesize}
\begin{tiny}
\node at (lv0) {$S^1$};
\node at (lv1) {$L^0_1$};
\node at (lv2) {$R^0_1$};
\node at (lv4) {$S^0$};
\node at (lv5) {$L^1_1$};
\node at (lv6) {$R^1_1$};

\end{tiny}
\draw[white,-] (-2.8,-2.8)--(-2.2,-2.8);
\draw[white,-] (2.5,2.8)--(2.2,2.8);
\end{tikzpicture}
\caption{\label{fig:Q-node-owner} 
To the left: ordering $\gamma$ of prime Q-node $Q$ binds $C$ in~$S^0$ and ordering $\delta$ binds $C$ in~$S^1$. 
To the right: ordering $\gamma$ of serial Q-node $Q$ binds $C$ in $S^0$, 
ordering $\sigma$ makes $C$ non-Helly.}
\end{figure}

Now, we are ready to characterize conformal models of $G$ 
in which $C$ satisfies the Helly property.
For this purpose, we say a node $N$ of $\pqmtree$ \emph{affects (the slot of)} $C$ if $N$ admits an ordering which binds $C$ in some slot of $S$.
\begin{lemma}
\label{lem:conformal-models-with-C-in-fixed-slot}
Suppose $C$ is a private clique which satisfies conditions \ref{prop:clique-type-clean}-\ref{prop:clique-type-vertices-of-C-from-outside-D} and let $j \in \{0,1\}$.
A conformal model $\phi$ of $G$ can be extended to a $C$-conformal model with the clique $C$ in the slot $S^j$ if and only if for every node $N$ affecting $C$ the ordering $\phi_{|N}$ binds $C$ in $S^j$.
\end{lemma}
\begin{proof}
Suppose $\phi$ is a $C$-conformal model of $G$ with $C$ in the slot $S^j$.
Lemmas~\ref{lem:deepest-owner-properties}.\eqref{item:deepest-owner-properties-model},
\ref{lem:M-node-owner-properties}.\eqref{item:M-node-owner-properties-model}, and~\ref{lem:Q-node-owner-properties}.\eqref{item:Q-node-owner-properties-model} assert that 
for every node $N$ affecting $C$ the ordering $\phi_{|N}$ binds $C$ in $S^j$.

Suppose $\phi$ is a conformal model of $G$ such that $\phi_{|N}$ binds $C$ in $S^j$ for every node $N$ affecting $C$.
Since, by Claim~\ref{claim:I_N-properties}.\eqref{item:I_N-properties-clique}, every element from $C$ is a member of some node $K \in \mathcal{I}(N)$ for some $N \in \owners \cup \{Q\}$, the properties of admissible orderings from $\Pi(N)$ binding $C$ in $S^j$
assert $\phi$ can be easily extended to a $C$-conformal model with $C$ in $S^j$ (we can choose any ordering for the nodes of $\pqmtree$ which do not affect $C$).
\end{proof}

Finally we determine the type of $C$.
\begin{theorem}
\label{thm:private_clique_classification}
Suppose $C$ is a private clique which satisfies conditions \ref{prop:clique-type-clean}-\ref{prop:clique-type-vertices-of-C-from-outside-D}.
If there is a node with an ordering which makes $C$ non-Helly or 
there are at least two nodes with orderings binding $C$, then $C$ is ambiguous.
Otherwise, $C$ is always Helly.
\end{theorem}
\begin{proof}
Lemmas \ref{lem:deepest-owner-properties}.\eqref{item:deepest-owner-properties-prime}, 
\ref{lem:M-node-owner-properties}.\eqref{item:M-node-owner-properties-prime}, and 
\ref{lem:Q-node-owner-properties}.\eqref{item:Q-node-owner-properties-prime} assert
every node $N$ affecting $C$ admits an ordering which binds $C$ in the slot $S^j$, for any 
$j \in \{0,1\}$.
Lemma~\ref{lem:conformal-models-with-C-in-fixed-slot} guarantees $G$ admits a conformal model in which $C$ satisfies the Helly property.

By Lemma~\ref{lem:conformal-models-with-C-in-fixed-slot}, if there is a serial node $N$ affecting $C$ with an ordering making $C$ non-Helly, the conformal models of $G$ which use this ordering on $N$ make $C$ non-Helly.
By the same reason, if we have two distinct nodes $N_0,N_1$ which bind $C$, 
the conformal models $\phi$ of $G$ in which $\phi_{|N_0}$ binds $C$ in $S^0$ and $\phi_{|N_{1}}$ binds $C$ in $S^1$, make $C$ non-Helly.
Otherwise, if there is exactly one node $N$ whose orderings binds $C$ either in $S^0$ or in $S^1$, 
then any conformal model $\phi$ of $G$ can be extended to a $C$-conformal model with $C$ in the slot $S^j$ for some unique $j \in \{0,1\}$
(that is, in $S^0$ if $\phi_{|N}$ binds $C$ in $S^0$ or in $S^1$ if $\phi_{|N}$ binds $C$ in $S^1$).
Finally, if there are no nodes affecting $C$, every conformal model of $G$ can be extended to a $C$-conformal model with $C$ in $S^j$ for both $j \in \{0,1\}$.
This concludes the proof.
\end{proof}

Finally, we can observe that we can test the type of $C$ in polynomial time,
thus proving Theorem~\ref{thm:intro:clique-type}.
For this purpose we check whether $C$ satisfies properties~\ref{prop:clique-type-clean}-\ref{prop:clique-type-vertices-of-C-from-outside-D}.
All those properties are easy to test except of checking whether $C$
contains a rigid non-Helly clique of size $\geq 4$.
For this we can use the fact that $C$ contains a rigid non-Helly clique of size $\geq 4$ iff the vertices of $C$ induce a non-chordal graph in the overlap graph $G_{ov}$.

\section{FPT algorithm for \hcp}
\label{sec:fpt-Helly-Cliques-problem}
The \hcp\ problem ($\HCP$) is defined as follows.

\medskip

\begin{tabular}{rl}
\textbf{Problem:} & \hcp \\
\textbf{Input:}& A circular-arc graph $G$ and some of its cliques $C_1,\ldots,C_k$ \\
\textbf{Question:} & Is there a normalized circular-arc model of $G$ in which all \\
&the cliques $C_1,\ldots,C_k$ satisfy the Helly property?
\end{tabular}
\medskip

Equivalently, in \hcp, for a given circular arc graph $G$ and its cliques $C_1,\ldots,C_k$  we need to decide whether $G$ admits a $\{C_1,\ldots,C_k\}$-conformal model. 

Let $(G,C_1,\ldots,C_k)$ be an instance of \hcp.
If $C_i$ is an inclusion-wise maximal clique in $G$, 
then the instance $(G, C_1, \ldots, C_k)$ is equivalent to the instance $(G', C_1, \ldots, C_{i-1},C_{i+1},\ldots,C_k)$, 
where $G'$ is obtained by extending $G$ by a vertex~$v$ with $N(v)=C_i$. 
Hence, if all the cliques on the input are maximal in $G$, 
\hcp\ can be solved by a linear time algorithm 
as it can be reduced to the recognition problem of 
circular-arc graphs.

In the rest of this section we will prove Theorem \ref{thm:helly-cliques},
which extends Theorem~\ref{thm:intro:helly-cliques-par}.\eqref{item:intro:helly-cliques-fpt} as follows: 
\begin{theorem} \hcp\ can be solved:
\label{thm:helly-cliques}
\begin{enumerate}
\item \label{item:helly-cliques-prime-parallel} in time $2^{k}n^{\Oh{1}}$ if we restrict to the instances $(G,C_1,\ldots,C_k)$
in which $V(G)$ is prime or parallel in $\strongModules(G_{ov})$.
\item \label{item:serial} in time $2^{\Oh{k\log{k}}}n^{\Oh{1}}$ if we restrict to the instances $(G,C_1,\ldots,C_k)$ in which $V(G)$ is serial in $\strongModules(G_{ov})$.
\end{enumerate}
In particular, \hcp\ can be solved in time $2^{\Oh{k\log{k}}}n^{\Oh{1}}$.
\end{theorem}

Suppose $(G,C_1,\ldots,C_k)$ is an instance of \hcp. 
If some clique $C_i$ is always-non-Helly, we immediately conclude that we are dealing with \no-instance.
Also, if $C_i$ is always-Helly, then $(G,C_1,\ldots,C_k)$ is equivalent to $(G,C_1,\ldots,C_{i-1},C_{i+1},\ldots,C_k)$.
So, we assume all the cliques $C_1,\ldots,C_k$ are ambiguous.

Suppose $\pqmtree$ is the PQM-tree of $G = (V,E)$ and
$G_{ov}=(V,{\sim})$ is the overlap graph of $G$.
Since $C_1,\ldots,C_k$ are ambiguous, $C_1,\ldots,C_k$ satisfy properties~\ref{prop:clique-type-clean}--\ref{prop:clique-type-no-rigid-non-Helly-subclique}. 
So, we assume $C_i$ is clean, we have $C_i \subseteq Q_i$ for a component $Q_i$ of $G_{ov}$, and $C_i$ contains no rigid non-Helly subclique, for every $i \in [k]$.

We consider two cases depending on the type of the module $V$ in $\strongModules(G_{ov})$.

\subsection{$V$ is prime or parallel in $\strongModules(G_{ov})$}
\label{subsec:fpt-first-case}
Since $V$ is prime/parallel, every $Q$-node in $\pqmtree$ is prime.
Since $C_1,\ldots,C_k$ are ambiguous and $
V$ is prime, Theorem~\ref{thm:public_clique_classification} asserts that all the cliques $C_1,\ldots,C_k$ are private, 
and since they are ambiguous, all of them satisfy properties \ref{prop:clique-type-clean}--\ref{prop:clique-type-vertices-of-C-from-outside-D}. 
We assume $C_i$ is private for a CA-module $S_i$, for $i \in [k]$.

Our algorithm iterates over all tuples in the set $\{0,1\}^{[k]}$ and 
for every tuple $(j_1,\ldots,j_k) \in \{0,1\}^{[k]}$ it verifies whether
there exists a $\{C_1,\ldots,C_k\}$-conformal model $\phi$ of $G$ such that 
the clique letter $C_i$ is contained in the slot $S^{j_i}_i$ for $i \in [k]$.

Let $(j_1,\ldots,j_k)$ be a tuple in $\{0,1\}^{[k]}$.
By Lemma~\ref{lem:conformal-models-with-C-in-fixed-slot} there is a $\{C_1,\ldots,C_k\}$-conformal model of $G$ with $C_i$ in $S^{j_i}$ if and only if for every $i \in [k]$ and for every node $N$ affecting $C_i$ the ordering
$\phi_{|N}$ binds $C$ in $S^{j_i}$.
Hence, our task is to check whether for every prime/serial node $N$ in $\pqmtree$ there is an ordering $\pi \in \Pi(N)$ such that for every $i \in [k]$
the ordering $\pi$ binds $C_i$ in $S^{j_i}$ whenever $N$ affects~$C_i$.

This is easy when $N$ is prime: ordering $\pi$ does not exist if and only if there are two cliques $C_{i_1}, C_{i_2} \in \{C_1,\ldots,C_k\}$ affected by $N$ and there is no ordering in $\Pi(N)$ which binds $C_{i_1}$ in $S^{j_{i_1}}$ and $C_{i_2}$ in $S^{j_{i_2}}$.

If $N$ is serial, then $N$ is an M-node as we are in the case when $V$ is prime or parallel. 
Suppose $N \subseteq S$ for a CA-module $S \in \camodules$.
Recall that the orderings $\pi$ in $\Pi(N)$ are in the correspondence with the linear orderings ${\prec_N}$ 
of the children of $N$ in $\pqmtree_S$.
Now, note that each constraint of the form \emph{$\pi$ binds $C_i$ in the slot $S^{i}$} enforces ${\prec_N}$-relation between some children of $N$ (see Lemmas \ref{lem:deepest-owner-properties}--\ref{lem:M-node-owner-properties}).
Eventually, the question of whether there exists $\pi \in \Pi(N)$ which binds $C_i$ in $S^{j_i}$ for every $C_i$ affected by $N$ is equivalent to testing whether a digraph $(N, {\prec_N})$ consisting of the ${\prec_N}$-edges enforced by the binding constraints, is acyclic.
Clearly, this condition can be tested in polynomial-time.

This completes the proof of Theorem \ref{thm:helly-cliques}.\eqref{item:helly-cliques-prime-parallel}.

\subsection{$V$ is serial in $\strongModules(G_{ov})$}
First observe that, contrary to the previous case, 
we can not assume that all the input cliques $C_1,\ldots,C_k$ are private;
indeed, Theorem~\ref{thm:public_clique_classification} asserts that a public clique $C_i$ intersecting at least three CA-modules of $G$ is ambiguous.
So, in this case we construct an algorithm which exploits PQM-tree of $G$ and the Trapezoid Lemma.

First, for every $S \in \camodules$ we let
$$
\begin{array}{rcl}
\cliques(S) & = & \{C_i: C_i \cap S \neq \emptyset \}, \\
\priv(S) &=& \{C_i: C_i \text{ is private for }S\}.
\end{array}
$$
Clearly, $C_i$ is public if and only if $C_i$ is contained in no set $\priv(S)$ for every $S \in \camodules$
and $C_i$ is private if and only if $C_i$ is contained in exactly one set $\priv(S)$ for some $S \in \camodules$.

A circular ordering $C_{i_1}\ldots C_{i_k}$ of the clique set $\{C_1,\ldots,C_k\}$ is \emph{good} if there exists a $\{C_1,\ldots,C_k\}$-conformal model $\phi$ of $G$ 
such that $$\phi|\{C_1,\ldots,C_k\} \equiv C_{i_1} \ldots C_{i_k}.$$
Model $\phi$ satisfying the above property is called \emph{$C_{i_1} \ldots C_{i_k}$-conformal}.

Let $\phi$ be a $C_{i_1}\ldots C_{i_k}$-conformal model of $G$.
Note that we can always extend the model $\phi$ by an oriented chord $s^0s^1$, called a \emph{stabilizer}, which satisfies the property:
\begin{itemize}
 \item For every $S \in \camodules$ the chord $s^0s^1$ has the slots $\phi|S^0$ and $\phi|S^1$ on its different sides.
\end{itemize}
Clearly, such an extension of $\phi$ is always possible:
for example, we may pick any $S \in \camodules$ and add the chord $s^0s^1$ such that $s^0$ is just before the slot $\phi|S^0$ and $s^1$ is just before the slot $\phi|S^1$.
Usually we draw the extended models $\phi$ on two parallel lines $A$ and $B$, with $A$ above $B$, putting the points of $\phi$ strictly between $s^0$ and $s^1$ on the line $A$, and the points of $\phi$ strictly between $s^1$ and $s^0$ on the line $B$,
and drawing the chord $s^0s^1$ in parallel and in between $A$ and $B$.
See Figure~\ref{fig:extended-model} for an illustration.

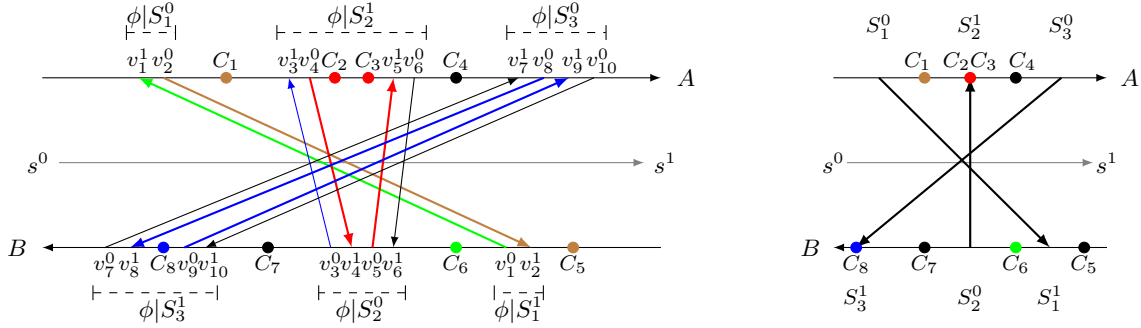
\begin{figure}[htp!]
\centering
\begin{tikzpicture}[yscale=0.75,xscale=1.1,>=latex]
\coordinate (label) at (1.5,-1) {};

\coordinate (a1) at (-2.55,3) {};
\coordinate (a2) at (-2.25,3) {};

\draw[|-|,dashed] (-2.7,3.8)--(-2.1,3.8);
\coordinate (laS1) at (-2.4,4.1) {};

\coordinate (a3) at (-1.5,3) {};

\coordinate (a4) at (-0.75,3) {};
\coordinate (a5) at (-0.5,3) {};
\coordinate (a6) at (-0.2,3) {};
\coordinate (a7) at (0.2,3) {};
\coordinate (a8) at (0.5,3) {};
\coordinate (a9) at (0.75,3) {};

\draw[|-|,dashed] (-0.9,3.8)--(0.9,3.8);
\coordinate (laS2) at (0,4.1) {};

\coordinate (a10) at (1.25,3) {};

\coordinate (a11) at (2.,3) {};
\coordinate (a12) at (2.3,3) {};
\coordinate (a13) at (2.6,3) {};
\coordinate (a14) at (2.9,3) {};

\draw[|-|,dashed] (1.85,3.8)--(3.05,3.8);
\coordinate (laS3) at (2.45,4.1) {};

\coordinate (b1) at (2.65,0) {};
\coordinate (b2) at (2.15,0) {};
\coordinate (b3) at (1.85,0) {};

\draw[|-|,dashed] (2.3,-.8)--(1.7,-.8);
\coordinate (lbS1) at (2,-1.1) {};

\coordinate (b4) at (1.25,0) {};

\coordinate (b5) at (0.5,0) {};
\coordinate (b6) at (0.25,0) {};
\coordinate (b7) at (0.,0) {};
\coordinate (b8) at (-0.25,0) {};

\draw[|-|,dashed] (0.65,-.8)--(-0.4,-.8);
\coordinate (lbS2) at (0.1,-1.1) {};

\coordinate (b9) at (-1,0) {};

\coordinate (b10) at (-1.75,0) {};
\coordinate (b11) at (-2,0) {};

\coordinate (b12) at (-2.25,0) {};

\coordinate (b13) at (-2.65,0) {};
\coordinate (b14) at (-2.95,0) {};

\draw[|-|,dashed] (-3.1,-.8)--(-1.6,-.8);
\coordinate (lbS3) at (-2.25,-1.1) {};


\coordinate (la1) at (-2.55,3.3) {};
\coordinate (la2) at (-2.25,3.3) {};
\coordinate (la3) at (-1.5,3.3) {};

\coordinate (la4) at (-0.75,3.3) {};
\coordinate (la5) at (-0.5,3.3) {};
\coordinate (la6) at (-0.2,3.3) {};
\coordinate (la7) at (0.2,3.3) {};

\coordinate (la8) at (0.5,3.3) {};
\coordinate (la9) at (0.75,3.3) {};
\coordinate (la10) at (1.25,3.3) {};

\coordinate (la11) at (2.,3.3) {};
\coordinate (la12) at (2.3,3.3) {};
\coordinate (la13) at (2.65,3.3) {};
\coordinate (la14) at (3,3.3) {};

\coordinate (lb1) at (2.65,-0.3) {};
\coordinate (lb2) at (2.15,-0.3) {};
\coordinate (lb3) at (1.85,-0.3) {};

\coordinate (lb4) at (1.25,-0.3) {};

\coordinate (lb5) at (0.5,-0.3) {};
\coordinate (lb6) at (0.25,-0.3) {};
\coordinate (lb7) at (0.,-0.3) {};
\coordinate (lb8) at (-0.25,-0.3) {};

\coordinate (lb9) at (-1,-0.3) {};

\coordinate (lb10) at (-1.65,-0.3) {};
\coordinate (lb11) at (-1.95,-0.3) {};

\coordinate (lb12) at (-2.25,-0.3) {};

\coordinate (lb13) at (-2.65,-0.3) {};
\coordinate (lb14) at (-2.95,-0.3) {};

\coordinate (phiA) at (4,3);
\coordinate (phiB) at (-4,0);
\coordinate (ls0) at (-3.75,1.5);
\coordinate (ls1) at (3.75,1.5);

\draw[->] (-3.7,3) -- (3.7,3);
\draw[<-] (-3.7,0) -- (3.7,0);
\draw[->,gray] (-3.5,1.5) -- (3.5,1.5);

\draw[<-,thick,green] (a1)--(b3);
\draw[->,thick,brown] (a2)--(b2);

\draw[<-,blue] (a4)--(b8);
\draw[->,red,thick] (a5)--(b7);
\draw[<-,red,thick] (a8)--(b6);
\draw[->,black] (a9)--(b5);

\draw[<-,black] (a11)--(b14);
\draw[->,blue,thick] (a12)--(b13);
\draw[<-,blue,thick] (a13)--(b11);
\draw[->,black] (a14)--(b10);

\tikzstyle{every node}=[circle,minimum size=4pt,inner sep=0pt,draw,fill]
\node[brown] at (a3) {};
\node[red] at (a6) {};
\node[red] at (a7) {};
\node at (a10) {};

\node[brown] at (b1) {};
\node[green] at (b4) {};
\node at (b9) {};
\node[blue] at (b12) {};

\tikzstyle{every node}=[inner sep=1pt]
\begin{scriptsize}
\node at (phiA) {$A$};
\node at (phiB) {$B$};
\node at (ls0) {$s^0$};
\node at (ls1) {$s^1$};

\node at (laS1) {$\phi|S^0_1$};
\node at (laS2) {$\phi|S^1_2$};
\node at (laS3) {$\phi|S^0_3$};

\node at (lbS1) {$\phi|S^1_1$};
\node at (lbS2) {$\phi|S^0_2$};
\node at (lbS3) {$\phi|S^1_3$};

\end{scriptsize}

\begin{tiny}
\node at (la1) {$v^1_1$};
\node at (la2) {$v^0_2$};
\node at (la3) {$C_1$};
\node at (la4) {$v^1_3$};
\node at (la5) {$v^0_4$};
\node at (la6) {$C_2$};
\node at (la7) {$C_3$};
\node at (la8) {$v^1_5$};
\node at (la9) {$v^0_6$};
\node at (la10) {$C_4$};
\node at (la11) {$v^1_7$};
\node at (la12) {$v^0_8$};
\node at (la13) {$v^1_9$};
\node at (la14) {$v^0_{10}$};

\node at (lb1) {$C_5$};
\node at (lb2) {$v^1_2$};
\node at (lb3) {$v^0_1$};
\node at (lb4) {$C_6$};
\node at (lb5) {$v^1_6$};
\node at (lb6) {$v^0_5$};
\node at (lb7) {$v^1_4$};
\node at (lb8) {$v^0_3$};
\node at (lb9) {$C_7$};
\node at (lb10) {$v^1_{10}$};
\node at (lb11) {$v^0_9$};
\node at (lb12) {$C_8$};
\node at (lb13) {$v^1_{8}$};
\node at (lb14) {$v^0_7$};
\end{tiny}
\draw[white] (-4.5,0)--(-4.5,-1.5);
\draw[white] (4.5,3)--(4.5,4.5);

\end{tikzpicture}
\hspace{0.5cm}
\begin{tikzpicture}[yscale=0.75,xscale=0.6,>=latex]
\coordinate (label) at (1.5,-1) {};

\coordinate (a1) at (1,3) {};
\coordinate (a2) at (2,3) {};
\coordinate (a3) at (3,3) {};
\coordinate (a4) at (4,3) {};
\coordinate (a5) at (5,3) {};

\coordinate (b1) at (5.5,0) {};
\coordinate (b2) at (4.75,0) {};
\coordinate (b3) at (4,0) {};
\coordinate (b4) at (3,0) {};
\coordinate (b5) at (2,0) {};
\coordinate (b6) at (0.5,0) {};

\coordinate (la1) at (1,3.3) {};
\coordinate (la2) at (1.85,3.3) {};
\coordinate (la3) at (3,3.3) {};
\coordinate (la4) at (4.15,3.3) {};
\coordinate (la5) at (5,3.3) {};

\coordinate (laa1) at (1,3.9) {};
\coordinate (laa2) at (1.85,3.9) {};
\coordinate (laa3) at (3,3.9) {};
\coordinate (laa4) at (4.15,3.9) {};
\coordinate (laa5) at (5,3.9) {};

\coordinate (lb1) at (5.5,-0.3) {};
\coordinate (lb2) at (4.75,-0.3) {};
\coordinate (lb3) at (4,-0.3) {};
\coordinate (lb4) at (3,-0.3) {};
\coordinate (lb5) at (2,-0.3) {};
\coordinate (lb6) at (0.5,-0.3) {};

\coordinate (lbb1) at (5.5,-0.9) {};
\coordinate (lbb2) at (4.75,-0.9) {};
\coordinate (lbb3) at (4,-0.9) {};
\coordinate (lbb4) at (3,-0.9) {};
\coordinate (lbb5) at (2,-0.9) {};
\coordinate (lbb6) at (0.5,-0.9) {};

\coordinate (phiA) at (6.5,3);
\coordinate (phiB) at (-0.5,0);
\coordinate (ls0) at (0,1.5);
\coordinate (ls1) at (6,1.5);

\draw[->] (0,3) -- (6,3);
\draw[<-] (0,0) -- (6,0);
\draw[->,gray] (0.3,1.5) -- (5.7,1.5);

\draw[->,thick] (a1)--(b2);
\draw[<-,thick] (a3)--(b4);
\draw[->,thick] (a5)--(b6);

\tikzstyle{every node}=[circle,minimum size=4pt,inner sep=0pt,draw,fill]
\node[brown] at (a2) {};
\node[red] at (a3) {};
\node at (a4) {};

\node at (b1) {};
\node[green] at (b3) {};
\node at (b5) {};
\node[blue] at (b6) {};

\tikzstyle{every node}=[inner sep=1pt]
\begin{scriptsize}
\node at (phiA) {$A$};
\node at (phiB) {$B$};
\node at (ls0) {$s^0$};
\node at (ls1) {$s^1$};



\end{scriptsize}

\begin{tiny}
\node at (laa1) {$S^0_1$};
\node at (la2) {$C_1$};
\node at (la3) {$C_2C_3$};
\node at (laa3) {$S^1_2$};
\node at (la4) {$C_4$};
\node at (laa5) {$S^0_3$};

\node at (lb1) {$C_5$};
\node at (lbb2) {$S^1_1$};
\node at (lb3) {$C_6$};
\node at (lbb4) {$S^0_2$};
\node at (lb5) {$C_7$};
\node at (lb6) {$C_8$};
\node at (lbb6) {$S^1_3$};
\end{tiny}
\draw[white] (-1,0)--(-1,-1.5);
\draw[white] (7,3)--(7,4.5);

\end{tikzpicture}

\caption{
\label{fig:extended-model}
An extended model $\phi$ and its skeleton $\phi^R$.
We have $v_1 \in C_6$, $v_2 \in C_1 \cap C_5$, $v_3 \in C_8$, 
$v_4, v_5 \in C_2 \cap C_3$. 
We have $\tau^A = C_1C_2C_3C_4$ and $\tau^B=C_5C_6C_7C_8$, $\tau^A(S_2) = C_2C_3$, $\tau^B(S_2) = \emptyset$, where $\emptyset$ denotes an empty word,
$\priv^A(S_2) = \{C_2,C_3\}$ and $\priv^B(S_2) = \emptyset$.
The model $(\phi|S^0_2, \phi|S^1_2) = (v_6^1v_5^0v_4^1v_3^0, v_3^1v_4^0C_2C_3v_5^1v_6^0)$ is $(\emptyset,C_2C_3)$-admissible for~$\SSS_2$.
We have $\mu^A = C_1 \{C_2,C_3\}C_4$ and $\mu^B = C_5C_6C_7C_8$, red dot is the clique-point of 
$C_2$ and $C_3$, green dot is the clique-point of $C_6$.
}
\end{figure}


Let $\phi$ be a $C_{i_1}\ldots C_{i_k}$-conformal model of $G$ extended by a stabilizer $s^0s^1$.
Let $\phi^{A}$ and $\phi^B$ be the contiguous subwords of $\phi$ 
containing the letters of $\phi$ strictly between $s^0$ and $s^1$ and 
strictly between $s^1$ and $s^0$, respectively.
Let $\tau(\phi), \tau^A, \tau^B$ be the restriction of $\phi, \phi^A, \phi^B$ 
to the set $\{C_1,\ldots,C_k\} \cup \{s^0,s^1\}$, respectively.
Also, for $S \in \camodules$ let $\tau^{A}(S)$ and $\tau^{B}(S)$
be the subwords of $\tau^A$ and $\tau^B$ containing the cliques from $\priv(S)$.
Note that $\tau^{A}(S)$ and $\tau^{B}(S)$, if non-empty, are contiguous in $\tau^A$ and $\tau^B$, respectively.
Next, let $\priv^A(S)$ and $\priv^B(S)$ be the cliques occurring in the words $\tau^A(S)$ and $\tau^B(S)$, respectively.
Note that $\priv(S) = \priv^A(S) \cup \priv^B(S)$.
Finally, let $\mu^A$ ($\mu^B$) be a word obtained from 
$\tau^A$ ($\tau^B$, respectively) by replacing every non-empty word $\tau^A(S)$ ($\tau^B(S)$, respectively) by the set $\priv^A(S)$ ($\priv^B(S)$, respectively), for $S \in \camodules$. 
See Figure~\ref{fig:extended-model} for an illustration.

Now, since $\phi$ is $\{C_1,\ldots,C_k\}$-conformal, the following property holds:
\begin{description}
 \item [\namedlabel{prop:P:admissible-models-for-CA-modules}{(P1)}] 
 For every $S \in \camodules$:
 \begin{itemize}
 \item If $\phi|S^0$ is above $s^0s^1$, then 
 $(\phi|S^0,\phi|S^1)$ is an \emph{$(\tau^A(S),\tau^B(S))$-admissible model for~$\SSS$}, which means that $(\phi|S^0,\phi|S^1)$ restricted to $S^*$ is admissible for~$\SSS$, restricted to $\Priv(S)$ coincides with $(\tau^A(S),\tau^B(S))$, and for every $C \in \Priv(S)$ and $v \in C \cap S$ the clique $C$ is on the left side of the chord $\phi(v)$,
 \item If $\phi|S^1$ is above $s^0s^1$, then 
 $(\phi|S^0,\phi|S^1)$ is an \emph{$\big{(}\tau^B(S),\tau^A(S)\big{)}$-admissible model for~$\SSS$}.
 \end{itemize}
\end{description}

Next, let $\phi_R$ denote a circular word that arises from $\phi$ by replacing, for every $S \in \camodules$, 
the subwords $\phi|S^0$ and $\phi|S^1$ by the letters $S^0$ and $S^1$, respectively.
The circular word $\phi_R$ obtained this way is called the \emph{skeleton} of $\phi$.
Clearly,
\begin{description}
 \item [\namedlabel{prop:P:chords_of_skeleton_are_pairwise_intersecting}{(P2)}] The chords $\phi_R(S)$ for $S \in \camodules$ are pairwise intersecting and all of them intersect the stabilizer $s^0s^1$.
\end{description}
See Figure~\ref{fig:extended-model} to the right for an illustration.

Now, we are going to describe some geometric properties of the skeleton $\phi_R$.
For this purpose, by slightly abusing our notation we represent every clique $C_i$ from $\{C_1,\ldots,C_k\}$ as the \emph{clique-point $C_i$} on the line $A$ or $B$, as follows.
We represent the letters from $\mu^A$ as points on the line $A$ and the letters from $\mu^B$ as points on the line $B$.
Then, the clique-point $C_i$ coincides with:
\begin{itemize}
\item the point $\priv^A(S)$ if $C_i \in \priv^A(S)$ for some $S \in \camodules$,
\item the point $\priv^B(S)$ if $C_i \in \priv^B(S)$ for some $S \in \camodules$,
\item the point $C_i$, otherwise.
\end{itemize}
See Figure~\ref{fig:extended-model} to the right for an illustration.

Now, note that we can easily assert the following properties of the oriented chord $\phi_R(S)$ for every $S \in \camodules$:
\begin{description}
 \item [\namedlabel{prop:T:trapezoids-not-private-cliques}{(T1)}]
For every $C_i \in \big{(}\cliques(S) \setminus \priv(S)\big{)}$, if the chords from $S \cap C$ are $(S^0,S^1)$-oriented ($(S^1,S^0)$-oriented), then the oriented chord $\phi_R(S)$ has the clique-point $C_i$ on its left side (right side, respectively).
\end{description}
\begin{description}
 \item [\namedlabel{prop:T:trapezoids-private-cliques}{(T2)}] If $\priv^A(S) \neq \emptyset$, then the upper tip of $\phi_R(S)$ coincides with the point $\priv^A(S)$. 
 If $\priv^B(S) \neq \emptyset$, then the lower tip of $\phi_R(S)$ coincides with the point $\priv^B(S)$.
\end{description}
Now let us consider the region where the oriented chord $\phi_R(S)$ can be placed in order to satisfy properties \ref{prop:T:trapezoids-not-private-cliques}--\ref{prop:T:trapezoids-private-cliques}.
Clearly, such region depends solely on the orientation of $\phi_R(S)$ but
in both cases it forms a trapezoid (possibly empty) with one base
(forming a possibly degenerated interval, open or closed on each side) in $A$ and the second one in $B$; we call such a trapezoid as \emph{spanned between $A$ and $B$}.  
Let $T_{S^0}$ ($T_{S^1}$) denote the largest trapezoid spanned between $A$ and $B$ such that
whenever $\phi_R(S)$ is downward oriented chord contained in $T_{S^0}$ (upward oriented chord contained in $T_{S^1}$, respectively), then $\phi_R(S)$ satisfies properties~\ref{prop:T:trapezoids-not-private-cliques}--\ref{prop:T:trapezoids-private-cliques}. 
So, since $\phi_R(S)$ satisfies properties~\ref{prop:T:trapezoids-not-private-cliques}--\ref{prop:T:trapezoids-private-cliques}, the following holds:
\begin{description}
 \item [\namedlabel{prop:P:chords_in_trapezoids}{(P3)}]
For every $S \in \camodules$:
\begin{itemize}
\item If $\phi|S^0$ is above $s^0s^1$, then $\phi_R(S) \subseteq T_{S^0}$.
\item If $\phi|S^0$ is below $s^0s^1$, then $\phi_R(S) \subseteq T_{S^1}$.
\end{itemize}
\end{description}
See Figure~\ref{fig:trapezoids} which shows trapezoids for slots shown in Figure~\ref{fig:extended-model}.

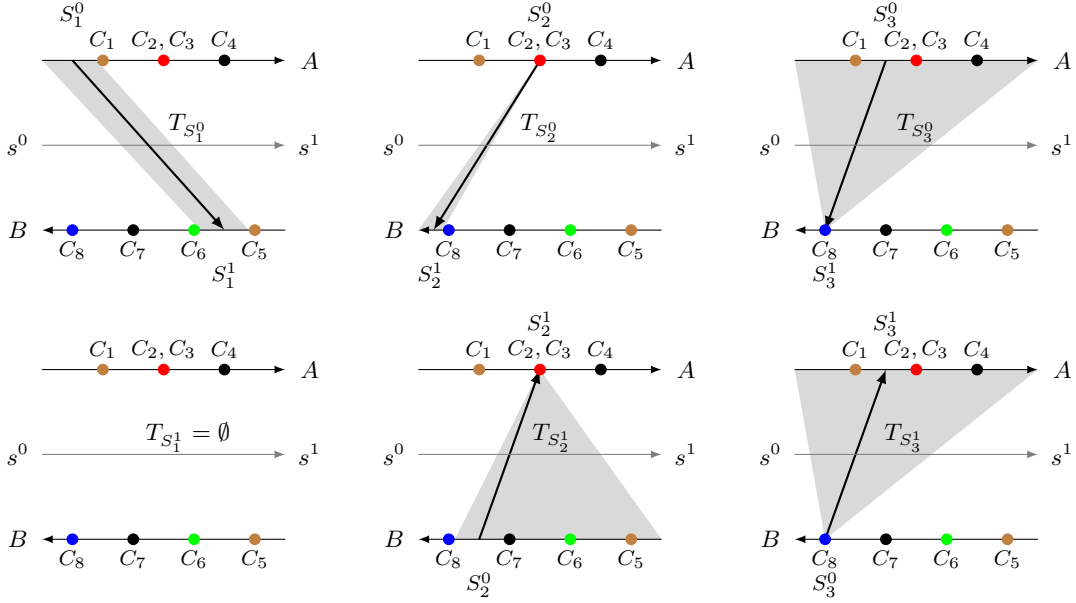
\begin{figure}[htp!]
\centering
\begin{tikzpicture}[yscale=0.75,xscale=0.8,>=latex]
\coordinate (label) at (1.5,-1) {};

\coordinate (a1) at (1,3) {};
\coordinate (a2) at (2,3) {};
\coordinate (a3) at (3,3) {};

\coordinate (b1) at (0.5,0) {};
\coordinate (b2) at (1.5,0) {};
\coordinate (b3) at (2.5,0) {};
\coordinate (b4) at (3.5,0) {};

\coordinate (la1) at (1,3.35) {};
\coordinate (la2) at (2,3.35) {};
\coordinate (la3) at (3,3.35) {};

\coordinate (lb1) at (0.5,-0.35) {};
\coordinate (lb2) at (1.5,-0.35) {};
\coordinate (lb3) at (2.5,-0.35) {};
\coordinate (lb4) at (3.5,-0.35) {};

\coordinate (laS) at (0.5,3.8) {};
\coordinate (lbS) at (3,-0.8) {};

\coordinate (A) at (4.4,3);
\coordinate (B) at (-0.4,0);
\coordinate (ls0) at (-0.4,1.5);
\coordinate (ls1) at (4.4,1.5);

\coordinate (T) at (2.4,1.8);

\draw[fill=gray!30, draw=none] (0,3) -- (0.9,3) -- (3.4,0) -- (2.6,0) -- cycle;

\draw[->] (0,3) -- (4,3);
\draw[<-] (0,0) -- (4,0);
\draw[->,gray] (0,1.5) -- (4,1.5);

\draw[->, thick] (0.5,3) -- (3.,0);

\tikzstyle{every node}=[circle,minimum size=4pt,inner sep=0pt,draw,fill]
\node[brown] at (a1) {};
\node[red] at (a2) {};
\node at (a3) {};

\node[blue] at (b1) {};
\node at (b2) {};
\node[green] at (b3) {};
\node[brown] at (b4) {};

\tikzstyle{every node}=[inner sep=1pt]
\begin{scriptsize}
\node at (A) {$A$};
\node at (B) {$B$};
\node at (ls0) {$s^0$};
\node at (ls1) {$s^1$};
\node at (T) {$T_{S^0_1}$};

\end{scriptsize}

\begin{tiny}
\node at (la1) {$C_1$};
\node at (la2) {$C_2,C_3$};
\node at (la3) {$C_4$};

\node at (lb1) {$C_8$};
\node at (lb2) {$C_7$};
\node at (lb3) {$C_6$};
\node at (lb4) {$C_5$};

\node at (laS) {$S^0_1$};
\node at (lbS) {$S^1_1$};
\end{tiny}
\draw[white] (-1,0)--(-1,-1.2);
\draw[white] (5,3)--(5,4.2);

\end{tikzpicture}
\begin{tikzpicture}[yscale=0.75,xscale=0.8,>=latex]
\coordinate (label) at (1.5,-1) {};

\coordinate (a1) at (1,3) {};
\coordinate (a2) at (2,3) {};
\coordinate (a3) at (3,3) {};

\coordinate (b1) at (0.5,0) {};
\coordinate (b2) at (1.5,0) {};
\coordinate (b3) at (2.5,0) {};
\coordinate (b4) at (3.5,0) {};

\coordinate (la1) at (1,3.35) {};
\coordinate (la2) at (2,3.35) {};
\coordinate (la3) at (3,3.35) {};

\coordinate (lb1) at (0.5,-0.35) {};
\coordinate (lb2) at (1.5,-0.35) {};
\coordinate (lb3) at (2.5,-0.35) {};
\coordinate (lb4) at (3.5,-0.35) {};

\coordinate (laS) at (2,3.8) {};
\coordinate (lbS) at (0.2,-0.8) {};

\coordinate (A) at (4.4,3);
\coordinate (B) at (-0.4,0);
\coordinate (ls0) at (-0.4,1.5);
\coordinate (ls1) at (4.4,1.5);

\coordinate (T) at (2,1.8);

\draw[fill=gray!30, draw=none] (2,3) -- (0.4,0) -- (0,0) -- cycle;
\draw[->, thick] (2,3) -- (0.25,0);

\draw[->] (0,3) -- (4,3);
\draw[<-] (0,0) -- (4,0);
\draw[->,gray] (0,1.5) -- (4,1.5);

\tikzstyle{every node}=[circle,minimum size=4pt,inner sep=0pt,draw,fill]
\node[brown] at (a1) {};
\node[red] at (a2) {};
\node at (a3) {};

\node[blue] at (b1) {};
\node at (b2) {};
\node[green] at (b3) {};
\node[brown] at (b4) {};

\tikzstyle{every node}=[inner sep=1pt]
\begin{scriptsize}
\node at (A) {$A$};
\node at (B) {$B$};
\node at (ls0) {$s^0$};
\node at (ls1) {$s^1$};
\node at (T) {$T_{S^0_2}$};

\end{scriptsize}

\begin{tiny}
\node at (la1) {$C_1$};
\node at (la2) {$C_2,C_3$};
\node at (la3) {$C_4$};

\node at (lb1) {$C_8$};
\node at (lb2) {$C_7$};
\node at (lb3) {$C_6$};
\node at (lb4) {$C_5$};

\node at (laS) {$S^0_2$};
\node at (lbS) {$S^1_2$};
\end{tiny}
\draw[white] (-1,0)--(-1,-1.2);
\draw[white] (5,3)--(5,4.2);

\end{tikzpicture}
\begin{tikzpicture}[yscale=0.75,xscale=0.8,>=latex]
\coordinate (label) at (1.5,-1) {};

\coordinate (a1) at (1,3) {};
\coordinate (a2) at (2,3) {};
\coordinate (a3) at (3,3) {};

\coordinate (b1) at (0.5,0) {};
\coordinate (b2) at (1.5,0) {};
\coordinate (b3) at (2.5,0) {};
\coordinate (b4) at (3.5,0) {};

\coordinate (la1) at (1,3.35) {};
\coordinate (la2) at (2,3.35) {};
\coordinate (la3) at (3,3.35) {};

\coordinate (lb1) at (0.5,-0.35) {};
\coordinate (lb2) at (1.5,-0.35) {};
\coordinate (lb3) at (2.5,-0.35) {};
\coordinate (lb4) at (3.5,-0.35) {};

\coordinate (laS) at (1.5,3.8) {};
\coordinate (lbS) at (0.5,-0.8) {};

\coordinate (A) at (4.4,3);
\coordinate (B) at (-0.4,0);
\coordinate (ls0) at (-0.4,1.5);
\coordinate (ls1) at (4.4,1.5);

\coordinate (T) at (2,1.8);

\draw[fill=gray!30, draw=none] (0,3) -- (4,3) -- (0.5,0) -- cycle;
\draw[->, thick] (1.5,3) -- (0.5,0);

\draw[->] (0,3) -- (4,3);
\draw[<-] (0,0) -- (4,0);
\draw[->,gray] (0,1.5) -- (4,1.5);

\tikzstyle{every node}=[circle,minimum size=4pt,inner sep=0pt,draw,fill]
\node[brown] at (a1) {};
\node[red] at (a2) {};
\node at (a3) {};

\node[blue] at (b1) {};
\node at (b2) {};
\node[green] at (b3) {};
\node[brown] at (b4) {};

\tikzstyle{every node}=[inner sep=1pt]
\begin{scriptsize}
\node at (A) {$A$};
\node at (B) {$B$};
\node at (ls0) {$s^0$};
\node at (ls1) {$s^1$};
\node at (T) {$T_{S^0_3}$};

\end{scriptsize}

\begin{tiny}
\node at (la1) {$C_1$};
\node at (la2) {$C_2,C_3$};
\node at (la3) {$C_4$};

\node at (lb1) {$C_8$};
\node at (lb2) {$C_7$};
\node at (lb3) {$C_6$};
\node at (lb4) {$C_5$};

\node at (laS) {$S^0_3$};
\node at (lbS) {$S^1_3$};
\end{tiny}
\draw[white] (-1,0)--(-1,-1.2);
\draw[white] (5,3)--(5,4.2);

\end{tikzpicture}

\begin{tikzpicture}[yscale=0.75,xscale=0.8,>=latex]
\coordinate (label) at (1.5,-1) {};

\coordinate (a1) at (1,3) {};
\coordinate (a2) at (2,3) {};
\coordinate (a3) at (3,3) {};

\coordinate (b1) at (0.5,0) {};
\coordinate (b2) at (1.5,0) {};
\coordinate (b3) at (2.5,0) {};
\coordinate (b4) at (3.5,0) {};

\coordinate (la1) at (1,3.35) {};
\coordinate (la2) at (2,3.35) {};
\coordinate (la3) at (3,3.35) {};

\coordinate (lb1) at (0.5,-0.35) {};
\coordinate (lb2) at (1.5,-0.35) {};
\coordinate (lb3) at (2.5,-0.35) {};
\coordinate (lb4) at (3.5,-0.35) {};

\coordinate (laS) at (0.5,3.8) {};
\coordinate (lbS) at (3,-0.8) {};

\coordinate (A) at (4.4,3);
\coordinate (B) at (-0.4,0);
\coordinate (ls0) at (-0.4,1.5);
\coordinate (ls1) at (4.4,1.5);

\coordinate (T) at (2.4,1.8);


\draw[->] (0,3) -- (4,3);
\draw[<-] (0,0) -- (4,0);
\draw[->,gray] (0,1.5) -- (4,1.5);


\tikzstyle{every node}=[circle,minimum size=4pt,inner sep=0pt,draw,fill]
\node[brown] at (a1) {};
\node[red] at (a2) {};
\node at (a3) {};

\node[blue] at (b1) {};
\node at (b2) {};
\node[green] at (b3) {};
\node[brown] at (b4) {};

\tikzstyle{every node}=[inner sep=1pt]
\begin{scriptsize}
\node at (A) {$A$};
\node at (B) {$B$};
\node at (ls0) {$s^0$};
\node at (ls1) {$s^1$};
\node at (T) {$T_{S^1_1} = \emptyset$};

\end{scriptsize}

\begin{tiny}
\node at (la1) {$C_1$};
\node at (la2) {$C_2,C_3$};
\node at (la3) {$C_4$};

\node at (lb1) {$C_8$};
\node at (lb2) {$C_7$};
\node at (lb3) {$C_6$};
\node at (lb4) {$C_5$};

\end{tiny}
\draw[white] (-1,0)--(-1,-1.2);
\draw[white] (5,3)--(5,4.2);

\end{tikzpicture}
\begin{tikzpicture}[yscale=0.75,xscale=0.8,>=latex]
\coordinate (label) at (1.5,-1) {};

\coordinate (a1) at (1,3) {};
\coordinate (a2) at (2,3) {};
\coordinate (a3) at (3,3) {};

\coordinate (b1) at (0.5,0) {};
\coordinate (b2) at (1.5,0) {};
\coordinate (b3) at (2.5,0) {};
\coordinate (b4) at (3.5,0) {};

\coordinate (la1) at (1,3.35) {};
\coordinate (la2) at (2,3.35) {};
\coordinate (la3) at (3,3.35) {};

\coordinate (lb1) at (0.5,-0.35) {};
\coordinate (lb2) at (1.5,-0.35) {};
\coordinate (lb3) at (2.5,-0.35) {};
\coordinate (lb4) at (3.5,-0.35) {};

\coordinate (laS) at (2,3.8) {};
\coordinate (lbS) at (1,-0.8) {};

\coordinate (A) at (4.4,3);
\coordinate (B) at (-0.4,0);
\coordinate (ls0) at (-0.4,1.5);
\coordinate (ls1) at (4.4,1.5);

\coordinate (T) at (2.2,1.8);

\draw[fill=gray!30, draw=none] (2,3) -- (0.6,0) -- (4,0) -- cycle;
\draw[<-, thick] (2,3) -- (1,0);

\draw[->] (0,3) -- (4,3);
\draw[<-] (0,0) -- (4,0);
\draw[->,gray] (0,1.5) -- (4,1.5);

\tikzstyle{every node}=[circle,minimum size=4pt,inner sep=0pt,draw,fill]
\node[brown] at (a1) {};
\node[red] at (a2) {};
\node at (a3) {};

\node[blue] at (b1) {};
\node at (b2) {};
\node[green] at (b3) {};
\node[brown] at (b4) {};

\tikzstyle{every node}=[inner sep=1pt]
\begin{scriptsize}
\node at (A) {$A$};
\node at (B) {$B$};
\node at (ls0) {$s^0$};
\node at (ls1) {$s^1$};
\node at (T) {$T_{S^1_2}$};

\end{scriptsize}

\begin{tiny}
\node at (la1) {$C_1$};
\node at (la2) {$C_2,C_3$};
\node at (la3) {$C_4$};

\node at (lb1) {$C_8$};
\node at (lb2) {$C_7$};
\node at (lb3) {$C_6$};
\node at (lb4) {$C_5$};

\node at (laS) {$S^1_2$};
\node at (lbS) {$S^0_2$};
\end{tiny}
\draw[white] (-1,0)--(-1,-1.2);
\draw[white] (5,3)--(5,4.2);
\end{tikzpicture}
\begin{tikzpicture}[yscale=0.75,xscale=0.8,>=latex]
\coordinate (label) at (1.5,-1) {};

\coordinate (a1) at (1,3) {};
\coordinate (a2) at (2,3) {};
\coordinate (a3) at (3,3) {};

\coordinate (b1) at (0.5,0) {};
\coordinate (b2) at (1.5,0) {};
\coordinate (b3) at (2.5,0) {};
\coordinate (b4) at (3.5,0) {};

\coordinate (la1) at (1,3.35) {};
\coordinate (la2) at (2,3.35) {};
\coordinate (la3) at (3,3.35) {};

\coordinate (lb1) at (0.5,-0.35) {};
\coordinate (lb2) at (1.5,-0.35) {};
\coordinate (lb3) at (2.5,-0.35) {};
\coordinate (lb4) at (3.5,-0.35) {};

\coordinate (laS) at (1.5,3.8) {};
\coordinate (lbS) at (0.5,-0.8) {};

\coordinate (A) at (4.4,3);
\coordinate (B) at (-0.4,0);
\coordinate (ls0) at (-0.4,1.5);
\coordinate (ls1) at (4.4,1.5);

\coordinate (T) at (1.8,1.8);

\draw[fill=gray!30, draw=none] (0,3) -- (4,3) -- (0.5,0) -- cycle;
\draw[<-, thick] (1.5,3) -- (0.5,0);

\draw[->] (0,3) -- (4,3);
\draw[<-] (0,0) -- (4,0);
\draw[->,gray] (0,1.5) -- (4,1.5);

\tikzstyle{every node}=[circle,minimum size=4pt,inner sep=0pt,draw,fill]
\node[brown] at (a1) {};
\node[red] at (a2) {};
\node at (a3) {};

\node[blue] at (b1) {};
\node at (b2) {};
\node[green] at (b3) {};
\node[brown] at (b4) {};

\tikzstyle{every node}=[inner sep=1pt]
\begin{scriptsize}
\node at (A) {$A$};
\node at (B) {$B$};
\node at (ls0) {$s^0$};
\node at (ls1) {$s^1$};
\node at (T) {$T_{S^1_3}$};

\end{scriptsize}

\begin{tiny}
\node at (la1) {$C_1$};
\node at (la2) {$C_2,C_3$};
\node at (la3) {$C_4$};

\node at (lb1) {$C_8$};
\node at (lb2) {$C_7$};
\node at (lb3) {$C_6$};
\node at (lb4) {$C_5$};

\node at (laS) {$S^1_3$};
\node at (lbS) {$S^0_3$};
\end{tiny}
\draw[white] (-1,0)--(-1,-1.2);
\draw[white] (5,3)--(5,4.2);

\end{tikzpicture}

\caption{
\label{fig:trapezoids}
Trapezoids $T_{S^0_1}$, $T_{S^1_1}$, $T_{S^0_2}$, $T_{S^1_2}$, $T_{S^0_3}$, $T_{S^1_3}$ for the instance 
shown in Figure~\ref{fig:extended-model}:
$\phi_R(S_1)$ must have $C_1$ and $C_5$ to the left and $C_6$ to the right (and hence $T_{S^1_1} = \emptyset$ as we can not have $\phi_R(S_1)$ upward oriented and having $C_5$ to the left and $C_6$ to the right),
$\phi_R(S_2)$ must have one end at $C_2,C_3$ and must have $C_8$ to the left,
$\phi_R(S_3)$ must have one end at $C_8$.
}
\end{figure}

Our algorithm iterates over all circular orders of the clique set $\{C_1,\ldots,C_k\}$ and 
for every circular order $C_{i_1} \ldots C_{i_k}$ it decides, in polynomial time, whether $C_{i_1}\ldots C_{i_k}$ is good.
Clearly, we accept $(G,C_1,\ldots,C_k)$ if at least one of the circular orders is good.
To test whether $C_{i_1}\ldots C_{i_k}$ is good, 
we iterate over all extensions of the circular word $C_{i_1},\ldots,C_{i_k}$ by the letters
$\{s^0,s^1\}$ and for each such extension $\tau$ we check whether there is an extended conformal model $\phi$ of 
$G$ such that $\tau = \tau(\phi)$.
Given~$\tau$, the algorithm computes the word~$\mu$, which defines the clique-point $C_i$ for every 
$C_i \in \{C_1,\ldots,C_k\}$, 
and the trapezoids $T_{S^0}$ and~$T_{S^1}$ for every $S \in \camodules$, 
exactly as we have computed~$\mu$ and $T_{S^0}$ and~$T_{S^1}$ from the word~$\tau(\phi)$.
Also, given $\tau$, we compute the sets $\priv^A(S)$, $\priv^B(S)$, and the words 
$\tau^A(S)$ and $\tau^B(S)$ for $S \in \camodules$.
Clearly, we reject $\tau$ if for some $S \in \camodules$ the cliques from $\priv^A(S)$ or from $\priv^B(S)$ are not consecutive in $\tau$. 

Our main task is to check whether there exits an extended $C_{i_1}\ldots C_{i_k}$-conformal model~$\phi$ of $G$ such that $\tau(\phi) = \tau$.
The most important question that arises when trying to construct such a model 
is whether we should place the slot $S^0$ above or below the stabilizer, for every $S \in \camodules$.
Clearly, we cannot check all possibilities as there can be exponentially many of them.
To avoid this problem, we construct a 2-SAT formula~$\Phi$, with variable $x_{S^0}$ for every CA-module $S \in \camodules$,
which has the property that the sets of true variables in the assignments satisfying $\Phi$
correspond to the sets of CA-modules~$S$ with the slot $S^0$ above the stabilizer in the conformal models $\phi$ of $G$ extending $\tau$.
Formally, we will construct $\Phi$ such that the following lemma is satisfied.
\begin{lemma}\ 
\label{lem:Phi-formula}
\begin{enumerate}
\item \label{item:Phi-formula-model-implies-assignment} For every $C_{i_1}\ldots C_{i_k}$-conformal model $\phi$ of $G$ extending $\tau$ an assignment $\alpha_{\phi}$
 given by
 $$\alpha_{\phi}(x_{S^0}) = 
 \left\{
 \begin{array}{lcl}
\text{T}& \text{if}& \text{$\phi|S^0$ is above $s^0s^1$} \\
\text{F}& \text{if}& \text{$\phi|S^0$ is below $s^0s^1$} \\
\end{array}
\right.
\quad \text{for $S \in \camodules$,}
$$
makes the formula $\Phi$ true.
\item \label{item:Phi-formula-assignment-implies-model} For every assignment $\alpha$ satisfying $\Phi$ 
there exists an extended $C_{i_1}\ldots C_{i_k}$-model $\phi$ of $G$ such that $\tau(\phi) = \tau$ and $\alpha_{\phi} = \alpha$.
\end{enumerate}
\end{lemma}
\noindent 
We first we define the clauses of $\Phi$ and then we prove that $\Phi$ satisfies the properties of Lemma~\ref{lem:Phi-formula}.

The clauses of $\Phi$ are defined as follows:
\begin{itemize}
\item For every $S \in \camodules$:
\begin{itemize}
 \item If there exists no $(\tau^A(S),\tau^B(S))$-admissible model for $\SSS$, 
 we add the clause $\lnot x_{S^0}$ to~$\Phi$.
 \item If there exists no $(\tau^B(S),\tau^A(S))$-admissible model for $\SSS$, 
 we add the clause $x_{S^0}$ to~$\Phi$.
\end{itemize}
We discuss how to test whether there is an $(\tau^A(S),\tau^B(S))$-admissible model ($(\tau^B(S),\tau^A(S))$-admissible) for $\SSS$ in the last paragraph of this section.
\item For every two distinct $S_1,S_2 \in \camodules$:
\begin{itemize}
 \item If $T_{S^0_1} \cap T_{S^0_2} = \emptyset$, we add the clause $\lnot (x_{S^0_1} \land x_{S^0_2})$ to $\Phi$.
 \item If $T_{S^0_1} \cap T_{S^1_2} = \emptyset$, we add the clause $\lnot (x_{S^0_1} \land \lnot x_{S^0_2})$ to $\Phi$.
 \item If $T_{S^1_1} \cap T_{S^0_2} = \emptyset$, we add the clause $\lnot ( \lnot x_{S^0_1} \land x_{S^0_2})$ to $\Phi$.
 \item If $T_{S^1_1} \cap T_{S^1_2} = \emptyset$, we add the clause $\lnot (\lnot x_{S^0_1} \land \lnot x_{S^0_2})$ to $\Phi$.
\end{itemize}
\end{itemize}

In the proof of Lemma~\ref{lem:Phi-formula} we use so-called \emph{Trapezoid Lemma}, whose proof (purely geometrical) can be found in the Appendix.
To state Trapezoid Lemma, we need a short definition.
We say that two trapezoids $T_1$ and $T_2$ spanned between $A$ and $B$ \emph{nicely intersect} if 
there is a segment $t_1$ spanned between $A$ and $B$ and contained in $T_1$ and a segment $t_2$ spanned between $A$ and $B$ and contained in $T_2$ such that $t_1$ and $t_2$ are intersecting strictly between the lines $A$ and $B$.
The Trapezoid Lemma asserts that for every family of pairwise nicely intersecting trapezoids one can pick a segment spanned between $A$ and $B$ from each trapezoid such that all the picked segments have different endpoints and are pairwise intersecting.

\begin{proof}[Proof of Lemma~\ref{lem:Phi-formula}]
Suppose $\phi$ is an extended $(C_{i_1},\ldots,C_{i_k})$-conformal model of $G$ such that $\tau(\phi) = \tau$ 
and suppose $\phi_R$ is the skeleton of $\phi$ satisfying properties \ref{prop:P:admissible-models-for-CA-modules}--\ref{prop:P:chords_in_trapezoids}.
We will show that the assignment $\alpha_{\phi}$ satisfies formula $\Phi$.

Let $S \in \camodules$. 
Suppose $x_{S^0}=T$ in $\alpha_{\phi}$, which means $\phi|S^0$ is above $s^0s^1$. 
Then, property \ref{prop:P:admissible-models-for-CA-modules} asserts $(\phi|S^0,\phi|S^1)$ is an $(\tau^A(S),\tau^B(S))$-admissible model of $\SSS$, 
and hence we do not have the clause $\lnot x_{S^0}$ in $\Phi$.
If $x_{S^0}=F$ in $\alpha_{\phi}$, which means $\phi|S^0$ is below $s^0s^1$, 
then $(\phi|S^0,\phi|S^1)$ is an $(\tau^B(S),\tau^A(S))$-admissible model of $\SSS$, 
and hence we do not have the clause $x_{S^0}$ in $\Phi$.
This proves $\alpha_{\phi}$ satisfies all one-literal clauses in $\Phi$.

Let $S_1,S_2$ be two distinct CA-modules from $\camodules$.
Suppose $x_{S_1^0}=x_{S_2^0}=T$ in $\alpha_{\phi}$, which means $\phi|S_1^0$ and $\phi|S_2^0$ are above $s^0s^1$ in $\phi$. 
Clearly, we have $\phi_R(S_1) \subseteq T_{S^0_1}$ and $\phi_R(S_2) \subseteq T_{S^0_2}$ by property \ref{prop:P:chords_in_trapezoids},
and hence $T_{S_1^0} \cap T_{S_2^0} \neq \emptyset$ as $\phi_R(S_1)$ and $\phi_R(S_2)$ are intersecting by property~\ref{prop:P:chords_of_skeleton_are_pairwise_intersecting}.
In particular, we do not have the clause $\lnot (x_{S^0_1} \land x_{S^0_2})$ in $\Phi$.
The remaining cases are proven analogously.
This shows $\alpha_{\phi}$ satisfies all two-literal clauses in $\Phi$ and hence satisfies $\Phi$.
This completes the proof of statement~\eqref{item:Phi-formula-model-implies-assignment}.

Now, suppose $\alpha$ is an assignment that satisfies $\Phi$.
For every $S \in \camodules$ let 
$$S^{A} = 
\left\{
\begin{array}{cl}
S^0& \text{if $x_{S^0} = T$,} \\
S^1& \text{if $x_{S^0} = F$.} 
\end{array}
\right.
$$
and let $S^{B}$ be such that $\{S^A,S^B\} = \{S^0,S^1\}$.
We will show that there is a $C_{i_1} \ldots C_{i_k}$-admissible model $\phi$ of $G$ extending $\tau$
such that $\phi|S^A$ is above $s^0s^1$ for every $S \in \camodules$.

Let $S \in \camodules$.
The one-literal clauses of $\Phi$ assert there exists a model $(\phi^0_S,\phi^1_S)$ of $\SSS$ which is admissible for $(\tau^A(S),\tau^B(S))$ if $S^A=S^0$ and
admissible for $(\tau^B(S),\tau^A(S))$ if $S^A=S^1$.
Indeed, if $S^A=S^0$, then $x_{S^0} = T$ in $\alpha$, we do not have the clause $\lnot x_{S^0}$ in $\Phi$, and hence
there is $(\tau^A(S),\tau^B(S))$-admissible model $(\phi^0_S,\phi^1_S)$ for $\SSS$.
If $S^A=S^1$, then we have $x_{S^0} = F$ in $\alpha$, we do not have the clause $x_{S^0}$ in $\Phi$, and hence
there is $(\tau^B(S),\tau^A(S))$-admissible model $(\phi^0_S,\phi^1_S)$ for $\SSS$.

The two-literal clauses of $\Phi$ assert that for $S \in \camodules$ the trapezoids $T_{S^A}$ are pairwise intersecting.
Suppose $S_1$ and $S_2$ are two distinct CA-modules from $\camodules$.
Indeed, if $S^A_1 = S^0_1$, $S^A_2 = S^0_2$, then we have $x_{S^0_1} = x_{S^0_2} = T$ in $\alpha$, 
we do not have the clause $\lnot (x_{S^0_1} \land x_{S^0_2})$ in $\Phi$, 
and hence we have $T_{S^0_1} \cap T_{S^0_2} \neq \emptyset$.
The other cases are proven analogously.

Now, properties \ref{prop:P:chords_of_skeleton_are_pairwise_intersecting}--\ref{prop:P:chords_in_trapezoids} assert that the trapezoids $T_{S^A}$ for $S \in \mathcal{S}$ are pairwise intersecting.
We can easily observe that they are also pairwise nicely intersecting.
Hence, Trapezoid Lemma asserts that for every $S \in \camodules$ we can draw a segment $\phi_R(S)$ spanned between $A$ and $B$ in every trapezoid $T_{S^A}$ such that the drawn segments have different endpoints and are pairwise intersecting. 
We orient $\phi_R(S)$ such that $S^A$ is above $s^0s^1$, for every $S \in \camodules$.
Now, for $S \in \camodules$ we replace the letters $S^0$ and $S^1$ in $\phi_R$ by 
the words $\phi^0_S$ and $\phi^1_S$, respectively, obtaining the circular word $\phi$.
One can easily check that $\phi$ is an extended $C_{i_1},\ldots,C_{i_k}$-admissible model for $G$ such that $\tau(\phi) = \tau$ and $\alpha_{\phi} = \alpha$.
\end{proof}

We end this section by sketching a poly-time algorithm testing whether there is 
a $(\tau^A(S), \tau^B(S))$-admissible model for $\SSS$, for $S \in \camodules$.
Suppose we search for a $(\tau^A(S), \tau^B(S))$-admissible model for $\SSS$ (the other case is similar).
Recall that $\tau^A(S)$ and $\tau^B(S)$ are permutations of $\priv^A(S)$ and $\priv^B(S)$, respectively, 
and $\priv^A(S) \cup \priv^B(S) = \priv(S)$.
Using the same ideas as in Subsection \ref{subsec:fpt-first-case} we add a set of restrictions on
admissible orderings of some M-nodes of $\pqmtree_S$ which assert that every admissible model 
for $\SSS$ admitting those restrictions can be extended to $(\mu^A(S),\mu^B(S))$-admissible model of~$\SSS$, 
where $\mu^X(S)$ is a permutation of $\Priv^X(S)$  for $X \in \{A,B\}$.
We want to test whether any of those models can be turned into $(\tau^A(S), \tau^B(S))$-admissible model for $\SSS$.
Clearly, there is no such model if for some two cliques $C,C' \in \priv^A(S)$ with $C'$ occurring before $C$ in $\tau^A(S)$ there are $a,b \in S$ such that $b <_S a$, $b^1 \in S^0$, $a^0 \in S^0$, $b \in C$, and $a \in C'$
(see Figure~\ref{fig:admissible-models-forb-structures} to the left).
So, suppose that such configurations do not exist for cliques in $\tau^A(S)$ 
and analogously for cliques in $\tau^B(S)$.
Now, we add additional restrictions on the set of admissible orderings $(\phi^0_S,\phi^1_S)$ of $\SSS$ 
which assert that for every two cliques $C,C' \in \priv^A(S)$ with $C'$ occurring before $C$ in $\tau^A(S)$
we do not have $b,a \in S$ such that $b \sim a$, $b^1 \in S^0$, $a^0 \in S^0$, $b \in C$, $a \in C'$, 
and $b^1$ is before $a^0$ in $\phi^0_S$ (see Figure~\ref{fig:admissible-models-forb-structures} to the right).
We add analogous restrictions for every two cliques from $\tau^B(S)$.
Clearly, there is no $\big{(}\tau^A(S), \tau^B(S)\big{)}$-admissible model for $\SSS$ if the set of models satisfying all restrictions is empty.
Otherwise, we start with any $(\mu^A(S),\mu^B(S))$-admissible model $(\phi^0_S,\phi^1_S)$ for $\SSS$, 
where $\mu^X(S)$ is a permutation of $\Priv^X(S)$ for $X \in \{A,B\}$.
One can easily observe that whenever the order of two neighbouring cliques $C$ and $C'$ in $\phi^0_S$ is inconsistent with its order in $\tau^A(S)$, we can shift and swap $C$ and $C'$ inside $\phi^0_S$ 
as we have excluded all the configurations which could forbid such a swap.
We proceed analogously with the cliques in $\tau^B(S)$.
Clearly, the above ideas allow us to design a poly-time algorithm testing existence of
a $\big{(}\tau^A(S), \tau^B(S)\big{)}$-admissible model for $\SSS$ with $S^0$ above the stabilizer.

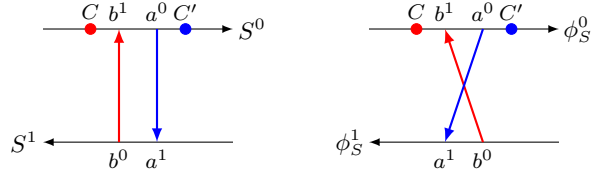
\begin{figure}[htp!]
\centering
\begin{tikzpicture}[yscale=0.5,xscale=0.5,>=latex]
\coordinate (label) at (1.5,-1) {};

\coordinate (a1) at (1.25,3) {};
\coordinate (a2) at (2,3) {};
\coordinate (a3) at (3,3) {};
\coordinate (a4) at (3.75,3) {};

\coordinate (b1) at (2,0) {};
\coordinate (b2) at (3,0) {};

\coordinate (lb1) at (2,-0.45) {};
\coordinate (lb2) at (3,-0.45) {};

\coordinate (la1) at (1.25,3.45) {};
\coordinate (la2) at (2,3.45) {};
\coordinate (la3) at (3,3.45) {};
\coordinate (la4) at (3.75,3.45) {};

\coordinate (S0) at (5.5,3);
\coordinate (S1) at (-0.5,0);

\draw[->] (0,3) -- (5,3);
\draw[<-] (0,0) -- (5,0);

\draw[<-, thick, red] (a2) -- (b1);
\draw[->, thick, blue] (a3) -- (b2);

\tikzstyle{every node}=[circle,minimum size=4pt,inner sep=0pt,draw,fill]
\node[red] at (a1) {};
\node[blue] at (a4) {};

\tikzstyle{every node}=[inner sep=1pt]
\begin{scriptsize}
\node at (S0) {$S^0$};
\node at (S1) {$S^1$};
\end{scriptsize}

\begin{tiny}
\node at (la1) {$C$};
\node at (la2) {$b^1$};
\node at (la3) {$a^0$};
\node at (la4) {$C'$};

\node at (lb1) {$b^0$};
\node at (lb2) {$a^1$};

\end{tiny}
\draw[white] (-1,0)--(-1,-0.8);
\draw[white] (6,3)--(6,3.8);
\end{tikzpicture}
\hspace{0.5cm}
\begin{tikzpicture}[yscale=0.5,xscale=0.5,>=latex]
\coordinate (label) at (1.5,-1) {};

\coordinate (a1) at (1.25,3) {};
\coordinate (a2) at (2,3) {};
\coordinate (a3) at (3,3) {};
\coordinate (a4) at (3.75,3) {};

\coordinate (b1) at (2,0) {};
\coordinate (b2) at (3,0) {};

\coordinate (lb1) at (2,-0.45) {};
\coordinate (lb2) at (3,-0.45) {};

\coordinate (la1) at (1.25,3.45) {};
\coordinate (la2) at (2,3.45) {};
\coordinate (la3) at (3,3.45) {};
\coordinate (la4) at (3.75,3.45) {};

\coordinate (S0) at (5.5,3);
\coordinate (S1) at (-0.5,0);

\draw[->] (0,3) -- (5,3);
\draw[<-] (0,0) -- (5,0);

\draw[<-, thick, red] (a2) -- (b2);
\draw[->, thick, blue] (a3) -- (b1);

\tikzstyle{every node}=[circle,minimum size=4pt,inner sep=0pt,draw,fill]
\node[red] at (a1) {};
\node[blue] at (a4) {};

\tikzstyle{every node}=[inner sep=1pt]
\begin{scriptsize}
\node at (S0) {$\phi^0_S$};
\node at (S1) {$\phi^1_S$};
\end{scriptsize}

\begin{tiny}
\node at (la1) {$C$};
\node at (la2) {$b^1$};
\node at (la3) {$a^0$};
\node at (la4) {$C'$};

\node at (lb1) {$a^1$};
\node at (lb2) {$b^0$};

\end{tiny}
\draw[white] (-1,0)--(-1,-0.8);
\draw[white] (6,3)--(6,3.8);

\end{tikzpicture}
\caption{
\label{fig:admissible-models-forb-structures}
Forbidden configurations for the existence of $(\tau^A(S), \tau^B(S))$-admissible model for $\SSS$,
where $C,C'$ are two cliques from $\priv^A(S)$ with $C'$ before $C$ in $\tau^A(S)$.}
\end{figure}

\section{Polynomial kernel for the Helly Clique Problem}
\label{sec:helly-cliques-kernel}
In this section we prove the following.
\begin{theorem}
\hcp\ admits a kernel of size $O(k^6)$.
\end{theorem}

Suppose $G,C_1,\ldots,C_k$ is an input to the Helly Clique problem.
As in the previous section, we can assume $C_1,\ldots,C_k$ are ambiguous and 
satisfy properties~\ref{prop:clique-type-clean}--\ref{prop:clique-type-no-rigid-non-Helly-subclique}.
For every private clique $C_i \in \{C_1,\ldots,C_k\}$ by $S_{C_i}$ we denote the 
CA-module of $G$ such that $C_i \in \priv(S_{C_i})$.
We note that $S_{C_i}$ is not defined when $C_i$ is public, which 
might only happen when $V(G)$ is serial in $\strongModules(G_{ov})$.
For convenience, we slightly abuse our convention and we assume $\priv(V(G)) = \{C_1,\ldots,C_k\}$ if $V(G)$ is serial in $\strongModules(G_{ov})$, and $\priv(Q) = \bigcup\{\priv(S): S \in \camodules(Q)\}$ for every prime Q-node $Q$.
Also, we extend the notation from the previous section and for every M-node $M$ we let
$$
\begin{array}{rcl}
\cliques(M) & = & \{C_i: C_i \cap M \neq \emptyset \}, \\
\priv(M) &=& \{C_i: C_i \text{ is private for }M\}. \\
\end{array}
$$


Suppose $N$ is a prime node (M-node or Q-node) in $\pqmtree$ and $C,C'$ are two distinct cliques from $\{C_1,\ldots,C_k\}$ such that $C,C' \in \priv(N)$.
Note that $S_{C} = S_{C'}$ if $N$ is an M-node.
We say that: 
\begin{itemize}
\item \emph{$N$ binds $C$ and $C'$ on the same side} if any order from $\Pi(N)$ binds $C$ in $S^j_{C}$ and binds $C'$ in $S^j_{C'}$ for some $j \in \{0,1\}$, 
\item \emph{$N$ binds $C$ and $C'$ on different sides} if any order from  $\Pi(N)$ binds 
$C$ in $S^j_{C}$ and $C'$ in $S^{1-j}_{C'}$ for some $j \in \{0,1\}$.
\item \emph{$N$ binds $C$ and $C'$} if $N$ binds $C$ and $C'$ on the same side or on different sides.
\end{itemize}
Lemmas~\ref{lem:deepest-owner-properties}.\eqref{item:deepest-owner-properties-prime}, 
\ref{lem:M-node-owner-properties}.\eqref{item:M-node-owner-properties-prime} and~\ref{lem:Q-node-owner-properties}.\eqref{item:Q-node-owner-properties-prime} 
assert that $N$ binds $C$ and $C'$ if and only if $N$ binds $C$ and $C'$ either on the same or on different sides.  

Suppose $N$ is a serial node (M-node or Q-node) in $\pqmtree$ and $C,C'$ are two distinct cliques from $\{C_1,\ldots,C_k\}$ such that $C,C' \in \priv(N)$.
We say that:
\begin{itemize}
\item \emph{$N$ binds $C$ and $C'$ on the same side} if there are two distinct children $K,L$ of $N$ such that
$C,C' \in \priv(K)$, $C,C' \in \cliques(L)$, and the vertices from $C \cap L$ and $C' \cap L$ have the same orientation in $\LLL$,
\item \emph{$N$ binds $C$ and $C'$ on different sides} if there are two distinct children $K,L$ of $N$ in $\pqmtree_Q$ such that
$C,C' \in \priv(K)$, $C,C' \in \cliques(L)$, and the vertices from $C \cap L$ and $C' \cap L$ have different orientation in $\LLL$,
\item \emph{$N$ binds $C$ and $C'$} if $N$ binds $C$ and $C'$ either on the same or on different sides. 
\end{itemize}
Note that, since $C \in \priv(K)$ ($C' \in \priv(K)$), the vertices from $C \cap L$ ($C' \cap L$, respectively) have the same orientation in $\LLL$.

Lemma~\ref{lem:conformal-models-with-C-in-fixed-slot} yields the following:
\begin{observation}
\label{obs:block-partition}
Let $C, C'$ be two cliques bound by a node $N$ from $\pqmtree$ and let $\phi$ be a $\{C,C'\}$-conformal model of $G$.
Then:
\begin{itemize}
 \item If $C$ and $C'$ are bound on the same side, then there is $j \in \{0,1\}$ such that $C$ is contained in $\phi|S^j_{C}$ and $C'$ is contained in $\phi|S^j_{C'}$.
 \item If $C$ and $C'$ are bound on different sides, then there is $j \in \{0,1\}$ such that $C$ is contained in $\phi|S^j_{C}$ and $C'$ is contained in $\phi|S^{1-j}_{C'}$.
\end{itemize}
\end{observation}

Our kernelization algorithm traverses the trees $\pqmtree_Q$ in the bottom-up order and marks every inner node of $\pqmtree_Q$ either as \emph{important} or \emph{irrelevant}.
Also, when the algorithm marks some node as important, it also marks some of its vertices as \emph{important} 
(the vertices which remain unmarked can be considered as irrelevant).
Roughly speaking, the important vertices form a subset of $V(G)$ of size $O(k^6)$ 
which preserves all relevant information allowing to conclude that we are dealing with \no-instance in the case where there is no $\{C_1,\ldots,C_k\}$-conformal model of $G$.
Formally, the algorithm returns as a kernel a reduct of $G$ with respect to the set of important vertices and the cliques $C_1,\ldots,C_k$ restricted to the set of important vertices.
We will show that the restricted instance is equivalent to the input instance.

Now we show how we process a single Q-node $Q$ in $\pqmtree$.
We traverse the tree $\pqmtree_Q$ in the bottom-up order and 
for every inner node $N \in \pqmtree$ such that $\priv(N) \neq \emptyset$ 
we compute a partition $\blocks(N)$ of the set 
$\priv(N)$ into so-called \emph{blocks} of $N$.
Additionally, for every block $B \in \blocks(N)$ we maintain
a partition $\sides(B)$ of the set $B$ into so-called \emph{sides} of~$B$.

Suppose the sets $\blocks(K)$ and their sides are already computed for every child $K$ of~$N$ in $\pqmtree_Q$.
Note that the sets $\priv(K)$ are pairwise disjoint.
We compute the set $\blocks(N)$ as follows.
We first initialize the set $\blocks(N)$ such that it contains all the sets $\blocks(K)$ of the children $K$ of $N$ and blocks $\{C\}$ with one side empty for every $C$ such that $N$ is the deepest owner of $C$ (if this holds then we say that \emph{$C$ is introduced by $N$}).
Next, we iterate over all the pairs $C',C'' \in \priv(N)$ and we do the following:
\begin{itemize}
 \item if $C'$ and $C''$ are bound by $N$ on the same side, then:
 \begin{itemize}
    \item if $C'$ and $C''$ are in the same block $B \in \blocks(N)$ but 
    in different sides of $B$, then we reject the instance $G,C_1,\ldots,C_k$,
    \item if $C'$ and $C''$ are in different blocks, say
    $B' \in \blocks(N)$ and $B'' \in \blocks(N)$ respectively, then 
    we form a new block $B = B' \cup B''$ with two sides, one obtained by merging the side of $B'$ containing 
    $C'$ and the side of $B''$ containing $C''$, and the second by merging the two remaining blocks of $B'$ and $B''$.
    If the above holds, then we say $N$ \emph{merges two blocks through the cliques $C'$ and $C''$}.
 \end{itemize}
 \item if $C'$ and $C''$ are bound by $N$ on different sides, then:
 \begin{itemize}
    \item if $C'$ and $C''$ are in the same block $B \in \blocks(N)$ and 
    the same side of $B$, then we reject the instance $G,C_1,\ldots,C_k$,
    \item if $C'$ and $C''$ are in different blocks, say
    $B' \in \blocks(N)$ and $B'' \in \blocks(N)$ respectively, then 
    we form a new block $B = B' \cup B''$ with two sides, one obtained by merging the side of $B'$ containing 
    $C'$ and the side of $B''$ containing no $C''$, and the second by merging the two remaining blocks in $B'$ and $B''$.
    If the above holds, then we say $N$ \emph{merges two blocks through the cliques $C'$ and $C''$}.
 \end{itemize}
\end{itemize}
Observation~\ref{obs:block-partition} asserts there is no $\{C_1,\ldots,C_k\}$-conformal model of $G$
in the case we reject the input instance. 
Otherwise (the set $\blocks(N)$ is computed),
\begin{description}
 \item [\namedlabel{prop:IN:important-node-blocks}{(N)}] 
 We mark $N$ as irrelevant if $\blocks(N) = \blocks(K)$ for some child $K$ of $N$; otherwise, we mark $N$ as important.
\end{description}
When we mark $N$ as important, we mark some vertices of $N$ as \emph{important}, 
according to the following rules.

If $N$ is prime or parallel and $N$ was marked as important, then:
\begin{description}
 \item[\namedlabel{prop:IV:important-vertices-deepest-owner}{(V1)}] 
 If $C$ is introduced by $N$, we mark two vertices $a,b$ from $C \cap N$ with different orientations in $\NNN$. 
 Moreover, if $N$ affects $C$, we choose $a$ and $b$ such that $a \sim b$.
\end{description}
\begin{description}
 \item[\namedlabel{prop:IV:important-vertices-binding}{(V2)}] If $N$ merges two blocks through $C'$ and $C''$, then: 
 \begin{itemize}
 \item if $C'$ is not introduced by $N$, we mark any vertex $c' \in C' \cap N$ 
 from any child $N'$ of $N$ such that $C' \notin \priv(N')$.
 \item if $C''$ is not introduced by $N$, we mark any vertex $c'' \in C'' \cap N$ 
 from any child $N''$ of $N$ such that $C'' \notin \priv(N)$.
 \end{itemize}
\end{description}

If $N$ is serial and $N$ was marked as important,
the marking procedure of important vertices in $N$ is performed in two phases.
First, we first mark some children of $N$ as \emph{weakly important}.
\begin{description}
\item [\namedlabel{prop:IN:important-node-serial-child-private}{(W1)}] 
If $K$ is a child of $N$ such that $\priv(K) \neq \emptyset$, we mark $K$ as weakly important.
\end{description}
Also, we might mark as weakly important some children $K$ of $N$ such that $\priv(K) = \emptyset$.
The marking procedure for this case requires some preparation.
First, we define the set $\signatures(N)$ of \emph{signatures} of~$N$:
$$\signatures(N) = \big{\{} (A,B): A,B \subseteq \priv(N),\ A \cap B = \emptyset,\ |A \cup B| \leq 4   \big{\}}.$$ 
Next, for every child $K$ of $N$ such that $\priv(K) = \emptyset$ and $\cliques(K) \cap \priv(N) \neq \emptyset$ we define a partition of the set $\cliques(K) \cap \priv(N)$ into two sets $\cliques_{01}(K)$ and $\cliques_{10}(K)$:
$$
\begin{array}{ccl}
\cliques_{01}(K) &=& \{C \in \cliques(K) \cap \priv(N): \text{the chords from $C \cap K$ are $(K^0,K^1)$-oriented in $\KKK$}\}, \\
\cliques_{10}(K) &=& \{C \in \cliques(K) \cap \priv(N): \text{the chords from $C \cap K$ are $(K^1,K^0)$-oriented in $\KKK$}\},
\end{array}
$$
and the set $\signatures_K(N)$ of signatures of $N$ occurring in the node $K$:
$$
\signatures_K(N) = \big{\{} (A,B) \in \signatures(N): A \subseteq \cliques_{01}(K),\ B \subseteq \cliques_{10}(K) \big{\}}.$$ 
See Figure~\ref{fig:signatures} for an illustration.
Now, we can describe the procedure which marks as weakly important some children $K$ of $N$ with $\priv(K) \neq \emptyset$.
First, we denote all the signatures in the set $\signatures(N)$ as not visited.
Next, we iterate over the children $K$ of $N$ such that $\priv(K) \neq \emptyset$ (in any order) and we do the following:
\begin{description}
\item [\namedlabel{prop:IN:important-node-serial-child-no-priv}{(W2)}] 
For every signature $(A,B)$ occurring in the set $\signatures_K(N)$ which is not yet visited, 
we mark $(A,B)$ as visited and we mark the node $K$ as weakly important. 
\end{description}
Note that child $K$ of $N$ is not marked as weakly important if and only if $\priv(K) = \emptyset$ and all signatures from 
$\signatures_K(N)$ have been visited by some weakly important children of $N$.
Finally, we mark some vertices in the children of $N$ marked as weakly important:
\begin{description}
\item [\namedlabel{prop:IV:important-vertices-serial-node-priv}{(V3)}] 
For every weakly important child $K$ of $N$ and every $C \in \cliques(K) \cap \priv(N)$ 
mark a vertex from $C \cap K$.
\end{description}
This completes the description of the marking procedure.

\begin{figure}[htp!]
\centering
\begin{tikzpicture}[yscale=0.85,xscale=1.6,>=latex]

\coordinate (av0) at (-0.2,3) {};
\coordinate (av1) at (0.1,3) {};
\coordinate (ac1) at (0.5,3) {};
\coordinate (av2) at (0.9,3) {};
\coordinate (av3) at (1.1,3) {};
\coordinate (ac2) at (1.5,3) {};
\coordinate (av4) at (1.85,3) {};
\coordinate (av5) at (2.0,3) {};
\coordinate (ac3) at (2.5,3) {};
\coordinate (av6) at (2.9,3) {};
\coordinate (av7) at (3.1,3) {};
\coordinate (ac4) at (3.5,3) {};
\coordinate (av8) at (3.9,3) {};
\coordinate (av9) at (4.2,3) {};

\coordinate (bv0) at (3.9,0) {};
\coordinate (bv1) at (4.2,0) {};
\coordinate (bc5) at (3.5,0) {};
\coordinate (bv2) at (2.9,0) {};
\coordinate (bv3) at (3.1,0) {};
\coordinate (bc6) at (2.5,0) {};
\coordinate (bv4) at (2.05,0) {};
\coordinate (bv5) at (2.2,0) {};
\coordinate (bc7) at (1.5,0) {};
\coordinate (bv6) at (0.9,0) {};
\coordinate (bv7) at (1.1,0) {};
\coordinate (bc8) at (0.5,0) {};
\coordinate (bv8) at (-0.2,0) {};
\coordinate (bv9) at (0.1,0) {};

\coordinate (lav0) at (-0.2,3.3) {};
\coordinate (lav1) at (0.1,3.3) {};
\coordinate (laN1) at (0,3.8) {};
\coordinate (lac1) at (0.5,3.3) {};
\coordinate (lav2) at (0.85,3.3) {};
\coordinate (lav3) at (1.15,3.3) {};
\coordinate (laN2) at (1,3.8) {};
\coordinate (lac2) at (1.5,3.3) {};
\coordinate (lav4) at (1.8,3.3) {};
\coordinate (laN3) at (1.9,3.8) {};
\coordinate (lav5) at (2.05,3.3) {};
\coordinate (lac3) at (2.5,3.3) {};
\coordinate (lav6) at (2.85,3.3) {};
\coordinate (laN4) at (3,3.8) {};
\coordinate (lav7) at (3.15,3.3) {};
\coordinate (lac4) at (3.5,3.3) {};
\coordinate (lav8) at (3.9,3.3) {};
\coordinate (lav9) at (4.2,3.3) {};
\coordinate (laN5) at (4,3.8) {};

\coordinate (lbv0) at (3.9,-0.3) {};
\coordinate (lbv1) at (4.2,-0.3) {};
\coordinate (lbN1) at (4,-0.8) {};
\coordinate (lbc5) at (3.5,-0.3) {};
\coordinate (lbv2) at (2.85,-0.3) {};
\coordinate (lbN2) at (3,-0.8) {};
\coordinate (lbv3) at (3.15,-0.3) {};
\coordinate (lbc6) at (2.5,-0.3) {};
\coordinate (lbv4) at (2.0,-0.3) {};
\coordinate (lbN3) at (2.1,-0.8) {};
\coordinate (lbv5) at (2.25,-0.3) {};
\coordinate (lbc7) at (1.5,-0.3) {};
\coordinate (lbv6) at (0.85,-0.3) {};
\coordinate (lbN4) at (1,-0.8) {};
\coordinate (lbv7) at (1.15,-0.3) {};
\coordinate (lbc8) at (0.5,-0.3) {};
\coordinate (lbv8) at (-0.2,-0.3) {};
\coordinate (lbv9) at (0.1,-0.3) {};
\coordinate (lbN5) at (0,-0.8) {};

\coordinate (laK) at (2.3,3.8) {};
\coordinate (lbK) at (1.7,-0.8) {};

\coordinate (N0) at (4.75,3);
\coordinate (N1) at (-0.75,0);

\draw[fill=gray!30, draw=none] (2.15,3) -- (2.35,3) -- (1.85,0) -- (1.65,0) -- cycle;

\draw[->] (-0.5,3) -- (4.5,3);
\draw[<-] (-0.5,0) -- (4.5,0);

\draw[<-,thick] (av0)--(bv0);
\draw[->,thick] (av1)--(bv1);

\draw[<-,thick] (av2)--(bv2);
\draw[->,thick] (av3)--(bv3);

\draw[<-,thick, red] (av4)--(bv4);
\draw[->,thick, green] (av5)--(bv5);

\draw[<-,thick] (av6)--(bv6);
\draw[->,thick] (av7)--(bv7);

\draw[<-,thick] (av8)--(bv8);
\draw[->,thick] (av9)--(bv9);

\tikzstyle{every node}=[circle,minimum size=4pt,inner sep=0pt,draw,fill]
\node at (ac1) {};
\node[red] at (ac2) {};
\node[green] at (ac3) {};
\node[green] at (ac4) {};

\node[green] at (bc5) {};
\node[green] at (bc6) {};
\node[red] at (bc7) {};
\node at (bc8) {};

\tikzstyle{every node}=[inner sep=1pt]
\begin{scriptsize}
\node at (N0) {$N^0$};
\node at (N1) {$N^1$};
\end{scriptsize}

\begin{tiny}
\node at (lac1) {$C_1$};
\node at (lac2) {$C_2$};
\node at (lac3) {$C_3$};
\node at (lac4) {$C_4$};

\node at (lbc5) {$C_5$};
\node at (lbc6) {$C_6$};
\node at (lbc7) {$C_7$};
\node at (lbc8) {$C_8$};

\node at (lav0) {$v^1_0$};
\node at (lav1) {$v^0_1$};
\node at (lav2) {$v^1_2$};
\node at (lav3) {$v^0_3$};
\node at (lav4) {$v^1_4$};
\node at (lav5) {$v^0_5$};
\node at (lav6) {$v^1_6$};
\node at (lav7) {$v^0_7$};
\node at (lav8) {$v^1_8$};
\node at (lav9) {$v^0_9$};

\node at (lbv0) {$v^0_1$};
\node at (lbv1) {$v^1_1$};
\node at (lbv2) {$v^0_2$};
\node at (lbv3) {$v^1_3$};
\node at (lbv4) {$v^0_4$};
\node at (lbv5) {$v^1_5$};
\node at (lbv6) {$v^0_6$};
\node at (lbv7) {$v^1_7$};
\node at (lbv8) {$v^0_8$};
\node at (lbv9) {$v^1_9$};

\node at (laN1) {$N^0_1$};
\node at (laN2) {$N^0_2$};
\node at (laN3) {$N^0_3$};
\node at (laN4) {$N^0_4$};
\node at (laN5) {$N^0_5$};
\node at (laK) {$N^0_6$};

\node at (lbN1) {$N^1_1$};
\node at (lbN2) {$N^1_2$};
\node at (lbN3) {$N^1_3$};
\node at (lbN4) {$N^1_4$};
\node at (lbN5) {$N^1_5$};
\node at (lbK) {$N^1_6$};
\end{tiny}
\draw[white] (-0.5,0)--(-0.5,-0.5);
\draw[white] (4.5,3)--(4.5,3.5);

\end{tikzpicture}
\caption{
\label{fig:signatures}
Serial M-node $N$ and its six children $N_1,\ldots,N_5,N_6$
with $\cliques_{01}(N_1)=\cliques_{10}(N_5)=\{C_1,C_2,C_3,C_4\}$, $\cliques_{10}(N_1)=\cliques_{01}(N_5)=\{C_5,C_6,C_7,C_8\}$, $\cliques_{01}(N_2)=\{C_2\}$ and $\cliques_{10}(N_2)=\{C_1\}$,
$\cliques_{01}(N_3) = \{C_3,C_4,C_5,C_6\}$ and $\cliques_{10}(N_3) = \{C_2,C_7\}$
(asserted by $v_4 \in C_2,C_7$ and $v_5 \in C_3,C_4,C_5,C_6$), 
$\cliques_{01}(N_4) = \{C_4\}$ and $\cliques_{10}(N_4) = \{C_3\}$, 
$\cliques_{01}(N_6) =\{C_3,C_4,C_6\}$ and $\cliques_{10}(N_6) = \{C_2,C_7\}$.
So we have $\priv(N) = \{C_1,\ldots,C_8\}$ and $\priv(N_i) = \emptyset$ for $i \in [6]$.
The set $\signatures_{N_3}(N)$ contains the pairs $(A, B)$,
where $A$ is a subset of $\{C_3,C_4,C_5,C_6\}$ of size at most $2$ and $B$ is a subset of $\{C_2,C_7\}$.
Suppose the children of $N$ are processed in the order $N_1,\ldots,N_6$.
Then $N_1,N_2,N_3, N_4,N_5$ are marked as weakly important as they mark as visited 
signatures $(\{C_1\},\{C_5\})$, $(\{C_2\},\{C_1\})$, $(\{C_3\},\{C_2\})$, $(\{C_4\},\{C_3\})$,
and $(\{C_8,C_4\})$, respectively, $N_6$ is not weakly important as $sgn_{N_6}(N) \subseteq sgn_{N_3}(N)$.
Since $\tau^0 = C_2C_3C_4$ and $\tau^1=C_7C_6$, we have $C^0_L=C_2$, $C^1_L=C_7$,
$C^0_R=C_3$, $C^1_R=C_6$.
Since $(\{C_3,C_6\}, \{C_2,C_7\})$ is in the set $\signatures_{N_3}(N)$, we can place the chords from
$\phi(N_6)$ next to $N_2$ (inside the gray trapezoid).
}
\end{figure}

\begin{observation}
The algorithm marks $O(k)$ nodes of $\pqmtree$ as important, $O(k^5)$ nodes as weakly important, 
and $O(k^6)$ vertices of $G$ as important.
\end{observation}
\begin{proof}
We mark at most $O(k)$ nodes as important as the algorithm introduces $k$ singleton blocks and merges blocks at most $k-1$ times.
We mark $O(k^5)$ nodes as weakly important as the set $\signatures(N)$ has size $O(k^4)$.
Consequently, we mark $O(k^6)$ vertices as important.
\end{proof}

Let $R$ denote the set of important vertices in $G$.
We now formally define the \emph{reduct} of $G$ with respect to $R$.
\begin{definition}
We say that a circular-arc graph $G'=(V',E')$ is a \emph{reduct} of $G$ with respect to $R$ if:
\begin{enumerate}
    \item $R\subseteq V'$.
    \item For every conformal model $\phi$ of $G$ there is a conformal model $\phi'$ of $G'$ such that $\phi | R^* \equiv \phi' | R^*$.
    \item For every conformal model $\phi'$ of $G'$ there is a conformal model $\phi$ of $G$ such that $\phi | R^* \equiv \phi' | R^*$. 
\end{enumerate}
\end{definition}
As we show in the Appendix in Section~\ref{subsec:reducts}, there is a polynomial-time algorithm computing a reduct $G'=(V',E')$ of $G$ with respect to $R$ such that $|V'|\leq 12 \cdot |R|$.
To finish the kernelization, we let $G^*$ be a reduct of $G$ respecting $R$ and let  
$C^*_i$ be $C_i$ restricted to $R$.
Note that $C^*_i$ is non-empty for every $i \in [k]$.
\begin{lemma}
$(G^*,C^*_1,\ldots,C^*_k)$ is a kernel of size $O(k^6)$ for the instance $(G,C_1,\ldots,C_k)$ of \hcp.
\end{lemma}
\begin{proof}
Since the reduct $G^*$ of $G$ respecting $R$ has size linear in $|R|$,
we deduce the instance $(G^*,C^*_1,\ldots,C^*_k)$ has size $O(k^6)$.

We need to show that $(G,C_1,\ldots,C_k)$ is a \yes-instance if and only if 
$(G^*,C^*_1,\ldots,C^*_k)$ is a \yes-instance.

Suppose $(G,C_1,\ldots,C_k)$ is a \yes-instance of \hcp.
Let $\phi$ be a $\{C_1,\ldots,C_k\}$-conformal model of $G$.
Let $\phi^*$ be a conformal model of $G^*$ such that $\phi^*|R \equiv \phi|R$;
such a model exists as $G^*$ is the reduct of $G$ respecting $R$.
One can easily observe that $\phi^*$ can be easily extended to the $\{C^*_1,\ldots,C^*_k\}$-conformal model of $G^*$ as 
we can insert the clique letters $\{C^*_1,\ldots,C^*_k\}$ into $\phi^*$ so as its position relative to 
the chords $\phi^*(R)$ is the same as the relative position of the clique letters $\{C_1,\ldots,C_k\}$ to 
the chords $\phi(R)$ in the model $\phi$.

Suppose $(G^*,C^*_1,\ldots,C^*_k)$ is a \yes-instance of \hcp.
Let $\phi^*$ be a $\{C^*_1,\ldots,C^*_k\}$-conformal model of $G^*$.
Let $\phi$ be an admissible model of $G$ such that $\phi|R \equiv \phi^*|R$ -- such a model exists as $G^*$ is a reduct of $G$.
Next, we extend $\phi$ by the letters $\{C^*_1,\ldots,C^*_k\}$ so as $\phi$ is a $\{C^*_1,\ldots,C^*_k\}$-conformal model for $G$, and finally we replace the letter $C^*_i$ by $C_i$ in $\phi$ for $i \in [k]$.
As far we know that $C_i$ is on the left side of $\phi(v)$ for every $v \in C_i^*$ and every $i \in [k]$.

Now, for every Q-node $Q$ we proceed the trees $\pqmtree_Q$ bottom up (in any order) and we show we can modify $\phi$ so as for every node $N \in \pqmtree_Q$ with $\priv(N) \neq \emptyset$ the following properties hold:
\begin{description}
 \item [\namedlabel{prop:M:model-blocks-consistency}{(M1)}] For every block $B \in \blocks(N)$ with $\sides(B) = \{B',B''\}$ there
 is $j \in \{0,1\}$ such that for every $C' \in B'$ the clique $C'$ occurs in $\phi|S^j_{C}$ and for every $C'' \in B''$ the clique $C''$ occurs in $\phi|S^{1-j}_{C''}$.
\end{description}
\begin{description}
\item [\namedlabel{prop:M:model-cliques-consistency}{(M2)}] For every $C \in \priv(N)$ and every $c \in C \cap N$ the point $C$ is on the left side of the chord $\phi(c)$.
\end{description}

Suppose $N$ is a prime M-node which was marked as irrelevant by rule \ref{prop:IN:important-node-blocks}.
In particular, it means there is a child $K$ of $N$ such that $\blocks(N) = \blocks(K)$.
So, we have also $\priv(N) = \priv(K)$.
Suppose $N \subseteq S$ for CA-module $S$ of $G$ and suppose $\{A',A''\}$ is a partition of the cliques affected by $N$ such that
any order $\pi \in \Pi(N)$ binds the cliques from $A'$ in $S^j$ and binds the cliques from $A''$ in $S^{1-j}$, for some $j \in \{0,1\}$.
Since $\blocks(N) = \blocks(K)$ and since $N$ did not result in rejection of the input instance and in merging two blocks of $N$,
we deduce that there is a block $B$ in $\blocks(K)$ which contains $A'$ and $A''$ in different sides.
Since the model $\phi$ satisfies \ref{prop:M:model-blocks-consistency} for the node $K$, there is $j \in \{0,1\}$ such that $B'$ is contained in $\phi|K^j$ and $B''$ is contained in $\phi|K^{1-j}$ for some $j \in \{0,1\}$.
Now, Lemma \ref{lem:M-node-owner-properties}.\eqref{item:M-node-owner-properties-prime} asserts we can adjust the ordering of $N$ in $\phi$ to obtain a model which satisfies conditions~\ref{prop:M:model-blocks-consistency}--\ref{prop:M:model-cliques-consistency} for $N$.

We can use an analogous argument (basing on Lemma~\ref{lem:Q-node-owner-properties}.\eqref{item:Q-node-owner-properties-prime}) to assert properties
\ref{prop:M:model-blocks-consistency} and \ref{prop:M:model-cliques-consistency} of $\phi$ for the case when $N$ is a prime irrelevant Q-node.

Suppose $N$ is a prime M-node which was marked as important by rule \ref{prop:IN:important-node-blocks}.
Suppose $N \subseteq S$ for CA-module $S$ of $G$.
Suppose $\{A',A''\}$ is a partition of the cliques affected by $N$ such that
any order $\pi \in \Pi(N)$ binds the cliques from $A'$ in $S^j$ and the cliques from from $A''$ in $S^{1-j}$, for some $j \in \{0,1\}$.
We claim that $\phi$ already satisfies condition \ref{prop:M:model-blocks-consistency} for $N$.
Clearly, $\phi$ satisfies condition \ref{prop:M:model-blocks-consistency} for every child $K$ of $N$.
Consider the case when $N$ merges two blocks $B_1$ and $B_2$ from $\blocks(N)$ through cliques $C'$ and $C''$. 
Suppose $C'$ and $C''$ are bound on the same side by $N$.
Note that rules \ref{prop:IV:important-vertices-deepest-owner} and \ref{prop:IV:important-vertices-binding} assert 
there are important vertices $a,b \in C' \cap N$ and $c,d \in C'' \cap N$ such that
$a \sim b$, $c \sim d$, $a$ and $b$ are in distinct children of $N$,  
and $c$ and $d$ are in distinct children of $N$.
Since $\phi(a)$ and $\phi(b)$ have $C'$ on the left side, $\phi(c)$ and $\phi(d)$ have $C''$ on the left side,
and since $\phi_{|N}$ binds $C$ and $C'$ in the same slot, we deduce $C$ and $C'$ occur in $\phi|S^j$ for some $j \in \{0,1\}$.
Since $\phi$ was consistent with the blocks of $N$ before the merge, 
we deduce $\phi$ is also consistent after the merge.
Finally, since $N$ did not result in rejection, 
we deduce there is a block of $N$ which contains $A'$ and $A''$ in different sides.
This shows $\phi$ satisfies \ref{prop:M:model-blocks-consistency} for $M$.
Note that, by Lemma~\ref{lem:M-node-owner-properties}.\eqref{item:M-node-owner-properties-prime}, property \ref{prop:M:model-cliques-consistency} is satisfied for every clique $C$ not introduced by $N$ as 
there are important vertices $a,b \in C \cap N$ contained in different children of $N$.
Also, by Lemma~\ref{lem:deepest-owner-properties}.\eqref{item:deepest-owner-properties-prime}, 
we can easily assert property \ref{prop:M:model-cliques-consistency} for every clique $C$ introduced by~$N$, by possibly shifting the clique $C$ inside the slot $\phi|N^j$ in which $C$ occurs.

We can use an analogous argument (basing on Lemma~\ref{lem:Q-node-owner-properties}.\eqref{item:Q-node-owner-properties-prime}) to assert property~\ref{prop:M:model-cliques-consistency} of $\phi$ for the case when $N$ is a prime important Q-node.

Suppose $N$ is a serial M-node and suppose $N$ was marked as irrelevant by rule~\ref{prop:IN:important-node-blocks}.
In particular, there is a child $K$ of $N$ such that $\blocks(N) = \blocks(K)$.
Suppose $N \subseteq S$ for CA-module $S$ of $G$.
Clearly, since $N$ did not result in rejection and $\phi$ satisfies property \ref{prop:M:model-blocks-consistency} for $K$, 
we deduce $\phi$ satisfies property \ref{prop:M:model-blocks-consistency} for $N$. 
Now, we show we can modify~$\phi$ in order to satisfy~\ref{prop:M:model-cliques-consistency} for $N$.
Let $L$ be a child of $N$ different than $K$.
Claim~\ref{claim:I_N-properties}.\eqref{item:I_N-properties-consistent-orientation} asserts that for every clique $C \in \cliques(L) \cap \priv(N)$ the chords 
from $C \cap L$ have the same orientation in~$\LLL$.
Clearly, if $\cliques(L) \cap \priv(N) = \{C\}$, we can shift the chords $\phi(L)$ in $\phi$ to fulfil property~\ref{prop:M:model-cliques-consistency} for vertices $c \in C \cap L$.
Suppose now that $|\cliques(L) \cap \priv(N)| \geq 2$ and suppose $A'$ and $A''$ 
is the partition of the set $\cliques(L) \cap \priv(N)$ such that
the elements from $C \cap L$ are $(L^0,L^1)$-oriented for $C \in A'$ 
and the elements from $C \cap L$ are $(L^1,L^0)$-oriented for $C \in A''$.
In particular, $N$ binds the cliques from $A'$ (or from $A''$) on the same side and the cliques from $A'$ and the cliques from $A''$ on the opposite sides. 
Again, since $\blocks(N) = \blocks(K)$ and since $N$ did not result in rejection, we deduce $A'$ and $A''$ are contained in different sides of a block $B$ from $\blocks(N)$.
Since $\blocks(N) = \blocks(K)$, we deduce there is $j \in \{0,1\}$ such that $A'$ is contained in $\phi|K^j$ and $A''$ is contained in $\phi|K^{1-j}$.
In particular, we can place the chords of $\phi(L)$ to satisfy condition ~\ref{prop:M:model-cliques-consistency} for every vertex $c \in C \cap L$ and every clique $C \in \cliques(L) \cap \priv(N)$.
This shows we can modify $\phi$ so as it satisfies property~\ref{prop:M:model-cliques-consistency} for $N$.

We can use an analogous argument to assert properties
\ref{prop:M:model-blocks-consistency} and \ref{prop:M:model-cliques-consistency} 
for the case when $N$ is a serial irrelevant Q-node.
Note that in this case $N=V$, $V$ is serial in $\strongModules(G_{ov})$, and
$\priv(S) = \{C_1,\ldots,C_k\}$ for some CA-module $S$ of $G$.

Suppose $N$ is a serial M-node and suppose $N$ was marked as important.
Suppose that $N \subseteq S$ for CA-module $S$ of $G$.
First note that $\phi$ satisfies the following properties with respect to every weakly important child $K$:
\begin{itemize}
 \item $\phi$ is $\priv(K)$-conformal model of $G$,
 \item for every $C \in \cliques(K) \cap \priv(N)$ the clique $C$ is to the left of every chord from $\phi(C \cap K)$.
\end{itemize}

Suppose $K$ was not marked as weakly important.
In particular, we have $\priv(K) = \emptyset$. 
Let $\tau^0$ be the sequence of all cliques from $\cliques(K)$ occurring in $\phi|N^0$ and let $\tau^1$ be the reverse of the sequence of all cliques from $\cliques(K)$ occurring in $\phi|N^1$ -- see Figure \ref{fig:signatures} for an illustration.

Suppose that the word $\tau^0$ ($\tau^1$) contains a clique from the set $\cliques_{01}(K)$; 
let $C^0_R$ ($C^1_R$, respectively) be the leftmost clique from $\tau^0$ (in $\tau^1$, respectively) with this property.
Suppose that the word $\tau^0$ ($\tau^1$) contains a clique from the set $\cliques_{10}(K)$; 
let $C^0_L$ ($C^1_L$, respectively) be the rightmost clique from $\tau^0$ (from $\tau^1$, respectively) with this property.
See Figure~\ref{fig:signatures} and the node $K=N_6$ for an illustration.
We claim that $C^0_L$ is before $C^0_R$ in the word $\tau^0$ and 
$C^1_L$ is before $C^1_R$ in $\tau^1$.
Indeed, note that $(\{C^0_R, C^1_R\},\{C^0_L, C^1_L\})$ is a signature in
the set $\signatures_K(N)$.
Since $K$ was not marked as weakly important, 
there is a weakly important child $L$ of $N$ which has 
the signature $(\{C^0_R, C^1_R\},\{C^0_L, C^1_L\})$ in the set $\signatures_L(N)$.
Now, we can place the chords from $\phi(K)$ just next to the chords $\phi(L)$ so
as $\phi$ is a valid conformal model of $G$ -- see Figure~\ref{fig:signatures}.
We proceed analogously in the case when some of the cliques $C^0_L$, $C^1_L$,$C^0_R$, $C^1_R$ is not defined.
Indeed, using similar arguments as above we can prove that the cliques from $\cliques_{10}(K)$ occur before the cliques from $\cliques_{01}(K)$ in $\tau^0$ and $\tau^1$ and there 
are chords $\phi(L)$ for some weakly important child $L$ of $N$ which separate those two sets.
Clearly, we can place the chords from $\phi(K)$ next to the chords from $\phi(L)$.
After performing the above procedure for every not weakly important child of $N$, we get a model $\phi$ which satisfies properties~\ref{prop:M:model-blocks-consistency}-\ref{prop:M:model-cliques-consistency} with respect to $N$.

Finally, assume $N$ is a serial $Q$-node which was marked as important.
It means that $N=V$ and $V$ is serial in $\strongModules(G_{ov})$.
In this case we proceed similarly as for important serial M-nodes.
Clearly, we need to take care about children of $V$ which are not weakly important.
Let $K$ be one of such child.

First, suppose $\cliques_{01}(K)\neq \emptyset$ and $\cliques_{10}(K)\neq \emptyset$.
Let $\tau$ be the sequence of all cliques from $\cliques(K)$ occurring in $\phi$.
We need to show that the cliques from $\cliques_{01}(K)$ form a contiguous subword in 
$\tau$ and the cliques from $\cliques_{10}(K)$ form a contiguous subword in $\tau$.
Suppose this is not the case. 
It means that there are $C_1,C_3 \in \cliques_{01}(K)$ and $C_2,C_4 \in \cliques_{10}(K)$ such that $C_1C_2C_3C_4$ is a subword of $\tau$.
However, $(\{C_1,C_3\},\{C_2,C_4\})$ is a signature in $\signatures(V)$,
which was marked as visited by some weakly important child $L$ of $V$ such that $\priv(L) = \emptyset$.
Since $\phi|L^0$ and $\phi|L^1$ do not contain cliques and since for every $C \in \{C_1,\ldots,C_k\}$ the chords from $\phi(C \cap L)$ are to the left of $C$,
we deduce $C_1C_2C_3C_4$ can not be a subword of $\phi$.
We use similar arguments to show there is a child $L$ of $V$ with $\priv(L) = \emptyset$ which separates the cliques from $\cliques_{01}(K)$ and from $\cliques_{10}(K)$. 
Clearly, we can place the chords $\phi(K)$ next to $\phi(L)$ to assert condition 
\ref{prop:M:model-cliques-consistency} for $K$.

Suppose now that $\cliques_{01}(K)=\emptyset$ (the case $\cliques_{10}(K)=\emptyset$ is analogous).
Note that this implies $\cliques(K)=\cliques_{10}(K)$.
Let $L$ be any weakly important child of $V$ such that $\priv(L) = \emptyset$.
Let $\tau_L$ and $\tau_R$ be the sequences of all cliques from $\cliques(K)$ lying on different sides of $\phi(L)$, respectively.
If $\tau_L$ or $\tau_R$ is empty, we can clearly place the chords $\phi(K)$ next to $\phi(L)$ to assert condition \ref{prop:M:model-cliques-consistency} for $K$.
So, suppose both $\tau_L$ and $\tau_R$ are non-empty.
Let $C_1,C_2$ be the first and last elements of $\tau_L$ and let $C_3,C_4$ be the first and last elements of $\tau_R$, respectively.
Since $(\emptyset, \{C_1,C_2,C_3,C_4\})$ is a signature in $\signatures(V)$, it was marked as visited by some weakly important child $L'$ of $V$ such that $\priv(L')=\emptyset$.
Since $\phi(L')$ must also intersect $\phi(L)$, we can clearly place the chords $\phi(K)$ next to $\phi(L')$ to assert condition \ref{prop:M:model-cliques-consistency} for $K$.

Proceeding this way for every not weakly important child of $V$, we obtain a model 
$\phi$ which satisfies property~\ref{prop:M:model-cliques-consistency} for $V$.

\end{proof}

\section{\hcp\ is \NP-complete}
\label{sec:npc}
In this section we show that \hcp\ is \NP-complete.
Moreover, it remains \NP-complete even if we fix the type of module formed by the vertex set in the modular decomposition of the overlap graph.
This is an improvement over the \NP-completeness proof of Agaoglu {\c{C}}agirici and Zeman~\cite{AgaogluZeman22}, as it required that the type of module was serial.

We are reducing from the \total problem.
In \total we are given a universe $S$ of size $n$ and a set $\mathbb{T}$ consisting of $m$ triples $(x, y, z)$ in $S^3$. 
We want to determine whether there exists a linear ordering $(S,{<})$ of $S$ such that for every $(x,y,z) \in \mathbb{T}$ we have either $x< y < z$ or $x > y > z$ in $(S,{<})$. 
Opatrny proved that the Total Ordering problem is \NP-complete~\cite{Opa79}.
Also, the proof given in \cite{Opa79} asserts that, assuming Exponential Time Hypothesis (ETH), the Total Ordering problem cannot be solved in time $2^{o(m)}n^{O(1)}$.

\begin{proof}
Let $S = [n]$ and $\mathbb{T}$ be an instance of the Total Ordering problem.
Now, let $G$ be any circular arc graph such that the PQM-tree of $G$ contains a serial M-node $M$ with the children $M_1,\ldots,M_n$, where for every $i \in [n]$ we have 
$M_i = \{u_i,v_i\}$ for some two non-adjacent vertices $u_i,v_i$ in $G$, 
$u_i$ is $(M^0,M^1)$-oriented and $v_i$ is $(M^1,M^0)$-oriented in~$\MMM$.
In particular, for every $i \in [n]$ we have that $M_i$ is a parallel child of $M$ and for every two distinct $i,j$ from $[n]$ we have that $M_i \sim M_j$.

Now, for every triple $(x,y,z)$ in $\mathbb{T}$ we define two cliques $C_{(x,y,z)}$ and $C'_{(x,y,z)}$ in $G$, where
$$C_{(x,y,z)} = \{v_x,u_y,u_z\} \quad \text{and} \quad C'_{(x,y,z)} = \{v_z,u_x,u_y\}.$$
See Figure \ref{fig:npc}.

\begin{figure}[htp!]
\centering
\begin{tikzpicture}[yscale=0.85,xscale=1.28,>=latex]

\coordinate (av1) at (0.9,3) {};
\coordinate (au1) at (1.1,3) {};
\coordinate (av2) at (1.9,3) {};
\coordinate (au2) at (2.1,3) {};
\coordinate (aC') at (2.5,3) {};
\coordinate (av3) at (2.9,3) {};
\coordinate (au3) at (3.1,3) {};

\coordinate (bv1) at (2.9,0) {};
\coordinate (bu1) at (3.1,0) {};
\coordinate (bC) at (2.5,0) {};
\coordinate (bv2) at (1.9,0) {};
\coordinate (bu2) at (2.1,0) {};
\coordinate (bv3) at (0.9,0) {};
\coordinate (bu3) at (1.1,0) {};

\coordinate (lav1) at (0.85,3.3) {};
\coordinate (lau1) at (1.15,3.3) {};
\coordinate (laM1) at (1.0,3.8) {};

\coordinate (lav2) at (1.85,3.3) {};
\coordinate (lau2) at (2.15,3.3) {};
\coordinate (laM2) at (2.0,3.8) {};

\coordinate (lav3) at (2.85,3.3) {};
\coordinate (lau3) at (3.15,3.3) {};
\coordinate (laM3) at (3.0,3.8) {};

\coordinate (lbv1) at (2.85,-0.3) {};
\coordinate (lbu1) at (3.15,-0.3) {};
\coordinate (lbM1) at (3.0,-0.8) {};

\coordinate (lbv2) at (1.85,-0.3) {};
\coordinate (lbu2) at (2.15,-0.3) {};
\coordinate (lbM2) at (2.0,-0.8) {};

\coordinate (lbv3) at (0.85,-0.3) {};
\coordinate (lbu3) at (1.15,-0.3) {};
\coordinate (lbM3) at (1.0,-0.8) {};

\coordinate (M0) at (3.8,3);
\coordinate (M1) at (0.2,0);

\draw[->] (0.5,3) -- (3.5,3);
\draw[<-] (0.5,0) -- (3.5,0);

\draw[<-,thick,red] (av1)--(bv1);
\draw[->,thick,blue] (au1)--(bu1);

\draw[<-,thick] (av2)--(bv2);
\draw[->,thick,red] (au2)--(bu2);
\draw[->,thick,blue] ($(au2) + (0.035,0)$)--($(bu2) + (0.035,0)$);

\draw[<-,thick,blue] (av3)--(bv3);
\draw[->,thick,red] (au3)--(bu3);

\tikzstyle{every node}=[circle,minimum size=4pt,inner sep=0pt,draw,fill]
\node[blue] at (aC') {};
\node[red] at (bC) {};

\tikzstyle{every node}=[inner sep=1pt]
\begin{scriptsize}
\node at (M0) {$M^0$};
\node at (M1) {$M^1$};
\end{scriptsize}

\begin{tiny}

\node at (lav1) {$v^1_1$};
\node at (lbv1) {$v^0_1$};
\node at (lau1) {$u^0_1$};
\node at (lbu1) {$u^1_1$};

\node at (lav2) {$v^1_2$};
\node at (lbv2) {$v^0_2$};
\node at (lau2) {$u^0_2$};
\node at (lbu2) {$u^1_2$};

\node at (lav3) {$v^1_3$};
\node at (lbv3) {$v^0_3$};
\node at (lau3) {$u^0_3$};
\node at (lbu3) {$u^1_3$};

\node at (laM1) {$M^0_1$};
\node at (laM2) {$M^0_2$};
\node at (laM3) {$M^0_3$};

\node at (lbM1) {$M^1_1$};
\node at (lbM2) {$M^1_2$};
\node at (lbM3) {$M^1_3$};

\end{tiny}
\draw[white] (0,0)--(-0,-0.5);
\draw[white] (4,3)--(4,3.5);
\end{tikzpicture}
\begin{tikzpicture}[yscale=0.85,xscale=1.28,>=latex]

\coordinate (av1) at (0.9,3) {};
\coordinate (au1) at (1.1,3) {};
\coordinate (av2) at (1.9,3) {};
\coordinate (au2) at (2.1,3) {};
\coordinate (aC') at (2.5,3) {};
\coordinate (av3) at (2.9,3) {};
\coordinate (au3) at (3.1,3) {};

\coordinate (bv1) at (2.9,0) {};
\coordinate (bu1) at (3.1,0) {};
\coordinate (bC) at (2.5,0) {};
\coordinate (bv2) at (1.9,0) {};
\coordinate (bu2) at (2.1,0) {};
\coordinate (bv3) at (0.9,0) {};
\coordinate (bu3) at (1.1,0) {};

\coordinate (lav1) at (0.85,3.3) {};
\coordinate (lau1) at (1.15,3.3) {};
\coordinate (laM1) at (1.0,3.8) {};

\coordinate (lav2) at (1.85,3.3) {};
\coordinate (lau2) at (2.15,3.3) {};
\coordinate (laM2) at (2.0,3.8) {};

\coordinate (lav3) at (2.85,3.3) {};
\coordinate (lau3) at (3.15,3.3) {};
\coordinate (laM3) at (3.0,3.8) {};

\coordinate (lbv1) at (2.85,-0.3) {};
\coordinate (lbu1) at (3.15,-0.3) {};
\coordinate (lbM1) at (3.0,-0.8) {};

\coordinate (lbv2) at (1.85,-0.3) {};
\coordinate (lbu2) at (2.15,-0.3) {};
\coordinate (lbM2) at (2.0,-0.8) {};

\coordinate (lbv3) at (0.85,-0.3) {};
\coordinate (lbu3) at (1.15,-0.3) {};
\coordinate (lbM3) at (1.0,-0.8) {};

\coordinate (M0) at (3.8,3);
\coordinate (M1) at (0.2,0);

\draw[->] (0.5,3) -- (3.5,3);
\draw[<-] (0.5,0) -- (3.5,0);

\draw[<-,thick,blue] (av1)--(bv1);
\draw[->,thick,red] (au1)--(bu1);

\draw[<-,thick] (av2)--(bv2);
\draw[->,thick,red] (au2)--(bu2);
\draw[->,thick,blue] ($(au2) + (0.035,0)$)--($(bu2) + (0.035,0)$);

\draw[<-,thick,red] (av3)--(bv3);
\draw[->,thick,blue] (au3)--(bu3);

\tikzstyle{every node}=[circle,minimum size=4pt,inner sep=0pt,draw,fill]
\node[red] at (aC') {};
\node[blue] at (bC) {};

\tikzstyle{every node}=[inner sep=1pt]
\begin{scriptsize}
\node at (M0) {$M^0$};
\node at (M1) {$M^1$};
\end{scriptsize}

\begin{tiny}

\node at (lav1) {$v^1_3$};
\node at (lbv1) {$v^0_3$};
\node at (lau1) {$u^0_3$};
\node at (lbu1) {$u^1_3$};

\node at (lav2) {$v^1_2$};
\node at (lbv2) {$v^0_2$};
\node at (lau2) {$u^0_2$};
\node at (lbu2) {$u^1_2$};

\node at (lav3) {$v^1_1$};
\node at (lbv3) {$v^0_1$};
\node at (lau3) {$u^0_1$};
\node at (lbu3) {$u^1_1$};

\node at (laM1) {$M^0_3$};
\node at (laM2) {$M^0_2$};
\node at (laM3) {$M^0_1$};

\node at (lbM1) {$M^1_3$};
\node at (lbM2) {$M^1_2$};
\node at (lbM3) {$M^1_1$};

\end{tiny}
\draw[white] (0,0)--(-0,-0.5);
\draw[white] (4,3)--(4,3.5);
\end{tikzpicture}
\begin{tikzpicture}[yscale=0.85,xscale=1.28,>=latex]

\coordinate (av1) at (0.9,3) {};
\coordinate (au1) at (1.1,3) {};
\coordinate (av2) at (1.9,3) {};
\coordinate (au2) at (2.1,3) {};
\coordinate (aC') at (2.5,3) {};
\coordinate (av3) at (2.9,3) {};
\coordinate (au3) at (3.1,3) {};

\coordinate (bv1) at (2.9,0) {};
\coordinate (bu1) at (3.1,0) {};
\coordinate (bC) at (2.5,0) {};
\coordinate (bv2) at (1.9,0) {};
\coordinate (bu2) at (2.1,0) {};
\coordinate (bv3) at (0.9,0) {};
\coordinate (bu3) at (1.1,0) {};

\coordinate (lav1) at (0.85,3.3) {};
\coordinate (lau1) at (1.15,3.3) {};
\coordinate (laM1) at (1.0,3.8) {};

\coordinate (lav2) at (1.85,3.3) {};
\coordinate (lau2) at (2.15,3.3) {};
\coordinate (laM2) at (2.0,3.8) {};

\coordinate (lav3) at (2.85,3.3) {};
\coordinate (lau3) at (3.15,3.3) {};
\coordinate (laM3) at (3.0,3.8) {};

\coordinate (lbv1) at (2.85,-0.3) {};
\coordinate (lbu1) at (3.15,-0.3) {};
\coordinate (lbM1) at (3.0,-0.8) {};

\coordinate (lbv2) at (1.85,-0.3) {};
\coordinate (lbu2) at (2.15,-0.3) {};
\coordinate (lbM2) at (2.0,-0.8) {};

\coordinate (lbv3) at (0.85,-0.3) {};
\coordinate (lbu3) at (1.15,-0.3) {};
\coordinate (lbM3) at (1.0,-0.8) {};

\coordinate (M0) at (3.8,3);
\coordinate (M1) at (0.2,0);

\draw[->] (0.5,3) -- (3.5,3);
\draw[<-] (0.5,0) -- (3.5,0);

\draw[<-,thick] (av1)--(bv1);
\draw[->,thick,red] (au1)--(bu1);
\draw[->,thick,blue] ($(au1) + (0.04,0)$)--($(bu1) + (0.04,0)$);

\draw[<-,thick,red] (av2)--(bv2);
\draw[->,thick,blue] (au2)--(bu2);

\draw[<-,thick,blue] (av3)--(bv3);
\draw[->,thick,red] (au3)--(bu3);

\tikzstyle{every node}=[circle,minimum size=4pt,inner sep=0pt,draw,fill]
\node[blue] at (aC') {};

\tikzstyle{every node}=[inner sep=1pt]
\begin{scriptsize}
\node at (M0) {$M^0$};
\node at (M1) {$M^1$};
\end{scriptsize}

\begin{tiny}

\node at (lav1) {$v^1_2$};
\node at (lbv1) {$v^0_2$};
\node at (lau1) {$u^0_2$};
\node at (lbu1) {$u^1_2$};

\node at (lav2) {$v^1_1$};
\node at (lbv2) {$v^0_1$};
\node at (lau2) {$u^0_1$};
\node at (lbu2) {$u^1_1$};

\node at (lav3) {$v^1_3$};
\node at (lbv3) {$v^0_3$};
\node at (lau3) {$u^0_3$};
\node at (lbu3) {$u^1_3$};

\node at (laM1) {$M^0_2$};
\node at (laM2) {$M^0_1$};
\node at (laM3) {$M^0_3$};

\node at (lbM1) {$M^1_2$};
\node at (lbM2) {$M^1_1$};
\node at (lbM3) {$M^1_3$};

\end{tiny}
\draw[white] (0,0)--(-0,-0.5);
\draw[white] (4,3)--(4,3.5);
\end{tikzpicture}
\caption{
\label{fig:npc}
For a triple $(1,2,3)$ in $\TTT$ we have two cliques: $C_{(1,2,3)}$ depicted in red and $C'_{(1,2,3)}$ depicted in blue.
If $\phi_{|M}$ has $M^0_2$ between $M^0_1$ and $M^0_3$ (to the left) or
between $M^0_3$ and $M^0_1$ (in the center), then both $C_{(1,2,3)}$ and $C'_{(1,2,3)}$ are Helly (witnessed by a red and blue point, respectively); otherwise, either $C_{(1,2,3)}$ or $C'_{(1,2,3)}$ is non Helly (to the right).
}
\end{figure}
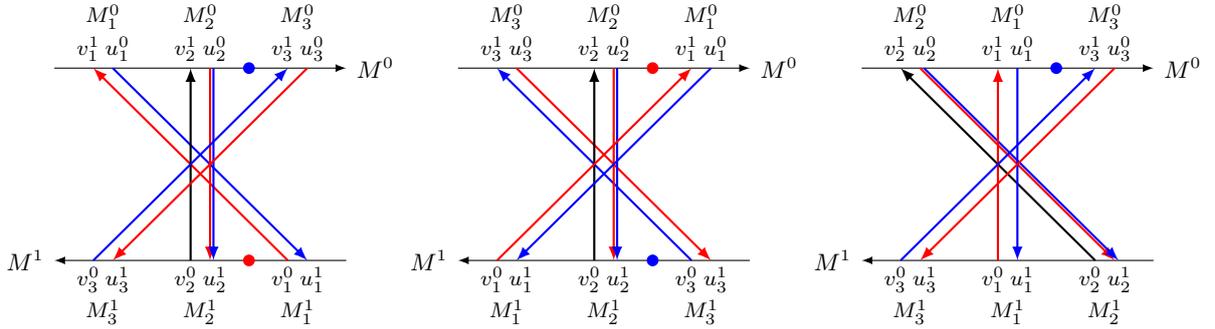

Let $\phi$ be a conformal model of $G$ and let $\phi_{|M} = (\pi^0_M,\pi^1_M)$ 
be an admissible ordering of $M$ in $\phi$.
Note that both $C_{(x,y,z)}$ and $C'_{(x,y,z)}$ satisfy the Helly property in $\phi$ if and only if $M^0_y$ occurs between $M^0_x$ and $M^0_z$ or between $M^0_z$ and $M^0_x$ in $\pi^0_M$ -- see Figure~\ref{fig:npc}.
Finally, let $$\mathbb{C} = \{ C_{(x,y,z)}, C'_{(x,y,z)}: (x,y,z) \in \mathbb{T} \}.$$
Clearly, $\mathbb{C}$ contains $2m$ cliques in $G$. 
The observation made above asserts that $G, \mathbb{C}$ is a \yes-instance of \hcp\ if and only if $S, \mathbb{T}$ is a \yes-intance of the Total Ordering problem.
\end{proof}

Due to the above construction we obtain the following theorem.
\begin{theorem}
Assuming ETH, \hcp\ cannot be solved in time $\Oh{2^{o(k)}|V(G)|^{O(1)}}$ even if we fix the type of the module formed by the vertex set $V$ in the modular decomposition of the overlap graph $G_{ov}$ of the input graph $G=(V,E)$.
\end{theorem}
\begin{proof}

Assuming ETH, the Total Ordering problem cannot be solved in time $\Oh{2^{o(m)}n^{O(1)}}$.

Suppose for the sake of contradiction that we can solve \hcp\ in time $\Oh{2^{o(k)}|V(G)|^{O(1)}}$.
Since in the construction given above we create $2m$ cliques, we can obtain an algorithm solving the Total Ordering problem in time $\Oh{2^{o(m)}n^{O(1)}}$, 
which is not possible if ETH holds.
\end{proof}

\section{Appendix}
\label{sec:appendix}

\subsection{Reducts of circular-arc graphs}
\label{subsec:reducts}

\begin{definition}
Let $G=(V, E)$ be a circular-arc graph and let $U\subseteq V$ be a subset of vertices of $G$.
We say that a circular-arc graph $G'=(V',E')$ is a \emph{reduct} of $G$ with respect to $U$ if:

\begin{enumerate}
    \item $U\subseteq V'$.
    \item For every conformal model $\phi$ of $G$ there is a conformal model $\phi'$ of $G'$ such that $\phi | U^* \equiv \phi' | U^*$.
    \item For every conformal model $\phi'$ of $G'$ there is a conformal model $\phi$ of $G$ such that $\phi | U^* \equiv \phi' | U^*$. 
\end{enumerate}
\end{definition}

In this section we prove the following lemma.
\begin{lemma}
\label{lem:reduct}
Let $G=(V, E)$ be a circular-arc graph and let $U\subseteq V$ be a subset of vertices of $G$. 
Then, there is a polynomial-time algorithm computing a reduct $G'=(V',E')$ of $G$ with respect to $U$ such that $|V'|\leq 12 \cdot |U|$.
\end{lemma}

To prove the lemma, we introduce some definitions.

First, for every CA-module $S \in \camodules$ we extend the notion of admissible models
$\Phi(N)$ on every node $N$ in $\pqmtree_S$.
When $N=\{u\}$ ($N$ is a leaf in $\pqmtree_S$), then $\Phi(N)$ contains a single oriented permutation model:
$(u^0,u^1)$ if $u$ is $(S^0,S^1)$-oriented and $(u^1,u^0)$ if $u$ is $(S^1,S^0)$-oriented in $\SSS$.
If $N$ is an M-node, then the admissible models in $\Phi(N)$ 
are obtained by first picking any admissible ordering $(\pi^0,\pi^1)$ in $\Pi(N)$ 
and then by replacing $L^0$ in $\pi^0$ by $\tau^0_L$ and $L^1$ in $\pi^1$ by $\tau^1_L$, where $L$ is any child $L$ of $N$ and $(\tau^0_L,\tau^1_L)$ is any admissible model from $\Phi(L)$.
Note that $\Phi(S)$ defined this way contains all the admissible models for CA-module~$S$.

Next, we root $\pqmtree$ in some Q-node $R$; 
we denote $\pqmtree$ rooted in $R$ by $\hat{\pqmtree}_R$.
Note that for any Q-node $Q$ and any $S \in \camodules(Q)$, CA-module $S$ and the slots $S^0,S^1$ are the children of $Q$ in $\pqmtree_{R}$.
For any PQ-node $N$ in $\hat{\pqmtree}_R$, by $V(N)$ we denote the vertices contained in the components (Q-nodes) descending $N$ in $\hat{\pqmtree}_R$ (including $N$) and by 
$U(N)$ we denote the set $V(N) \cap U$.
Similarly, we set $V(N)=N$ if $N$ is a node in $\pqmtree_S$ for some $S \in \camodules$ and we set $U(N) = V(N) \cap U$.
We say a node $N$ in $\hat{\pqmtree}_R$ distinct than a slot is \emph{important} if $U(N) \neq \emptyset$; slots $S^0,S^1$ are \emph{important} if $S$ is important.

We now extend the notion of admissible models $\Phi(N)$ on every
PQ-node $N$ in $\hat{\pqmtree}_R$.
If $N$ is a Q-node, then the admissible models in $\Phi(N)$ are obtained 
by first picking any admissible ordering $\pi_N$ in $\Pi(N)$ and then
by replacing in $\pi_N$:
\begin{itemize}
\item $S^0$ by $\tau^0_S$ and $S^1$ by $\tau^1_S$, where $S$ is any CA-module from $\camodules(N)$ and $(\tau^0_S,\tau^1_S)$ is any admissible model from~$\Phi(S)$,
\item $P$ by $\tau_P$, where $P$ is any P-node child of $N$ and $\tau_P$ is any admissible model from~$\Phi(P)$.
\end{itemize}
Similarly, if $P$ is a P-node, then the admissible models in $\Phi(N)$ are obtained 
by picking first an admissible ordering $\pi_N$ in $\Pi(N)$ and then
by replacing in $\pi_N$ the letter $Q$ by $\tau_Q$, where $Q$ is any Q-node child $Q$ and $\tau_Q$ is any admissible model from $\Phi(Q)$.

Note that for every PQ-node $N$ with the parent $N'$ in $\hat{\pqmtree}_R$, 
an admissible model from $\Phi(N)$ is a circular word over $V^*(N) \cup \{N'\}$.
In this case, for some technical reasons we replace in every member of $\Phi(N)$ the letter $N'$ by $R$.
Note that $\Phi(R)$ coincides with the set of all conformal models of $G$.

The general idea behind the construction of the reduct of $G$ is as follows. 
We traverse the tree $\hat{\pqmtree}_R$ bottom-up and we compress the tree $\hat{\pqmtree}_R$ to 
a new tree PQM-tree $\hat{\pqmtree}'_{R'}$ representing the conformal models of the reduct $G'$ of $G$. 
In particular, with every important node $N$ in $\hat{\pqmtree}_R$ 
a node $N'$ in $\hat{\pqmtree}'_{R'}$ is associated, called the \emph{reduct} of~$N$.
Roughly speaking, the reduct $N'$ of $N$ contains $U(N)$, has size linear in the size of $U(N)$, and 
satisfies the property that the set of chord configurations representing the vertices of $U(N)$ 
occurring in the admissible models in $\Phi(N)$ and in the admissible models in $\Phi(N')$ coincide 
(as in the definition of the reduct of $G$).

Now we proceed to the construction of the reduct $G'$ of $G$ and its PQM-tree $\hat{\pqmtree}'_{R'}$.

First, for every important CA-module $S$ of $G$ we process the tree $\pqmtree_S$ bottom-up and 
we construct~$\pqmtree'_S$. 
When processing an important node $N$ in $\pqmtree_S$, we construct also the reduct $N'$ of $N$ in $\pqmtree'_S$, which is required to satisfy the following properties:
\begin{itemize}
 \item $U(N) \subseteq N'$ and $|N'| \leq 12\cdot |U(N)|-10$,
 \item for every admissible model $(\tau^0,\tau^1)$ in $\Phi(N)$ there is 
 an admissible model $(\lambda^0, \lambda^1)$ in $\Phi(N')$ such that 
 $\tau^j|U^*(N) = \lambda^j|U^*(N)$ for $j \in \{0,1\}$,
 \item for every admissible model $(\lambda^0,\lambda^1)$ in $\Phi(N')$ there is 
 an admissible model $(\tau^0, \tau^1)$ in $\Phi(N)$ such that 
 $\lambda^j|U^*(N) = \tau^j|U^*(N)$ for $j \in \{0,1\}$.
\end{itemize}

Suppose we process an important node $N$ in $\pqmtree_S$. 
If $N$ is a leaf in $\pqmtree_S$, then we set $N'=N \cup \{v\}$, where $v$ is a vertex of $G'$ non-adjacent to $N$.
Note that $N'$ has two children, which are leaves in $\hat{\pqmtree}'_{R'}$.
Clearly, $N'$ satisfies the properties of the reduct of $N$.
Suppose now that $N$ is an important node in $\pqmtree_S$. 
If $N$ has exactly one important child $N_1$ in $\pqmtree_S$,
we set $N' = K_1$, where $K_1$ is the reduct of $N_1$. 
Clearly, $N'$ satisfies the properties of the reduct of $N$.
Now, suppose $N$ has at least two important children.
Suppose $N_1,\ldots,N_k$ are the important children of $N$ and 
$K_1,\ldots,K_k$ are their reducts, respectively.


If $N$ is serial or parallel, we define the reduct $N'$ of $N$ by first setting
$K = \bigcup_{i=1}^k K_i$ with the set $\Pi(K)$ obtained from $\Pi(N)$ by restricting the members of $\Pi(N)$ to the letters $\{N^0_1,N^1_1,\ldots,N^0_k,N^1_k\}$ and by replacing $N^j_i$ by $K^j_i$.
Then, we create the reduct $N'$ of $N$ by extending $K$ by the sets $E_1,E_2,E_3$, where $E_i$ contains two non-adjacent vertices of $G'$, 
and we set $\lambda=(E^0_1E^0_2E^0_3K^0,E^1_3E^1_1K^1E^1_2)$  
as an admissible order of $N'$ -- see Figure~\ref{fig:M-node-reduct}.(a)

When $N$ is prime, we proceed as follows.
Let $\pi = (\pi^0,\pi^1)$ be an admissible ordering from $\Pi(N)$.
Suppose $N_1, N_2, \ldots, N_k$ are enumerated such that $N^0_i$ occurs before $N^0_j$ in $\pi^0$ for every $i < j$ in $[k]$.
We define $N'$ such that $N' = \big{(}\bigcup_{i=1}^k (K_i \cup L_i \cup M_i) \big{)} \cup M$, where each of the sets $L_1,M_1, \ldots, L_{k}, M_k,M$ contains 
two non-adjacent vertices of $G'$, and we define an admissible ordering $\lambda = (\lambda^0,\lambda^1)$ of $\Pi(N')$ such that:
\begin{itemize}
 \item $\lambda^0 = L^0_1 K^0_1 \ldots L^0_k K^0_k M^0 M^0_1\ldots M^0_k$,
 \item the suffix of $\lambda^1$ contains the letters $K^1_1,\ldots,K^1_k$ which occur in the same
 order as the letters $N^1_1,\ldots,N^1_k$ in $\pi^1$,
 the prefix of~$\lambda^1$ equals to $M^1L^1_kM^1_k \ldots L^1_1M^1_1$.
\end{itemize}
See Figure~\ref{fig:M-node-reduct}.(b) for an illustration.
One can easily check that the module $N'$ in the graph $(N',{\sim})$ is prime and has $K_1,L_1,M_1,\ldots,K_k,L_k,M_k$ and $M$ as its children. 
We can also easily verify that $N'$ satisfies the properties of the reduct of~$N$.

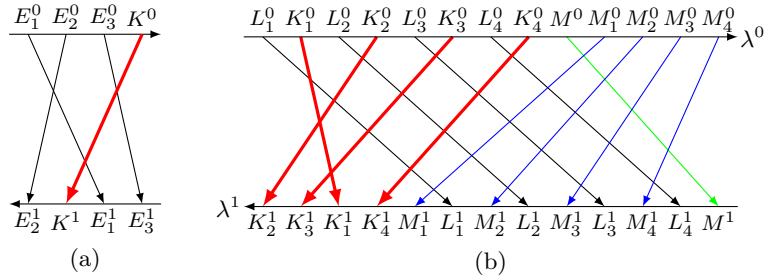
\begin{figure}[htp!]
\centering
\begin{tikzpicture}[yscale=0.75,xscale=0.5,>=latex]
\coordinate (label) at (1.5,-1) {};

\coordinate (a_v1) at (0,3) {};
\coordinate (a_v2) at (1,3) {};
\coordinate (a_v3) at (2,3) {};
\coordinate (a_v4) at (3,3) {};

\coordinate (b_v1) at (0,0) {};
\coordinate (b_v2) at (1,0) {};
\coordinate (b_v3) at (2,0) {};
\coordinate (b_v4) at (3,0) {};

\coordinate (la_v1) at (0,3.3) {};
\coordinate (la_v2) at (1,3.3) {};
\coordinate (la_v3) at (2,3.3) {};
\coordinate (la_v4) at (3,3.3) {};

\coordinate (lb_v1) at (0,-0.3) {};
\coordinate (lb_v2) at (1,-0.3) {};
\coordinate (lb_v3) at (2,-0.3) {};
\coordinate (lb_v4) at (3,-0.3) {};

\coordinate (tau0) at (3.9,3);
\coordinate (tau1) at (-0.9,0);

\draw[->] (-0.5,3) -- (3.5,3);
\draw[<-] (-0.5,0) -- (3.5,0);

\draw[->,black] (a_v1)--(b_v3);
\draw[->,black] (a_v2)--(b_v1);
\draw[->,black] (a_v3)--(b_v4);
\draw[->,red,very thick] (a_v4)--(b_v2);

\begin{tiny}
\node at (la_v1) {$E^0_1$};
\node at (la_v2) {$E^0_2$};
\node at (la_v3) {$E^0_3$};
\node at (la_v4) {$K^0$};

\node at (lb_v1) {$E^1_2$};
\node at (lb_v2) {$K^1$};
\node at (lb_v3) {$E^1_1$};
\node at (lb_v4) {$E^1_3$};
\end{tiny}

\begin{scriptsize}
\node at (label) {(a)};
\end{scriptsize}
\draw[white] (-1,-1.3)--(0,-1.3);
\end{tikzpicture}
\hspace{0.35cm}
\begin{tikzpicture}[yscale=0.75,xscale=0.5,>=latex]
\coordinate (label) at (6,-1) {};
\coordinate (a_v1) at (0,3) {};
\coordinate (a_v2) at (1,3) {};
\coordinate (a_v3) at (2,3) {};
\coordinate (a_v4) at (3,3) {};
\coordinate (a_v5) at (4,3) {};
\coordinate (a_v6) at (5,3) {};
\coordinate (a_v7) at (6,3) {};
\coordinate (a_v8) at (7,3) {};
\coordinate (a_v9) at (8,3) {};
\coordinate (a_v10) at (9,3) {};
\coordinate (a_v11) at (10,3) {};
\coordinate (a_v12) at (11,3) {};
\coordinate (a_v13) at (12,3) {};

\coordinate (b_v1) at (0,0) {};
\coordinate (b_v2) at (1,0) {};
\coordinate (b_v3) at (2,0) {};
\coordinate (b_v4) at (3,0) {};
\coordinate (b_v5) at (4,0) {};
\coordinate (b_v6) at (5,0) {};
\coordinate (b_v7) at (6,0) {};
\coordinate (b_v8) at (7,0) {};
\coordinate (b_v9) at (8,0) {};
\coordinate (b_v10) at (9,0) {};
\coordinate (b_v11) at (10,0) {};
\coordinate (b_v12) at (11,0) {};
\coordinate (b_v13) at (12,0) {};

\coordinate (la_v1) at (0,3.3) {};
\coordinate (la_v2) at (1,3.3) {};
\coordinate (la_v3) at (2,3.3) {};
\coordinate (la_v4) at (3,3.3) {};
\coordinate (la_v5) at (4,3.3) {};
\coordinate (la_v6) at (5,3.3) {};
\coordinate (la_v7) at (6,3.3) {};
\coordinate (la_v8) at (7,3.3) {};
\coordinate (la_v9) at (8,3.3) {};
\coordinate (la_v10) at (9,3.3) {};
\coordinate (la_v11) at (10,3.3) {};
\coordinate (la_v12) at (11,3.3) {};
\coordinate (la_v13) at (12,3.3) {};

\coordinate (lb_v1) at (0,-0.3) {};
\coordinate (lb_v2) at (1,-0.3) {};
\coordinate (lb_v3) at (2,-0.3) {};
\coordinate (lb_v4) at (3,-0.3) {};
\coordinate (lb_v5) at (4,-0.3) {};
\coordinate (lb_v6) at (5,-0.3) {};
\coordinate (lb_v7) at (6,-0.3) {};
\coordinate (lb_v8) at (7,-0.3) {};
\coordinate (lb_v9) at (8,-0.3) {};
\coordinate (lb_v10) at (9,-0.3) {};
\coordinate (lb_v11) at (10,-0.3) {};
\coordinate (lb_v12) at (11,-0.3) {};
\coordinate (lb_v13) at (12,-0.3) {};

\coordinate (tau0) at (12.9,3);
\coordinate (tau1) at (-0.9,0);

\draw[->,black] (a_v1)--(b_v6);
\draw[->,red,very thick] (a_v2)--(b_v3);
\draw[->,black] (a_v3)--(b_v8);
\draw[->,red,very thick] (a_v4)--(b_v1);
\draw[->,black] (a_v5)--(b_v10);
\draw[->,red,very thick] (a_v6)--(b_v2);
\draw[->,black] (a_v7)--(b_v12);
\draw[->,red,very thick] (a_v8)--(b_v4);
\draw[->,green] (a_v9)--(b_v13);
\draw[->,blue] (a_v10)--(b_v5);
\draw[->,blue] (a_v11)--(b_v7);
\draw[->,blue] (a_v12)--(b_v9);
\draw[->,blue] (a_v13)--(b_v11);

\draw[->] (-0.5,3) -- (12.5,3);
\draw[<-] (-0.5,0) -- (12.5,0);

\tikzstyle{every node}=[inner sep=1pt]
\begin{tiny}
\node at (la_v1) {$L^0_1$};
\node at (la_v2) {$K^0_1$};
\node at (la_v3) {$L^0_2$};
\node at (la_v4) {$K^0_2$};
\node at (la_v5) {$L^0_3$};
\node at (la_v6) {$K^0_3$};
\node at (la_v7) {$L^0_4$};
\node at (la_v8) {$K^0_4$};
\node at (la_v9) {$M^0_{\text{ }}$};
\node at (la_v10) {$M^0_1$};
\node at (la_v11) {$M^0_2$};
\node at (la_v12) {$M^0_3$};
\node at (la_v13) {$M^0_4$};

\node at (lb_v1) {$K^1_2$};
\node at (lb_v2) {$K^1_3$};
\node at (lb_v3) {$K^1_1$};
\node at (lb_v4) {$K^1_4$};
\node at (lb_v5) {$M^1_1$};
\node at (lb_v6) {$L^1_1$};
\node at (lb_v7) {$M^1_2$};
\node at (lb_v8) {$L^1_2$};
\node at (lb_v9) {$M^1_3$};
\node at (lb_v10) {$L^1_3$};
\node at (lb_v11) {$M^1_4$};
\node at (lb_v12) {$L^1_4$};
\node at (lb_v13) {$M^1_{\text{ }}$};
\end{tiny}

\begin{scriptsize}
\node at (tau0) {$\lambda^0$};
\node at (tau1) {$\lambda^1$};
\node at (label) {(b)};
\end{scriptsize}
\draw[white] (-1,-1.3)--(0,-1.3);
\end{tikzpicture}

\caption{
\label{fig:M-node-reduct}
Figures (a) shows the extended reduct of $N$ for the case when $N$ is serial.
Figure (b) shows the reduct of $N$ for the case when $N$ is prime: $N$ has 4 important children $N_1,\ldots,N_4$ with the reducts $K_1,\ldots,K_4$.} 
\end{figure}

Next, we process the tree $\hat{\pqmtree}_R$ bottom up and we construct the tree $\hat{\pqmtree}'_{R'}$.
When we process an important node $N$ in $\hat{\pqmtree}_R$, we construct also the reduct $N'$ of $N$ in $\hat{\pqmtree}'_{R'}$, which is required to satisfy the following properties:
\begin{itemize}
 \item for every admissible model $\tau$ in $\Phi(N)$ there is 
 an admissible model $\lambda$ in $\Phi(N')$ such that 
 $\tau|(U^*(N) \cup \{R\}) = \lambda|(U^*(N) \cup \{R\})$,
 \item  for every admissible model $\lambda$ in $\Phi(N')$ there is 
 an admissible model $\tau$ in $\Phi(N)$ such that 
 $\lambda|(U^*(N) \cup \{R\}) = \tau|(U^*(N) \cup \{R\})$,
 \item $|V(N')| \leq 12\cdot |U(N)|-10$, where $V(N')$ is the set of the vertices from the components descending $N'$ in $\hat{\pqmtree}'_{R'}$.
\end{itemize}
Then, the reduct of the root $R$ is simply the reduct $G'$ of $G$ satisfying $|V(G')|\leq 12\cdot |U|-10$.

Suppose we process an important Q-node $Q$ in $\hat{\pqmtree}_R$ which is not the root of $\hat{\pqmtree}_R$.
If $Q$ has exactly one important P-node child $P$ with exactly one important Q-node child $Q_1$, 
then we set $Q' = Q'_1$ (that is, we set the reduct of $Q_{1}$ as the reduct of $Q$).
Otherwise ($Q$ has at least two important children in $\hat{\pqmtree}_R$ or 
$Q$ has one important P-node child with at least two important Q-node children), 
we proceed as follows.
We choose an admissible ordering $\pi$ in $\pi(Q)$. 
Let $Rr_1r_2 \ldots r_k$ be the circular word that arises from 
$\pi$ by replacing the letter representing the parent of $Q$ in $\hat{\pqmtree}_R$ by $R$
and then by restricting $\pi$ to important children of $Q$ and to the letter $R$.
Note that each $r_i$ is either an important P-node child of $Q$ or 
a slot of an important CA-module contained in $Q$ -- see Figure~\ref{fig:Q-node-reduct}.
We define $Q'$ such that 
$$Q' = \bigcup \Big{\{} K: K \text{ is a reduct of an important CA-module in $\camodules(Q)$}\Big{\}} \cup \Big{(}\bigcup_{i=1}^k (L_i \cup M_i) \Big{)} \cup M,$$ 
where each of the sets $L_1,M_1, \ldots, L_{k}, M_k,M$ contains 
two non-adjacent vertices of $G'$, and we define an admissible ordering $\lambda$ of $\Pi(N')$ such that 
$$\lambda \equiv RL^0_1 r_1 L^0_2 r_2 \ldots L^0_k r_k L^1_k M^0_1L^1_{k-1} M^0_2 \ldots L^1_1 M^0_k M^0 M^1_k \ldots M^1_1 M^1.$$
See Figure~\ref{fig:Q-node-reduct} for an illustration.
Again, we can easily check that $L_1,M_1,\ldots,L_k,M_k,M$ as well as the reducts of important CA-modules from $\camodules(Q)$ are the CA-modules of the component $Q'$ of the overlap graph $G'_{ov}$ of $G'$.
Also, we can easily verify that $Q'$ defined this way satisfies all the conditions of the reduct of $Q$.
\begin{figure}[htp!]
\centering
\begin{tikzpicture}[xscale=1,yscale=1,>=latex,shorten >=-0.4pt,shorten <=-0.4pt]
\coordinate (center) at (0,0) {};
\coordinate (label) at (0,-2.5) {};

\coordinate (r) at ($(center)+(90:2.0cm)$) {};

\coordinate (s1) at ($(center)+(67:2.0cm)$) {};
\coordinate (s2) at ($(center)+(52:2.0cm)$) {};
\coordinate (s3) at ($(center)+(37:2.0cm)$) {};
\coordinate (s4) at ($(center)+(22:2.0cm)$) {};
\coordinate (s5) at ($(center)+(7:2.0cm)$) {};
\coordinate (s6) at ($(center)+(352:2.0cm)$) {};
\coordinate (s7) at ($(center)+(337:2.0cm)$) {};
\coordinate (s8) at ($(center)+(322:2.0cm)$) {};
\coordinate (s9) at ($(center)+(307:2.0cm)$) {};
\coordinate (s10) at ($(center)+(292:2.0cm)$) {};
\coordinate (s11) at ($(center)+(277:2.0cm)$) {};
\coordinate (s12) at ($(center)+(262:2.0cm)$) {};
\coordinate (s13) at ($(center)+(247:2.0cm)$) {};
\coordinate (s14) at ($(center)+(232:2.0cm)$) {};
\coordinate (s15) at ($(center)+(217:2.0cm)$) {};
\coordinate (s16) at ($(center)+(202:2.0cm)$) {};
\coordinate (s17) at ($(center)+(187:2.0cm)$) {};
\coordinate (s18) at ($(center)+(172:2.0cm)$) {};
\coordinate (s19) at ($(center)+(157:2.0cm)$) {};
\coordinate (s20) at ($(center)+(142:2.0cm)$) {};
\coordinate (s21) at ($(center)+(127:2.0cm)$) {};
\coordinate (s22) at ($(center)+(112:2.0cm)$) {};

\coordinate (lr) at ($(center)+(90:2.3cm)$) {};

\coordinate (ls1) at ($(center)+(67:2.3cm)$) {};
\coordinate (ls2) at ($(center)+(52:2.3cm)$) {};
\coordinate (ls3) at ($(center)+(37:2.3cm)$) {};
\coordinate (ls4) at ($(center)+(22:2.3cm)$) {};
\coordinate (ls5) at ($(center)+(7:2.3cm)$) {};
\coordinate (ls6) at ($(center)+(352:2.3cm)$) {};
\coordinate (ls7) at ($(center)+(337:2.3cm)$) {};
\coordinate (ls8) at ($(center)+(322:2.3cm)$) {};
\coordinate (ls9) at ($(center)+(307:2.3cm)$) {};
\coordinate (ls10) at ($(center)+(292:2.3cm)$) {};
\coordinate (ls11) at ($(center)+(277:2.3cm)$) {};
\coordinate (ls12) at ($(center)+(262:2.3cm)$) {};
\coordinate (ls13) at ($(center)+(247:2.3cm)$) {};
\coordinate (ls14) at ($(center)+(232:2.3cm)$) {};
\coordinate (ls15) at ($(center)+(217:2.3cm)$) {};
\coordinate (ls16) at ($(center)+(202:2.3cm)$) {};
\coordinate (ls17) at ($(center)+(187:2.3cm)$) {};
\coordinate (ls18) at ($(center)+(172:2.3cm)$) {};
\coordinate (ls19) at ($(center)+(157:2.3cm)$) {};
\coordinate (ls20) at ($(center)+(142:2.3cm)$) {};
\coordinate (ls21) at ($(center)+(127:2.3cm)$) {};
\coordinate (ls22) at ($(center)+(112:2.3cm)$) {};

\draw (0,0) circle (2cm);

\draw[very thick,red,<-] (s2) -- (s8);
\draw[thick,black,->] (s1) -- (s16);
\draw[thick,black,->] (s3) -- (s14);
\draw[thick,black,->] (s5) -- (s12);
\draw[thick,black,->] (s7) -- (s10);
\draw[thick,blue,->] (s9) -- (s21);
\draw[thick,blue,->] (s11) -- (s20);
\draw[thick,blue,->] (s13) -- (s19);
\draw[thick,blue,->] (s15) -- (s18);
\draw[thick,green,->] (s17) -- (s22);

\tikzstyle{every node}=[circle,minimum size=3.5pt,inner sep=0pt,draw,fill]
\node[red] at (s4) {};
\node[red] at (s6) {};
\node[black] at (r) {};

\tikzstyle{every node}=[inner sep=1pt]
\begin{tiny}
\node at (lr) {$R$};

\node at (ls1) {$L^0_1$};
\node at (ls2) {$K^1$};
\node at (ls3) {$L^0_2$};
\node at (ls4) {$P_1$};
\node at (ls5) {$L^0_3$};
\node at (ls6) {$P_2$};
\node at (ls7) {$L^0_4$};
\node at (ls8) {$K^0$};
\node at (ls9) {$M^0_1$};
\node at (ls10) {$L^1_4$};
\node at (ls11) {$M^0_2$};
\node at (ls12) {$L^1_3$};
\node at (ls13) {$M^0_3$};
\node at (ls14) {$L^1_2$};
\node at (ls15) {$M^0_4$};
\node at (ls16) {$L^1_1$};
\node at (ls17) {$M^0$};
\node at (ls18) {$M^1_4$};
\node at (ls19) {$M^1_3$};
\node at (ls20) {$M^1_2$};
\node at (ls21) {$M^1_1$};
\node at (ls22) {$M^1$};

\end{tiny}
\end{tikzpicture}
\caption{
\label{fig:Q-node-reduct}
An admissible ordering of the reduct $Q'$ of Q-node $Q$.
$Q$ has $4$ important children $P_1,P_2,S^0,S^1$ in $\pqmtree^{PQ}_R$,
$K$ with the slots $K^0,K^1$ is the reduct of the CA-module $S$.
}
\end{figure}
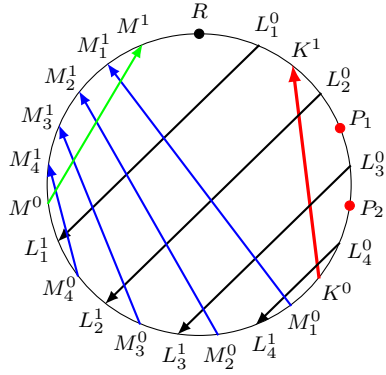

Now, suppose we process an important P-node $P$ in $\hat{\pqmtree}_R$.
If $P$ has exactly one important Q-node child $Q$ and $Q$ has exactly one important P-node child $P_1$, 
then we set $P' = P'_1$ (that is, we set the reduct of $P_1$ as the reduct of $Q$).
Otherwise ($P$ has at least two important Q-node children in $\hat{\pqmtree}_R$ or 
$P$ has one important Q-node child with at least two important children), 
we define the reduct $P'$ of $P$ by restricting $P$ and $\Pi(P)$ to 
important Q-node children of $P$.

Finally, we define the reduct of $R$, which completes the definition of $G'$.
If $R$ has exactly one important child $N$, then we set $R' = K$, where $K$ is the reduct of $N$.
Otherwise, $R$ has at least two important children and we proceed as follows.
Suppose $R$ is a serial Q-node (that is, $R$ is the only inner node in $\hat{\pqmtree}_R$).
Then, we set the reduct
$R'$ by restricting $R$ and $\Pi(R)$ to important CA-modules of $G$.
Finally, when $R$ is prime, we proceed similarly as for an important prime node different than $R$.

\subsection{Trapezoid Lemma}
\label{subsec:trapezoid-lemma}
We say that trapezoids $T_1$ and $T_2$ are \emph{nicely intersecting} if there exist a segment $s_1$ spanned between the bases of $T_1$ and a segment $s_2$ spanned between the bases of $T_2$ such that $s_1$ intersects $s_2$ and the endpoints of $s_1$ and $s_2$ are pairwise distinct.
We now prove the following lemma.
\begin{lemma}
Let $A$ and $B$ be two distinct horizontal lines on the plane, $A$ above $B$. 
Let $\mathcal{T} = \{T_1, \ldots, T_t\}$ be a set of $t$ (possibly degenerated) trapezoids 
with one base contained in $A$ (might be infinite and might be open/closed on each side) 
and one base contained in $B$ (might be infinite and might be open/closed on each side).
Then there exists a set of $t$ pairwise intersecting segments $\{s_1, \ldots, s_t\}$ with pairwise distinct endpoints spanned between $A$ and $B$
such that for each $i\in [t]$ the segment $s_i$ is contained in $T_i$ if and only if the trapezoids from $\mathcal{T}$ pairwise nicely intersect.
\end{lemma}
\begin{proof}
We assume that the bases of $T_1,\ldots,T_t$ are closed intervals. 
We leave to the reader to verify that the proof given below can be easily modified to handle all the cases.

One direction of the proof is trivial: if there is such a set of pairwise intersecting segments, then naturally the trapezoids from $\mathcal{T}$ have to pairwise nicely intersect.

For the other direction suppose the trapezoids from $\{T_1,\ldots,T_t\}$ pairwise nicely intersect. 
We now argue that there exists a desired choice of segments $\{s_1, \ldots, s_t\}$.

First, we define relations ${<_A}$ and ${<_B}$ on the set of trapezoids $\mathcal{T}$. 
For trapezoids $T_i, T_j \in \mathcal{T}$ we say $T_i <_A T_j$ when the top-right corner of $T_i$ is on the left or equal to the top-left corner of $T_j$ (see Figure~\ref{fig:trapezoids_relation}). 
Since $T_i$ and $T_j$ are nicely intersecting, if $T_i <_A T_j$ then 
the top-left corner of $T_i$ is on the left of the top-right corner of $T_j$.
In particular, this implies that there are no $T_i, T_j$ such that $T_i <_A T_j$ and $T_j <_A T_i$.
For trapezoids $T_i, T_j \in \mathcal{T}$ we say $T_i <_B T_j$ when the bottom-left corner of $T_i$ is on the right of equal to the bottom-right corner of $T_j$ (see Figure~\ref{fig:trapezoids_relation}).
Again, observe that there are no $T_i, T_j$ such that $T_i <_B T_j$ and $T_j <_B T_i$.
Observe that since $T_i$ and $T_j$ are nicely intersecting, if $T_i <_B T_j$ then the top-left corner of $T_i$ is on the left of the top-right corner of $T_j$.

\begin{center}
    \begin{figure}[htp!]

\begin{tikzpicture}[scale=1,>=latex,shorten >=-0.4pt,shorten <=-0.4pt]

\draw[] (-1, 3) -- (3, 3);
\draw[] (-1, 0) -- (3, 0);
\draw[fill=gray!30, draw=none, opacity=0.5] (-0.8,3) -- (0.3,3) -- (2,0) -- (0,0) -- cycle;

\draw[fill=gray!30, draw=none, opacity=0.5] (1.3,3) -- (2.4,3) -- (2.8,0) -- (1.3,0) -- cycle;

\node at (0.3, 1.5) {$T_i$};
\node at (2, 1.5) {$T_j$};
\node at (1, -1) {$T_i <_A T_j$};

\end{tikzpicture}
\hspace{3cm}
\begin{tikzpicture}[scale=1,>=latex,shorten >=-0.4pt,shorten <=-0.4pt]

\draw[] (-1, 3) -- (3, 3);
\draw[] (-1, 0) -- (3, 0);
\draw[fill=gray!30, draw=none, opacity=0.5] (-0.8,3) -- (2.2,3) -- (2.7,0) -- (2,0) -- cycle;

\draw[fill=gray!30, draw=none, opacity=0.5] (0.8,3) -- (1.5,3) -- (1,0) -- (-0.7,0) -- cycle;

\node at (2, 1.5) {$T_i$};
\node at (0.2, 0.7) {$T_j$};
\node at (1, -1) {$T_i <_B T_j$};

\end{tikzpicture}
\caption{
\label{fig:trapezoids_relation}
Examples of trapezoids in the relations ${<_A}$ and ${<_B}$.
}
\end{figure}
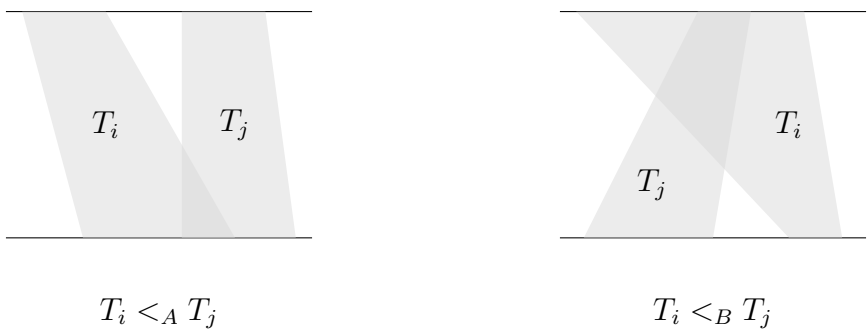
\end{center}

We now claim that the union ${<_A \cup <_B}$ is acyclic on $\mathcal{T}$.
Suppose for the sake of contradiction there is a cycle in ${<_A \cup <_B}$ 
and consider the shortest one. 
Since ${<_A}$ and ${<_B}$ are transitive, the cycle has to have even length and alternate between ${<_A}$ and ${<_B}$.
Due to the observations above, an even cycle containing a trapezoid $T_i$ would imply that the top-right corner of $T_i$ is on the left of itself, 
which is a contradiction.

Since ${<_A \cup <_B}$ is acyclic, there is a topological order of ${<_A \cup <_B}$. 
Let ${<}$ be any such order.

Suppose the trapezoids $T_1, \ldots, T_t$ are enumerated consistently with ${<}$.
We now inductively prove that there exists a valid choice of segments $s_1, \ldots, s_i$ for every $i\in [t]$ such that the top endpoint of $s_j$ is on the left of the top endpoint of $s_{j+1}$ for every $j\in [i-1]$.

The base case of $i=1$ is trivial, as any choice of $s_1$ inside $T_1$ is valid.

Suppose $i>1$. 
From the induction hypothesis there is a valid choice of segments $s_1, \ldots, s_{i-1}$ for $T_1, \ldots, T_{i-1}$. 
Consider the trapezoid $T_i$. 
Observe that if for any $j<i$ the top endpoint of $s_j$ is on the right or equal to the top-right corner of $T_i$, 
then $T_j$ must contain the top-right corner of $T_i$ since $T_j < T_i$.
Moreover, the top-right corner of $T_i$ is not equal to the top-left corner of $T_j$.
Similarly, if the bottom endpoint of $s_j$ is on the left of the bottom side of $T_i$, then $T_j$ must contain the bottom-left corner of $T_i$.
Moreover, the bottom-left corner of $T_i$ is not equal to the bottom-right corner of $T_j$.

For every $j<i$ denote the top endpoint of $s_j$ by $a_j$ and the bottom endpoint by $b_j$. Let $P$ be the top-right corner of $T_i$ and $Q$ be the bottom-left corner of $T_i$.

We create the valid set of segments for $T_1, \ldots, T_i$ the following way.
Let $j\in [i-1]$ be the smallest number such that $a_j$ is on the right of $P$.
We move the top endpoints $a_j, \ldots, a_{i-1}$ to the left side of $P$ but very close to it.
Similarly, let $j' \in [i-1]$ be the smallest number such that $b_{j'}$ is on the left of $Q$.
We move the bottom endpoints $b_{j'}, \ldots, b_{i-1}$ to the right side of $Q$ but very close to it.

Finally, we define $s_i$ as the segment with endpoints $P$ and $Q$. Note that $a_{i-1}$ is on the left of $P$ and $b_{i-1}$ is on the right side of $Q$. This completes the induction (see Figure~\ref{fig:insert}).

\begin{center}
    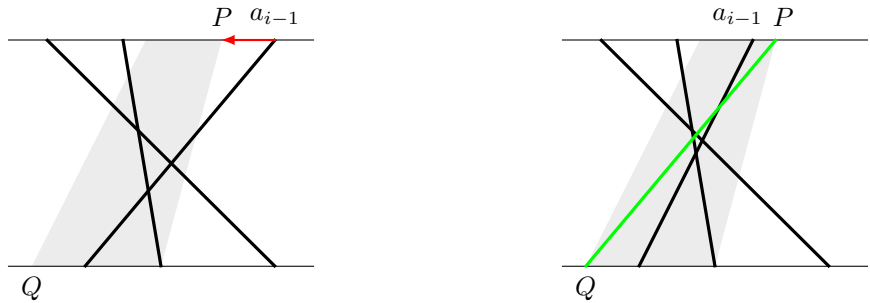
\begin{figure}[htp!]

\begin{tikzpicture}[scale=1,>=latex,shorten >=-0.4pt,shorten <=-0.4pt]

\draw[fill=gray!30, draw=none, opacity=0.5] (0.8,3) -- (1.8,3) -- (1,0) -- (-0.7,0) -- cycle;

\draw[] (-1, 3) -- (3, 3);
\draw[] (-1, 0) -- (3, 0);

\draw[very thick] (-0.5,3) -- (2.5, 0);
\draw[very thick] (2.5,3) -- (0, 0);
\draw[very thick] (0.5, 3) -- (1, 0);

\draw[draw=red, thick, ->] (2.5,3) -- (1.8,3);

\begin{footnotesize}
\node at (2.5, 3.3) {$a_{i-1}$};
\node at (1.8, 3.3) {$P$};
\node at (-0.7, -0.3) {$Q$};  
\end{footnotesize}

\end{tikzpicture}
\hspace{3cm}
\begin{tikzpicture}[scale=1,>=latex,shorten >=-0.4pt,shorten <=-0.4pt]

\draw[fill=gray!30, draw=none, opacity=0.5] (0.8,3) -- (1.8,3) -- (1,0) -- (-0.7,0) -- cycle;

\draw[] (-1, 3) -- (3, 3);
\draw[] (-1, 0) -- (3, 0);

\draw[very thick] (-0.5,3) -- (2.5, 0);

\draw[very thick] (0.5, 3) -- (1, 0);
\draw[very thick] (1.5, 3) -- (0, 0);
\draw[very thick, green] (1.8, 3) -- (-0.7, 0);

\begin{footnotesize}
\node at (1.3, 3.3) {$a_{i-1}$};
\node at (1.9, 3.3) {$P$};
\node at (-0.7, -0.3) {$Q$};  
\end{footnotesize}

\end{tikzpicture}
\caption{
\label{fig:insert}
Example of the induction step. The trapezoid $T_i$ is shaded. The endpoint $a_{i-1}$ is moved slightly to the left of $P$, which allows to insert a segment $s_i$ (depicted in green) inside $T_i$ with endpoints $P$ and $Q$.
}
\end{figure}
\end{center}

\end{proof}

\bibliographystyle{plain}
\bibliography{lit_short}

\end{document}